%% file: FLogit_2023-08-15.tex
\documentclass[11pt]{article}

\usepackage{graphicx}
\usepackage{color}
\usepackage{rotating}
\usepackage{amsmath}
\usepackage{amssymb}
\usepackage{amsthm}
\usepackage{float}
 \usepackage{url}
\usepackage{subfigure}

\usepackage{natbib} 

\usepackage{enumerate}
\usepackage{enumitem}

\usepackage{dsfont}
 
 \usepackage{algorithm}
\usepackage{algorithmic}
\usepackage{array,arydshln}

\newtheorem{theorem}{Theorem}[section]
\newtheorem{lemma}[theorem]{Lemma}
\newtheorem{corollary}[theorem]{Corollary}
\newtheorem{proposition}[theorem]{Proposition} 
\newtheorem{remark}{Remark}[section]

\usepackage{soul} %% PARA TACHAR CON \st
 
\newcommand\bb {\mathbf b}

\newcommand\bx {\mathbf x}

\newcommand\indica {\mathbb{I}}

\newcommand\bbG {\mathbb{G}}

\newcommand\wb {\widehat{{b}}}
\newcommand\wbb {\widehat{{\bb}}}

\newcommand\wpe {\widehat{p}}

\newcommand\wF {\widehat{F}}

\newcommand\wtb {\widetilde{b}}

\newcommand\wty {\widetilde{y}}

\newcommand\wtX {\widetilde{X}}

\newcommand\itA {{\mathcal{A}}}
\newcommand\itB {{\mathcal{B}}}
\newcommand\itC {{\mathcal{C}}}

\newcommand\itE {{\mathcal{E}}}
\newcommand\itF {{\mathcal{F}}}
\newcommand\itG {{\mathcal{G}}}
\newcommand\itH {{\mathcal{H}}}
\newcommand\itI {{\mathcal{I}}}

\newcommand\itM {{\mathcal{M}}}
\newcommand\itN {{\mathcal{N}}}

\newcommand\itT {{\mathcal{T}}}
\newcommand\itU {{\mathcal{U}}}
\newcommand\itV {{\mathcal{V}}}
\newcommand\itW {{\mathcal{W}}}
\newcommand\itX {{\mathcal{X}}}

\newcommand\bbe {\mbox{\boldmath $\beta$}}

\newcommand\bmu {\mbox{\boldmath $\mu$}}

\newcommand\bSi {\mbox{\boldmath $\Sigma$}}

\newcommand\walfa {\widehat{\alpha}}

\newcommand\wbeta {\widehat{\beta}}
\newcommand\wbbe {\widehat{\bbe}}

\newcommand\wbmu {\widehat{\bmu}}

\newcommand\wtheta {\widehat{\theta}}

\newcommand\wbSi {\widehat{\bSi}}

\newcommand\wtbeta {\widetilde{\beta}}

\newcommand\wtpi {\widetilde{\pi}}

\newcommand\wttheta {\widetilde{\theta}}

\newcommand\wtTheta {\widetilde{\Theta}}
\def\real{\mathbb{R}}
\def\natu{\mathbb{N}}
\def\qu{\mathbb{Q}}

\newcommand{\esp}{\mathbb{E}}
\newcommand{\prob}{\mathbb{P}}

\newcommand{\var}{\mbox{\sc Var}}

\newcommand{\convpp}{ \buildrel{a.s.}\over\longrightarrow}

\newcommand{\convprob  }{ \buildrel{p}\over\longrightarrow}

\newcommand{\trasp}{^{\mbox{\footnotesize \sc t}}}

\def\dst{\displaystyle}

\def\argmin{\mathop{\mbox{argmin}}}

\newcommand\noi{\noindent}
\def\dst{\displaystyle}

\def\square{\ifmmode\sqr\else{$\sqr$}\fi}
\def\sqr{\vcenter{
         \hrule height.1mm
         \hbox{\vrule width.1mm height2.2mm\kern2.18mm
\vrule width.1mm}
         \hrule height.1mm}}

\newcommand{\eme}{\mbox{\footnotesize \sc m}}

\newcommand{\clas}{\mbox{\scriptsize \sc cl}}

\newcommand{\wclBOX}{\mbox{\scriptsize \textsc{wcl-fbb}}}

\newcommand{\wclHR}{\mbox{\scriptsize \textsc{wcl-hr}}}

\newcommand{\wemeBOX}{\mbox{\scriptsize  \textsc{wm-fbb}}} 

\newcommand{\wemeHR}{\mbox{\scriptsize  \textsc{wm-hr}}}

\newcommand{\clasnorm}{\mbox{\sc cl}}
\newcommand{\wclBOXnorm}{\mbox{\textsc{wcl-fbb}}}

\newcommand{\wclHRnorm}{\mbox{\textsc{wcl-hr}}}

\newcommand{\wclnorm}{\mbox{\textsc{wcl}}}

\newcommand{\emenorm}{\mbox{\sc m}} 

\newcommand{\wemeBOXnorm}{\mbox{\textsc{wm-fbb}}} 

\newcommand{\wemeHRnorm}{\mbox{\textsc{wm-hr}}} 

\newcommand{\wemenorm}{\mbox{\textsc{wm}}}

\newcommand\trim {\mbox{\footnotesize  \sc tr}}

\newcommand{\out}{\mbox{\scriptsize \sc out}}

 \usepackage{xcolor}

\usepackage{hyperref} %PARA QUE AL PINCHAR VAYAS A LA ECUACION

\hypersetup{
    colorlinks,
    linkcolor= {blue!90!black}, %{red!80!black},
    citecolor={blue!60!black},
    urlcolor={blue!80!black}
}

\begin{document}

%%%%%%%%%%%%%%%%%% ABSTRACT%%%%%%%%%%%%%%%%%%%%%%%%%%%%%%%%

\title{Robust estimation for functional logistic regression models}
\author{Graciela Boente$^a$, Marina Valdora$^a$\\ 
  $^a$ Universidad de Buenos Aires and CONICET, Argentina  
}
\date{}

\maketitle

%%%%%%%%%%%%%%%%%% ABSTRACT%%%%%%%%%%%%%%%%%%%%%%%%%%%%%%%%
\begin{abstract}
 This paper addresses the problem of providing robust estimators under a functional logistic regression model.  Logistic regression  is a popular tool in classification problems with two populations. As in functional linear regression, regularization tools are needed to compute estimators for the functional slope. The traditional methods are based on dimension reduction or penalization combined with maximum likelihood or quasi--likelihood techniques and for that reason, they may be affected   by  misclassified points especially if they are associated to functional covariates with atypical behaviour. The proposal given in this paper  adapts some of the best practices used when the covariates are finite--dimensional to provide reliable estimations. Under regularity conditions, consistency of the resulting estimators and rates of convergence for the predictions are derived. A numerical study illustrates the finite sample performance of the proposed method and reveals its stability under different contamination scenarios.  A real data example is also presented.
\end{abstract}

\noi \textbf{Keywords:} $B$-splines; Functional Data Analysis; Logistic Regression Models;  Robust Estimation 

\noi\textbf{AMS Subject Classification:}  62F35; 62G25

\normalsize
\newpage
\section{Introduction}\label{sec:intro}
In many applications,  such as chemometrics,  image recognition and spectroscopy,  the observed data contain functional covariates, that is, variables originated by phenomena that are continuous in time or space and can be assumed to be smooth functions, rather than finite dimensional vectors.
 Functional data analysis aims to provide tools for analysing  such data  and has received considerable attention in recent years due to its high versatility and numerous applications. Different approaches either parametric, nonparametric or even semiparametric ones, were given to model data with functional predictors. Some well-known references in the treatment of functional data are the books of   \citet{ferraty:vieu:2006} and \citet{ferraty:romain:2010}, who carefully discuss non--parametric models, and also \citet{ramsay2005functional, ramsay2002applied}, \citet{horvath:kokoska:2012}  and \citet{hsing:eubank:2015} who place emphasis on parametric models such as the functional linear one. We also refer to  \citet{aneiros:etal:2017} and the reviews in \citet{cuevas:2014} and \citet{goia:vieu:2016} for some other  results in the area.  
 
Among regression models relating a scalar response with a functional covariate, the functional linear regression                                                                                                                                                                                                                                                                                                                                                                                                                model is one of the more popular ones. Several estimation aspects under this  model have been considered among others in \citet{cardot:etal:2003},
  \citet{cardot:sarda:2005}, \citet{cai:hall:2006}, \citet{hall:horowitz:2007}, see also  \citet{febrero:etal:2017} and \citet{reiss:etal:2017} for a review. Robust proposals for functional linear regression models using either $P-$, $B-$splines or functional principal components
were given in \citet{maronna:yohai:2013}, \citet{boente:salibian:vena:2020}, \citet{kalogridis:vanaelst:2023} and \citet{kalogridis:vanaelst:2019}, respectively.

Most of the papers mentioned above  consider the case where the response is a continuous variable. However, discrete responses arise in some practical problems such as in classification. In particular, when the interest relies on  the presence or absence of a condition of interest,   the response corresponds to a binary outcome.  \citet{reiss:etal:2017} includes a review on relevant papers treating the case of responses  whose conditional distribution belong to an exponential--family distributions, that is, estimation under a generalized functional linear regression model. As in functional linear models, the  naive approach to estimate the functional regression parameter considering as multivariate covariates  the values of the functional covariates observed  on a grid of points and ignoring their functional nature,  is not appropriate, since this approach leads to an ill-conditioned problem. The causes for this issue are, on the one hand,   the high correlation existing, within each trajectory, among observations  corresponding to close grid values and on the other hand,   the fact that the number of grid measurements may exceed the number of observations, see  \citet{marx:eilers:1999} and \citet{ramsay2004functional} for a discussion. 
As mentioned in \citet{wang:etal:2016}, one of the challenges in functional regression is the inverse nature of the problem, which causes estimation problems mainly generated by   the compactness of the covariance operator of $X$. For the reasons mentioned above, the extension from the  situation with   finite-dimensional predictors to the case of an infinite-dimensional one is not direct.    The usual practice  to solve this drawback is regularization which can be achieved in several ways, either  reducing the set of candidates  to a finite--dimensional one  or by adding a penalty term as when considering $P-$splines.
In particular, \citet{cardot:sarda:2005} proposed  estimators based on penalized likelihood and
spline approximations and derived consistency properties for them. A different approach was followed in \citet{muller:stadtmuller:2005}   where the estimators are obtained via  a truncated Karhunen-Lo\`eve expansion for the covariates. These authors also  provided a  theoretical study of the  properties of their estimators as well as an illustration of their proposal on a classification problem that analyzes the interplay of longevity and reproduction,   in short  and long--lived Mediterranean fruit flies.

Among generalized regression models, logistic regression is one of the best known and useful models in statistics. It has been extensively studied when euclidean covariates arise and several robust proposals were given in this setting, some of which will be mentioned below, in Section \ref{eq:logfinito}. The functional logistic regression model is a generalization of the finite--dimensional logistic regression model, which assumes that the observed covariates are functional data rather than vectors in $\real^p$. It is particularly relevant in discrimination problems for curve data. 
This model was already considered in  \citet{james:2002}   as a particular case of generalized linear models with functional predictors.   Beyond the procedures studied in the framework of generalized functional regression models which can be used in functional logistic regression models, some authors have considered specific estimation methods for this particular model. Among others, we can mention  \citet{escabias:etal:2004}, \citet{aguilera:etal:2008}, \citet{aguilera-morillo:etal:2013} and  \citet{Mousavi:2018} who provides a  revision and  comparison of different estimation methods for functional logistic regression.
  Among the numerous interesting applications of functional logistic regression that have been reported in the literature we can mention
 \citet{escabias:etal:2005}, who use it to model environmental data,
 \citet{ratcliffe:etal:2002}, who apply it to foetal heart rate data and \citet{reiss:etal:2005} who analyze pet imaging data. Besides, \citet{sorensen:etal:2013} present several medical applications of functional data, while \citet{wang:etal:2017} used  penalized Haar wavelet approach  for the classification of   brain images to assist in early diagnosis of Alzheimer's disease.

The framework in which this paper will focus corresponds to the functional logistic regression model with scalar response, also labelled as   scalar--on--function logistic regression in the literature. Under this model, the i.i.d. observations  $(y_i, X_i)$, $1\le i \le n$, are such that  the response $y_i\in \{0,1\}$ and the predictor $X_i\in L^2(\itI)$, with $\itI$ a compact interval, while the link function is the logistic one. More precisely, if  $Bi(1,p)$ stands for the Bernoulli distribution with success probability $p$ and  $F(t)=\exp(t)\left[ 1+\exp(t)\right]^{-1}$ denotes the logistic function, the functional logistic regression model assumes that 
\begin{equation}
y_i|X_i \sim Bi\left(1, p_i\right) \qquad \mbox{where} \qquad p_i=\esp(y_i|X_i)= F\left(\alpha_0+\langle X_i, \beta_0\rangle\right)\,,
\label{eq:FLOGIT}
\end{equation}  
with $\alpha_0\in\real$, $\beta_0\in L^2(\itI)$ and $\langle u, v\rangle =\int_{\itI} u(t) v(t) dt$ stands for the usual inner product in $L^2(\itI)$ and $\|\cdot\|$ for the corresponding norm.

 The estimators for functional logistic regression mentioned above are based on the method of maximum likelihood combined with some regularization tool. As it happens with finite dimensional covariates, these estimators are highly affected by the presence of  misclassified points specially when combined with high leverage covariates. 
 Robust methods, on the other hand, have the advantage of giving reliable results, even when a proportion of the data correspond to atypical data. As mentioned above, some robust methods for  functional linear regression models have been recently proposed and this area has shown great development in the last ten years. However, the literature on robust procedures for generalized functional linear models and specifically for   functional logistic regression ones is scarce. 
Up to our knowledge,  only few procedures have been considered and most of them lack  a careful study of the asymptotic properties of the proposal considered. The first attempt  to provide a robust method for functional logistic regression was given in  \citet{denhere:billor:2016}. This method is based on reducing the dimension of the covariates by using a robust principal components method proposed in \citet{Hubert:etal:2005}. The robustness of this method ensures that the functional principal component analysis is not influenced by outlying covariates. However, it  but does not take into account the problem of large deviance residuals  originated by incorrectly classified observations. To solve this problem, \citet{mutis:etal:2022} propose to combine a basis approximation with   weights computed using the Pearson residuals, in an approach related to that given by \citet{alin:agostinelli:2017} for finite--dimensional covariates. Recently, \citet{kalogridis:2023} introduced an approach based on divergence measures combined with penalizations and provided a careful study of its asymptotic properties,  for bounded covariates. 
 
 In this paper, we follow a different perspective, taking into account the sensitivity of these estimators to atypical observations and  based on the ideas given for euclidean covariates by  \citet{Bianco:yohai:1996} and \citet{Croux:H:2003}, we define robust estimators of the intercept $\alpha_0$ and the slope $\beta_0$ following a sieve approach combined with weighted $M-$estimators. More precisely, as done for instance in functional linear regression, we first reduce the set of candidates for estimating $\beta_0$  to those belonging to a finite--dimensional space spanned by a fixed basis  selected by the practitioner, such as the $B-$splines, Fourier, or wavelet bases. This enables to use the robust tools developed for finite--dimensional covariates in this infinite--dimensional framework. Clearly this regularization process involves the selection of the basis dimension which should increase with the sample size at a given rate and which must be chosen in a robust way. For that reason, in Section \ref{sec:BIC}, we   describe a resistant procedure to select the dimension of the approximating space.

The rest of the paper is organized as follows. The model and our proposed
estimators are described in Section \ref{sec:estimadores}. Theoretical
assurances regarding   consistency and convergence rates of our proposal are
provided in Section \ref{sec:consis}, while in  Section \ref{sec:simu} we
report the results of a simulation study to explore their finite-sample
properties. Section \ref{sec:ejemplos} contains the analysis of a real-data set, while 
final comments are given in Section \ref{sec:concl}. All
proofs are relegated to the Appendix.

\section{The estimators} \label{sec:estimadores}
As mentioned in the Introduction, our proposal for estimators under the functional logistic regression model \eqref{eq:FLOGIT}  is based on basis reduction. For that reason and for the sake of completeness,   in Section \ref{eq:logfinito}, we recall  some of the robust proposals given when the covariates belong to a finite--dimensional space.

\subsection{Some robust proposals for euclidean covariates}{\label{eq:logfinito}}
When the covariates are finite--dimensional, the practitioner  deals with   i.i.d. observations   $\left(y_i, \bx_i \right)$, $1\le i \le n$, where $\bx_i \in \real^p$, $y_i\in \{0,1\}$. In this case,   the well--known logistic regression model states   that $y_i|\bx_i \sim Bi(1, F(\alpha_0+\bx_i \trasp \bbe_0))$, where 
%$Bi(1,p)$ stands for the Bernoulli distribution with success probability $p$,  $F(t)=\exp(t)\left[ 1+\exp(t)\right]^{-1}$ is the logistic function and 
 $\bbe_0 \in \real^p$.
 As mentioned in the Introduction,    the maximum likelihood estimator   of the regression coefficients is very sensitive to outliers, meaning that we cannot accurately classify a new observation based on these estimators, neither identify those covariates with important information for assignation. To solve this drawback, different robust procedures have been considered. 
 
In particular,  consistent $M-$estimators bounding the deviance  were  defined in  \citet{Bianco:yohai:1996}, while in order to obtain bounded influence estimators a weighted  version was introduced in \citet{Croux:H:2003}.  For the family of $M-$estimators defined in \citet{Bianco:yohai:1996},  \citet{Croux:H:2003}  introduced a loss function that guarantees the existence of the resulting robust estimator  when the maximum likelihood estimators do  exist.    \citet{Basu:etal:1998}   considered a proposal based on minimum divergence. However, their approach can also be seen as a particular case of the \citet{Bianco:yohai:1996} estimator with a properly defined loss function. Other approaches were given in   \citet{Cantoni:ronchetti:2001} who consider a robust quasi--likelihood estimator, \citet{Bondell05, Bondell08} whose procedures incorporate a minimum distance perspective and \citet{Hobza:etal:2008} who defines a a median estimator by using an $L^1-$estimator of the smoothed responses. 
 
As pointed out in  \citet{maronna:etal:2019libro}, the use of redescending weighted $M-$estimators ensures estimators with good robustness properties. For that reason and taking into account that our proposal will combine dimension reduction with   weighted $M-$estimators,  we briefly describe the proposal introduced in \citet{Croux:H:2003}. From now on, denote  $d(y,t)$   the squared deviance function, that is, $d(y,t) =  - \log(F(t))  y - \log(1 - F(t)) (1-y)$ and let $\rho:\real_{\ge 0}\to \real$ be  a bounded, differentiable and  nondecreasing function with derivative $\psi = \rho^{\prime}$. 
Furthermore, define
\begin{eqnarray}
\phi(y, t)&=& \rho(d(y, t)) + G(F(t)) + G(1 - F(t)) \,,
\label{eq:phiBY}
\end{eqnarray}
where   $G(t) = \int_0^t \psi(-\log u) \, du$. The correction term $G(F(t)) + G(1 - F(t))$ was introduced in \citet{Bianco:yohai:1996}  to guarantee Fisher--consistency of the resulting procedure. It is worth mentioning that the function $\phi(y, t)$ can be written as 
\begin{equation}\label{eq:phiBY*}
\phi(y, t)=y\rho\left(\,-\,\log\left[F(t)\right]\right)+ G(F(t)) +(1-y)\rho\left(\,-\,\log\left[1-F(t)\right]\right)  + G(1 - F(t))\,.
\end{equation}
The weighted $M-$estimators are  the minimizers of $ 
 L_n^{\star}(a,\bb) =  \sum_{i = 1}^n \phi(y_i, a + \bx_i\trasp \bb) w(\bx_i)/n$, that is,
\begin{equation}
\label{eq:WBY}
(\walfa,\wbbe )=  \argmin_{a\in \real, \bb\in \real^p} L_n^{\star}(a, \bb)\,.
\end{equation} 
 The weights  $ w(\bx_i)$ are usually  based  on a robust Mahalanobis distance of the explanatory variables, that is, they depend on the distance between $\bx_i$    and a robust center of the data.  
With this notation, the minimum divergence estimators considered in \citet{Basu:etal:1998} correspond to the choice $\rho(t)=(1+1/c)(1-\exp(-\,c\,t)) $ and $w\equiv 1$.  

The asymptotic properties of the  $M-$estimators with $w\equiv 1$  were obtained in  \citet{Bianco:yohai:1996}, while the situation of a general weight function was studied in  \citet{Bianco:Martinez:2009}. It should be mentioned that these estimators are implemented in the package \texttt{RobStatTM}, through the functions \texttt{logregBY}, when the weights equal 1, and \texttt{logregWBY} when considering hard rejection weights   derived from the MCD estimator of the continuous explanatory variables. In both cases, the loss function is taken as the one introduced in \citet{Croux:H:2003}.

\subsection{The case of functional covariates}
In this section, we consider the situation where  $(y_i, X_i)$, $1\le i \le n$ are independent observations such that $y_i\in  \{0,1\}$ and $X_i\in L^2(\itT)$  with $\itT$ a compact interval, that, without loss of generality, we assume to be $\itT=[0,1]$. The model relating the responses to the covariates is  the functional   logistic regression model, that is, we assume  that \eqref{eq:FLOGIT} holds. 
 
As mentioned in the Introduction,  estimation under  functional linear regression or  functional logistic regression models is an ill--posed problem. To avoid this issue, one possibility is dimension reduction that   can be achieved by considering as possible candidates  for estimating $\beta_0$ the elements of   a finite--dimensional space spanned by a fixed basis. This is the approach we follow in this paper, that is, to define robust estimators of the intercept $\alpha_0$ and the slope $\beta_0$, we will use a sieve approach combined with weighted $M-$estimators. We do not restrict our attention to a particular basis as the $B-$spline basis considered, for instance, in  \citet{boente:salibian:vena:2020} for the functional semi--linear model or \citet{mutis:etal:2022} for the functional logistic regression one. Instead, we provide a general framework which allows the practitioner to choose the basis according to the smoothness knowledge or assumptions to be considered on $\beta_0$. 

Henceforth,  let $k=k_n$ stand for the dimension of the finite dimensional space spanned by the basis $\{B_j: 1\leq j\leq k_n\}$. The space of possible candidates correspond to $\real \times \itM_k$, where
$\itM_k=\left\{\sum_{j=1}^{k} b_{j} B_{j}, \bb \in \real^k\right\}$. 

From now on, for any $\bb \in \real^k$,    $\beta_{\bb}$ will stand for  $\beta_{\bb}=\sum_{j=1}^{k} b_{j} B_{j}$.
Then, for any possible candidate $\beta_{\bb} \in \itM_k$, the inner product $ \langle X_i,\beta_{\bb}  \rangle$ equals $ \sum_{j=1}^{k} b_{j} x_{ij}$ where  $x_{ij}=\langle X_i, B_j\rangle$ which suggests to use the robust estimators defined in Section \ref{eq:logfinito} taking as covariates $\bx_i=(x_{i1}, \dots, x_{i k})\trasp $.

More precisely,  the weighted estimators defined in \eqref{eq:WBY} over the finite--dimensional approximating spaces $\real \times \itM_k$ are the key tool for obtaining consistent estimators of $\beta_0$. For   any $\bb \in \real^k$  define
$$L_n(\alpha, \beta_{\bb})= \frac{1}{n} \sum_{i = 1}^n \phi(y_i, \alpha +\langle X_i, \beta_{\bb}\rangle) w(X_i)=\frac{1}{n} \sum_{i = 1}^n \phi(y_i, \alpha + \bx_i\trasp \bb) w(X_i)$$
 and 
 \begin{equation}
\label{eq:FWBY}
(\walfa,\wbb )=     \argmin_{\alpha\in \real, \bb\in \real^{k}} L_n(\alpha, \beta_{\bb}) \,.
\end{equation} 
Hence, the estimator of $\beta_0$ is given by 
$$\wbeta(t)=\beta_{\wbb}(t)=\sum_{j=1}^{k} \wb_j B_j(t)\,,$$ 
where $\wbb=\left(\wb_1, \ldots, \wb_{k}\right)\trasp$, meaning that 
$(\walfa,\wbeta )=     \argmin_{\alpha\in \real, \beta\in \itM_k} L_n(\alpha, \beta)$.   

The weights $w(X_i)$ in \eqref{eq:FWBY} may be computed as in \eqref{eq:WBY} using a weight function of the  robust Mahalanobis distance of the  projected variables $\bx_i$, in which case, $L_n(\alpha, \beta_{\bb})= L_n^{\star}(a,\bb)$. Another possibility is to compute the weights from the functional covariates, for instance, discarding observations which are declared as outliers by the functional boxplot or any other functional measure of atipicity.  In the simulation study reported in Section \ref{sec:simu}, we explore both possible choices for the weights. As in the finite--dimensional setting, for the sake of simplicity, when deriving consistency results, we will assume that the weight  function $w$ is not  data dependent.

\subsection{On the basis choice}{\label{sec:BASE}}
The basis choice depends on the knowledge or assumptions to be made on the slope parameter.
Some well known basis are $B-$splines, Bernstein and Legendre polynomials, Fourier basis or Wavelet ones. They vary in the way they approximate a function  as discussed, for instance, in \citet{boente:martinez:2023} and \citet{kalogridis:vanaelst:2023}.

When considering $B$-spline approximations, consistency results will require   the slope parameter to be   $r$-times continuously differentiable, that	is, $ \beta_0  \in C^r([0,1])$, where $r\le \ell -2$ and $\ell$ is the spline order. In particular, when cubic splines are considered, the results in 
Section  \ref{sec:consis} hold for twice continuously differentiable regression functions. Recall that a spline of order $\ell$ is    a polynomial of degree $\ell - 1$ within each  subinterval defined by the knots. As stated in Corollary 6.21 in \citet{schumaker1981spline}, if $\beta_0 \in C^r([0,1])$ with $r-$th derivative Lipschitz and  $r\le \ell -2$, under proper assumptions on the knots,  there exist a spline of order $\ell$, let's say $\wtbeta=\sum_{j=1}^k b_j B_j$,   such that $\|\wtbeta- \beta_0\|_\infty=O(k^{-{(r+1)}})$. It is worth mentioning that the approximation order has an impact on the rates of convergence derived in Theorems  \ref{teo:RATES1} and   \ref{teo:RATES} through assumption \ref{ass:approxorder}.

Bernstein polynomials are a possible alternative to $B-$splines. They are defined as 
$$B_j(t)= \binom{k}{j} t^{j} (1-t)^{k-j}\quad \mbox{for } j=0,\dots, k\,.$$ 
Weierstrass  Theorem ensures that if $\beta_0\in C([0,1])$, there exists $\wtbeta=\sum_{j=1}^k b_j B_j$, where $b_j=\beta_0(j/k)$ such that $\|\wtbeta- \beta_0\|_\infty \to 0$. Furthermore,   Theorem 3.2 in \citet{Powell:1981} guarantees  that when considering Bernstein polynomials  of order $k$  we also get that $\|\wtbeta- \beta_0\|_\infty=O(k^{-r})$, whenever $\beta_0 \in C^r([0,1])$.  

Legendre polynomials define an  orthogonal basis in $L^2([0,1])$. As mentioned in \citet{boente:martinez:2023},  the  convergence rates derived in Theorem 2.5 from \citet{Wang:Xiang:2012} allow to  show that, if $\beta_0 \in C^r([0,1])$ and $\wtbeta$ stands for the truncated Legendre series
expansion of $\beta_0$ of order $k$, then   $\|\wtbeta- \beta_0\|_\infty=O\left(k^{- \, r+1/2}\right)$. Note that in this case, the approximation rate is lower than for the other two basis mentioned above and this will affect the  rates provided in Theorem \ref{teo:RATES}. More generally, when polynomials basis of order $k$ are considered,   Jackson's Theorem (see Theorem 3.12 in \citealp*{schumaker1981spline}) ensures that if $\beta_0 \in W^{r,2}([0,1])$, the $L^2-$Sobolev space of order $r$ as defined below, there exists a polynomial $\wtbeta$ of order $k\ge r$ such that  $\|\wtbeta- \beta_0\|=O(k^{-r})$  and this improved $L^2$ approximation order is enough to guarantee a better rate of convergence for the predictions, when polynomial bases are considered.

Finally, the Fourier basis  is the natural basis in $L^2([0,1])$ and it is usually considered when approximating periodic functions. Clearly, the finite expansion  $\wtbeta_0(t)=b_0+ \sum_{j=1}^k b_{j,1} \sin(2\pi\,j\,t) + b_{j,2} \cos(2\pi\,j\,t)$, where $b_0= \int_0^1 \beta_0(t)  \, dt$, $b_{j,1}=\int_0^1 \beta_0(t) \sin(2\pi\,j\,t)\, dt$ and $b_{j,2}=\int_0^1 \beta_0(t) \cos(2\pi\,j\,t)\, dt$ converges to $\beta_0$ in $L^2([0,1])$. When   $\beta_0 \in W^{r,2}([0,1])$, Corollary 2.4 in Chapter 7 from \citet{devore:lorentz:1993} ensures that $\|\wtbeta- \beta_0\|=O(k^{-r})$ 

Wavelet basis may be useful when seeking for sparse slope function estimates. \citet{zhao:etal:2012} provide conditions ensuring $L^2$ approximations of    order $O(k^{-r})$,  when $\beta_0 \in W^{r,2}([0,1])$ and one--dimensional Wavelets are used, see also \citet{mallat:2009}.

\subsection{Selecting the size of the  basis}{\label{sec:BIC}}
 The number of elements of the  basis plays the role of   regularization parameter  in our estimation procedure. 
 The importance of considering a robust criterion to select the regularization parameter has been discussed by several authors who report how  standard  model selection methods can be highly affected by a small proportion of  outliers. The sensitivity to atypical data of classical basis selectors may be inherited by the final regression estimators even when robust procedure is considered.   

To deal with these problems, when the covariates belong to $\real^p$, \citet{ronchetti:1985}  and \citet{thar:claes:2013}  provide some robust approaches when considering  linear regression models.  Besides,    under a sparse logistic regression model, \citet{Bianco:Boente:Chebi:2022} report in their supplement a numerical study that reveals   the importance of considering a robust criterion  in order to achieve reliable predictions. 
  Finally, for functional covariates and under a semi--linear and a  functional linear model, \citet{boente:salibian:vena:2020}, \citet{kalogridis:vanaelst:2019,kalogridis:vanaelst:2023}, respectively, discuss robust criteria for selecting the regularization parameters.
 
In our framework,  the basis dimension $k=k_n$   may be determined by a
  model selection criterion  such as a robust version of the Akaike criterion used in \citet{lu:2015}  or the robust \citet{schwartz:1978}  criterion considered in \citet{he:shi:1996} and \citet{he:zhu:fung:2002} for semi--parametric regression models. Suppose that $  (\walfa^{(k)},\wbb^{(k)})$ is the  solution of \eqref{eq:FWBY} when we use a $k-$dimensional linear space and denote as $\wbeta^{(k)}=\beta_{\wbb^{(k)}}$.  We define  a robust $BIC$ criterion, whose large values indicate a poor fit,  as
\begin{equation}
RBIC(k) =  L_n(\walfa^{(k)}, \wbeta^{(k)})  
+ k\,  \frac{\log n}{n}\,.
  \label{eq:bic1}
  \end{equation}
For instance, when considering  $B-$spline procedures, in order to obtain an  optimal rate of convergence, we let 
the number of knots increase slowly with the sample size. 
 Theorem \ref{teo:RATES} below  shows  that when $\beta_0$ is twice continuously differentiable and is
approximated with cubic splines ($\ell=4$),  the  size $k_n$ of the bases can be taken of order $n^{1/5}$.  Hence, a possible way to select $k_n$ is 
to search for the first local
minimum of $RBIC(k_n)$ in the range $\max(n^{1/5}/2, 4) \le k\le 8 + 2\, n^{1/5}$.
Note that for cubic splines the smallest possible number of knots is 4. 

 \section{Consistency results}{\label{sec:consis}}
To provide a unified approach in which the basis gives approximations either in $L^2(0,1)$ or in $C([0,1])$, equipped by their respective norms $\|\cdot\|$ and $\|\cdot\|_{\infty}$, we will denote $\itH$ the space $L^2([0,1])$ or  $C([0,1])$ and $\|\cdot\|_{\itH}$ the corresponding norm. Hence, we have that $\|f\|\le \|f\|_{\itH}$. Furthermore, $\itW^{r,\itH}$ will stand for the H\"older space
$$
\itW^{r,\infty}([0,1]) \, = \, \biggl\{ f\in C^r \left([0, 1] \right):
  \big \|f^{(j)} \big \|_\infty<\infty, \;0\leq j\leq r,  \mbox{ and } 
  \sup_{z_1\ne z_2} \frac{\big |f^{(r)}(z_1) - f^{(r)}(z_2) \big|}{|z_1- z_2|}
  <\infty \biggr\} \, ,
$$
when $\itH=C([0,1])$, while when  $\itH=L^2([0,1])$, we label  $\itW^{r,\itH}$  the Sobolev space
\begin{align*}
\itW^{r,2}([0,1])=& \{f\in L^2([0,1]):    \mbox{the weak derivatives of $f$ up to order $r$ exist  and } \\
&  \|f^{(j)}\|^2= \int_0^1 \left\{f^{(j)}(t)\right\}^2\, dt<\infty \mbox{ for } j=1\dots, r\}\,.
  \end{align*} 
We denote  $\|\cdot\|_{\itW^{r,\itH}}$ the corresponding norm, that is,  
$$\|f\|_{\itW^{r,\itH}} = \max_{1\le j\le r} \big \|f^{(j)} \big \|_{\infty}+ \sup_{x \neq y,
  x,y \in (0,1)} \frac{\big |f^{(r)}(x)-f^{(r)}(y) \big|}{|x-y|} \,$$ 
  in the former case and $\|f\|_{\itW^{r,\itH}}^2=\|f\|^2+ \sum_{j=1}^r \|f^{(j)}\|^2$ in the latter one.

To derive consistency results, we will need the following assumptions. 

 \begin{enumerate}[label=\textbf{A\arabic*}]
\setcounter{enumi}{0}
 \item \label{ass:rho_bounded_derivable} $\rho: \real_{\geq 0} \to \real$ is a bounded, continuously differentiable function with   bounded derivative $\psi$ and $\rho(0) = 0$.
 
\item \label{ass:rho_derivative_positive}$\psi(t) \geq 0$ and there exists some $a \geq \log 2$ such that $\psi(t) > 0$ for all $0 < t < a$.
 
\item \label{ass:psi_strictly_positive} $\psi(t) \geq 0$ and there exist values $a \geq \log 2$ and $A_0 > 0$ such that $\psi(t) > A_0$ for every $0 < t < a$.
 \item\label{ass:psidif}  $\psi$ is continuously differentiable function with   bounded derivative $\psi^{\prime}$.

\item \label{ass:funcionw} $w$ is a non--negative bounded function with support $\itC_w$ such that $\prob(X\in \itC_w)>0$.  Without loss of generality, we assume that $\|w\|_{\infty}=1$.

\item 
\begin{enumerate}
\item \label{ass:wx}  $\esp w(X)\,\|X\|<\infty$.
\item \label{ass:wx2}  $\esp w(X)\,\|X\|^2<\infty$.
 \end{enumerate}
  
\item\label{ass:beta0}   The basis functions are such that   $B_j \in \itW^{1,\itH}$.

 \item\label{ass:dimbasis}  The basis dimension $k_n$ is such that $k_n\to \infty$, $k_n/n\to 0$.  

 \item\label{ass:approx}  There  exists an element $\wtbeta_{k}\in \itM_k$,  $\wtbeta_{k}=\sum_{j=1}^{k} \wtb_j  B_j(x)$ such that  $\|\wtbeta_{k}-\beta_{0}\|_{\itH}\to 0$ as $k\to \infty$. 
 
  \item\label{ass:approxorder}  There  exists an element $\wtbeta_{k}\in \itM_k$,  $\wtbeta_{k}=\sum_{j=1}^{k} \wtb_j  B_j(x)$ such that  $\|\wtbeta_{k}-\beta_{0}\|_{\itH}=O(k^{-r})$, for some $r>0$. Furthermore, the basis dimension $k_n$ is of order $O(n^{\varsigma})$ where $ \varsigma< 1/r$.  

 \item \label{ass:probaX*}   The following condition holds:
 \begin{equation}{\label{eq:probaX*}}
 \prob(\langle X,\beta\rangle+ \alpha =0)=0 \,, \mbox{ for any $\beta\in    \itH^{\star}$, $\alpha \in \real$,  such that $(\beta,\alpha)\ne 0$}\,,
 \end{equation}
where $\itH^{\star}={\itH}$ or $\itH^{\star}=\itW^{1,\itH}$,  depending on whether $\beta_0$ belongs to $ {\itH}$ or $\itW^{1,\itH}$, respectively.

\end{enumerate}
 
Our first result states that, under mild assumptions, the estimators obtained minimizing $L_n(\alpha, \beta)$ over $(\alpha, \beta)\in \real\times \itM_k$ produce consistent estimators of the conditional success probability with respect to the  weighted mean square error of the   differences between predicted probabilities defined as
$$\pi_\prob^2(\theta_1 ,\theta_2 )= \esp\left(w(X) \left[F(\alpha_1+\langle X,  \beta_1 \rangle)- F(\alpha_2+\langle X,  \beta_2 \rangle )\right]^2\right)\,,$$ 
where for $j=1,2$, $\theta_j=(\alpha_j, \beta_j)\in \Theta=\real\times \itH$.  Henceforth, to simplify the notation, we denote $\wtheta=(\walfa,  \wbeta)$ and $\theta_0=(\alpha_0,\beta_0)$.

\begin{theorem}\label{teo:RATES1} 
Let $\rho$ be a function satisfying  \ref{ass:rho_bounded_derivable} and  \ref{ass:psi_strictly_positive} and $w$ a  weight function satisfying \ref{ass:funcionw}.    
\begin{enumerate}
\item[(a)] Under \ref{ass:dimbasis} and \ref{ass:approx}, $\pi_\prob^2(\wtheta ,\theta_0 ) \convpp 0$. 
\item[(b)] If $w$  satisfies \ref{ass:wx}, we have that $ \pi_\prob^2(\wtheta ,\theta_0 )=O_\prob(\sqrt{k_n/n}+\|\wtbeta_{k}-\beta_{0}\|_{\itH}  )$. Moreover, if     \ref{ass:approxorder} also holds, $ \pi_\prob (\wtheta ,\theta_0 )=O_\prob(n^{- {\lambda}} )$, where   $\lambda = \min( \varsigma\, r/2, (1-\varsigma)/4)$.  
\item[(c)] If $w$  satisfies \ref{ass:wx2} and $\psi$ satisfies \ref{ass:psidif}, we have that $ \pi_\prob^2(\wtheta ,\theta_0 )=O_\prob(\sqrt{k_n/n}+\|\wtbeta_{k}-\beta_{0}\|_{\itH}^2  )$. Moreover, if     \ref{ass:approxorder} also holds, $ \pi_\prob (\wtheta ,\theta_0 )=O_\prob(n^{ -{\omega} } \;)$, where   $\omega= \min(  \,\varsigma\, r, (1-\varsigma)/4)$.  
\end{enumerate}
 \end{theorem}
 
\begin{remark}{\label{remark:consistencia}}
Denote  $\wpe (X)=F(\walfa+\langle X,  \wbeta \rangle )$ and $p_0(X)=F(\alpha_0+\langle X,  \beta_0 \rangle )$. When $w(X)\equiv 1$,  Theorem \ref{teo:RATES1}(a) implies that   for any $\epsilon>0$,  
$$\prob\left(\left|\wpe(X)-p_0(X)\right|>\epsilon\right)=\esp\left(\prob\left(\left|\wpe(X)-p_0(X)\right|>\epsilon\Big|_{(y_1,X_1),\dots, (y_n,X_n)}\right)\right)\le \esp\left(\frac{\pi_\prob^2(\wtheta ,\theta_0 )}{\epsilon^2}\right)\,,$$
where the right hand side of the inequality converges to 0  as $ n\to \infty$.Thus, $\wpe (X)\convprob p_0(X)$ allowing to consistently classify a new observation. Moreover, using that $F^{-1}$ is continuous, we also conclude that $\walfa+\langle X,  \wbeta \rangle  $   converges in probability to $  \alpha_0+\langle X,  \beta_0 \rangle$.
  However, the infinite--dimensional structure of the covariates does not allow  to derive the consistency of $\wbeta$, which is instead obtained in Theorem \ref{teo:CONSIST}.  The rates obtained in Theorem \ref{teo:RATES1} (b) provide a preliminary rate  that will be improved in Theorem \ref{teo:RATES}. 
\end{remark}

Theorem \ref{teo:CONSIST}  establishes strong consistency of the intercept and slope parameter, which clearly implies that of the predicted probability, that is, $F(\walfa+\langle X,  \wbeta \rangle )\convpp F(\alpha_0+\langle X,  \beta_0 \rangle )$. This result provides an improvement over the one  obtained in Theorem \ref{teo:RATES1}, but requires additional assumptions on the covariates, namely, assumption \ref{ass:probaX*} which is discussed in Remark \ref{remark:comentarios}.

\begin{theorem}\label{teo:CONSIST} 
Let $\rho$ be a function satisfying  \ref{ass:rho_bounded_derivable} and  \ref{ass:rho_derivative_positive}, and $w$ a weight function satisfying \ref{ass:funcionw}. Assume that \ref{ass:beta0} to \ref{ass:approx} hold. 
 \begin{enumerate}
 \item[(a)]If in addition \ref{ass:probaX*}  holds, we have that $|\walfa- \alpha_0|+ \| \wbeta - \beta_0\|_{\itH}\convpp 0$.
 \item[(b)] Assume that $\itH=L^2([0,1])$, that  the basis functions are such that  $B_j \in \itW^{1,2}$, $1\le j\le k_n$,  providing approximations in \ref{ass:approx} in $L^2([0,1])$ and that \ref{ass:probaX*} holds with $\itH^{\star}=    \itW^{1,2}$, i.e., $\beta_0\in \itW^{1,2}$.
 \item[] If, in addition,   $\beta_0\in C([0,1])$ and the basis elements   are also continuous, i.e., $B_j \in C([0,1])$, $1\le j\le k_n$,    then $|\walfa- \alpha_0|+ \| \wbeta - \beta_0\|_{\infty}\convpp 0$.
 \end{enumerate}
\end{theorem}

Theorem \ref{teo:CONSIST}(b) is useful for situation where the slope parameter is continuous but we use a smooth basis that provides an approximation in $L^2([0,1])$, such as the Fourier one.

\begin{remark}[\textbf{Comments on assumptions}]{\label{remark:comentarios}}
Assumptions \ref{ass:rho_bounded_derivable}, \ref{ass:rho_derivative_positive}, \ref{ass:funcionw} and \ref{ass:probaX*} are needed to ensure Fisher--consistency of the proposal, see Lemma \ref{lema:FC} in the Appendix. In the finite--dimensional case, \ref{ass:rho_bounded_derivable} and \ref{ass:rho_derivative_positive} were also required in \citet{Bianco:yohai:1996} who considered $M-$estimators, while assumption \ref{ass:funcionw}  corresponds to assumption A2 in \citet{Bianco:Martinez:2009} who studied the asymptotic behaviour of weighted $M-$estimators. Assumption \ref{ass:probaX*} is the infinite--dimensional counterpart of assumptions C1 in  \citet{Bianco:yohai:1996} and A3 in \citet{Bianco:Martinez:2009}. For the functional logistic regression model considered this assumption is stronger than the one required for functional linear regression models in \citet{boente:salibian:vena:2020} and \citet{kalogridis:vanaelst:2023} which states that $\prob(\langle X,\beta\rangle+ \alpha =0)<c<1$. However, it is weaker than assumption C2 in  \citet{kalogridis:vanaelst:2019}  who defined robust estimators based on principal components under a functional linear model and assumed that the process $X$ has a finite--dimensional Karhunen-Lo\`eve expansion with scores having a joint density function. It is worth mentioning that condition \ref{ass:probaX*} is related to the fact that the slope parameter is not identifiable if the kernel of the covariance operator of $X$ does not reduce to $\{0\}$. Instead of requiring the condition over all possible elements $\beta\in L^2([0,1])$, depending on the smoothness of $\beta_0$, the set of values $\beta\ne 0$ over which the probability $\prob(\langle X,\beta\rangle+ \alpha =0)$ equals 0 may be reduced.

Furthermore, assumptions   \ref{ass:rho_bounded_derivable} and \ref{ass:rho_derivative_positive} hold for the loss function introduced in \citet{Croux:H:2003}  and for $\rho(t)=(1+1/c)(1-\exp(-\,c\,t)) $ which is related to the minimum divergence estimators defined in \citet{Basu:etal:1998} and in both cases, $a$ can be taken as $+\infty$. 
Note that if $\psi(0)\ne 0$ and assumptions \ref{ass:rho_bounded_derivable} and  \ref{ass:rho_derivative_positive} hold for some constant $ a>\log(2)$, then  condition \ref{ass:psi_strictly_positive} is fulfilled. This situation arises, for example, for the two loss functions mentioned above. It is worth mentioning that  \ref{ass:psi_strictly_positive} is a key point to derive that   $L(\theta)-L(\theta_0)\ge C_0\,\pi_\prob^2(\theta,\theta_0)$, for any $\theta$ and  some constant $C_0>0$, where  $L(\theta)= L(\alpha, \beta)=\esp\left(\phi\left(y, \alpha + \langle X, \beta\rangle \right) w(X) \right)$. This inequality  allows to  derive  convergence rates for the weighted mean square error of the
prediction differences  from those obtained for the empirical process $\sup_{\theta \in \real \times \itM_k} |L_n(\theta)-L(\theta)|$.

Assumption \ref{ass:dimbasis} gives a rate at which the dimension of the finite--dimensional space $\itM_k$ should increase. It is a standard condition when a sieve approach is considered. Furthermore, in assumption \ref{ass:approxorder} a stronger convergence rate is required to the basis dimension in order  obtain rates of convergence.  

Assumption \ref{ass:approx} states that the true slope may be approximated by an
element of $\itM_k$. Conditions under which this assumption holds for some basis choices were discussed in Section \ref{sec:BASE}, where conditions ensuring a given rate for this approximation were also given. The approximation rate required in \ref{ass:approxorder} plays a role  when deriving rates of convergence for the   predicted probabilities. Note also that under \ref{ass:beta0}, the approximating element   $\wtbeta_{k}\in \itM_k$, given in  \ref{ass:approx} and \ref{ass:approxorder},  also belongs to $\itW^{1,\itH}$.  A first attempt to obtain these rates is given in   Theorem  \ref{teo:RATES1}, but better ones will be obtained in Theorem  \ref{teo:RATES} below. 
 \end{remark}

\subsection{Rates of Consistency}  {\label{sec:tasas}}              
  To derive rates of convergence for the estimators, we define  the pseudo-distance given by
$$\wtpi_\prob^2(\theta_1 ,\theta_2 )= \esp\left(w(X) \left[\alpha_1 -\alpha_2+\langle X, (\beta_1-\beta_2)\rangle\right]^2\right)\,,$$ 
where for $j=1,2$, $\theta_j=(\alpha_j, \beta_j)\in \Theta=\real\times \itH$.
The following additional assumption  will be required
 \begin{enumerate}[resume,label=\textbf{A\arabic*}]
  \item\label{ass:cotainf}: There exists $\epsilon_0>0$ and a positive constant $C_0^{\star}$, such that for any $\theta=(\alpha,\beta) \in \real\times  \itH$ with $|\alpha-\alpha_0|+\|\beta-\beta_0\|<\epsilon_0$ we have  $L(\theta)-L(\theta_0)\ge C_0^{\star}\,\wtpi_\prob^2(\theta,\theta_0)$.
 \end{enumerate}
  
  Note that since $F^{\prime}(t)=F(t)\left(1-F(t)\right)$ is bounded by 1, $\pi_\prob^2(\theta_1 ,\theta_2 )\le \wtpi_\prob^2(\theta_1 ,\theta_2 )$, so the 
weighted mean square error of the differences between predicted probabilities inherits the rates of converges obtained in Theorem \ref{teo:RATES} for the distance $\wtpi_\prob$.

\begin{theorem}\label{teo:RATES} 
Let $\rho$ be a function satisfying  \ref{ass:rho_bounded_derivable}, \ref{ass:psi_strictly_positive} and  \ref{ass:psidif}, and $w$ a weight function satisfying \ref{ass:funcionw} and \ref{ass:wx2}. Assume that \ref{ass:beta0} and \ref{ass:approxorder} to \ref{ass:cotainf} hold.  Then,  $\gamma_n\wtpi_\prob(\wtheta ,\theta_0 )=O_\prob(1)$, whenever $\gamma_n=O(n^{r\,\varsigma})$ and $\gamma_n \sqrt{\log(\gamma_n)} = O(n^{(1-\varsigma)/2})$.
 \end{theorem}
 
 \vskip0.1in

The lower bound given in assumption  \ref{ass:cotainf} is a requirement  that is fulfilled when the covariates are bounded as shown in Proposition \ref{prop:cotainf} below. Moreover, one consequence of Proposition \ref{prop:cotainf}  is that for bounded functional covariates, the pseudo-distances $\wtpi_\prob$ and $\pi_\prob$ are equivalent.

\vskip0.1in

 \begin{proposition}\label{prop:cotainf}
 Assume that  assumptions  \ref{ass:rho_bounded_derivable}   and \ref{ass:psi_strictly_positive} hold and that for some positive constant $ C>0 $, $\prob(\|X\|\le C)=1$. Then, there exists a constant $C_1>0$  such that  $\pi_\prob^2(\theta,\theta_0)\ge C_1\,\wtpi_\prob^2(\theta,\theta_0)$, for any $\theta=(\alpha,\beta) \in \real\times  \itH$ with $|\alpha-\alpha_0|+\|\beta-\beta_0\|<1$. Moreover, \ref{ass:cotainf} holds. 
 \end{proposition}
 
 \vskip0.1in
 
 As a consequence of Proposition \ref{prop:cotainf} and Theorem \ref{teo:RATES}, we get the following result that improves the rates given in Theorem \ref{teo:RATES1}.
 
 \vskip0.1in
 
\begin{corollary}\label{teo:RATESpi} 
Let $\rho$ be a function satisfying  \ref{ass:rho_bounded_derivable},  \ref{ass:psi_strictly_positive} and  \ref{ass:psidif}, and $w$ a weight function satisfying \ref{ass:funcionw}. Assume that \ref{ass:beta0}, \ref{ass:approxorder} and \ref{ass:probaX*} hold and that  for some positive constant $ C>0 $, $\prob(\|X\|\le C)=1$.   Then,  $\gamma_n\pi_\prob(\wtheta ,\theta_0 )=O_\prob(1)$, whenever $\gamma_n=O(n^{r\,\varsigma})$ and $\gamma_n \sqrt{\log(\gamma_n)} = O(n^{(1-\varsigma)/2})$. In particular,   $$\frac{n^{\eta}}{\sqrt{\log(n)}}\pi_\prob(\wtheta ,\theta_0 )=O_\prob(1)\,,$$
where  $\eta=\min(r\,\varsigma, (1-\varsigma)/2)$.
 \end{corollary}

\begin{remark}
 As mentioned in  \citet{boente:salibian:vena:2020}, who obtained rates under a semi-linear functional regression model, if $\varsigma = 1/(1 + 2r)$ in \ref{ass:approxorder}, one can choose $\gamma_n=O(n^{\frac{r}{1+2r}-\delta})$, for some $\delta > 0$ arbitrarily small, which yields   a
convergence rate arbitrarily close to the optimal one. Then, as mentioned above, when considering cubic splines, if $\beta_0$ is twice continuously  differentiable one has that $r=2$ in  \ref{ass:approxorder}. Hence, taking  $\varsigma=1/5$, i.e., if the basis dimension $k_n$ has order $n^{1/5}$, we ensure that the   convergence rate for $\wtpi_\prob(\wtheta ,\theta_0 )$ and $\pi_\prob(\wtheta ,\theta_0 )$ is arbitrarily close to $n^{2/5}$.

Clearly, one may select in Theorem \ref{teo:RATES}, $\gamma_n=n^{\eta}/\sqrt{\log(n)}$ where $\eta=\min(r\,\varsigma, (1-\varsigma)/2)$. Hence, if $\varsigma = 1/(1 + 2r)$ in  \ref{ass:approxorder}, we have that  $\gamma_n=n^{ {r}/({1+2r})}/\sqrt{\log(n)}$. Hence, for bounded covariates, both the weighted mean square error of the predictions and the weighted mean square error between the predicted probabilities are such that $\wtpi_\prob(\wtheta ,\theta_0 )=O_{\prob}\left(n^{-\, {r}/({1+2r})}\; \sqrt{\log(n)}\right)$  and $\pi_\prob(\wtheta ,\theta_0 )=O_{\prob}\left(n^{-\, {r}/({1+2r})}\; \sqrt{\log(n)}\right)$, leading to a convergence rate is suboptimal with respect to the one obtained for instance in nonparametric regression models, see \citet{stone:1982, stone:1985}. Furthermore, this rate  equals the one obtained for penalized estimators in \citet{kalogridis:2023}, when considering  $B-$splines with $k_n=O(n^{1/(2r+1)})$ and the slope function is $r$ times continuously differentiable, i.e., $\beta_0\in \itW^{r-1,\infty}$. As mentioned therein, the term $\log(n)$ is related to the fact that we are considering  infinite--dimensional covariates. Recall that when $X\in L^2(0,1)$ and to ensure identifiability, the eigenvalues of its covariance operator are non--null but converge to $0$, enabling us to provide a lower bound for $\wtpi_\prob(\theta_1 ,\theta_2 )$ in terms of  $\|\beta_1-\beta_2\|$.

It also is worth mentioning that the rate derived in   Theorem \ref{teo:RATES1}(c), allow to conclude that,  when $k_n=O(n^{1/(4r+1)}$, that is, when $\varsigma=1/(4r+1)$, we have $\pi_\prob (\wtheta ,\theta_0 )=O_\prob(n^{-\,r/(4r+1)})$. Hence, if the functional covariates are bounded,  from Proposition \ref{prop:cotainf} we get that $\wtpi_\prob (\wtheta ,\theta_0 )=O_\prob(n^{- \, r/(4r+1)})$. This   rate of convergence corresponds to the one obtained in  \citet{cardot:sarda:2005} for their penalized estimators  and is slower than the rate $\wtpi_\prob(\wtheta ,\theta_0 )= O\left(n^{-\, {r}/({1+2r})}\; \sqrt{\log(n)}\right)$  derived from  Theorem \ref{teo:RATES}.

\end{remark}

%%%%%%%%%%%%%%%%%%%%%%%%
%MONTE CARLO
%%%%%%%%%%%%%%%%%%%%%%%%%%%%%%%%
\section{Simulation study}{\label{sec:simu}}
We performed a Monte Carlo study to investigate the finite-sample properties 
of our proposed estimators for the functional logistic regression model. For that purpose, we generated a training sample $\itM$ of observations $(y_i, X_i)$ i.i.d. such that
$y_i\sim Bi(1, F(\alpha_0+\langle \beta_0, X_i\rangle)$ where $\alpha_0=0$. The true regression parameter  was set equal to   $\beta_0(t) = \sum_{j=1}^{50} b_{j,0} \phi_j$, where $\{\phi_j\}_{j=1}^{50}$ correspond to elements of the Fourier basis, more precisely,  $\phi_1(t) \equiv 1$, $\phi_j(t) = \sqrt{2} \cos ((j-1)\pi t)$, $j\geq 2$, and  
the coefficients $b_{1,0} = 0.3$ and $b_{j,0} = 4(-1)^{j+1}j^{-2}$, $j \geq 2$.
The process that generates the functional covariates $X_i(t)$ was Gaussian with mean 0 and 
covariance operator with eigenfunctions  $\phi_j(t)$. 
For uncontaminated samples, the scores  $\xi_{ij} $  were generated as  independent Gaussian 
random variables $\xi_{ij}\sim N(0,j^{-2})$. We denote the distribution of this Gaussian process $\itG(0, \Gamma)$.  Taking into account that $\var(\xi_{ij})\le 1/2500$ when $j > 50$, the process was approximated numerically using the first 50 terms of its  Karhunen-Lo\`{e}ve representation.

We chose as basis the $B-$spline basis for all the procedures considered, denoted $\{B_j\}_{j=1}^{k_n}$. We compared four estimators: the  procedure based on using the deviance after dimension  reduction, that is using $\rho(t)=t$ in \eqref{eq:phiBY}, labelled the classical estimators and denoted \textsc{cl}, the one that uses $M-$estimators   denoted \textsc{m}, and their weighted versions. The $M-$estimators and  weighted $M-$estimators were computed using the loss function introduced in \citet{Croux:H:2003} and defined as 
$$
\rho(t) =\left\{
\begin{array}{ll}
t\; e^{-\sqrt{c}} & \hbox{if }  \,\, t\leq c\\
-2e^{-\sqrt{t}}\left(  1+\sqrt{t}\right)  +e^{-\sqrt{c}}\left(  2\left(
1+\sqrt{c}\right)  + c \right)  & \hbox{if }  \,\, t>c \, ,
\end{array}
\right. $$
with tuning constant $c=0.5$. For the former the weights equal  1 for all observations, while for the latter, as for the weighted deviance estimators, two different type of weight   functions were considered.
\begin{itemize}
	\item[a)] For the first one, after dimension reduction, that is, after computing $x_{ij}=\langle X_i, B_j\rangle$, we evaluated the Donoho--Stahel location and scatter estimators, denoted $\wbmu$ and $\wbSi$, respectively, of the sample $\bx_1, \dots, \bx_n$ with $\bx_i=(x_{i1}, \dots, x_{i k_n})\trasp $. The weights are then  defined as $w(X_i)=1$ when the squared Mahalanobis distance $ d_i^2= ( \bx_i - \wbmu )\trasp \wbSi^{-1}( \bx_i - \wbmu )$ is less than or equal to $\chi_{0.975,k_n}$ and $0$ otherwise, where   $\chi_{_{\alpha,p}}$ stands  for the $\alpha-$quantile of a chi-square distribution with $p$ degrees of freedom. Hence,  for this family we used hard rejection weights and for that reason, the weighted estimators based on the deviance and the weighted $M-$estimators   are denoted ${\wclHRnorm}$  and  ${\wemeHRnorm}$, respectively.  Note that  ${\wclHRnorm}$ are related to the Mallows--type estimators introduced in \citet{carroll:pederson:1993}.
	
	\item[b)] The second family of weight functions  is based on the functional boxplots as defined by \citet{sun:genton:2011}.  Again, we choose   hard rejection weights but   on the functional space by taking $w(X) = 0$ if $X$ was declared an outlier by the functional boxplot  and $1$ otherwise.  The functional boxplot was computed using the function \texttt{fbplot} of the library \texttt{fda} taking as method "Both" which orders the observations according to the band depth and then breaks ties with the modified band depth, as defined  in \citet{lopez-pintado:romo:2009}. 
In this case the weighting is done before the  $B-$spline approximation.  The resulting weighted classical and $M-$estimators are denoted  ${\wclBOXnorm}$  and  ${\wemeBOXnorm}$, respectively.
\end{itemize}

For each setting we generated $n_R = 1000$ samples of size $n = 300$ and 
used cubic splines with equally spaced knots. For the robust  estimators we  
selected the size of the spline basis, $ k_{n }$,
by minimizing $RBIC(k)$ in equation \eqref{eq:bic1} over the  grid $4 \leq k \leq 14$.  
For  the classical estimator, we used the standard $BIC$ criterion, that is, we  chose  $\rho(t)=t$  in equation \eqref{eq:bic1}.

\begin{figure}[ht!] 
  \begin{center} 
    \begin{tabular}{ccc}
     $C_0$ & $C_{1,0.05}$ & $C_{2,0.05}$\\
      \includegraphics[width=0.35\textwidth]{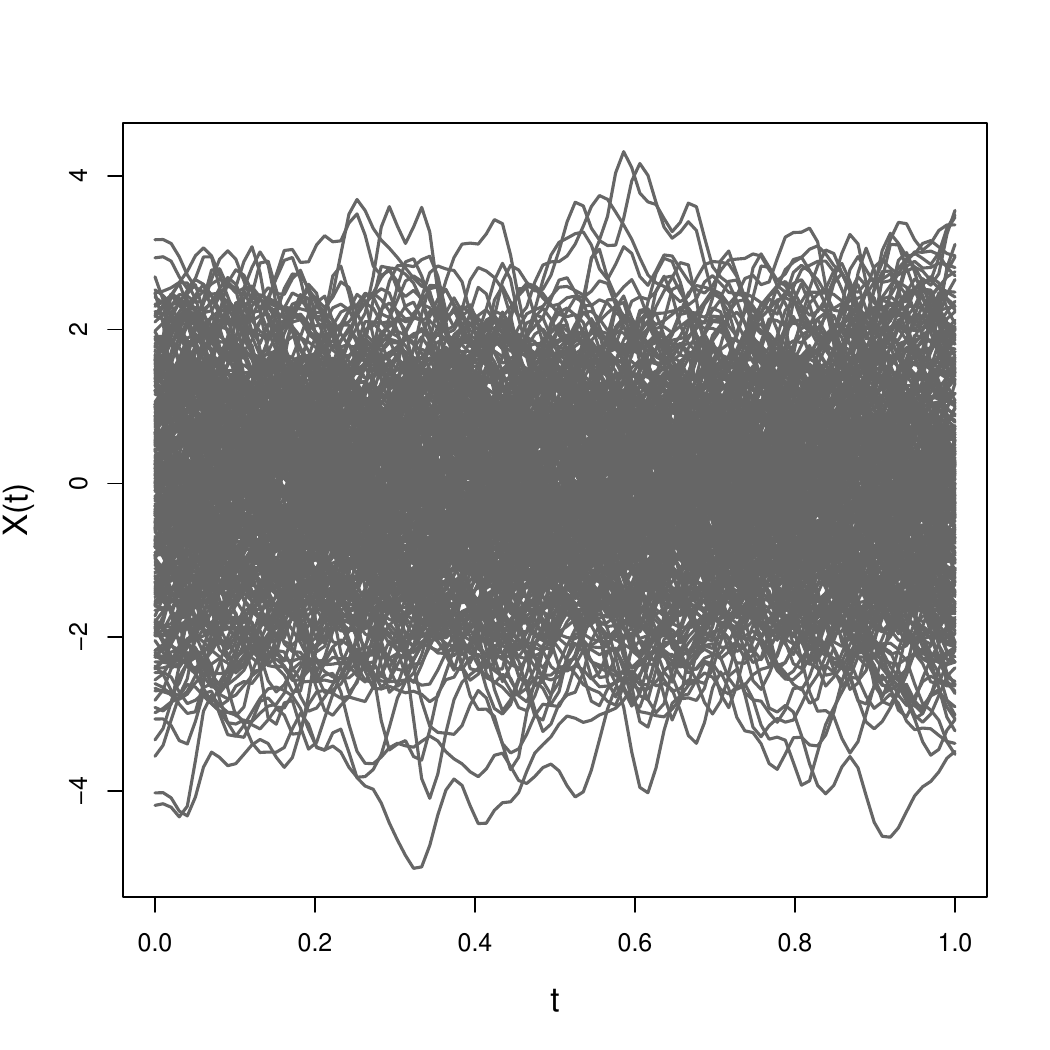} &
      \includegraphics[width=0.35\textwidth]{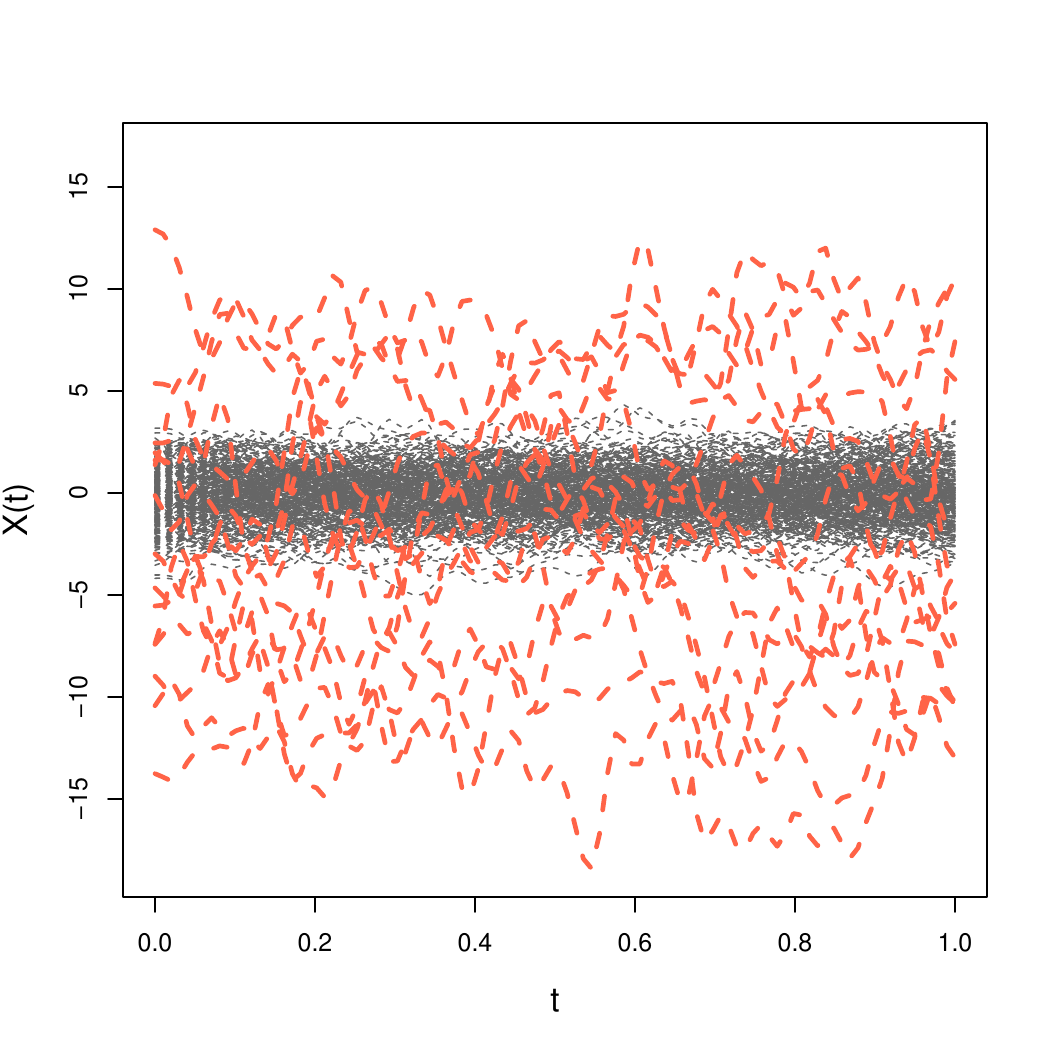} &
      \includegraphics[width=0.35\textwidth]{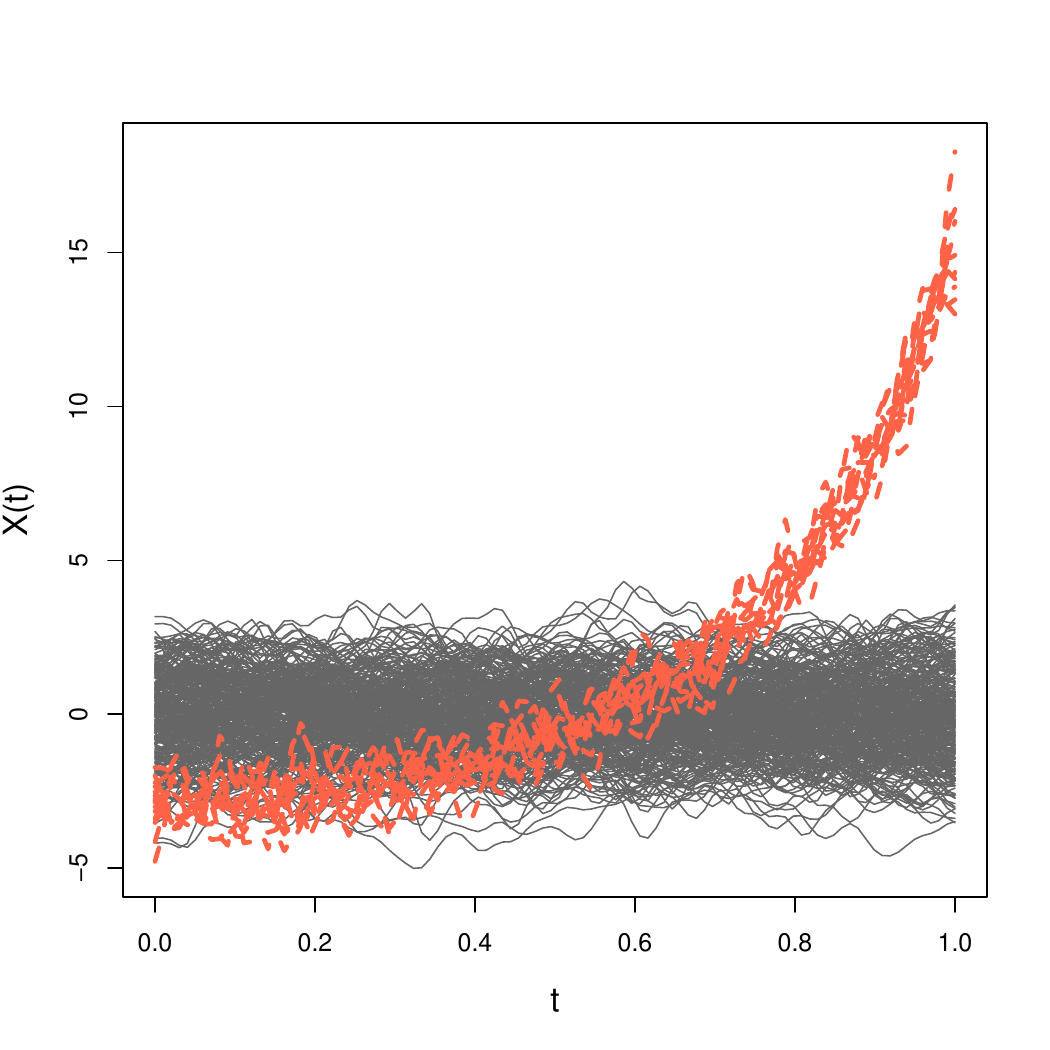} \\
      $C_{3,0.05}$ & $C_{4,0.05}$ & $C_{5,0.05}$\\
      \includegraphics[width=0.35\textwidth]{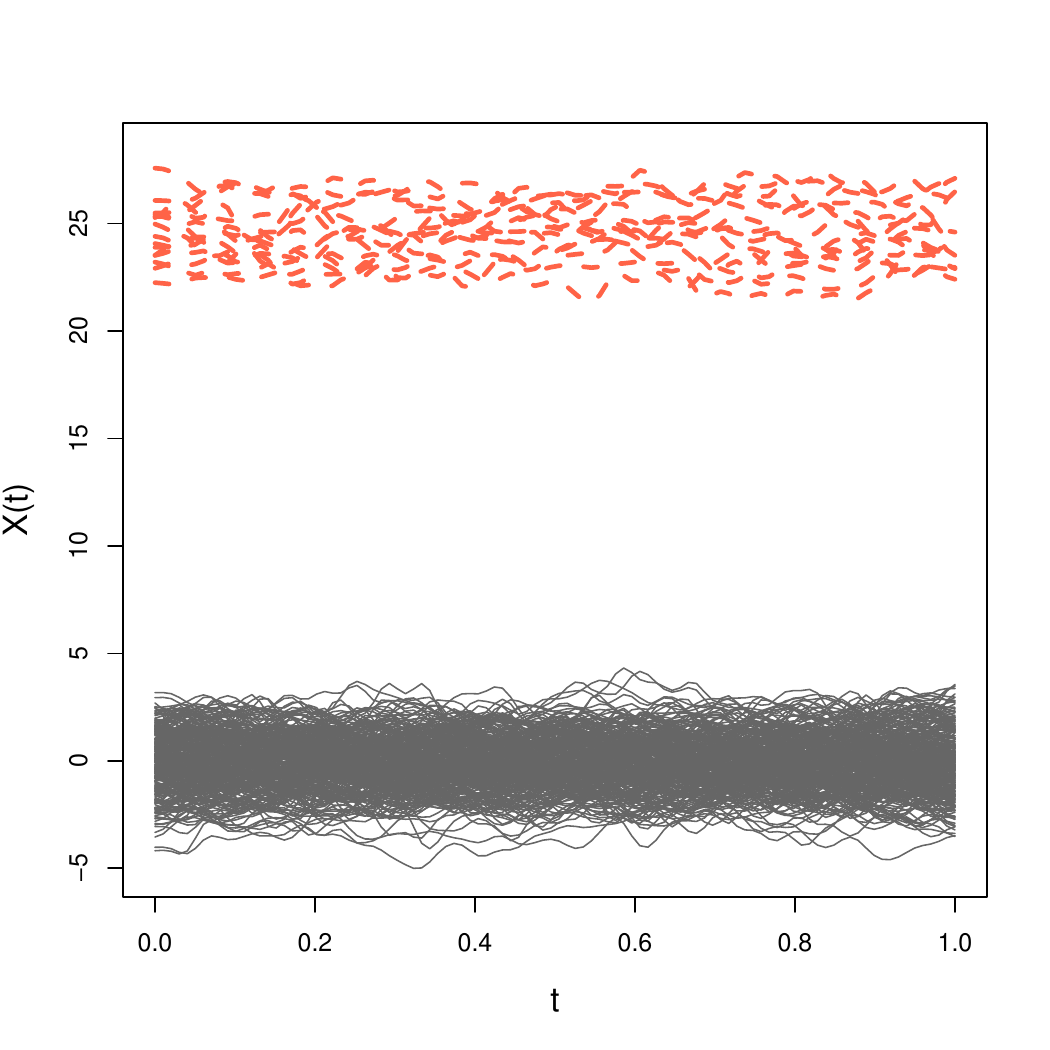} &
      \includegraphics[width=0.35\textwidth]{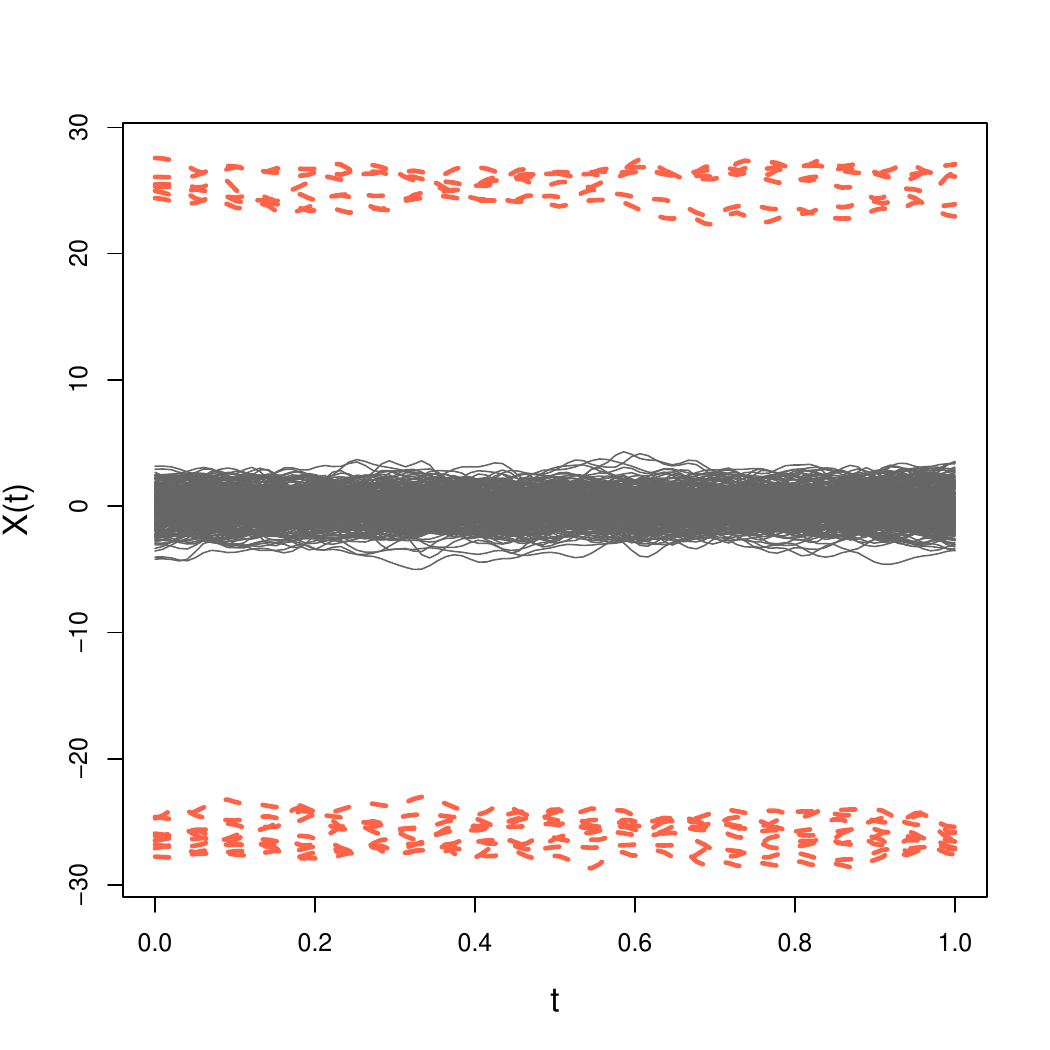} &
      \includegraphics[width=0.35\textwidth]{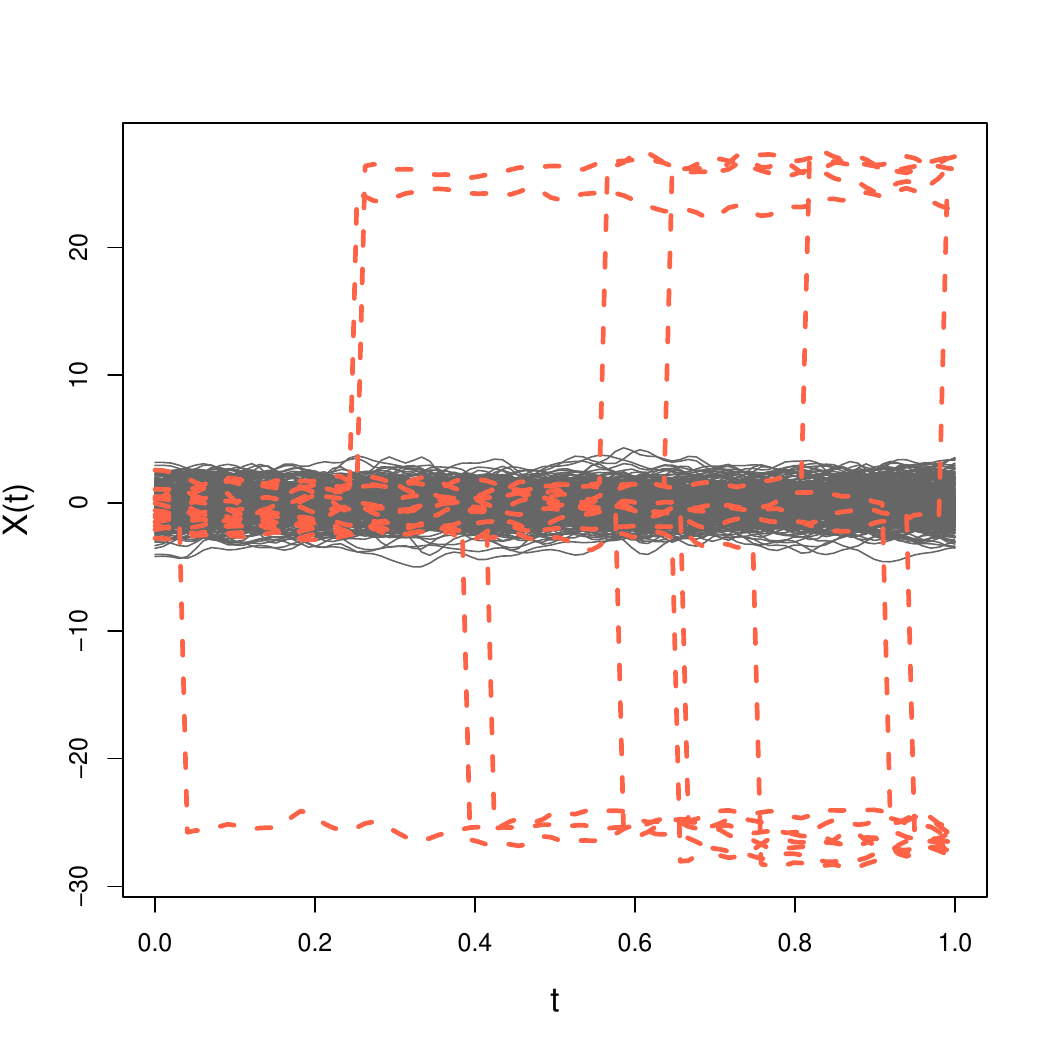} 
   \end{tabular}
   \end{center}
 \vskip-0.2in \caption{\small \label{fig:fplm-trayectoria} Trajectories $X_i(t)$ with and without contamination. The red dotted lines correspond to the added  contaminated covariates.}
\end{figure}

We considered different contamination scenarios  by adding a proportion $\epsilon$ of atypical points. We denote these scenarios $C_{j,\epsilon}$, for $1\le j\le 5$ and we chose   $\epsilon=0.05$ and $0.10$.
\begin{itemize}
	\item In the first scenario, denoted $C_{1,\epsilon}$, we  generated $n_{\out}=\epsilon \; n$ misclassified points  $(\wty, \wtX)$, where $\wtX \sim   \itG(0, 25 \Gamma)$ and $\wty = 1 $ when $\alpha_0+\langle\wtX,  \beta_0\rangle < 0$ and $\wty=0$, otherwise.
	\item Under $C_{2,\epsilon}$,   we  have tried to adapt to the functional framework  the damaging effect of the high leverage points  considered  by  \citet{Croux:H:2003}. For that purpose, given $m>0$, we generated   $\wtX \sim \itG( m\; \beta_0, 0.01 \Gamma)$. The response  $\wty$, related to  $\wtX$,  was  always taken equal to $0$. It is worth noticing that $\langle \wtX, \beta_0\rangle$ is very close to  $ m  \|\beta_0\|^2$, thus the leverage of the added points increases with  $m$. We chose $m=4$.
	\item Contamination $C_{3,\epsilon}$ generates extreme outliers as $\wtX \sim   \itG(\mu,  \Gamma)$ with $\mu(t)=25$, for all $t$. As in $C_{1,\epsilon}$, we chose $\wty = 1 $ when $\alpha_0+\langle\wtX,  \beta_0\rangle < 0$ and $\wty=0$, otherwise.
	\item Setting $C_{4,\epsilon}$ aims  to construct trajectories with extreme symmetric outliers. For that purpose, we define $\wtX\sim 0.5   \itG(\mu,  \Gamma)+ 0.5 \itG(- \mu,  \Gamma)$ with $\mu(t)=25$, for all $t$. As in $C_{1,\epsilon}$, we chose $\wty = 1 $ when $\alpha_0+\langle\wtX,  \beta_0\rangle < 0$ and $\wty=0$, otherwise.
	\item The purpose of $C_{5,\epsilon}$ is to add trajectories with a partial contamination. We generated   $\wtX(t) =Z(t) + 25\, B   \indica_{T<t}$ where $Z \sim   \itG(0,  \Gamma)$, $\prob(B=1)=\prob(B=-1)=0.5$ and $T\sim \itU(0,1)$ and $Z,B$ and $T$ are independent. As in $C_{1,\epsilon}$, we chose $\wty = 1 $ when $\alpha_0+\langle\wtX,  \beta_0\rangle < 0$ and $\wty=0$, otherwise.
\end{itemize}
The way contaminated trajectories are constructed under settings $C_{3,\epsilon}$  to $C_{5,\epsilon}$ corresponds  to the contaminations considered in \citet{denhere:billor:2016}. However, we force the atypical trajectories to correspond to bad leverage points. 

To illustrate the  type of outliers generated, 
Figure \ref{fig:fplm-trayectoria} shows the obtained functional covariates $X_i(t)$, 
for one sample generated under each scheme.

To compare the estimators of $\alpha_0$, we computed their biases and  standard deviations, $s_{\walfa}$ which are reported in  Table  \ref{tab:tabla2-alfa-C0-C5}. We also present in Figures \ref{fig:boxplotsalpha} and \ref{fig:boxplotsalpha-C3} their boxplots   for the considered  contaminations.

 \begin{table}[ht!]
  \centering
  \small
   \renewcommand{\arraystretch}{1.2}
\begin{tabular}{  c |  rc  | rc  | rc  |  rc|  rc|}
\hline  
  & $\walfa-\alpha_0$ & $s_{\walfa}$  & $\walfa-\alpha_0$ & $s_{\walfa}$ & $\walfa-\alpha_0$ & $s_{\walfa}$& $\walfa-\alpha_0$ & $s_{\walfa}$ & $\walfa-\alpha_0$ & $s_{\walfa}$ \\ 
  \hline
  &\multicolumn{2}{c|}{ } & \multicolumn{2}{c|}{ } & \multicolumn{2}{c|}{$C_0$} &\multicolumn{2}{c|}{ } &\multicolumn{2}{c|}{} \\
\hline
${\clasnorm}$ &  & &  & & -0.0040 & 0.1250 & &  & &\\ 
${\emenorm}$ & & &  & & -0.0072 & 0.1528 & &  & &  \\ 
${\wclHRnorm}$ & & &  & & -0.0046 & 0.1250 & &  & &  \\ 
${\wemeHRnorm}$ & & &  & & -0.0045 & 0.1278 & &  & &   \\  
${\wclBOXnorm}$  & & &  & &-0.0040 & 0.1250 & &  & & \\  
${\wemeBOXnorm}$ & & &  & & -0.0083 & 0.1259 & &  & &    \\ 
  \hline 
  &\multicolumn{2}{c|}{$C_{1,0.05}$} &\multicolumn{2}{c|}{$C_{2,0.05}$}&\multicolumn{2}{c|}{$C_{3,0.05}$}&\multicolumn{2}{c|}{$C_{4,0.05}$}&\multicolumn{2}{c|}{$C_{5,0.05}$}\\ 
  \hline 
${\clasnorm}$ &  -0.0024 & 0.1160  & -0.0540 & 0.1200 &  -0.0200 & 0.1230  & -0.0064 & 0.1216  & -0.0034 & 0.1166   \\ 
${\emenorm}$ &   -0.0038 & 0.1589 &  -0.1613 & 0.1139 &  -0.0580 & 0.1086 &  -0.0041 & 0.1170 &  -0.0050 & 0.1138 \\ 
${\wclHRnorm}$ &     -0.0029 & 0.1265 & -0.0028 & 0.1263 & -0.0029 & 0.1263 & -0.0065 & 0.1247  & -0.0029 & 0.1254\\ 
${\wemeHRnorm}$ &   0.0007 & 0.1331 & 0.0008 & 0.1331  & 0.0009 & 0.1333 & -0.0053 & 0.1304 &0.0042 & 0.1282  \\  
${\wclBOXnorm}$  &   -0.0030 & 0.1208 & -0.0080 & 0.1247 & -0.0177 & 0.1228 & -0.0061 & 0.1234 &  -0.0028 & 0.1234 \\  
${\wemeBOXnorm}$ &   -0.0002 & 0.1352 &  -0.0068 & 0.1326  & -0.0168 & 0.1313  & -0.0084 & 0.1283 & 0.0034 & 0.1259  \\   
\hline
&\multicolumn{2}{c|}{$C_{1,0.10}$} &\multicolumn{2}{c|}{$C_{2,0.10}$}&\multicolumn{2}{c|}{$C_{3,0.10}$}&\multicolumn{2}{c|}{$C_{4,0.10}$}&\multicolumn{2}{c|}{$C_{5,0.10}$}\\ 
\hline
 ${\clasnorm}$ &  0.0011 & 0.1117 &  -0.0550 & 0.1220 & -0.0154 & 0.1219  & -0.0017 & 0.1219 & 0.0035 & 0.1174\\ 
${\emenorm}$ &    -0.0092 & 0.1564 & -0.1680 & 0.0952 &-0.0397 & 0.1110 & -0.0024 & 0.1128 & -0.0006 & 0.1064\\ 
${\wclHRnorm}$ & 0.0022 & 0.1234 &   0.0018 & 0.1240 &  0.0021 & 0.1241 & -0.0006 & 0.1250 & 0.0053 & 0.1247 \\ 
${\wemeHRnorm}$ & -0.0067 & 0.1292 &   -0.0069 & 0.1302 &  -0.0065 & 0.1299 &  -0.0032 & 0.1292 &  0.0078 & 0.1279 \\  
${\wclBOXnorm}$  & 0.0030 & 0.1156 &  -0.0057 & 0.1235 & -0.0154 & 0.1219 &  -0.0014 & 0.1220  &  0.0051 & 0.1229 \\  
${\wemeBOXnorm}$  & -0.0037 & 0.1262 &   -0.0145 & 0.1295 &  -0.0240 & 0.1297  & -0.0017 & 0.1245 &  0.0031 & 0.1296\\ 
\hline
   \end{tabular}
\caption{ \small \label{tab:tabla2-alfa-C0-C5} 
   Bias and standard deviations  for the estimators of $\alpha_0$, over $n_R = 1000$ for clean and contaminated samples of size $n = 300$.}
\end{table}

\begin{figure}[ht!]
 \begin{center}
 \footnotesize
 \renewcommand{\arraystretch}{0.2}
 \newcolumntype{M}{>{\centering\arraybackslash}m{\dimexpr.01\linewidth-1\tabcolsep}}
   \newcolumntype{G}{>{\centering\arraybackslash}m{\dimexpr.35\linewidth-1\tabcolsep}}
%\begin{tabular}{MGG}
\begin{tabular}{GG}
\multicolumn{2}{c}{$C_{0}$}   \\
\multicolumn{2}{c}{ \includegraphics[scale=0.35]{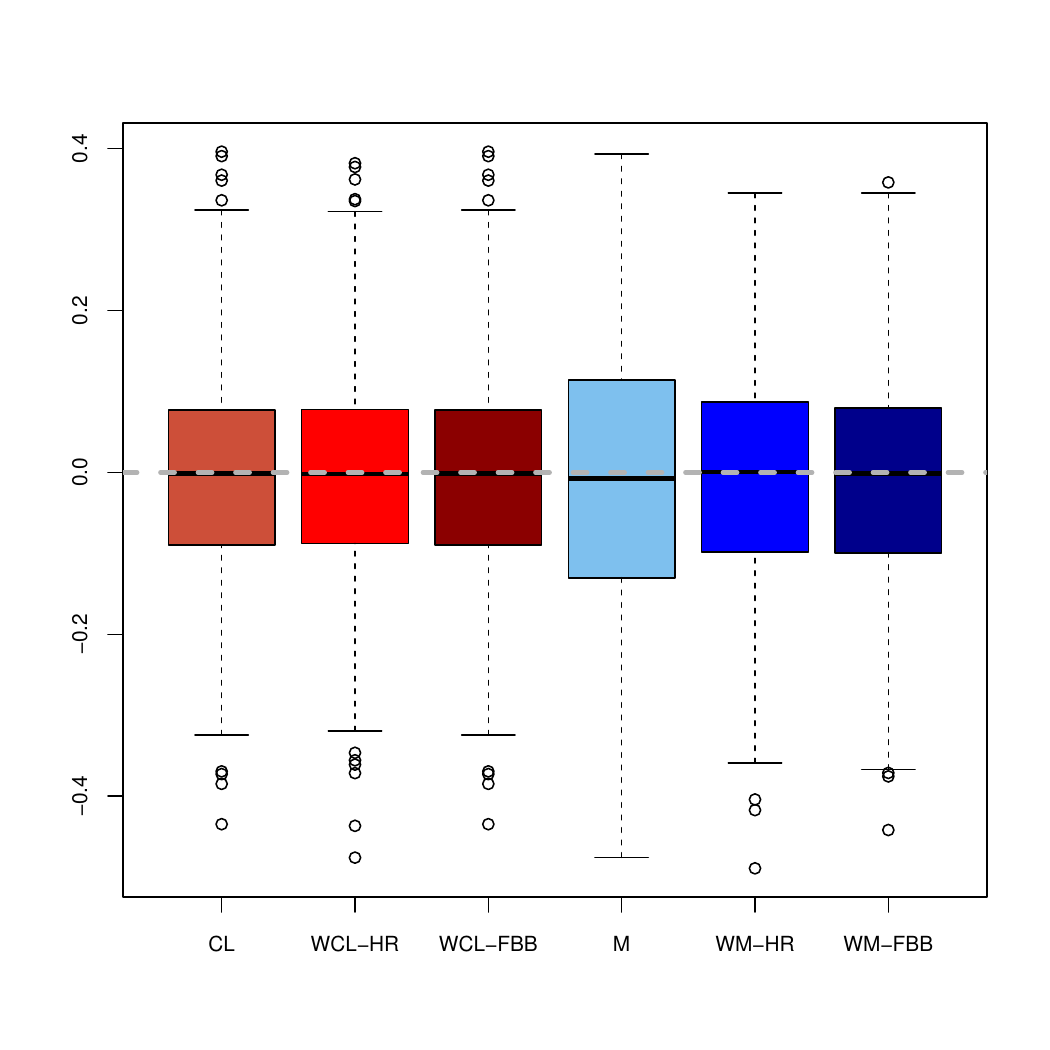} }   \\
 $C_{1,0.05}$ & $C_{1,0.10}$   \\
     \includegraphics[scale=0.35]{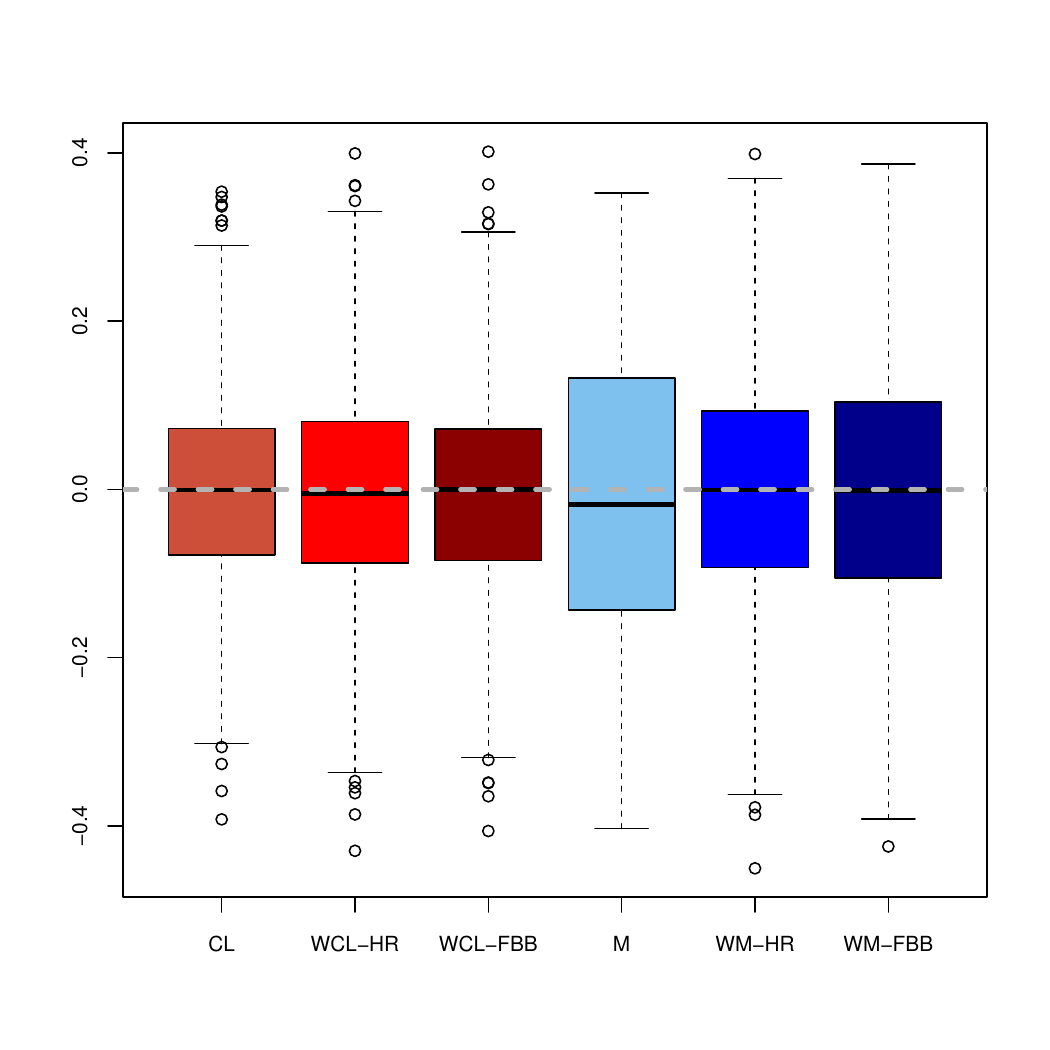} 
&  \includegraphics[scale=0.35]{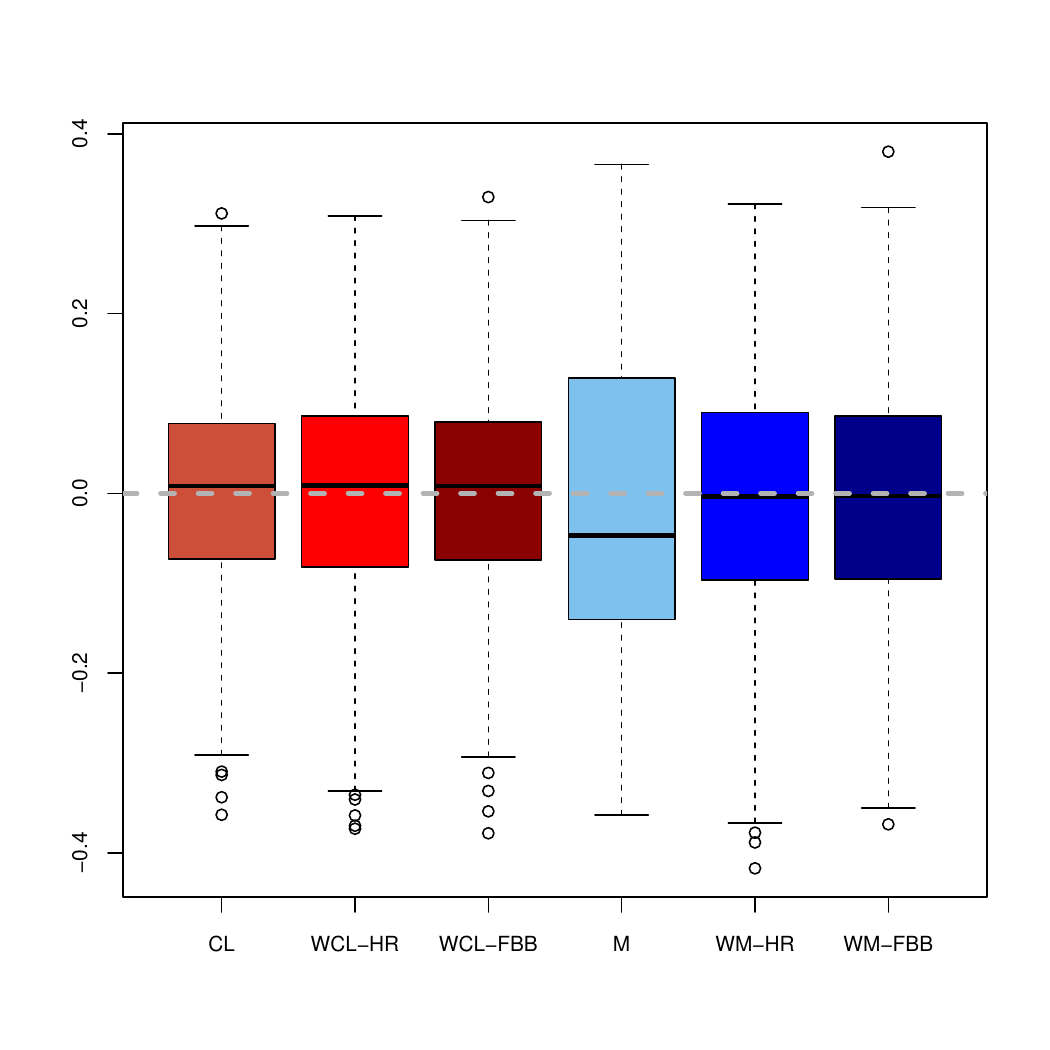} \\
 $C_{2,0.05}$ & $C_{2,0.10}$   \\
   \includegraphics[scale=0.35]{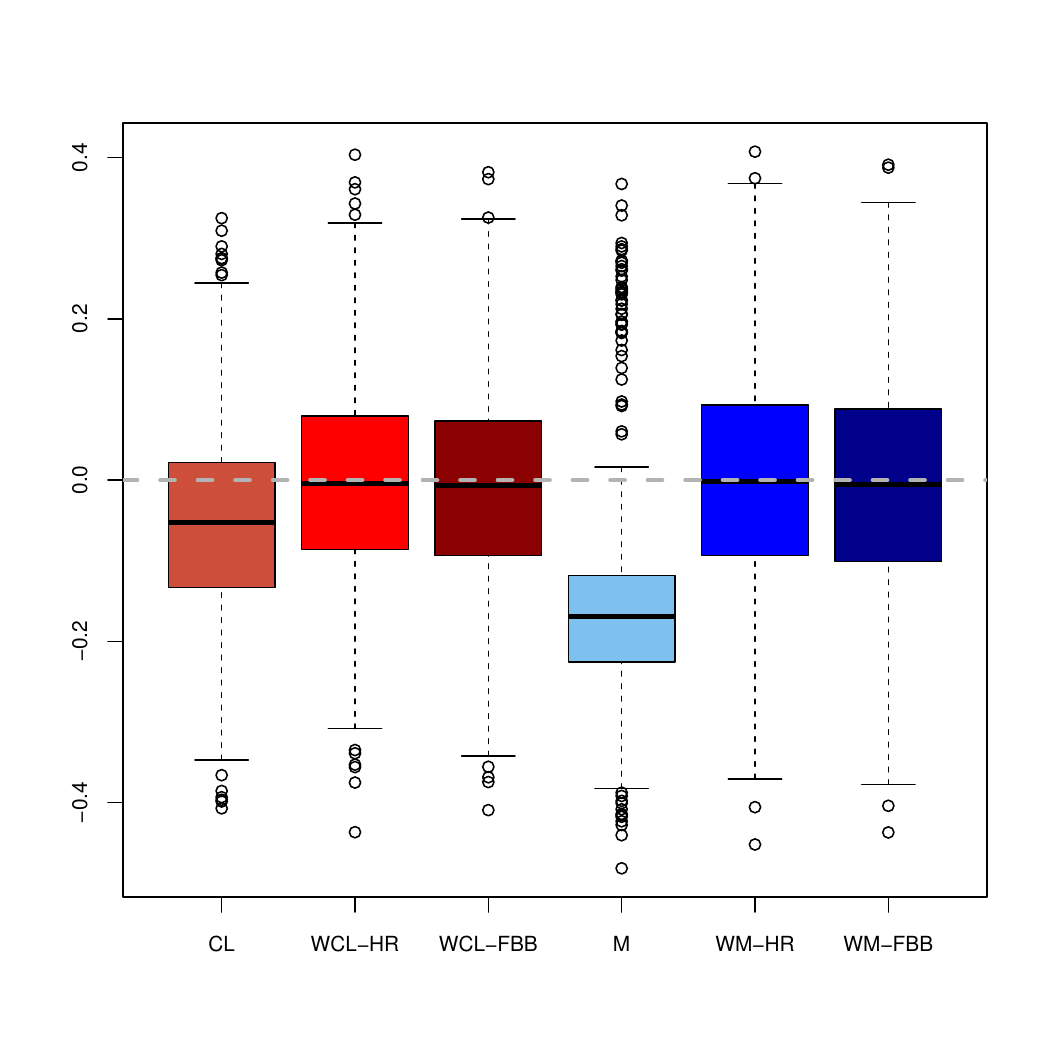}
&  \includegraphics[scale=0.35]{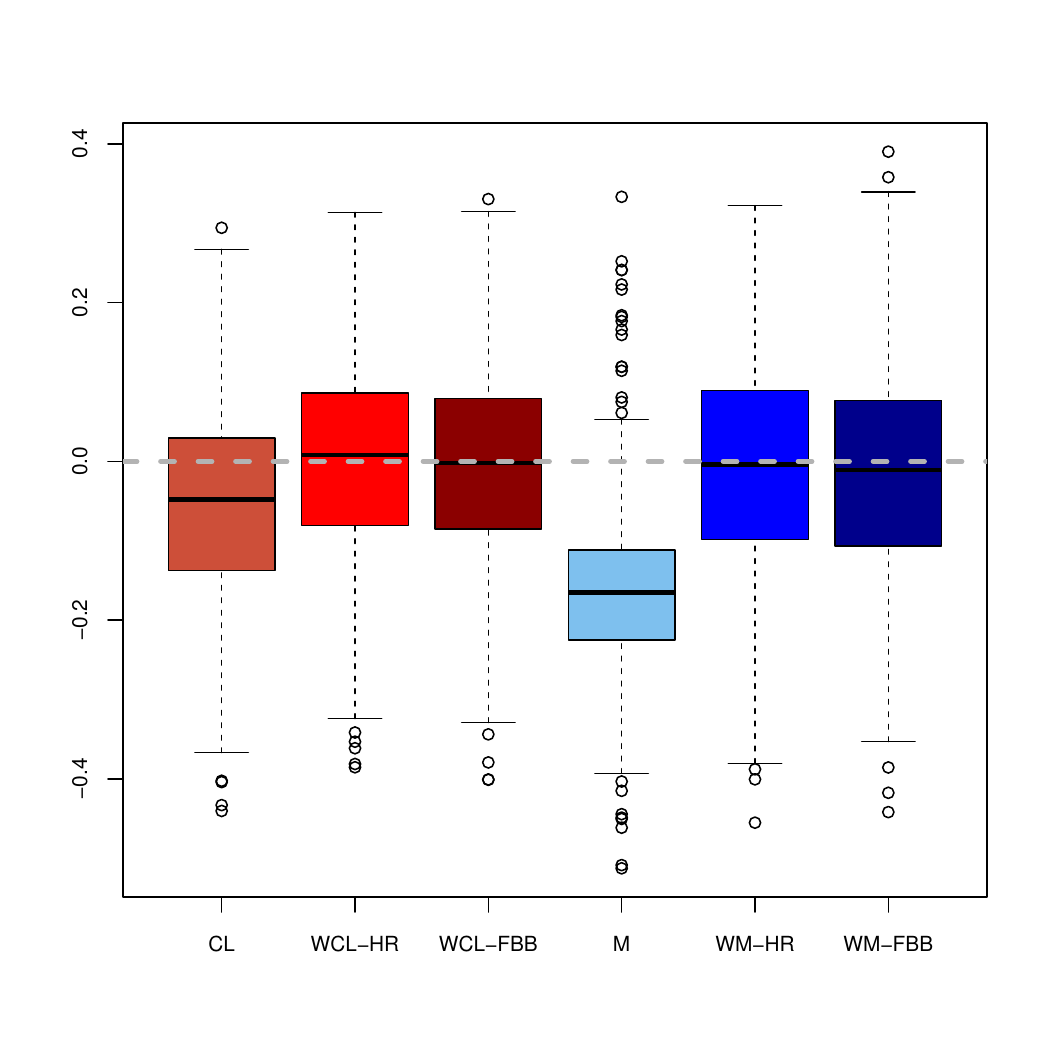} 
 \end{tabular}
\caption{\small \label{fig:boxplotsalpha}  Boxplot of the estimators for $\alpha_0$ under $C_0$ to $C_{2,0.10}$.}
\end{center} 
\end{figure}

\begin{figure}[ht!]
 \begin{center}
 \footnotesize
 \renewcommand{\arraystretch}{0.2}
 \newcolumntype{M}{>{\centering\arraybackslash}m{\dimexpr.01\linewidth-1\tabcolsep}}
   \newcolumntype{G}{>{\centering\arraybackslash}m{\dimexpr.35\linewidth-1\tabcolsep}}
%\begin{tabular}{MGG}
\begin{tabular}{GG} $C_{3,0.05}$ & $C_{3,0.10}$   \\[-4ex]
   \includegraphics[scale=0.35]{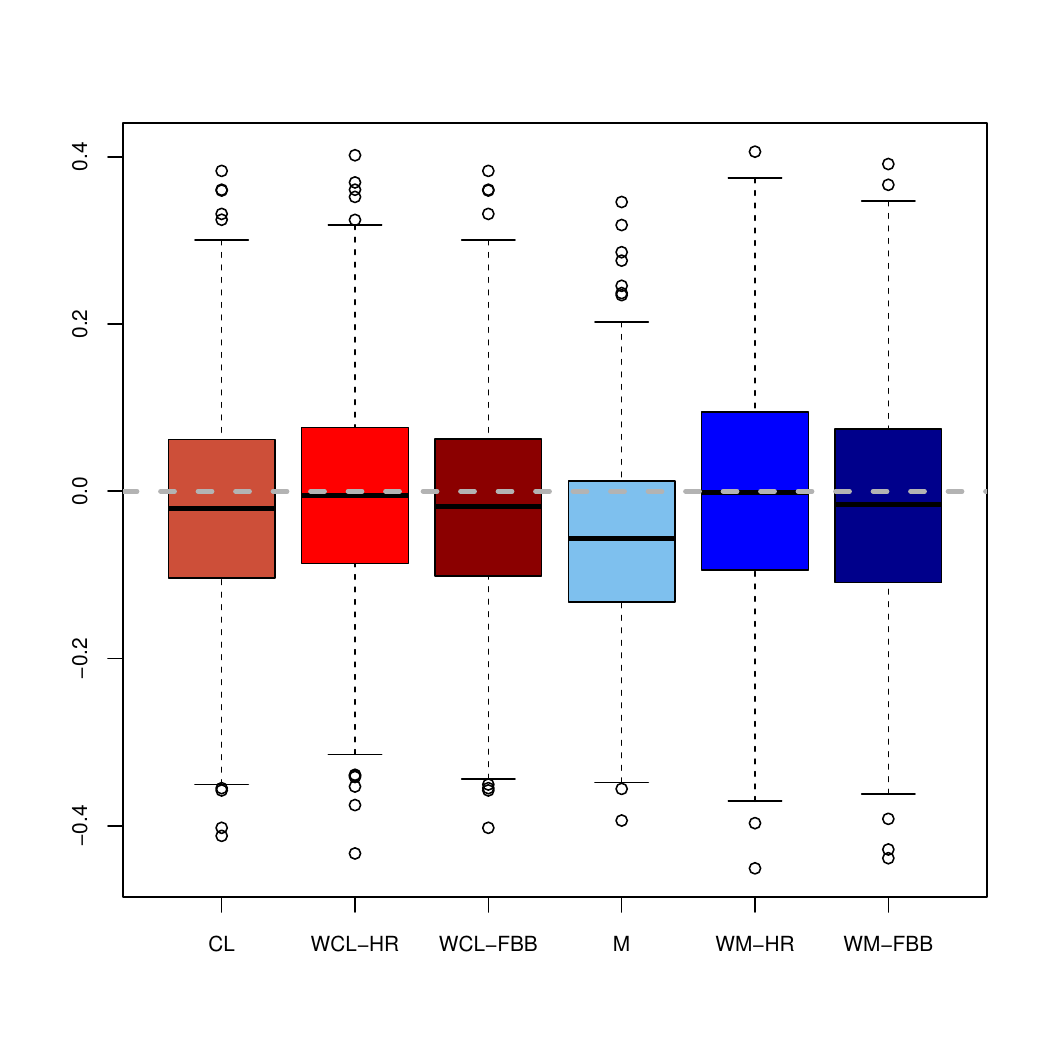}
&  \includegraphics[scale=0.35]{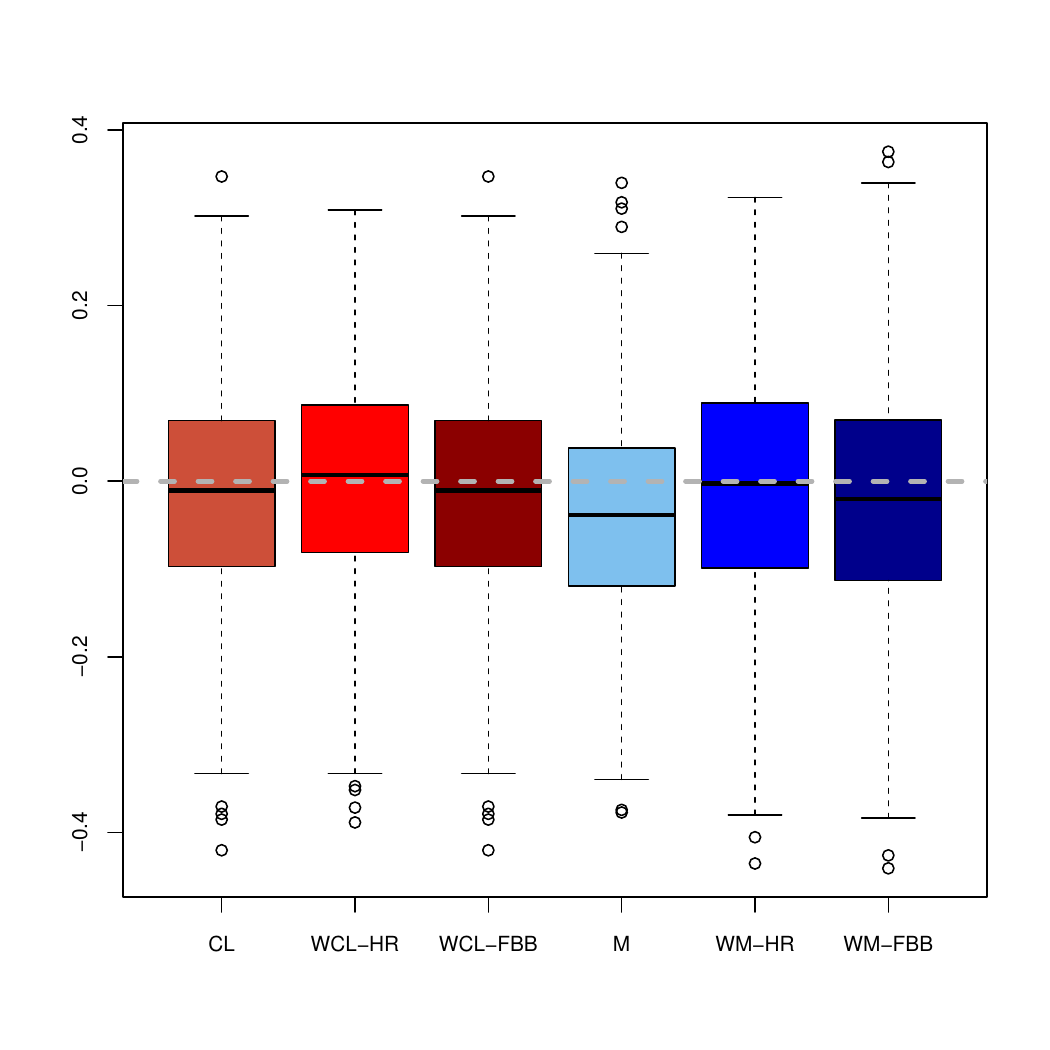} \\[-3ex]
 $C_{4,0.05}$ & $C_{4,0.10}$   \\[-4ex]
   \includegraphics[scale=0.35]{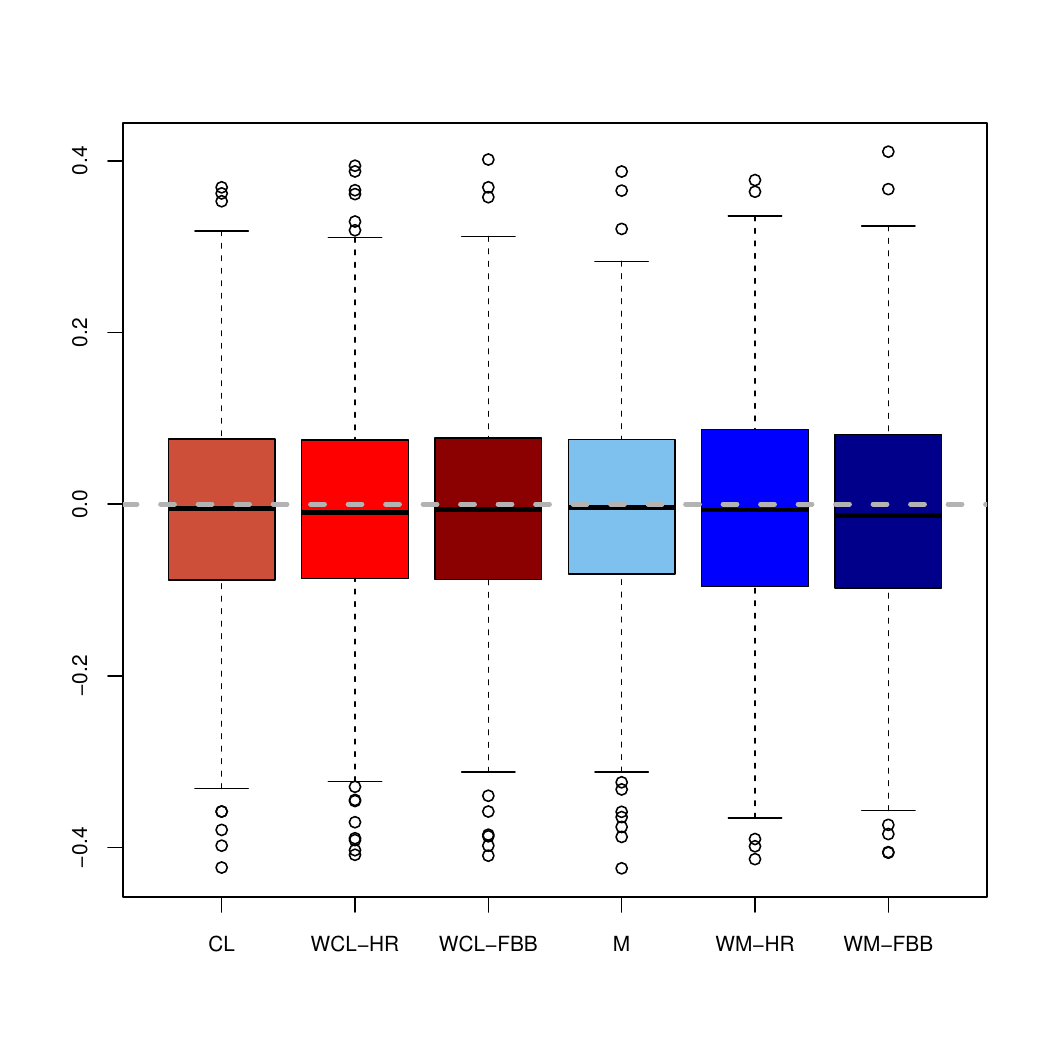}
&  \includegraphics[scale=0.35]{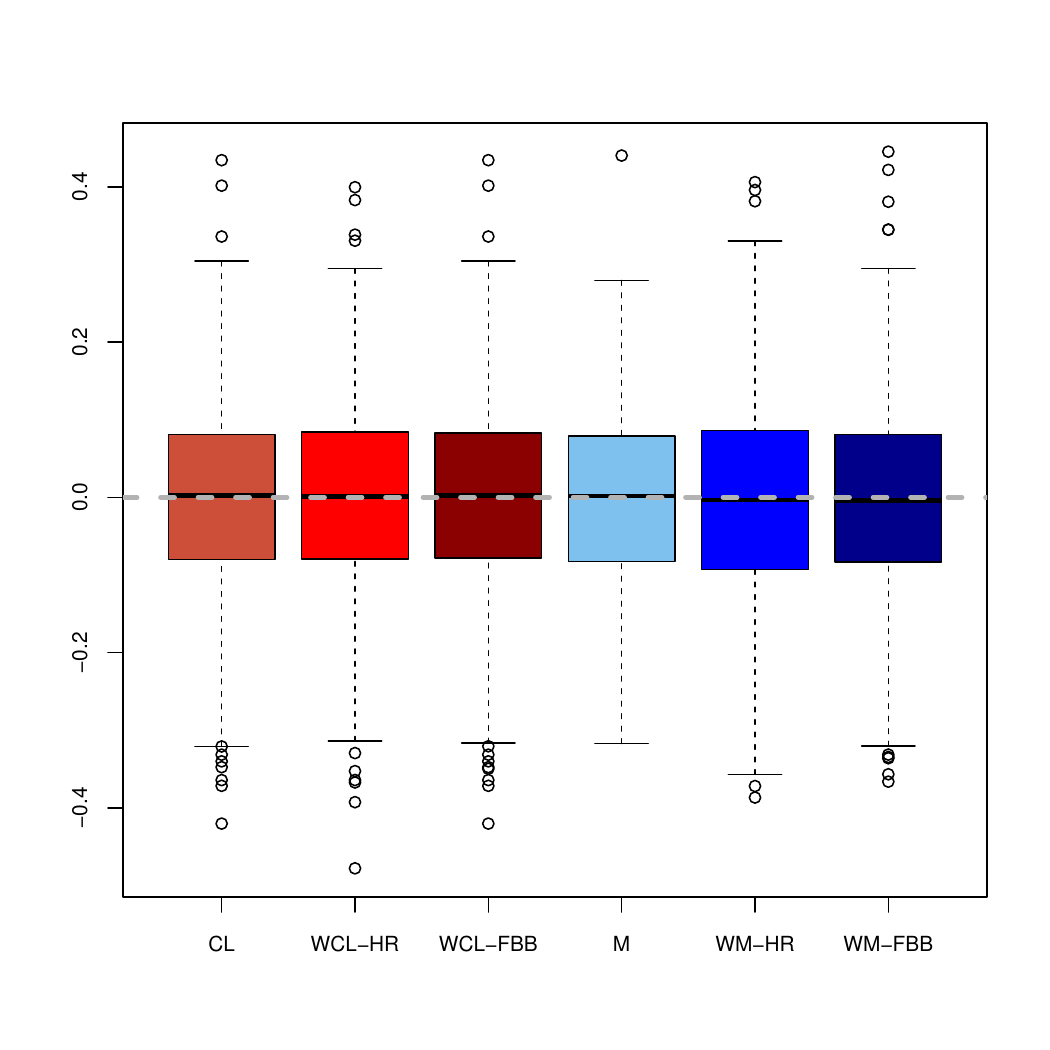} \\[-3ex]
 $C_{5,0.05}$ & $C_{5,0.10}$   \\[-4ex]
   \includegraphics[scale=0.35]{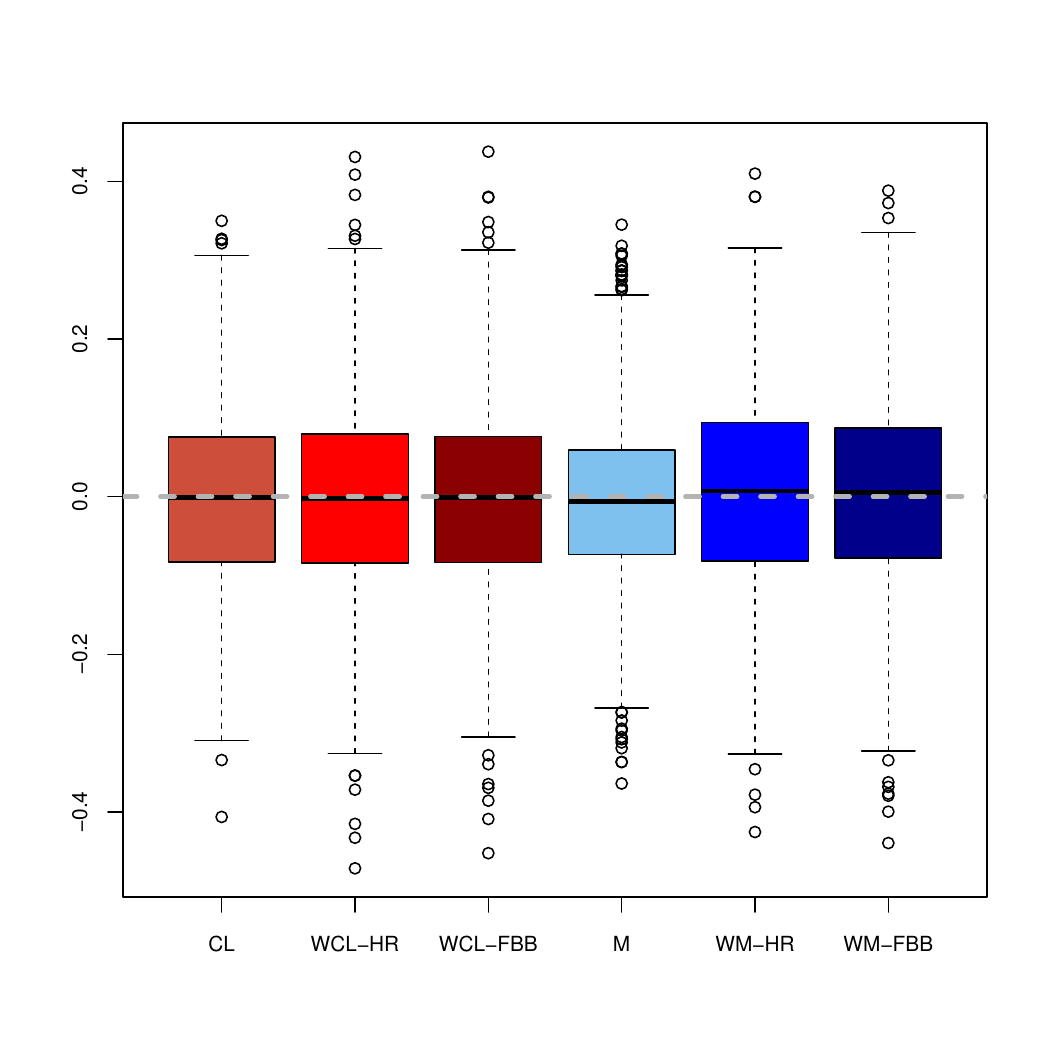}
&  \includegraphics[scale=0.35]{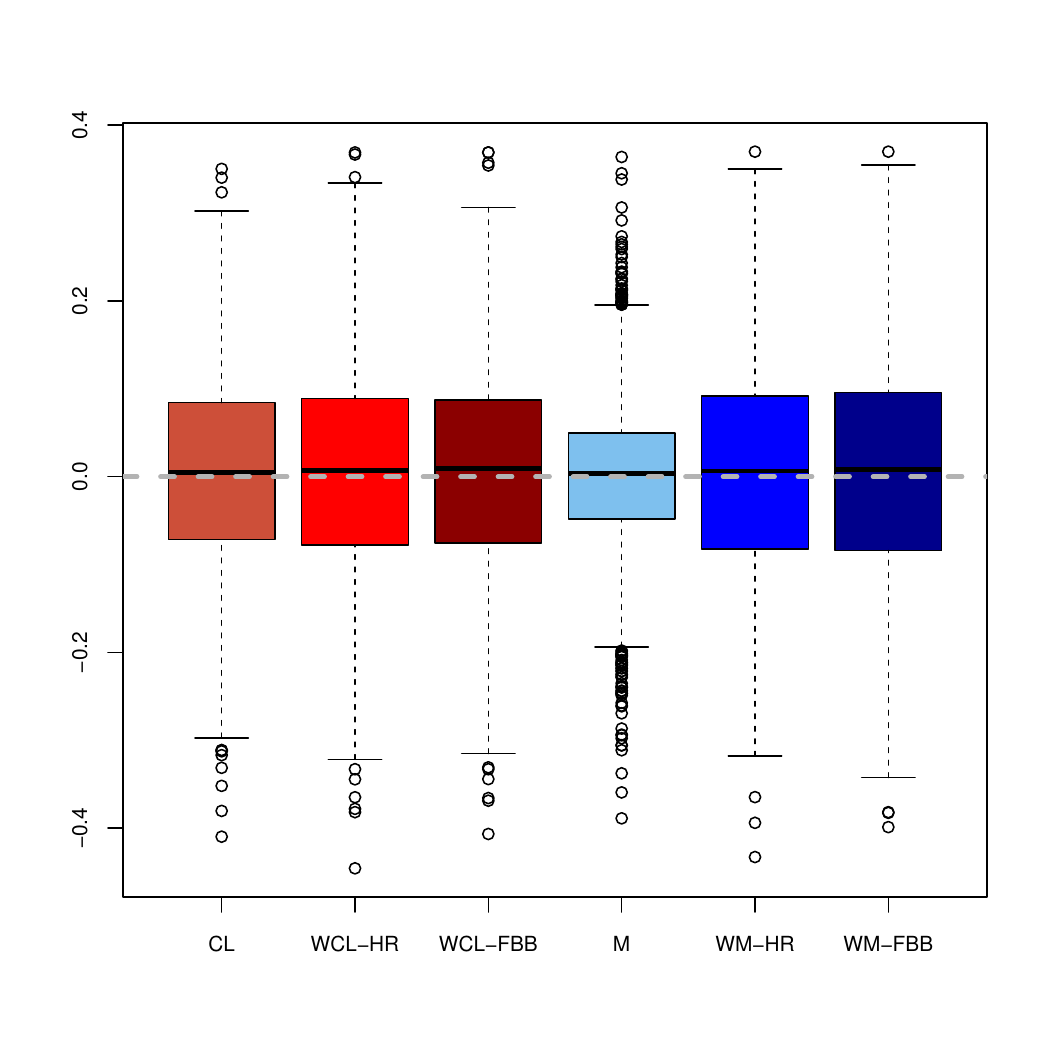} 

 \end{tabular}
\vskip-0.2in
\caption{\small \label{fig:boxplotsalpha-C3}  Boxplot of the estimators for $\alpha_0$ under $C_{3,0.05}$ to $C_{5,0.10}$.}
\end{center} 
\end{figure}

As expected, for clean samples, all the methods perform similarly. It should be noted that the $M$ and weighted $M-$estimators show slightly lower biases than the classical procedure based on the deviance, but their standard deviations are larger due to the loss of efficiency.  
For all the considered contamination schemes, the obtained results reflect the stability of the weighted estimators. In contrast, the classical procedure  and also the $M-$estimator are affected by some of the these schemes.  When considering the   performance across contaminations, we observe that $C_{2,\epsilon}$ is the one with the larger effect on the bias of the classical estimator. This contamination, as well as $C_{3,\epsilon}$, is also damaging for the $M-$estimator due to the presence of extreme high leverage outliers.

Regarding the performance of the weighted estimators under $C_{1,\epsilon}$ to  $C_{5,\epsilon}$ the weighted proposal which downweights observations with large robust Mahalanobis distance provide the best results especially when looking at the bias. Note that under some contaminating schemes, the obtained biases for     ${\wemeBOXnorm}$  are considerably larger than those of  ${\wemeHRnorm}$, for instance under $C_{3,0.10}$ it is more than 10 times larger, while their standard deviations are comparable. The same behaviour arises when comparing the biases of ${\wclBOXnorm}$  and ${\wclHRnorm}$.

To evaluate the performance of the different estimators of $\beta_0$, as in \citet{qingguo:2015} and \citet{boente:salibian:vena:2020}, one possibility is to consider numerical approximations of their 
integrated squared bias and mean integrated squared error,
computed on a grid of equally spaced points on $[0, 1]$. 
However, as mentioned in \citet{he:shi:1998}, these measures   may be influenced by numerical errors at or near the boundaries of the grid. For that reason, we only report here trimmed versions of the above summaries computed without the $q$ first and last points of the grid. More specifically,  if $\wbeta_j$ is the estimate of the function $\beta_0$ obtained with the $j$-th sample ($1 \le j \le n_R$) and  $t_1\le \dots \le t_M$ are equispaced points on $[0,1]$,  we evaluated
 \begin{eqnarray*}
 \mbox{Bias}_{\trim}^2(\wbeta)  & = &  \frac 1{M-2q} \sum_{s=q+1}^{M-q}
 \left( \frac1{n_R} \sum_{j=1}^{n_R} \wbeta_j(t_s) - \beta_0 (t_s) \right )^2 \, ,
 \\
 \mbox{MISE}_{\trim}(\wbeta)  & = &  \frac 1{M-2q} \sum_{s=q+1}^{M-q}  
 \frac1{n_R} \sum_{j=1}^{n_R} \left( \wbeta_j(t_s) - \beta_0 (t_s) \right )^2 \, .
 \end{eqnarray*}
We chose $M = 100$ and $q=[M\times 0.05]$, which uses the central 90\% interior points in the grid. 
 
We also considered another summary measure which aims to evaluate the estimator predictive capability. With that purpose, denote  $(\walfa_j,\wbeta_j)$ the estimator obtained at replication $j$, $1\le j\le n_R$. At the $j-$th replication, we also generated,   independently from the   sample used to compute the estimator,    a new sample   $\itT_j =\{ (y_{i,\itT_j}, X_{i,\itT_j})\}_{i = 1}^{n}$  distributed as  \textbf{C0}. We then computed the probability mean squared errors   defined  as  
$$
\text{PMSE}  (\walfa_,\wbeta)   = \frac{1}{n_R}\sum_{j=1}^{n_R} \frac1{n_{\itT_j}} \sum_{i = 1}^{n} \left\{F\left(\langle X_{i,\itT_j} ,  \beta_0\rangle +\alpha_0\right) - F\left(\langle X_{i,\itT_j} ,  \wbeta_j\rangle +\walfa_j\right)\right\}^2\,.
$$
 Table \ref{tab:tabla-PMSE-C0-C5}   presents the results obtained for the  $\text{PMSE}$, while Table \ref{tab:tabla2-poda5-C0-C5} reports  the obtained summary measures $\mbox{Bias}^2_{\trim}(\wbeta)$ and $\mbox{MISE}_{\trim}(\wbeta)$.

 \begin{table}[ht!]
  \centering
  \small
   \renewcommand{\arraystretch}{1.2}
    \setlength{\tabcolsep}{4pt}
\begin{tabular}{  c |c|  c    | c   |c    |c   |c    |c    | c   |c    |c   |c    |}
\hline  
 & $C_0$ & {$C_{1,0.05}$} & {$C_{1,0.10}$} & {$C_{2,0.05}$} & {$C_{2,0.10}$} & {$C_{3,0.05}$}& {$C_{3,0.10}$} & {$C_{4,0.05}$}& {$C_{4,0.10}$} & {$C_{5,0.05}$}& {$C_{5,0.10}$}\\ 
  \hline  
${\clasnorm}$ & 0.0039 &   0.0155 & 0.0251  &   0.0223 & 0.0265 &   0.0118 &  0.0130 &  0.0118 & 0.0129 & 0.0263 & 0.0296 \\ 
${\emenorm}$ & 0.0043 &   0.0151 & 0.0244 &   0.0222 & 0.0256 & 0.0105 & 0.0113 &  0.0106 &  0.0113 & 0.0242 &  0.0266 \\ 
${\wclHRnorm}$ & 0.0042 &     0.0042 &  0.0040 & 0.0042 & 0.0041 &  0.0042 &  0.0040 &  0.0042  & 0.0040 &  0.0041 &  0.0041 \\ 
${\wemeHRnorm}$ & 0.0042 &    0.0042 &  0.0041 &   0.0041  & 0.0042 &  0.0042 & 0.0041 & 0.0042 & 0.0041 &  0.0042 & 0.0040\\  
${\wclBOXnorm}$  & 0.0039 &    0.0055 &  0.0115  &  0.0051 & 0.0060 & 0.0108 &  0.0130 &  0.0056 &  0.0108  & 0.0039 &  0.0039\\  
${\wemeBOXnorm}$ & 0.0039 &    0.0053 &  0.0109 &  0.0053 & 0.0060  & 0.0104 & 0.0120 & 0.0056 &  0.0102 & 0.0038 & 0.0039\\ 
  \hline 
\end{tabular}
\caption{ \small \label{tab:tabla-PMSE-C0-C5} 
  Probability mean squared errors, PMSE,   over $n_R = 1000$ for clean and contaminated samples of size $n = 300$. }
\end{table}

 \begin{table}[ht!]
  \centering
  \small
   \renewcommand{\arraystretch}{1.2}
   \setlength{\tabcolsep}{4pt}
\begin{tabular}{  c |  cc   | cc  |cc   |cc   |cc   |}
\hline  
  & $\mbox{Bias}_{\trim}^2 $ & $\mbox{MISE}_{\trim}  $   & $\mbox{Bias}_{\trim}^2 $ & $\mbox{MISE}_{\trim}  $ & $\mbox{Bias}_{\trim}^2 $ & $\mbox{MISE}_{\trim}  $ & $\mbox{Bias}_{\trim}^2 $ & $\mbox{MISE}_{\trim}  $ & $\mbox{Bias}_{\trim}^2 $ & $\mbox{MISE}_{\trim}  $  \\ 
  \hline
  &\multicolumn{2}{c|}{ } & \multicolumn{2}{c|}{ } & \multicolumn{2}{c|}{$C_0$} &\multicolumn{2}{c|}{ } &\multicolumn{2}{c|}{} \\
\hline
${\clasnorm}$ &  & &  & &0.0029 & 0.3305 &  &  &  & \\ 
${\emenorm}$ &  & &  & & 0.0024 & 0.3259 & & & &\\ 
${\wclHRnorm}$ &  & &  & & 0.0019 & 0.3795 &   & & & \\ 
${\wemeHRnorm}$ &  & &  & &   0.0021 & 0.3510 &  & & &\\  
${\wclBOXnorm}$  &  & &  & & 0.0029 & 0.3305 &    & & &\\  
${\wemeBOXnorm}$ &  & &  & &  0.0030 & 0.3257 &    & & &  \\ 
  \hline 
&\multicolumn{2}{c|}{$C_{1,0.05}$} &\multicolumn{2}{c|}{$C_{2,0.05}$}&\multicolumn{2}{c|}{$C_{3,0.05}$}&\multicolumn{2}{c|}{$C_{4,0.05}$}&\multicolumn{2}{c|}{$C_{5,0.05}$}\\ 
  \hline  
${\clasnorm}$ & 0.5726 & 0.7706 &    1.2394 & 1.5669  &  0.1416 & 0.4412  & 0.1435 & 0.4246  & 0.9499 & 1.0056 \\ 
${\emenorm}$ & 0.5222 & 0.7338 &  1.0635 & 1.3688 &  0.1290 & 0.3230  &   0.1300 & 0.3323 & 0.8710 & 0.8908  \\ 
${\wclHRnorm}$ &   0.0019 & 0.3566 &  0.0020 & 0.3523 &  0.0020 & 0.3535  &   0.0020 & 0.3444 &  0.0025 & 0.3534  \\ 
${\wemeHRnorm}$ &  0.0016 & 0.3350 & 0.0014 & 0.3353  &   0.0015 & 0.3379 &  0.0016 & 0.3382 & 0.0018 & 0.3605 \\  
${\wclBOXnorm}$  &   0.0749 & 0.4007 & 0.0119 & 0.4016 &   0.1069 & 0.4258  & 0.0076 & 0.3386 &  0.0028 & 0.3190  \\  
${\wemeBOXnorm}$ &  0.0517 & 0.3878 &  0.0105 & 0.4046 &   0.1028 & 0.4173 &   0.0086 & 0.3452 &   0.0019 & 0.3179 \\ 
  \hline  
&\multicolumn{2}{c|}{$C_{1,0.10}$} &\multicolumn{2}{c|}{$C_{2,0.10}$}&\multicolumn{2}{c|}{$C_{3,0.10}$}&\multicolumn{2}{c|}{$C_{4,0.10}$}&\multicolumn{2}{c|}{$C_{5,0.10}$}\\ 
\hline
${\clasnorm}$ &    1.0181 & 1.1032   &   1.5130 & 1.8423 & 0.1620 & 0.4570  & 0.1659 & 0.4436  & 1.0601 & 1.0953  \\ 
${\emenorm}$ &   0.9764 & 1.0436  &   1.2632 & 1.5716 &  0.1440 & 0.3264  &  0.1414 & 0.3118 &  0.9227 & 0.9385 \\ 
${\wclHRnorm}$ &     0.0018 & 0.3484   & 0.0017 & 0.3497 &  0.0018 & 0.3475  & 0.0025 & 0.3344 & 0.0016 & 0.3381  \\ 
${\wemeHRnorm}$ &   0.0033 & 0.3354 &  0.0023 & 0.3440  &  0.0023 & 0.3385  & 0.0025 & 0.3272 & 0.0017 & 0.3274 \\  
${\wclBOXnorm}$  &      0.4184 & 0.6609 &  0.0270 & 0.4626 & 0.1617 & 0.4568 &   0.0997 & 0.4133  &  0.0016 & 0.3188\\  
${\wemeBOXnorm}$ &     0.3615 & 0.6406  &  0.0252 & 0.4493   & 0.1441 & 0.4293 &  0.0942 & 0.4012  & 0.0016 & 0.3099\\ 
  \hline    
\end{tabular}
\caption{ \small \label{tab:tabla2-poda5-C0-C5} 
   Trimmed version of the integrated squared bias  and mean integrated squared errors  for the estimators of $\beta_0$, over $n_R = 1000$ for clean and contaminated samples of size $n = 300$. }
\end{table}

In order to  visually  explore the performance of these estimators, 
Figures \ref{fig:wbeta-C0-poda0} to  \ref{fig:wbeta-C510-poda0}  contain functional 
 boxplots, see  \citet{sun:genton:2011}, for the $n_R = 1000$ realizations of the different estimators
 for $\beta_0$  under the
 contamination settings. As in standard boxplots, the magenta central box
 of these functional boxplots represents the 50\% inner band of curves, the solid black line indicates the 
 central (deepest) function and the dotted red lines indicate outlying 
 curves, that is, outlying estimates $\wbeta_j$ for some $1 \le j \le n_R$. We also indicate in blue lines the whiskers delimiting the non--outlying curves and  the  true  function  $\beta_0$  with 
 a dark green   line. 
To avoid boundary effects, we show in Figures \ref{fig:wbeta-C0-poda5}  to  \ref{fig:wbeta-C510-poda5}   the different estimates  
evaluated on the interior points of a grid of 100 equispaced points. 
In addition, to facilitate comparisons between contamination cases and estimation methods, 
the scales of the vertical axes are the same for all  Figures.

	In clean samples, all the estimators give similar results with slightly smaller values of the $\mbox{Bias}_{\trim}^2$ when considering the weighted estimators. It is worth mentioning   that the weighted $M-$estimators are remarkably efficient.  Furthermore, when considering the PMSE, the weighted estimators giving weight 0 to the observations with covariates detected as atypical by the functional boxplot, give results similar to those of the classical estimator. This behaviour may be explained by the fact than in most generated samples, under $C_0$, no atypical curves are detected among the covariates.

As expected, when misclassified observations are present, the procedure based on the deviance breaks--down, in particular when these responses are combined with extreme high leverage covariates as is the case for $C_{2,\epsilon}$.	Contamination schemes $C_{3,\epsilon}$ and  $C_{4,\epsilon}$ have less effect than $C_{1,\epsilon}$ on the integrated squared bias and mean integrated squared error of the classical estimator of the slope. This fact is also illustrated in Figures \ref{fig:wbeta-C35-poda0} to \ref{fig:wbeta-C410-poda0} and \ref{fig:wbeta-C35-poda5} to \ref{fig:wbeta-C410-poda5}, where the true curve is still in the band containing the 50\% deepest estimates. In contrast, under $C_{1,\epsilon}$, $C_{2,\epsilon}$ and particularly under $C_{5,\epsilon}$, the plot of the true function is beyond the limits of the functional boxplot, meaning that the obtained estimates become completely uninformative and do not reflect the shape of $\beta_0$, see Figures  \ref{fig:wbeta-C15-poda0} to  \ref{fig:wbeta-C210-poda0},  \ref{fig:wbeta-C55-poda0} and \ref{fig:wbeta-C510-poda0} for instance.

It is interesting to note that the   the unweighted $M$-estimator shows a  performance similar to that of the classical one, for the considered contamination schemes.	In contrast, weighted estimators give very good results in all the studied contamination settings, especially when considering ${\wclHRnorm}$ and ${\wemeHRnorm}$, which clearly outperform  ${\wclBOXnorm}$  and ${\wemeBOXnorm}$ in most cases. It is worth mentioning that the weighted estimators based on the deviance   are quite stable across contaminations, the only exception being $C_{1,0.10}$ where  ${\wclBOXnorm}$ is more sensitive. In this last case, as when considering the  ${\wemeBOXnorm}$ estimates, the true curve lies above the magenta central region for values of $t$ larger than 0.8 (see Figures \ref{fig:wbeta-C110-poda0}  and \ref{fig:wbeta-C110-poda5}), but is still included in the band limited by the blue whiskers.

In most cases, the weighted $M-$estimators improve on the weighted estimators obtained when $\rho(t)=t$ and the procedures with weights based on the Mahalanobis distance of the projected covariates outperform those whose weights are based on the functional boxplot.
In particular, contamination schemes $C_{1,\epsilon}$ and   $C_{3,\epsilon}$, affect more the PMSE of the the latter procedure than that of the former one. Note that, under  $C_{1,0.10}$, $C_{3,0.10}$ and also under $C_{4,0.10}$ the PMSE of ${\wclBOXnorm}$ are twice those obtained with ${\wclHRnorm}$ and similarly when comparing ${\wemeBOXnorm}$ with ${\wemeHRnorm}$ (see Table \ref{tab:tabla-PMSE-C0-C5}).

In summary, the ${\wclHRnorm}$ and specially the   ${\wemeHRnorm}$ estimators display a remarkably stable behaviour across the selected contaminations. Both estimators are comparable, but ${\wemeHRnorm}$ attains in general lower values of the  $\mbox{MISE}_{\trim}$. The good performance of   ${\wclHRnorm}$ may be explained by the fact that the hard--rejection weights mainly discard from the sample the observations with high leverage covariates, which in most cases correspond to missclassified ones. This behaviour was already noticed by \citet{Croux:H:2003} in the finite--dimensional setting.

 \include{graficos-poda0-sinC6}

 \include{graficos-poda5-sinC6}

\section{Real data example} \label{sec:ejemplos}
 
%%%%%%%%%%%%%%%%%%%%%%%%%%%%%%%%%%%%%%%%%%%%%%%%%
% CON HDS PARA DETECTAR OUTLIERS
%%%%%%%%%%%%%%%%%%%%%%%%%%%%%%%%%%%%%%%%%%%%%%%%%

 In this section, we consider   the electricity prices data set analysed in \citet{liebl:2013} in the context of electricity
price forecasting and we investigate the performance of the proposed estimators.

The data consist of hourly electricity prices in Germany between 1 January 2006 and
30 September 2008, as traded at the Leipzig European Energy Exchange, German electricity demand (as reported by the
European Network of Transmission System Operators for Electricity). These data  were also used in \citet{boente:salibian:vena:2020} who modelled the daily average hourly energy demand through a semi-functional linear regression model  using as euclidean covariate the
mean hourly amount of wind-generated electricity in the system for that day and as functional covariate  the curve of energy prices  observed hourly. 

In our analysis, the response measures high or low demand of electricity, that is, we define the binary variable $y=1$ if the average hourly demand exceeds 55000 (\lq\lq High Demand\rq\rq) and $0$ (\lq\lq Low Demand\rq\rq) otherwise. Furthermore, the functional covariates used to predict the conditional probability that a day has high demand, correspond to the curves of energy prices. These curves are observed hourly originating a matrix of dimension $638 \times 24$, after removing weekends, holidays and other non-working days. 
Hence, with these data we fit a functional logistic regression model, 
$$ \prob(y = 1|X)  = F\left(\alpha_0 + \left\langle X, \beta_0 \right \rangle\right)\,,$$
using the weighted robust estimators defined in this paper and their classical alternatives. 
The trajectories corresponding to $y=1$ and $y=0$ are given in Figure \ref{fig:maplotX-HDS}.

\begin{figure}[ht!]
 \begin{center}
    \newcolumntype{G}{>{\centering\arraybackslash}m{\dimexpr.5\linewidth-1\tabcolsep}}
    \begin{tabular}{GG}
			 Low demand ($y=0$) & High demand ($y=1$) \\[-3ex]   			
			\includegraphics[scale=0.4]{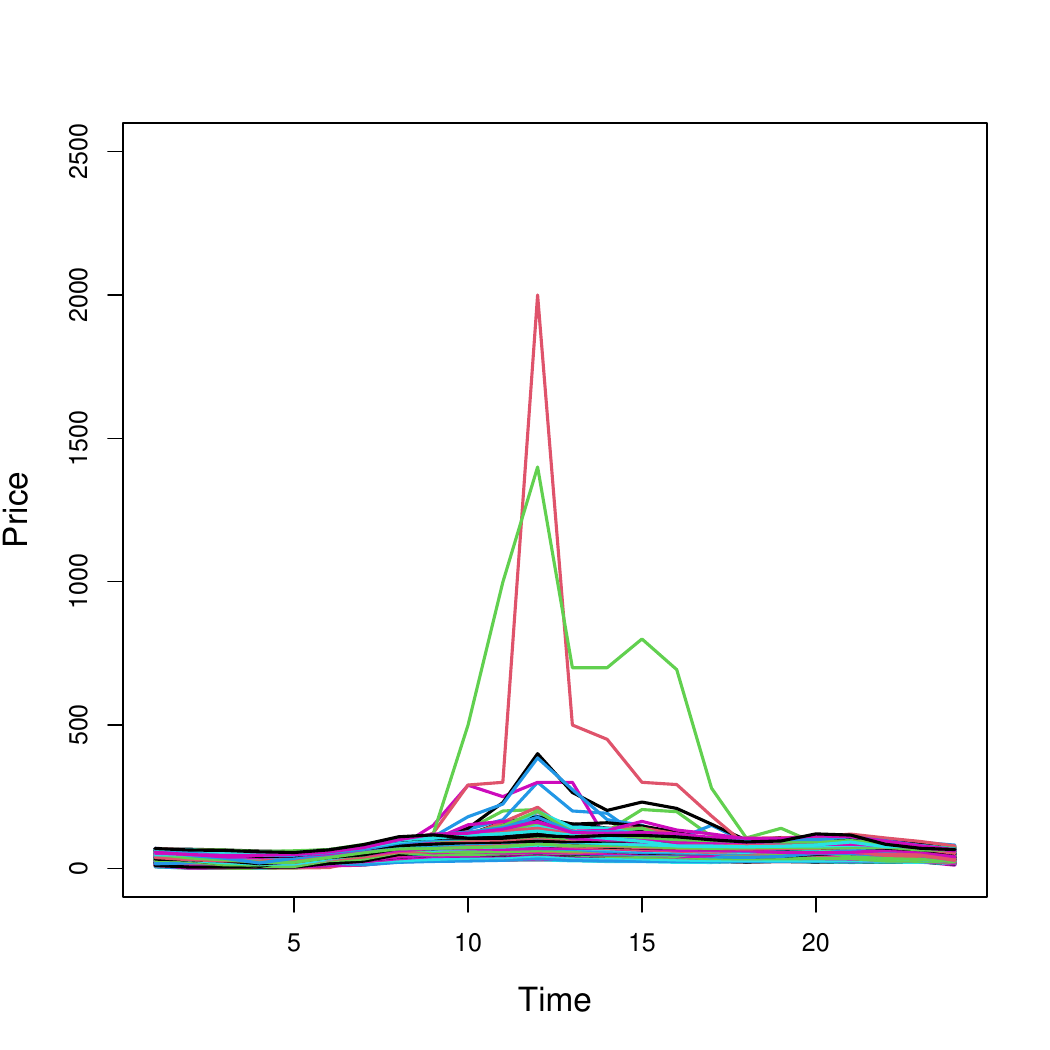}
			&   \includegraphics[scale=0.4]{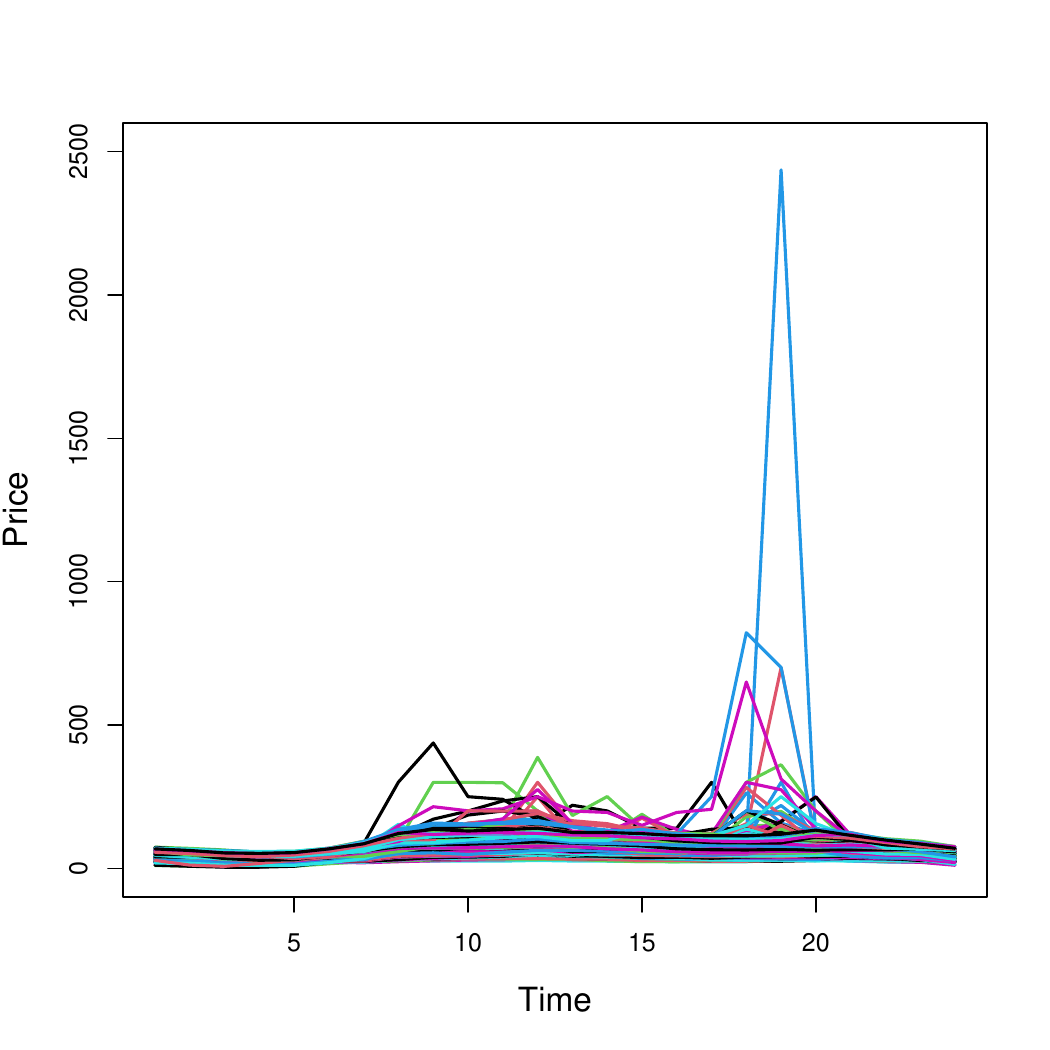} 
\end{tabular}
\caption{\small \label{fig:maplotX-HDS} Energy prices for the  German electricity data.}
\end{center} 
\end{figure}  

 To compute the estimators, we used $B-$splines to generate the space of possible candidates and as above we denote $\{B_j\}_{j=1}^{k_n}$ the $k_n-$dimensional basis. As in the simulation study, after dimension reduction, we computed the finite--dimensional coefficients through \eqref{eq:FWBY} using as loss $\rho(t)=t$ that leads to the deviance based estimators denoted ${\clasnorm}$ and the $\rho$ function introduced in \citet{Croux:H:2003} with tuning constant $0.5$ denoted $\emenorm$. We also computed the weighted versions of the previous estimators choosing as weights the  weights based on the robust Mahalanobis distance of the projected data, that is,  we evaluated the Donoho--Stahel location and scatter estimators, denoted $\wbmu$ and $\wbSi$, respectively, of the sample $\bx_1, \dots, \bx_n$ with $\bx_i=(x_{i1}, \dots, x_{i k_n})\trasp $, $x_{ij}=\langle X_i, B_j\rangle$ and we defined $w(X_i)=1$ when $ d_i^2= ( \bx_i - \wbmu )\trasp \wbSi^{-1}( \bx_i - \wbmu )$ is less than or equal to $\chi_{0.975,k_n}$ and $0$ otherwise.  For simplicity, in Table \ref{tab:german-alfas-HDS} below, these procedures are labelled ${\wclnorm}$ and  ${\wemenorm}$, respectively.   
As in the simulation, to select the basis dimension we use the criterion defined in Section \ref{sec:BIC}, that is,  we minimize the quantity $RBIC(k)$ defined in  \eqref{eq:bic1}. All procedures  choose $k_n=7$, except the $M-$estimator that selects $k_n=5$.

\begin{table}[ht]
\centering
\begin{tabular}{|cccc|}
  \hline
${\clasnorm}$ & $\emenorm$ & ${\wclnorm}$ &  ${\wemenorm}$ \\
\hline
-1.564 &  -1.588 & -1.693 & -1.811  \\
\hline
\end{tabular}
\caption{\small \label{tab:german-alfas-HDS} Estimates of $\alpha_0$.} 
\end{table} 

The estimates of $\alpha_0$ are reported in Table \ref{tab:german-alfas-HDS}, while those for $\beta_0$ are shown in Figure \ref{fig:german-betas-HDS}. In solid black and gray lines
we represent the ${\wemenorm}$ and $M-$estimators,  respectively, while red solid and dashed lines correspond to the ${\clasnorm}$ and ${\wclnorm}$  procedures. Comparing the obtained results, we note that the weighted estimators lead to slightly larger absolute values of the intercept. Regarding the estimation of $\beta_0$, it is clear that the $M-$estimator produces a smoother curve since the basis dimension is smaller. The classical procedure is almost equal to $0$ between 9am and 4pm, while all the weighted estimators detect a positive \lq\lq peak\rq\rq around 1pm and two  slumps near 8am and 17pm. All procedures detect the large peaks close 4am and 9pm, where prices seem to have a larger (in magnitude) association when predicting the demand, but for the classical procedure  the magnitude of the function $\wbeta$ is smaller
than that of the weighted ones in that range.

\begin{figure}[ht!]
\centering
\includegraphics[scale=0.4]{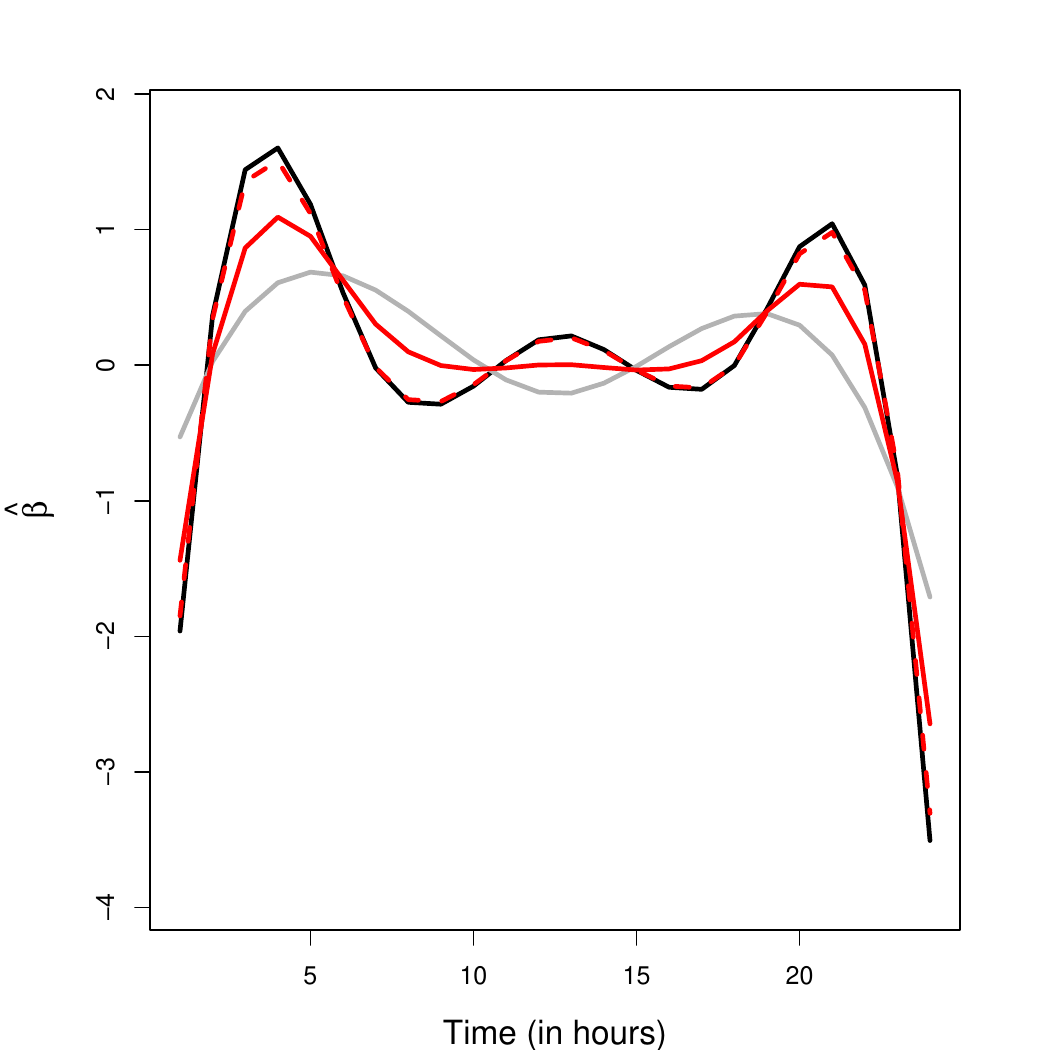}  
\caption{\small \label{fig:german-betas-HDS}  German Electricity: Estimates of $\beta_0$. Solid black and gray lines
are used for the weighted-$M$ and $M-$estimators, respectively, while red solid and dashed lines correspond to the classical   procedure and its weighted counterpart.}
\end{figure}

To identify potential atypical observations, we used the deviance QQ-plot defined in \citet{garciaben:yohai:2004}. For that purpose, we considered the deviance residuals based on the ${\wemenorm}-$estimator. More precisely, if  $(\walfa, \wbeta)$ stands for the weighted $M-$estimator, we compute the predicted probabilities and deviances $\wpe_i=F\left(\walfa+ \left\langle X_i, \wbeta\right \rangle\right)$ and  
$d_i = \text{sign}(y_i - \wpe_i)\sqrt{-2\left\{(1-y_i)\log(1-\wpe_i)+ y \log(\wpe_i))\right\}}\,.$ 
Let $\wF_D$ be the estimator of their distribution function, given by  
$$
\wF_D(d)=\frac{1}{n}\left\{\sum_{\wpe_i \in \itA} \wpe_i+\sum_{1-\wpe_i \in \itA}\left(1-\wpe_i\right)\right\}\,,
$$
where $\itA=\left\{t:(-2 \log t)^{1 / 2} \leq d\right\}$.  We consider outliers those observations with deviance residual $d_i$  smaller than $\wF_D^{-1}(0.005)$ or larger than $\wF_D^{-1}(0.995)$, leading to $19$ observations detected as possible atypical observations.  
These observations, together with the value of their residual deviance and predicted probability are reported in Table \ref{tab:resdeviance-HDS}. 

%% latex table generated in R 4.1.1 by xtable 1.8-4 package
% Sun Jun 11 20:08:04 2023
\begin{table}[ht]
\centering
\small
\begin{tabular}{|c|ccc|c|c|ccc|}
  \hline
   & Date & $d_i$ &  $\wpe_i$ & &  & Date & $d_i$ &  $\wpe_i$ \\ 
  \hline
  $i$ & \multicolumn{3}{c|}{$y_i=0$} & & $i$ & \multicolumn{3}{c|}{$y_i=1$}\\
  \hline
  134 & 2006-07-25 & -5.686 & 1.000 & &  66 & 2006-04-10 & 2.344 & 0.064\\ 
  136 & 2006-07-27 & -4.972 & 1.000 & &  75 & 2006-04-25 & 2.905 & 0.015 \\  
  137 & 2006-07-28 & -2.667 & 0.971 & &  76 & 2006-04-26 & 2.511 & 0.043 \\ 
  469 & 2008-01-10 & -2.493 & 0.955 & &  77 & 2006-04-27 & 2.313 & 0.069 \\ 
  502 & 2008-03-03 & -2.412 & 0.945 & & 445 & 2007-11-20 & 3.208 & 0.006 \\ 
  508 & 2008-03-11 & -2.282 & 0.926 & & 543 & 2008-05-06 & 2.393 & 0.057 \\  
  519 & 2008-03-28 & -3.052 & 0.990 & & 550 & 2008-05-16 & 2.435 & 0.052 \\  
  621 & 2008-09-05 & -2.512 & 0.957 & & 580 & 2008-07-01 & 2.196 & 0.090 \\ 
  622 & 2008-09-08 & -2.497 & 0.956 & & & & &\\ 
  627 & 2008-09-15 & -2.385 & 0.942 & & & & &\\ 
  638 & 2008-09-30 & -2.923 & 0.986 & & & & &\\   
   \hline
\end{tabular}
\caption{\small \label{tab:resdeviance-HDS} Outliers identified by weighted $M-$estimator  with weights based on the Mahalanobis distance of the projected data and a hard rejection weight function.} 
\end{table}

It is worth mentioning, that the observation corresponding to November 7 of 2006, which is an outlier in the covariate space is not detected as such by the deviance QQ-plot, since it does not correspond to a bad leverage point. 

\begin{figure}[ht!]
 \begin{center}
    \newcolumntype{G}{>{\centering\arraybackslash}m{\dimexpr.5\linewidth-1\tabcolsep}}
    \begin{tabular}{GG}
			 Low demand & High demand \\[-3ex]   			
			\includegraphics[scale=0.4]{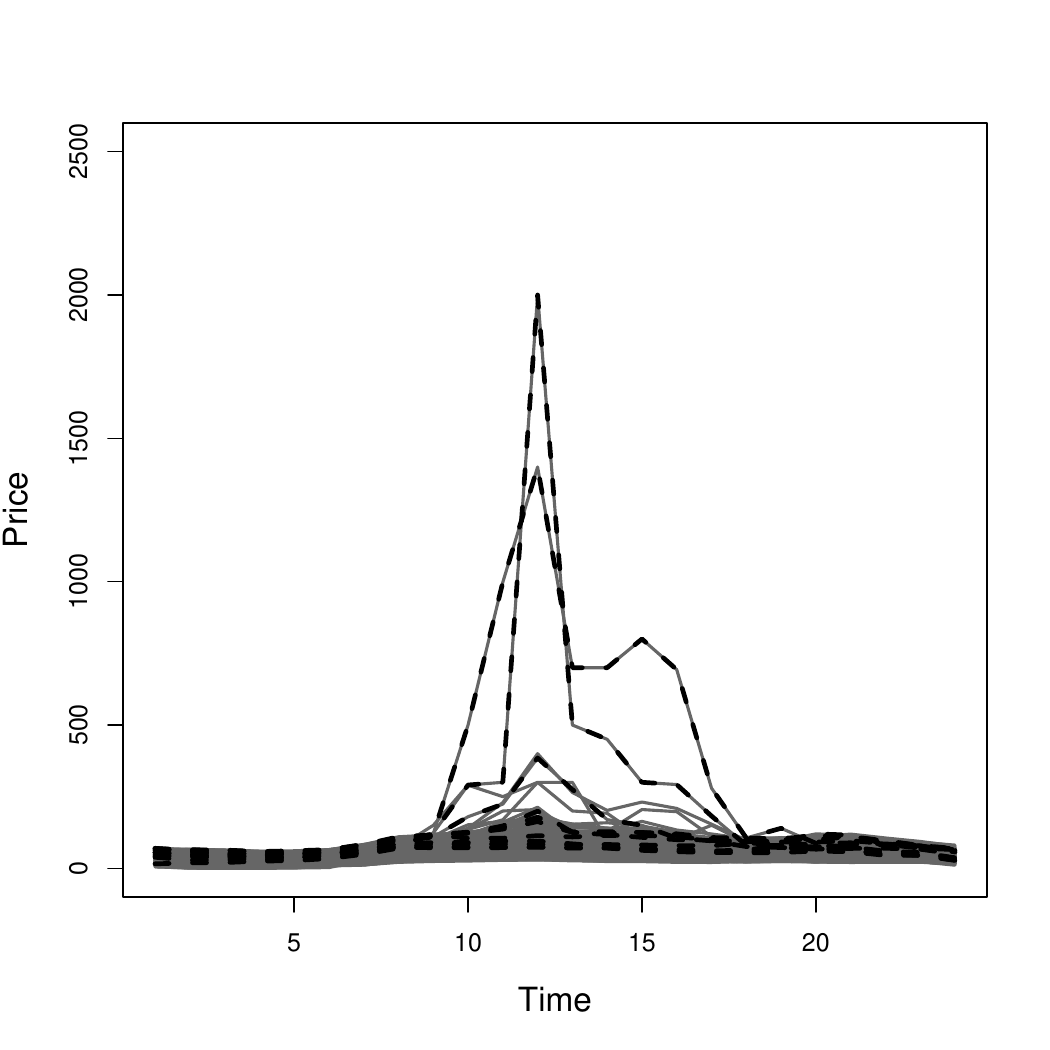}
			&   \includegraphics[scale=0.4]{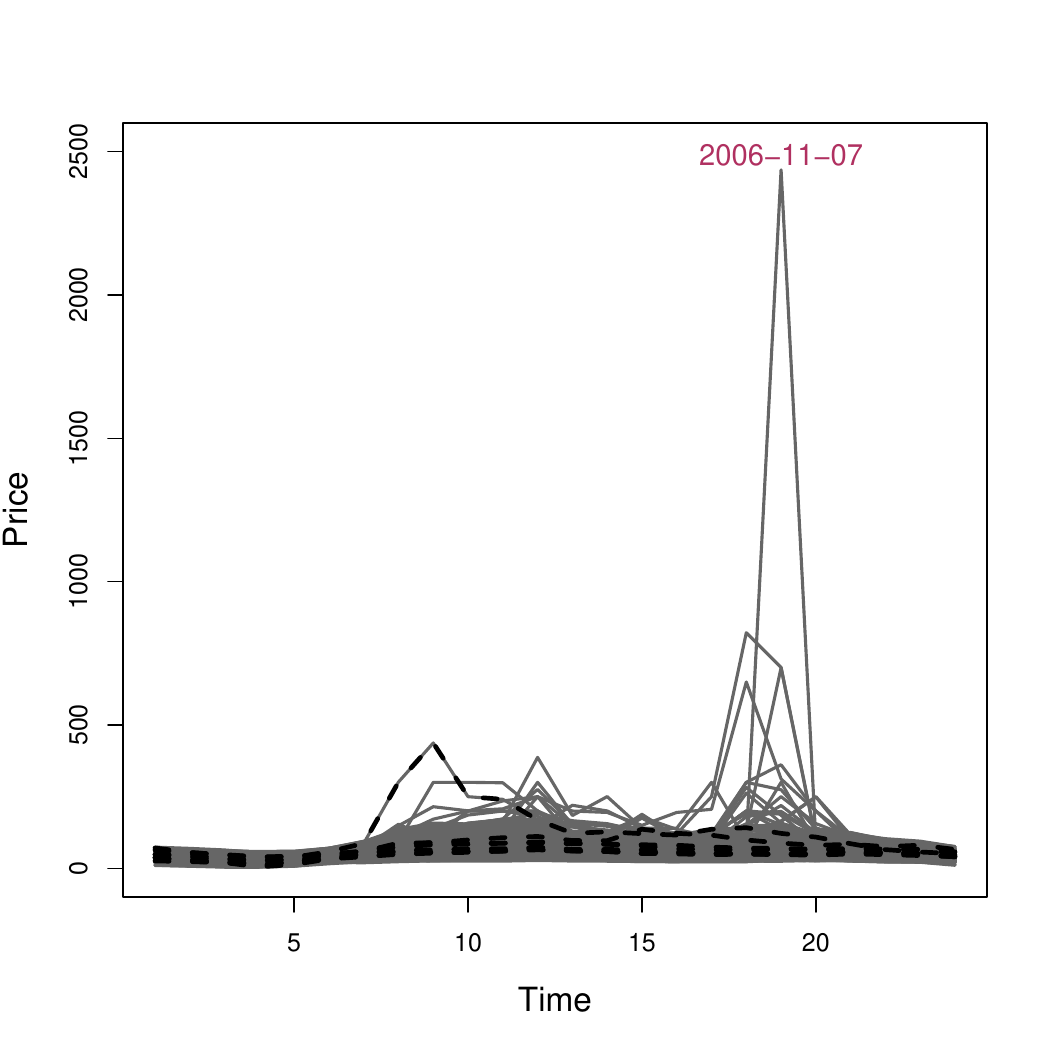} 
\end{tabular}
\vskip-0.2in
\caption{\small \label{fig:maplotX-outs-HDS} Energy prices for the  German electricity data. The covariates related to the detected atypical data are plotted in dashed black lines.}
\end{center} 
\end{figure}  

The covariates related to the detected atypical observations are presented in Figure \ref{fig:maplotX-outs-HDS} with dashed black lines. The left panel correspond to the covariates corresponding to $y=0$ and the right one to $y=1$, where we also indicate the curve corresponding to the 7th of November which  clearly appears as an outlier  in the covariate space.

We next compute the classical estimator after removing the identified outliers. The obtained value for the estimate of $\alpha_0$ is -1.852, a value closer to the one obtained when using the ${\wemenorm}$, as reported in Table \ref{tab:german-alfas-HDS}.  Figure \ref{fig:german-betas-clsout-HDS}  presents the obtained estimator in dotted magenta lines. To facilitate comparisons we present  in the left panel  all the estimators and in the right one, only the weighted $M-$estimate with solid black line and the classical one computed with the whole data set and without the outliers, with a red solid line  and a dotted magenta one. Note the similarity of the curve obtained by the ${\wemenorm}$ method and the estimate obtained when using  the $\clasnorm$ method after removing the atypical observations.

\begin{figure}[ht!]
\centering
\newcolumntype{G}{>{\centering\arraybackslash}m{\dimexpr.5\linewidth-1\tabcolsep}}
    \begin{tabular}{GG} 
   \includegraphics[scale=0.4]{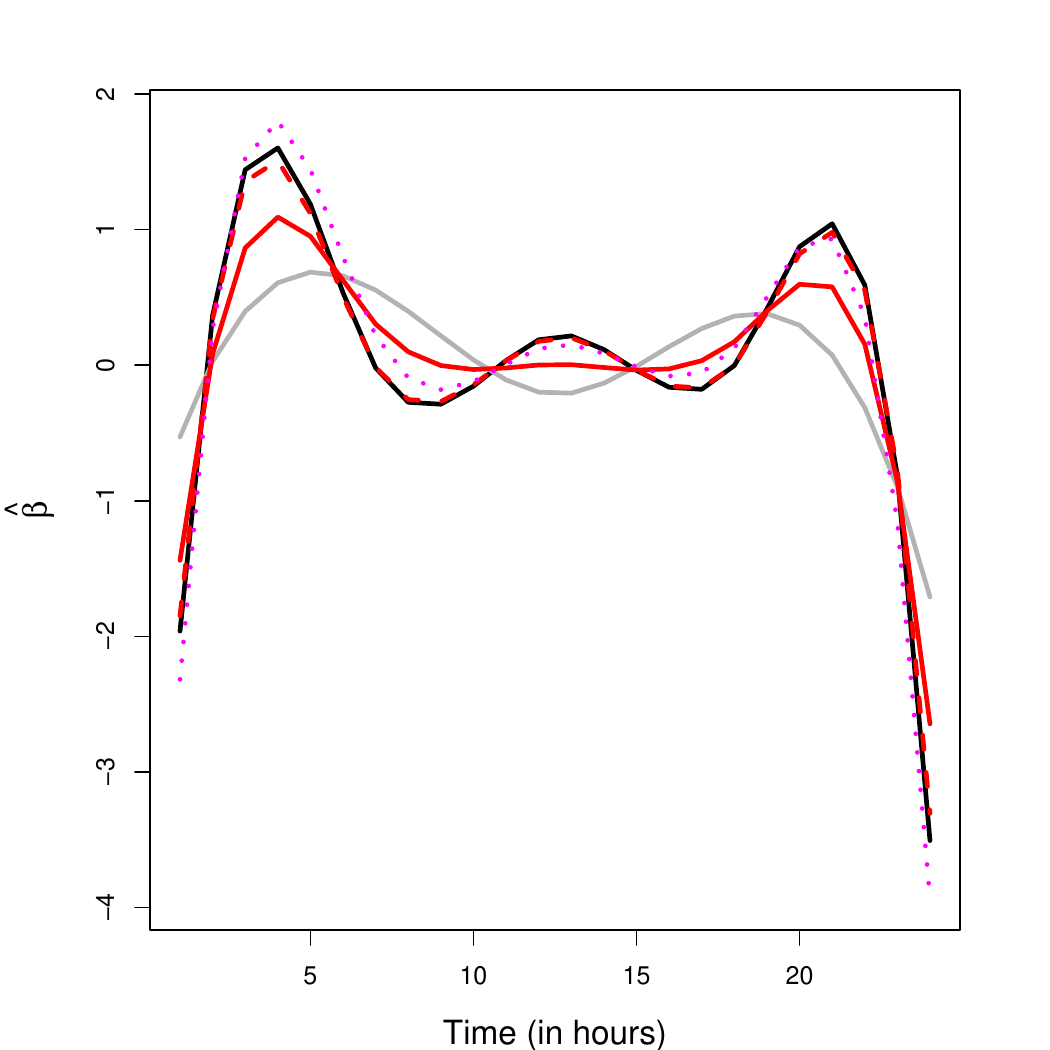} & 
   \includegraphics[scale=0.4]{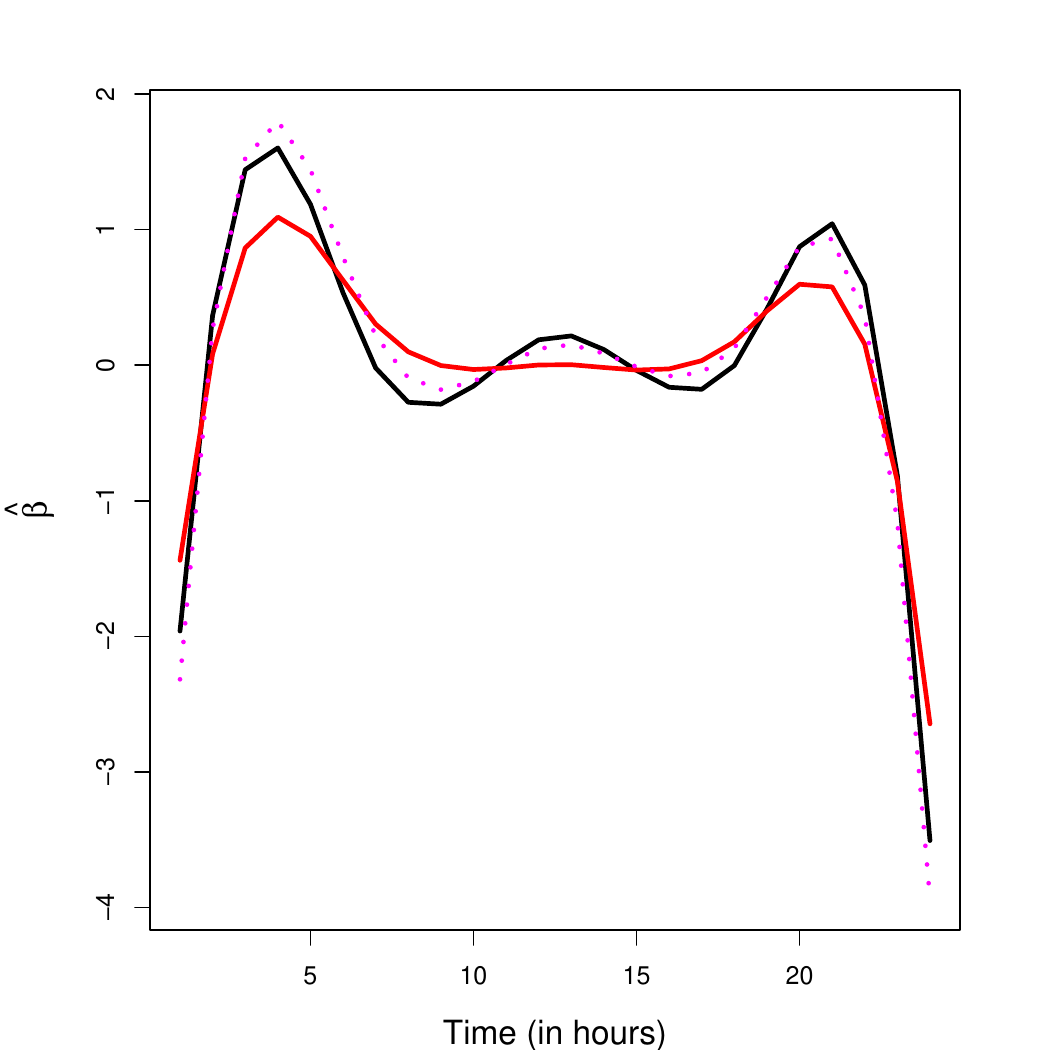}  
\end{tabular}
\caption{\small \label{fig:german-betas-clsout-HDS}  German Electricity: Estimates of $\beta_0$. A solid black and gray lines
are used for the weighted-$M$ and $M-$estimators, respectively, while red solid and dashed lines correspond to the classical   procedure and its weighted counterpart. The classical estimator computed without the detected outliers is given in magenta dotted lines.}
\end{figure}
 
\section{Conclusion}{\label{sec:concl}}
 This paper faces the problem of providing  robust estimators for the functional logistic regression model. We tackle the problem by basis dimension reduction, that is, we consider a finite--dimensional space of candidates and, once  the reduction is made, we use weighted $M-$estimators to down--weight the effect of bad leverage covariates. In this sense, our estimators are robust against outliers corresponding to miss--classified points, and also to outliers in  the functional explanatory variables that have a damaging effect on the estimation. 
We also propose a robust BIC-type criterion to select the optimal dimension of the splines bases that seems to work  well in practice.

Regarding the asymptotic behaviour of our proposal,  we prove that the estimators are strongly
consistent, under a very general framework which allows the practitioner to choose different basis as $B-$splines, Bernstein and Legendre polynomials, Fourier basis or Wavelet ones.   The convergence rates obtained depend on the capability of the basis to approximate the true slope function. It is worth mentioning that unlike the weak consistency results obtained in \citet{kalogridis:2023} for penalized estimators based on divergence measures, our achievements do not require   the covariates to be bounded. %, a condition that may seem unnatural in a robust setting. 
Moreover, strong consistency rather than convergence in probability of the sieve based weighted $M-$estimators is derived. 

The performance  in finite samples is studied by means of a simulation study and in a real data example. The simulation study reveals the good robustness properties of our proposal and the high sensitivity of the procedure based on the deviance when atypical observations are present.  In particular, weighted $M-$estimators show their advantage over the unweighted ones, a fact already observed by \citet{Croux:H:2003} when   finite--dimensional covariates instead of functional ones are used in the model. 

We also apply our method to a real data set  where the weighted $M-$estimator  remains reliable  even in presence of    misclassified   observations corresponding to atypical functional explanatory variables. The deviance residuals obtained from the robust fit provide a natural way to identify potential atypical observations. The classical estimators obtained minimizing the deviance after dimension reduction and computed over the \lq\lq cleaned\rq\rq  ~  data set, that is, the data set without the outliers present a similar shape to those obtained by the robust procedure which automatically down--weights these data.

As it has been extensively discussed, functional data are intrinsically infinite--dimensional and even when recorded at a finite grid of points, they should be considered as random elements of some functional space rather than multivariate observations.  For instance, in some applications, the functional covariates  may be viewed as longitudinal data, while in other cases,   the predictors are densely observed curves. Examples of these situations  are the Mediterranean fruit flies data   studied in \citet{muller:stadtmuller:2005}, for the former, and the Tecator data set studied in \citet{ferraty:vieu:2006} or the Danone data set reported in \citet{aguilera-morillo:etal:2013}, for the latter.
In more challenging cases, the grid at which the predictors are recorded may be sparse  or irregular and the observations may be possibly subject to measurement errors. The proposals considered in  \citet{james:2002} or \citet{muller:2005} for generalized functional regression models allow  for sparse recording and measurement errors. 
%In this paper,  to derive consistency and rates of convergence, we have assumed that the complete covariates trajectory is available. 
Even when our proposal  can  be implemented if the predictors are observed on a dense grid of points,  the asymptotic results derived make use of the whole process structure. The interesting situation of  sparse and irregular time grids which is common, for instance,  in longitudinal studies is beyond the scope of the paper and may be object of future research.
%%%%%%%%%%%%%%%%%%%%%%%%%%%%%%%%%%%
 %ACKNOWLEDGEMENT
%%%%%%%%%%%%%%%%%%%%%%%%%%%%%%
 
\section{Acknowledgements.}
 This research was partially supported by    Universidad de Buenos Aires [Grant  20020170100022\textsc{ba}]  and  \textsc{anpcyt} [Grant  \textsc{pict} 2021-I-A-00260] at  Argentina (Graciela Boente and Marina Valdora) and by the   Ministry of Economy, Industry
and Competitiveness, Spain (MINECO/AEI/FEDER, UE) [Grant  {MTM2016-76969P}] (Graciela Boente).

  \setcounter{section}{0}
\renewcommand{\thesection}{\Alph{section}}

\section{Appendix} {\label{sec:appendix}}
From now on, we denote   $\nu(t)$  the function 
 \begin{equation}\label{eq:funcionnu}
\nu(t)= \psi\left(-\log F(t)\right)\left[1- F(t)\right] + \psi\left(-\log\left[1- F(t)\right]\right) F(t) \,.
\end{equation} 
Under \ref{ass:rho_bounded_derivable} and \ref{ass:rho_derivative_positive}, the function  $\nu(\cdot)$ is continuous,  bounded and strictly positive.  

Denote $\Psi(y,t) = {\partial} \phi(y,t)/{\partial t}= -[y-F(t)] \nu(t)$ where $\nu(t)$ is given by   \eqref{eq:funcionnu}. 
Under \ref{ass:rho_bounded_derivable} and \ref{ass:rho_derivative_positive}, the function  $\Psi(y, \cdot)$ is continuous and  bounded and we will denote   $M_{\phi}=\sup_{y\in \{0,1\}, t\in \real} \phi(y,t)$ and $M_{\Psi}=\sup_{y\in \{0,1\}, t\in \real} |\Psi(y,t)|$,   where we have used that $\phi(y,t)\ge 0$ from assumptions \ref{ass:rho_bounded_derivable}  and \ref{ass:rho_derivative_positive}.

Henceforth, for any $(\alpha, \beta)\in \real\times L^2([0,1])$, we denote $\theta=(\alpha, \beta)$. Besides, $L(\theta)$ stands for the population counterpart of $L_n(\theta)$, that is,
\begin{equation*}
L( \theta) = \esp\left(\phi\left(y, \alpha + \langle X, \beta\rangle \right) w(X) \right)\,,
\label{eq:Ltheta}
\end{equation*}
where $y|X\sim Bi\left(1, F(\alpha_0+\langle X, \beta_0\rangle)\right)$. Recall that we have denoted $\wtheta=(\walfa,\wbeta)$ and $\theta_0=(\alpha_0,\beta_0)$.

\subsection{Some preliminary results}
Lemma \ref{lema:FC} states the Fisher--consistency of the estimators defined through \eqref{eq:FWBY}. It follows using similar arguments as those considered in Theorem S.1.1  from \citet{Bianco:Boente:Chebi:2022}. 
Note that Lemma \ref{lema:FC}(a) only states that the true value $\theta_0$ is one of the minimizers of the function $L( \theta)$, while in part (b) to derive that it is the unique minimizer an additional requirement, \eqref{eq:identifiafuerte}, is needed. As mentioned in Remark \ref{remark:comentarios}, this condition is related to the fact that the slope parameter is not identifiable if the kernel of the covariance operator of $X$ does not reduce to $\{0\}$.  Instead of asking \eqref{eq:identifiafuerte} to hold for any  $\beta \in L^2([0,1])$, we require that the condition holds only for $  \beta \in  \itH^{\star}$ whenever $  \beta_0 \in  \itH^{\star}$, where henceforth,  $ \itH^{\star}= \itH$ when $\beta_0\in \itH$, while $ \itH^{\star}=\itW^{1,\itH}$  when $\beta_0\in\itW^{1,\itH}$. It is worth mentioning that under \ref{ass:funcionw} and \ref{ass:probaX*}, condition \eqref{eq:identifiafuerte} is fulfilled.

\begin{lemma}\label{lema:FC}
Suppose  $(y, X)$ follows  the functional logistic regression model given in  \eqref{eq:FLOGIT}.	
Let $\phi:\real^2\to \real$ be the function given by \eqref{eq:phiBY}, where $\rho$ satisfies conditions \ref{ass:rho_bounded_derivable} and \ref{ass:rho_derivative_positive}    and let $w$ be a non--negative bounded function. 
\begin{itemize}
\item[a)] For any $\theta=(\alpha,\beta)\in \real\times L^2([0,1])$, we have that  $L(\theta_0)\le  L(\theta)$.
\item[b)] Furthermore,  assume that  $  \beta_0 \in  \itH^{\star}$ and
 \begin{equation}
       \prob\left(\langle X, \beta\rangle  = a \cup \{w(X)=0\}\right)<1, \qquad \forall   a \in \real, \beta \in  \itH^{\star}, (a,\beta) \neq 0 \,,
       \label{eq:identifiafuerte}
     \end{equation}
holds.  Then, for all $\alpha \in \real$, $\beta\in \itH^{\star}$, $\theta=(\alpha,\beta) \neq (\alpha_0,\beta_0)$, we have $L(\theta_0)< L(\theta)$. 
 \end{itemize}
\end{lemma}
\begin{proof}  
(a)  As in Theorem 2.2 in \citet{Bianco:yohai:1996}, taking conditional expectation, we have that
\begin{align*}
 L(\theta) & = \esp \left[\esp \phi(y,\alpha+\langle X, \beta\rangle)w(X)|X\right]=\esp \left\{ \phi(F(\alpha_0+\langle X, \beta_0\rangle  ),\alpha+\langle X, \beta\rangle  ) w(X)\right\}\\
 & = \esp\left\{  \phi(F(R_0(X)),R(X)) w(X)\right\}\,,
 \end{align*}
where for a fixed value $X$, we have denoted $R(X)=\alpha+\langle X, \beta\rangle $ and $R_0(X)=\alpha_0+\langle X, \beta_0\rangle  $.

We will show that, for any fixed $t_0$, the function $\phi( F(t_0),t)$ reaches its unique minimum when $t=t_0$. For simplicity, denote $ \Phi(t)=  \phi( F(t_0),t)$, then, $\Phi^{\prime}(t)= \Psi(F(t_0),t)= \,-\, (F(t_0)-F(t)) \nu(t)$,
where $\nu(t)$   defined in \eqref{eq:funcionnu} is positive.
Hence, $ \Phi^{\prime}(t_0)=0$. Furthermore, $\Phi^{\prime}(t)>0$, for $t>t_0$, and $\Phi^{\prime}(t)<0$ for $t<t_0$ which entails that $\Phi$ has a unique minimum at $t_0$.

Hence, for any $R(X)\ne R_0(X)$, $\phi(F(R_0(X)),R_0(X))< \phi(F(R_0(X)),R(X))$, that is,  
\begin{equation}\label{eq:eq1}
\phi(F(\alpha_0+\langle X, \beta_0\rangle ),\alpha_0 +\langle X,  \beta_0\rangle  )<  \phi(F(\alpha_0 +\langle X,  \beta_0\rangle ),\alpha +\langle X, \beta\rangle )\,,
 \end{equation}
  for any $\langle X,(\beta-\beta_0)\rangle + \alpha-\alpha_0\ne 0$	which concludes the proof of a).
  
The proof of b) follows immediately from \eqref{eq:eq1}  since \eqref{eq:identifiafuerte} holds and $\beta-\beta_0\in   \itH^{\star}$. 
\end{proof}

Lemma \ref{lemma:L_function}  provides an improvement over Lemma \ref{lema:FC} and will be helpful to derive Theorem \ref{teo:RATES1}. 
For its proof, we need the following Lemma which corresponds to   Lemma 3 in \citet{Bianco:Boente:Chebi:2023}. We state it here for completeness.

\begin{lemma} \label{lemma:L_function_Chebi}
Let $\Lambda(p, p_0)$ be defined as 
\begin{equation}
\Lambda(p, p_0) = p_0 \rho(- \log p) + (1 - p_0) \rho(- \log(1-p)) + G(p) + G(1- p)  \,
\label{eq:funcionM}
\end{equation} 
for $(p, p_0) \in (0,1) \times [0,1]$ and assume that  assumption  \ref{ass:rho_bounded_derivable} holds.
\begin{enumerate}
\item[(a)] If in addition  \ref{ass:rho_derivative_positive} holds, then the function $\Lambda(p, p_0)$ has a unique minimum at $p=p_0$. Furthermore, $\Lambda(p, p_0)$ can be extended to a continuous function on $[0,1] \times [0,1]$.
\item[(b)] Assume that, in addition, \ref{ass:psi_strictly_positive} holds. Then there exists a positive constant $C_0$ such that for each $0<p<1$, $\Lambda(p, p_0) - \Lambda(p_0, p_0) \geq  C_0  \,(p - p_0)^2$.
\end{enumerate}
\end{lemma}

The proof of Lemma \ref{lemma:L_function} is now a direct consequence of the previous result.

\begin{lemma} \label{lemma:L_function}
Assume that  assumptions  \ref{ass:rho_bounded_derivable}   and \ref{ass:psi_strictly_positive} hold. Then, there exists a constant $C_0>0$ such that  $L(\theta)-L(\theta_0)\ge C_0\,\pi_\prob^2(\theta,\theta_0)$. 
\end{lemma}

\begin{proof} 
For any $\theta=(\alpha,\beta)$, denote  $p(X)=F(\alpha+\langle X, \beta\rangle)$ and  let $p_0(X)= F(\alpha_0+\langle X, \beta_0\rangle)$. 

Using that $y|X\sim Bi\left(1, p_0(X)\right)$, \eqref{eq:phiBY*} and \eqref{eq:funcionM}, taking conditional expectation, we  immediately  get that $ L(\theta)  = \esp\left\{w(X)\,\Lambda\left( p(X) ,  p_0(X) \right)\right\} $. Lemma \ref{lemma:L_function_Chebi}(b) implies that 
 $$\Lambda\left( p(X) ,  p_0(X) \right) - \Lambda\left( p_0(X) ,  p_0(X) \right) \ge C_0\,\left\{p(X)-  p_0(X)\right\}^2= C_0 \left\{F(\alpha+\langle X, \beta\rangle)- F(\alpha_0+\langle X, \beta_0\rangle)\right\}^2\,,$$
 which entails  
 \begin{align*}
  L(\theta)-L(\theta_0) & \ge C_0 \esp\left[ w(X)\left\{F(\alpha+\langle X, \beta\rangle)- F(\alpha_0+\langle X, \beta_0\rangle) \right\}^2\right]=C_0 \, \,\pi_\prob^2(\theta,\theta_0)\,,
  \end{align*}
and concludes the proof.  
\end{proof}

From now on, for any measure $\qu$ and class of functions $\itF$, $N(\epsilon, \itF, L_s(\qu))$ and $N_{[\;]}(\epsilon, \itF, L_s(\qu))$ will denote  the covering  and bracketing  numbers  of the class $\itF $ with respect to the distance in $ L_s(\qu)$, as defined, for instance, in \citet{vanderVaart:wellner:1996}. Furthermore, we will make use of the empirical process  $\bbG_n f= \sqrt{n} (P_n-P)f$,   where $P_n$ stands for the empirical probability measure of $(y_i, X_i)$, $1\le i\le n$, and $P$ is the probability measure corresponding  of $(y,X)$ which follows the functional logistic regression model \eqref{eq:FLOGIT}.

We will first derive a result regarding Glivenko--Cantelli results for classes depending on $n$,  which will be helpful to derive Lemma \ref{lema:entropia}(b) and which is a slight modification of  Theorem 37 in \citet{Pollard:1984}, see also Lemma 2.3.3  in \citet{sara1988}.
 
 \vskip0.1in
\begin{lemma}\label{lema:nuevo}
 Let $z_1,\dots, z_n$, be i.i.d. random elements in a metric space $\itX$ and $\itF_n=\{f:\itX\to \real\}$ be the class of bounded functions depending on $n$, that is, for some positive constant $M$, $|f|\le M$, for any $f\in \itF_n$ and assume that  for any   $0<\epsilon<1$, there exists some constant $C>1$ independent of $n$ and $\epsilon$, such that
$$
\log N(M\, \epsilon , \itF_n,L_1(P_n))\leq  C q_n \log\left(\frac{1}{\epsilon}\right)    \;,
$$ 
where  $q_n$ is a non--random sequence of numbers such that $q_n/n\to 0$. 
Then, we have that 
$$\sup_{f \in \itF_n}  \left |(P_n-P)(f)\right| \convpp 0\,.$$
\end{lemma}

\begin{proof} 
%The proof follows the same steps as those considered in Theorem 37 in \citet{Pollard:1984}.
First of all, note that without loss of generality we may assume that $M=1$, in which case
$$\var\left(P_n f\right)=\frac{1}{n} \var\left(f(z_1)\right)\le \frac{1}{n} \esp\left(f^2(z_1)\right)\le \frac{1}{n} \,,$$
implying that for $n\ge 1/(8\epsilon^2)$, we have that
$$\frac{\var\left(P_n f\right)}{\left(4\,\epsilon\right)^2}\le \frac{1}{2}\,.$$
Thus, using Lemma 8  in \citet{Pollard:1984}, we obtain the inequality
$$\prob\left(\sup_{f \in \itF_n}  \left |(P_n-P)(f)\right| > 8\, \epsilon\right)\le 4 \prob\left(\sup_{f \in \itF_n}  \left | P_n^{0}f\right| > 2\, \epsilon\right)$$
where $P_n^{0}f= (1/n) \sum_{i=1}^n \xi_i f(z_i)$ and $\{\xi_i\}_{i=1}^n$ is a Rademacher sequence independent of $\{z_i\}_{i=1}^n$, that is, $ \xi_1,\dots \xi_n$ are i.i.d. $\prob(\xi_i=1)=\prob(\xi_i=-1)=1/2$. 

The covering number allows to choose functions $g_j$, $1\le j\le N= N( \epsilon , \itF_n,L_1(P_n)) $, such that for any $f\in \itF_n$, $\min_j P_n |f-g_j| \le \epsilon$. Denote $j(f)$ the index where the minimum is attained. Without loss of generality, we may assume that $g_j\in \itF_n$, so that $|g_j|\le 1$. Thus,    the approximation argument
and Hoeffding' s inequality lead to
\begin{align*}
\prob\left(\sup_{f \in \itF_n}  \left | P_n^{0}f\right| > 2\, \epsilon\,\Big|{\{z_1,\dots,z_n\}}\right)& \le
\prob\left(\sup_{f \in \itF_n}  \left | P_n^{0}g_{j(f)}\right|+   P_n \left |f-g_{j(f)}\right| > 2\, \epsilon\,\Big|{\{z_1,\dots,z_n\}}\right)
\\
& \le
\prob\left(\max_{1\le j\le N}  \left | P_n^{0}g_{j}\right|  >  \, \epsilon\,\Big|{\{z_1,\dots,z_n\}}\right)
\\
& \le
N \max_{1\le j\le N} \prob\left( \left | P_n^{0}g_{j}\right|  >  \, \epsilon\,\Big|{\{z_1,\dots,z_n\}}\right)
\\
&\le  2  N( \epsilon , \itF_n,L_1(P_n)) \exp\left\{-\, \frac{n\epsilon^2}{2}\right\}\le  2 \exp\left\{-\, \frac{n\epsilon^2}{2}+  C q_n \log\left(\frac{1}{\epsilon}\right) \right\} \,.
\end{align*}
Hence, using that $q_n/n\to 0$, we have that for $n$ large enough
$$\frac{q_n}{n}   \log\left(\frac{1}{\epsilon}\right)  <  \frac{ \epsilon^2}{4}\,,$$
which implies that, for $n$ large enough,
 $$\prob\left(\sup_{f \in \itF_n}  \left |(P_n-P)(f)\right| > 8\, \epsilon\right)\le  2 \exp\left\{-\, \frac{n\epsilon^2}{4}  \right\}\,,$$
so $\sum_{n\ge 1}\prob\left(\sup_{f \in \itF_n}  \left |(P_n-P)(f)\right| > 8\, \epsilon\right)<\infty$ and the result follows now from the Borel--Cantelli lemma.
\end{proof}

To avoid burden notation, when there is no doubt we will denote $k$ instead of $k_n$. To derive uniform results, Lemma \ref{lema:entropia}(a) below provides a bound to the covering number of the class of
functions
\begin{equation}
\itF_{n}=\left\{f(y, X)=\phi\left(y,\alpha+ \langle X, \beta\rangle   \right) \, w(X), \beta \in \itM_k, \alpha \in \real\right\}\,.
\label{eq:claseFn}
\end{equation}
Its proof relies on providing   a bound for the Vapnik-Chervonenkis  dimension for the set $\itF_n$, which follows from   Lemma  S.2.2 in \citet{Bianco:Boente:Chebi:2022}.  Besides, Lemma \ref{lema:entropia}(b) shows that $L_n(\theta)$ converges to $L(\theta)$ with probability one, uniformly over
$\theta=(\alpha,\beta)\in \real\times \itM_k$. Its proof uses standard empirical processes. This uniform law of large numbers will be crucial to obtain consistency results for our proposal as given in Theorems \ref{teo:RATES1}(a) and \ref{teo:CONSIST}.

\begin{lemma}\label{lema:entropia}
Let $\rho$ be a function satisfying  \ref{ass:rho_bounded_derivable} and  \ref{ass:rho_derivative_positive}, and $w$ a weight function satisfying \ref{ass:funcionw}. Let $\itF_n$ be the class of functions given in \eqref{eq:claseFn}. Then, 
\begin{enumerate} 
\item[(a)] for any measure $\qu$, $0<\epsilon<1$, there exists some constant $C>1$ independent of $n$ and $\epsilon$, such that
$$
N(M_{\phi}\, \epsilon , \itF_n,L_1(\qu))\leq  C q_n (16\,e)^{q_n}\left(\frac{1}{\epsilon}\right)^{q_n-1}   \;,
$$ 
 where $q_n=2k_n+6$ and $M_{\phi}=\sup_{y,t} |\phi(y,t)|$.
\item[b)] If in addition, \ref{ass:dimbasis} holds, we have $\sup_{\theta\in \real\times \itM_k}\left|L_n (\theta )-L(\theta )\right| \convpp 0$.
\item[c)] There exists a constant $C^{\star}$ that does not depend on $n$ nor $k_n$ such that
\begin{equation}\label{eq:cotaespFn}
\esp\left \{\sup_{f \in \itF_n}  \left |(P_n-P)(f) \right |\right \} \leq C^{\star} \sqrt{\frac{k_n}{n}}\,,
\end{equation}
which entails that $\sup_{f \in \itF_n}  \left |(P_n-P)(f)\right|=O_{\prob}\left(\sqrt{{k_n}/{n}}\right)$.
\end{enumerate}
\end{lemma}
\begin{proof}
(a) First of all note that since  \ref{ass:rho_bounded_derivable} and \ref{ass:funcionw} hold, we have that   $\phi(y, t) $ is   bounded and $w(X)$ is a bounded function with $\|w\|_{\infty}=1$, hence  $\itF_n$ has envelope $M_{\phi}$.  

Lemma S.2.2 from  \citet{Bianco:Boente:Chebi:2022} implies that the class of functions
$$ \itG_n=\left\{g(y, \bx)=\phi\left(y, \alpha +\bx\trasp \bb \right), \bb \in \real^k, \alpha \in \real\right\}\,,$$
where  $\bx=(\langle X, B_1 \rangle,\dots,  \langle X, B_1 \rangle)\trasp$, is   VC-subgraph with index smaller or equal than $ 2(k+1)+4=2k+6$, since we now include an intercept term $\alpha\in \real$. 

Note that for any $f\in \itF_n$  we have    $f(y, X)=\phi\left(y,\alpha + \langle X, \beta\rangle   \right) w(X)= \phi\left(y,\alpha + \bx\trasp \bb   \right) w(X)=g(y,\bx) w(X)$, where $\beta=\beta_{\bb}$ and $g\in \itG_n$. Hence, using  the permanence property of VC-classes when multiplying by a fixed function, in this case $w(X)$, we conclude that $V(\itF_n)\le  2k+6$ and the result follows now from Theorem 2.6.7 in \citet{vanderVaart:wellner:1996}.

\noi (b) From (a), using that $\log(2k_n+6)/(2k_n+6) <1$ and assuming that $C>1$ we get that 
\begin{align*}
\sup _{\qu} \log\left(N\left(M_{\phi}\,\epsilon, \itF_n, L_{1}(\qu)\right)\right) & \leq \log C + \log (2k_n+6)+ (2k_n+6)\log(16 e)+ (2k_n+5)\log \left(\frac{1}{\epsilon}\right)\\
& \le (2 k_n+6) \left[ \log C + 1+  \log(16 e)+  \log \left(\frac{1}{\epsilon}\right)\right]\,.
\end{align*}
 Hence, for any $\epsilon < \min( e^{-C}, (16 e)^{-1}, 1/e)$  we have that 
$$\frac{1}{n}\log\left(N\left(M_{\phi}\,\epsilon, \itF_n, L_{1}(P_n)\right)\right)  \leq 3\, \frac{2k_n + 6}{n}\log \left(\frac{1}{\epsilon}\right)\to 0\,,$$
since \ref{ass:dimbasis} holds. 
Therefore,  using Lemma \ref{lema:nuevo},  
we conclude that 
$$
\sup _{\theta\in \real\times \itM_k}\left|L_{n}(\theta)-L(\theta)\right| \convpp 0\,,
$$
completing the proof of (b).

\noi (c) As in (a), using that $V(\itF_n)\le  q_n=2k_n+6$ and  Theorem 2.6.7 in \citet{vanderVaart:wellner:1996},  we deduce that there exists some constant $C_0>1$ independent of $n$ and $\epsilon$, such that
\begin{equation}
N(M_{\phi}\, \epsilon , \itF_n,L_2(\qu))\leq  C_0 q_n (16\,e)^{q_n}\left(\frac{1}{\epsilon}\right)^{2q_n-2} = C_0 (2k_n+6) (16\,e)^{2k_n+6}\left(\frac{1}{\epsilon}\right)^{4k_n+10}   \;,
\label{eq:cotaN}
\end{equation} 
for any measure $\qu$, $0<\epsilon<1$. Theorem 2.14.1 in \citet{vanderVaart:wellner:1996} allows us to conclude that, for some universal constant $C_1>0$,
$$
\esp\left\{\sup_{f \in \itF_n}\left| \sqrt{n} (P_n-P)(f)  \right | \right \} \leq C_1\; M_{\phi}  \sup_{\qu} \int_0^1 \sqrt{1+\log N\left(M_{\phi}\,\epsilon , \itF_n, L_2(\qu)\right)} d\epsilon \,,
$$
where the supremum is taken over all discrete probability measures $\qu$. Using \eqref{eq:cotaN} and that $\log t\le t$ for $t\ge 1$ and  denoting $C_2=\log(C_0)+1+\log(16\; e)>1$, we get that
\begin{align*}
\sqrt{1+\log N\left(M_{\phi}\,\epsilon , \itF_n, L_2(\qu)\right)} 
& \le    \sqrt{1+  C_2 (2k_n+6)   + (4k_n+10) \log \left(\frac{1}{\epsilon}\right ) }\\
&  \le    \sqrt{1+ 16 k_n C_2     +  16 k_n \log \left(\frac{1}{\epsilon}\right ) }  \,,
\end{align*}
where we have used that $2k_n+6\le 4k_n+10\le   16 k_n$. Let  $C_4= 4\, C_1\; M_{\phi}\; C_3$ with  $C_3= C_2+1$. Then, we obtain that
\begin{eqnarray*}
\esp\left[\sup_{f \in \itF}\left| \sqrt{n} (P_n-P)(f)  \right | \right ]
 & \leq &  \sqrt{k_n}\; C_4 \, \int_0^1 \sqrt{1+ \log \left(\frac{1}{\epsilon}\right ) } d\epsilon \,,
\end{eqnarray*}
which  entails \eqref{eq:cotaespFn}, since   $\int_0^1 \sqrt{1 - \log(\epsilon)}\, d\epsilon < \infty$. Markov's inequality immediately leads to \linebreak $\sup_{f \in \itF_n}  \left |(P_n-P)(f)\right|=O_{\prob}\left(\sqrt{{k_n}/{n}}\right)$, concluding the proof.
\end{proof}

The following result is needed to prove the convergence rates stated in  Theorem \ref{teo:RATES}. From now on, $\Theta_n^{\star} =  \real\times \itM_{k_n} $.

\begin{lemma}\label{lema:bracketing}  
Let $\rho$ be a function satisfying  \ref{ass:rho_bounded_derivable} and  \ref{ass:rho_derivative_positive}, and $w$ a weight function satisfying \ref{ass:funcionw} and \ref{ass:wx2}.  Given  $\wtbeta_0 \in   \itM_{k}$, define $\wttheta_{0}=(\alpha_0, \wtbeta_0 )$ and the class of functions
\begin{align*}
\itG_{n, c, \wttheta_{0},\theta_0^*} & = \{f_{\theta}=V_{  \theta}-V_{\theta_{0}^*}:      \theta\in \Theta_n^{\star} \quad \mbox{and}\quad |\alpha-\alpha_0|+ \|\beta -\wtbeta_0\|_{\itH} \leq c\}\,,
\end{align*} 
where $\theta_0^*=(\alpha_0^*,\beta_0^*)$ is a fixed element in $\real\times \itH$, and $V_{\theta}=\phi\left(y,\alpha+\langle X, \beta \rangle\right)w(X) $,
 for $ \theta=(\alpha,\beta)$. Then,  there exists   some constant  $A >0$ independent of $n$, $\wttheta_{0}$ and $\epsilon$, such that
 $$N_{[\;]}(\epsilon, \itG_{n, c, \wttheta_{0},\theta_0^*}, L_2(P))\le   \left( \frac{A  c}{\epsilon} + 1 \right)^{k+1} \,.$$
 \end{lemma}

\begin{proof}
First note that $|f_{\theta}|\le 2  M_{\phi}$ and  denote $\wtTheta_{n,c}=\{\theta\in \Theta_n^{\star}: |\alpha-\alpha_0|+ \|\beta -\wtbeta_0\|_{\itH} \leq c \}$. Note that since  $ \|\beta -\wtbeta_0\|\le  \|\beta -\wtbeta_0\|_{\itH}$, $\wtTheta_{n,c} \subset \itI_c\times \itB_{n,c}$, where $\itI_c=[\alpha_0 -c, \alpha_0+c]$ and  $\itB_{n,c}= \{\beta\in \itM_k: \|\beta -\wtbeta_0\|  \leq c\}$. 

For any $\theta_\ell\in \wtTheta_{n,c}$, $\ell=1,2$ we have that
\begin{align}
|f_{\theta_1}-f_{\theta_2}| &= \left|V_{  \theta_1}- V_{  \theta_2}\right| = \left|\phi\left(y,\alpha_1+\langle X, \beta_1 \rangle\right)-\phi\left(y,\alpha_2+\langle X, \beta_2 \rangle\right)\right| w(X) 
\nonumber\\
& \le M_{\Psi} \left| \alpha_1-\alpha_2+\langle X, \beta_1 - \beta_2 \rangle \right|  w(X)
\nonumber\\
& \le  M_{\Psi}\left\{ \left| \alpha_1-\alpha_2\right|+ w(X) \|X\|\; \|\beta_1-\beta_2\|\right\}\,,
\label{eq:fth1th2}
\end{align}
where we recall that $M_{\Psi}=\sup_{y\in \{0,1\}, t\in \real} |\Psi(y,t)|$. Hence,
$$\esp \left(f_{\theta_1}-f_{\theta_2}\right)^2 \le   M_{\Psi}^2\left\{(\alpha_1-\alpha_2)^2+ \esp\left(w^2 (X) \|X\|^2 \right)\|\beta_1-\beta_2\|^2\right\}\,.$$
Taking into account that $\wtbeta_0\in \itM_k$, it can be written as $\wtbeta_0=\sum_{j=1}^{k} \wtb_{0,j} B_{j}$, so $\itB_{n,c}=\{\beta=\sum_{j=1}^{k} b_{j} B_{j}, \bb \in \real^k: \|\sum_{j=1}^{k} \left(b_{j}-\wtb_{0,j}\right) B_{j}\|  \leq c\}$. Thus, according to Corollary 2.6 in \citet{sara2000},  taking therein as measure $\qu$ the uniform measure on $\itT=[0,1]$, we get that $\itB_{n,c}$ can be covered by 
$$N_{\itB_{n,c},\delta}= \left(\frac{4 c + \delta}{\delta}\right)^k\,,$$
  balls with center  $\beta_{j,c,\delta}$ , $1\le j \le N_{\itB_{n,c},\delta}$, and radius $\delta$. Besides, the interval $\itI_c$ may also be covered by $N_{\itI_{c},\delta}= (4c+\delta)/\delta$ balls with center $\alpha_{j,c,\delta}$, $1\le j \le N_{\itI_{c},\delta}$, and radius $\delta$.

Given $\epsilon>0$, take $\delta=\epsilon/\left\{2\;M_{\Psi} \left[1+\left(\esp w^2 (X) \|X\|^2 \right)^{1/2}\right]\right\}$ and for $1\le j\le N_{\itI_{c},\delta}$ and $1\le m\le N_{\itB_{n,c},\delta}$, define the functions $f_{ j,m}=\phi\left(y, \alpha_{j,c,\delta} +\langle X, \beta_{m,c,\delta} \rangle\right)w(X) - V_{\theta_{0}}$ and 
\begin{align*}
f_{j,m}^{(U)}(y,X)& = f_{ j,m} + \delta\; M_{\Psi} \left(  1 + w(X) \|X\|\; \right)\\
f_{j,m}^{(L)}(y,X) & = f_{ j,m} - \delta\; M_{\Psi} \left(  1 + w(X) \|X\|\; \right)\,.
\end{align*}
Given $f_{\theta}\in \itG_{n, c, \wttheta_{0},\theta_0^*}$, let $j$ and $m$ be such that $|\alpha- \alpha_{j,c,\delta}|<\delta$ and $\|\beta- \beta_{m,c,\delta}\| <\delta$ and $\theta_{j,m}=(\alpha_{j,c,\delta},\beta_{m,c,\delta})$. Then, using \eqref{eq:fth1th2}, we obtain that
$$\left|f_{\theta}- f_{\theta_{j,m}}\right|\le  \delta \;M_{\Psi}\left\{ 1+ w(X) \|X\|\; \right\}\,,$$
so $f_{j,m}^{(L)} \le f_{\theta}\le  f_{j,m}^{(U)}$, since $f_{\theta_{j,m}}=f_{ j,m}$. Besides,
$$\|f_{j,m}^{(U)}-f_{j,m}^{(L)} \|= 2\delta    M_{\Psi}  \left\{\esp\left( 1 + w(X) \|X\|\right)^2 \right\}^{1/2}\le 2\delta    M_{\Psi}  \left\{1+ \left(\esp  w^2(X) \|X\|^2\right)^{1/2} \right\}=\epsilon \,,$$
which implies that
$$N_{[\;]}(\epsilon, \itG_{n, c, \wttheta_{0},\theta_0^*}, L_2(P))\le  N_{\itI_{c},\delta}  N_{\itB_{n,c},\delta} \le \left(\frac{4 c + \delta}{\delta}\right)^{k+1}
= \left( \frac{8\;M_{\Psi} \left[1+\left(\esp w^2 (X) \|X\|^2 \right)^{1/2}\right]  c}{\epsilon} + 1 \right)^{k+1}\,,$$
and the result follows taking $A=8\;M_{\Psi} \left[1+\left(\esp w^2 (X) \|X\|^2 \right)^{1/2}\right]$. 
\end{proof}

\subsection{Proof of Theorems \ref{teo:RATES1} and \ref{teo:CONSIST}}
Recall that $\wtheta=(\walfa,\wbeta)$, $\theta_0=(\alpha_0,\beta_0)$. Henceforth,   $\wttheta$ stands for $\wttheta=(\alpha_0,\wtbeta)$ with $\wtbeta=\wtbeta_{k}$ defined in assumptions \ref{ass:approx} and \ref{ass:approxorder}.

The following Lemma  is   useful to derive Theorems  \ref{teo:RATES1} and \ref{teo:CONSIST}.

\begin{lemma}\label{lema:M_betahat} 
Let $\rho$ be a function satisfying  \ref{ass:rho_bounded_derivable} and  \ref{ass:rho_derivative_positive}, and $w$ a weight function satisfying \ref{ass:funcionw}. Assume that \ref{ass:dimbasis} and \ref{ass:approx} hold. Then, we have that $L (\walfa, \beta_{\wbb})=L(\wtheta ) \convpp L(\theta_0)$.
\end{lemma}

\begin{proof} 
 From \ref{ass:dimbasis} and \ref{ass:approx}, we have that there exists $\wtbeta_{k_n}\in \itM_{k_n}$,  $\wtbeta_{k_n}=\sum_{j=1}^{k_n} \wtb_j  B_j(x)$ such that  $\|\wtbeta_{k_n}-\beta_{0}\|_{\itH}\to 0$ as  $n\to \infty$. As mentioned above, we denote $\wttheta=(\alpha_0,\wtbeta_{k_n})$. Using that $\wtbeta_{k_n}\in \itM_{k_n}$, we conclude that 
$L_n (\wtheta) \leq L_n(\wttheta)\,,$ 
 while from  Lemma \ref{lema:FC}(a) we have that $0\le L( \wtheta ) -  L(\theta_0)$. Thus,
\begin{align}
0\le L(\wtheta) - L(\theta_0) & = L(\wtheta) - L(\wttheta)+ L(\wttheta)- L(\theta_0)
\nonumber\\
 & = \left\{L_n(\wttheta)-L(\wttheta)\right\} -\left\{L_n(\wtheta)-L(\wtheta)\right\}  +\left\{L(\wttheta)- L(\theta_0)\right\} + L_n(\wtheta)- L_n(\wttheta)
 \nonumber\\
&  \leq  \left\{L_n(\wttheta)-L(\wttheta)\right\} -\left\{L_n(\wtheta)-L(\wtheta)\right\}  +\left\{L(\wttheta)-L(\theta_0)\right\} 
 \nonumber\\
 & \le 2 \, \sup_{f \in \itF_n}  \left |(P_n-P) f\right|  +\left\{L(\wttheta) - L(\theta_0)\right\}= 2\, A_n + B_n\,,
 \label{eq:cotaLwtheta}
\end{align}
where $\itF_n$ is defined in \eqref{eq:claseFn}. 

From Lemma \ref{lema:entropia} we obtain that $A_n\convpp 0$. On the other hand, from  
$\|\wtbeta_{k_n}-\beta_{0}\|_{\itH} \to 0$ as $n\to \infty$, the fact that for any $f\in \itH$, we have the bound $\|f\|\le \|f\|_{\itH}$ and the Cauchy-Schwartz inequality, we get  that for any $v\in L^2([0,1])$,   $\alpha_0+\langle v, \wtbeta_{k_n} \rangle   \to \alpha_0+\langle v, \beta_0 \rangle $. Thus, from the  Bounded Convergence Theorem, the continuity of $\phi\left(y, t \right)$ with respect to $t$ and its boundedness together with the boundedness of $w$, we conclude that $B_n=L(\wttheta) - L(\theta_0) \to 0$. Summarizing we have that
$$0\le L( \wtheta ) -  L(\theta_0) \le 2\, A_n + B_{n} \convpp 0\,,$$
which concludes the proof.
\end{proof}

\begin{proof}[Proof of Theorem \ref{teo:RATES1}]
(a) From Lemma \ref{lemma:L_function}, we have that there exists a constant $ C_0>0$   independent from $n$ such that
$$ L(\wtheta) - L(\theta_0) \ge C_0\,\pi_\prob^2(\wtheta,\theta_0)\,.$$
Then, the result follows from  Lemma \ref{lema:M_betahat} which implies that $L( \wtheta ) -  L(\theta_0)\convpp 0$. 
 
 \noi (b) Recall that from \eqref{eq:cotaLwtheta} and  Lemma \ref{lemma:L_function} we obtain that, for some constant $C_0$ independent of $n$,
 $$C_0\,\pi_\prob^2(\wtheta,\theta_0)\le   L(\wtheta) - L(\theta_0)\le 2 \, \sup_{f \in \itF_n}  \left |(P_n-P) f\right|  +\left\{L(\wttheta) - L(\theta_0)\right\}\,.$$
 Lemma \ref{lema:entropia}(c), entails that $
\sup_{f \in \itF_n}  \left |(P_n-P) f\right| =O_{\prob}\left(\sqrt{{k_n}/{n}}\right)$,
so, using again that $B_n= L(\wttheta) - L(\theta_0)\ge 0$, the proof will be concluded if we show that
\begin{equation}
  L(\wttheta) - L(\theta_0) = O(\|\wtbeta_{k_n}-\beta_{0}\|_{\itH})\,.
\label{eq:cotaele}
\end{equation}
To prove \eqref{eq:cotaele}, recall  that $\Psi(y,t)={\partial} \phi(y,t)/{\partial t}= -[y-F(t)] \nu(t)$ where $\nu(t)$ is given by   \eqref{eq:funcionnu} and that $M_{\Psi}=\sup_{y\in\{0,1\}, t\in \real} | \Psi(y,t)|$ is finite. Thus, 
\begin{align*}
0 & \le L(\wttheta) - L(\theta_0)= \esp\left\{w(X)\,\left[\phi(y,  \alpha_0+\langle X, \wtbeta\rangle)-\phi(y,  \alpha_0+\langle X, \beta_0\rangle \right]\right\}
\nonumber\\
& \le M_{\Psi} \esp\left\{w(X)\,\left|\langle X, \wtbeta-\beta\rangle  \right|\right\}\le  M_{\Psi} \|\wtbeta-\beta\|\;\esp\left\{w(X)\,\|X\| \right\} \,.
\label{eq:cotasupLwtheta}
\end{align*}
Then, \eqref{eq:cotaele} follows now from the fact that  $\|f\|\le \|f\|_{\itH}$.

 \noi (c) To derive (c) it will be enough to show that  $ L(\wttheta) - L(\theta_0) = O\left(\|\wtbeta_{k_n}-\beta_{0}\|_{\itH}^2\right)$.
 
 Using that  $\psi$ is continuously differentiable with bounded derivative, we obtain that the derivative of the function $\nu(t)$ defined in \eqref{eq:funcionnu} equals
\begin{align*}
\nu^{\,\prime}(t)  & = - \psi^{\prime}\left(-\log F(t)\right)\left[1- F(t)\right]^2-  \psi\left(-\log F(t)\right) F(t)\left[1-F(t)\right] \\
& + \psi\left(-\log\left[1- F(t)\right]\right) F(t)\left[1-F(t)\right]+ \psi^{\prime}\left(-\log\left[1- F(t)\right]\right) F^2(t)
\end{align*}
and  is bounded. Hence, ${\partial^2} \phi(y,t)/{\partial t^2}=  F(t)\left[1-F(t)\right] \nu(t)- [y-F(t)] \nu^{\,\prime}(t)$ is also bounded. Denote $A=\sup_{y\in\{0,1\}, t\in \real} |{\partial^2}  \phi(y,t)/{\partial  t^2}|$.
 
 To avoid burden notation, denote $R_0(X)=  \alpha_0+ \langle X, \beta_0\rangle$.
 Define the function $g:\real\to \real$ as $g(t)= L(t \wttheta+ (1-t) \theta_0)$. Then, 
 $$g(t)=  \esp \left\{w(X)\phi\left(F(R_0(X)), R_0(X)  + t \langle X, \wtbeta- \beta_0\rangle\right) \right\}\,.$$
Note that $g(0)=L(\theta_0)$, $g(1)=L( \wttheta)$ and $g(t)\ge g(0) $ for all $t$.

Recall that $\Psi(y,t) = {\partial} \phi(y,t)/{\partial t}= -[y-F(t)] \nu(t)$, then $\Psi(F(r_0),r_0)=0$, for any $r_0$. Therefore, we have that
\begin{align*}
g(1)-g(0) & =  \esp \left\{w(X)\;\Psi\left(F(R_0(X)),R_0(X) \right) \langle X, \wtbeta- \beta_0\rangle \right\}+ \esp \left\{w(X)\;\frac{\partial^2 \phi(F(R_0(X)), t)}{\partial t^2}\Big|_{t=\xi_{\varrho}} \langle X, \wtbeta- \beta_0\rangle^2\right\}\\
& =  \esp \left\{w(X)\;\frac{\partial^2 \phi(F(R_0(X)), t)}{\partial t^2}\Big|_{t=\xi_{\varrho}} \langle X, \wtbeta- \beta_0\rangle^2\right\}
\end{align*}
where  $\xi_{\varrho}= \alpha_0+ \langle X, \beta_0\rangle  + \varrho \langle X, \wtbeta- \beta_0\rangle$ for some $\varrho=\varrho(X)\in (0,1)$.

Therefore, using that ${\partial^2} \phi(y,t)/{\partial t^2}$ is bounded, we obtain  
\begin{equation}
 L(\wttheta) - L(\theta_0) = g(1)-g(0) \le A\, \esp \left\{w(X) \langle X, \wtbeta- \beta_0\rangle^2\right\}\le A \, \esp \left\{w(X)\|X\|^2\right\} \|\wtbeta- \beta_0\|_{\itH}^2\,,
 \label{eq:Lwttheta}
 \end{equation}
 where the last inequality follows from  the Cauchy-Schwartz inequality and the fact that, for any $f\in \itH$, we have the bound $\|f\|\le \|f\|_{\itH}$, concluding the proof.
\end{proof}

%%%%%%%%%%%%%%%%%%%%%%%%%%%%%%%%%%%%%%%
% PROOF OF THEOREM CONSISTENCY
%%%%%%%%%%%%%%%%%%%%%%%%%%%%%%%%%%%%%%

Lemma  \ref{lema:betaultimate} is  an intermediate step   to derive Theorem \ref{teo:CONSIST}.

\begin{lemma}\label{lema:betaultimate} 
Let $\rho$ be a function satisfying  \ref{ass:rho_bounded_derivable} and  \ref{ass:rho_derivative_positive}, and $w$ a weight function satisfying \ref{ass:funcionw}. Assume that \ref{ass:beta0} to \ref{ass:approx} and \ref{ass:probaX*} hold. Then,  we have that,  there exists $M>0$ such that $  \prob( \cup_{m\in \natu}\cap_{n\ge m}|\walfa |+\| \wbeta \|_{\itW^{1,\itH}}\le M)=1$. 
\end{lemma}
 
 \begin{proof}
 From now on,  $\itB_1$ denotes the unit ball in $\itW^{1,\itH}$  and  $\itB=\{(\alpha,\beta)\in \real\times \itW^{1,\itH}: |\alpha|+\|\beta\|_{\itW^{1,\itH}}=1\}\subset\  [-1,1]\times \itB_1$. The Rellich--Kondrachov theorem entails that  $\itW^{1,\itH}$ is compactly embedded in  $\itH$, hence  $\itB$ is compact in  $\real \times \itH$. 
	
To simplify the notation, let  $\theta=(\alpha, \beta)$ and  $\ell(X,\theta)=\alpha + \langle X,\beta \rangle  $. Furthermore, given $\theta_j\in\real\times \itH$, for $j=1,2$, denote $d(\theta_1,\theta_2) =  |\alpha_1 - \alpha_2|+ \|\beta_1-\beta_2\|_{\itH} $.

\noi \textsl{Step 1.} We begin proving that, given $0< \delta\le 1/4$, there exist $K>0$ and positive numbers $\varphi_1, \dots, \varphi_s$ such that for every $\theta\in \itB$, there exist $j\in \{1,\dots, s\}$ such that 
	\begin{equation} \label{eq:cota1}
	\prob\left(|\ell(X,\theta)|>\frac{\varphi_{j}}{2}\right)> \prob\left(|\ell(X,\theta)|> \frac{\varphi_{j}}{2} \cap \|X\|\le K\right)>1-2\delta\,.
	\end{equation} 
To derive 	\eqref{eq:cota1}, first given $\delta>0$, define $K_\delta$ such that, for any $K\geq K_\delta$ 
\begin{equation}
\prob(\|X\|\geq K)<\delta\,.
\label{eq:cotaKX}
\end{equation}
Fix now $\theta \in \itB$. Using that $\beta \in \itH$ and  \ref{ass:probaX*}, there exists   $\varphi_\theta >0$  a continuity point  of the distribution of $|\ell(X,\theta)|$ such that
\begin{equation}
\prob\left(|\ell(X,\theta)|<\varphi_\theta\right )< \delta\;.
\label{eq:cotavarphiellX}
\end{equation}
Note that from \eqref{eq:cotaKX} and  \eqref{eq:cotavarphiellX}, we conclude that $A(\theta) =  \prob\left(|\ell(X,\theta) |\geq \varphi_\theta\right) - \prob\left( \|X\|  \geq  K\right)>1-2\delta$ and $B(\theta) =  \prob\left(|\ell(X,\theta) |\geq \varphi_\theta\cap \|X\|  \leq  K\right)>1-2\delta$ .

Let $\tau_\theta$ stand for $\tau_\theta = {\varphi_\theta}/{(2(K + 1))}$. Then, given $\theta^*=(\alpha^*, \beta^*) \in  \real\times \itH$, such that $d(\theta^*,\theta)   <\tau_\theta$, and using that $\|f\|\le \|f\|_{\itH}$ and the Cauchy--Schwartz inequality, we obtain 
$$
| \ell(X,\theta)|   \leq  |\ell(X,\theta^*)|+|\ell(X,\theta-\theta^*)|  \leq  |\ell(X,\theta^*)| + \tau_\theta \left(\|X\| +1\right)\,.
 $$
Hence,     we get
\begin{align}
\prob\left( |\ell(X,\theta^*)|\geq   \frac{\varphi_\theta}{2}\cap \|X\|  \le K \right) & \ge B(\theta)> 1-2\delta \,.\label{eq:cotaprobavarphiellX}
\end{align}
Consider the covering of $\itB$ given by the open balls $\itB(\theta,\tau_\theta)=\{(a,f) \in \real\times \itH: d((a,f),\theta)   <\tau_\theta\}$, $\theta \in \itB$. 
Using that $\itB$ is compact in $\real\times \itH$, we get that there exist $\theta_1, \dots \theta_s$ such that $\theta_j \in \itB$ and $\itB\subset \cup_{j=1}^s \itB(\theta_j,\tau_{\theta_j})$. Therefore, from \eqref{eq:cotaprobavarphiellX}, we conclude that
\begin{equation} \label{eq:cota2}
\min _{1 \leq j \leq s} \inf_{d(\theta, \theta_j)<\tau_{\theta_j}} \prob\left(|\ell(X,\theta)|>\frac{\varphi_{\theta_j}}{2}\cap \|X\|  \le K\right)>1-2\delta\,,
\end{equation}
meaning that, for every $\theta\in \itB$, there exist $j\in \{j_1\dots j_s\} $ such that 
$$\prob\left(|\ell(X,\theta)|>\frac{\varphi_{\theta_j}}{2}\cap \|X\|  \le K \right)>1-2\delta\;,$$
  concluding the proof of \textit{Step 1.} 

\noi \textsl{Step 2.} We will show that for any $\theta=(\alpha,\beta) \in \itB$
\begin{equation}
\label{eq:Dpositivo}
\esp\left(\lim_{a\to \infty}  \phi(y, a \; \ell(X, \theta)) w(X)\right)> L(\theta_0)\,.
\end{equation}
Let us consider a sequence $a_m\to +\infty$ it is enough to show that $D>L(\theta_0)$ where \linebreak $D=\esp\left\{\lim_{m\to \infty}  \phi\left(y, a_m \;\ell(X, \theta) \right) w(X)\right\}$. Using that $\phi$ is a bounded function and the bounded convergence theorem, we get that
$D= \lim_{m\to \infty} D_m$, where $D_m=\esp\left\{   \phi\left(y, a_m \;\ell(X, \theta) \right) w(X)\right\} =L(a_m\,\theta)$. 
%Clearly, from Lemma \ref{lema:FC}, for each $m\ge m_0$ we have that $D_m> L(\theta_0)$, since \ref{ass:probaX*} holds and $ \beta\in \itW^{1,\itH}$, so $a_m\,\beta\in \itW^{1,\itH}$ too. We want to show the stronger result $D> L(\theta_0)$.  

Note that as in the proof of Lemma \ref{lema:FC}, we have that
\begin{align}
D_m  & = \esp \left\{ w(X)\,\phi\left(F(\ell(X,\theta_0)), a_m \; \ell(X, \theta)\right)\right\}= \esp  \left\{ w(X) \phi(F(R_0(X)),R_m(X)) \right\}\,,
\label{eq:Dm}
\end{align}
where,  as in the proof of Theorem \ref{teo:RATES1}, we  denote  $ R_m(X)=a_m \; \ell(X, \theta)$ and $ R_0(X)=\ell( X, \theta_0)$.

In the proof of Lemma \ref{lema:FC}, we have shown   that, for any fixed $r_0$, the function $ \Phi(t)=\phi( F(r_0),t)$ reaches its unique minimum when $t=r_0$ and $\Phi^{\prime}(t)>0$, for $t>r_0$, and that $\Phi^{\prime}(t)<0$ for $t<r_0$. Then, $\Phi$ is strictly increasing on $(r_0, +\infty)$ and strictly decreasing on $(-\infty, r_0)$, so   $\lim_{t\to +\infty}\Phi(t)> \Phi(r_0)$ and similarly $\lim_{t\to -\infty}\Phi(t)> \Phi(r_0)$ .

 Using \ref{ass:probaX*} and that $\beta\in \itW^{1,\itH}$, we have that with probability one
\begin{equation}
\lim_{m\to \infty}  R_m(X) =\lim_{m\to  \infty}   a_m \;\ell(X, \theta) =\left\{\begin{array}{rl}
+\infty & \mbox{ when } \ell(X, \theta)>0\\
-\infty & \mbox{ when } \ell(X, \theta)<0
\end{array}\right.\,.
\label{eq:limiteRm}
\end{equation} 
Fix $\omega\in \Omega$, such that $X_0=X(\omega)$ satisfies \eqref{eq:limiteRm} and  $F(R_0(X_0))$  is not $0$ or $1$. Let $r_0= R_0(X_0)$ and take as above $\Phi(t)=\phi(F(r_0),t)$. Then,
\begin{align*}
\lim_{m\to\infty} \phi(F(R_0(X_0)),R_m(X_0))= \lim_{m\to \infty} \Phi(R_m(X_0))  & =   \left\{\begin{array}{rl}
 \lim_{r\to +\infty}\Phi(r) & \mbox{ when } \ell(X_0, \theta)>0\\
 \lim_{r\to -\infty}\Phi(r) & \mbox{ when } \ell(X_0, \theta)<0
\end{array}\right.\\
& >\Phi(r_0)=\Phi\left(R_0(X_0)=\phi(F(R_0(X)),R_0(X))\right) \,.
\end{align*}
Using \eqref{eq:Dm}, \ref{ass:probaX*} and \ref{ass:funcionw}, we obtain that 
Thus, using again that $\phi$ is a bounded function and the bounded convergence theorem, we obtain that
\begin{align*}
D= \lim_{m\to \infty} D_m  & = \esp \left\{ w(X)\,\lim_{m\to \infty}   \phi(F(R_0(X)),R_m(X)) \right\}>\esp \left\{\phi(F(R_0(X)),R_0(X))\right\}  \,,
\end{align*}
which concludes the proof of  \eqref{eq:Dpositivo}.

\noi \textsl{Step 3.} Let us show that  there exists $\eta>0$ such that for any $\theta=(\alpha,\beta) \in \itB$, we have
\begin{equation}\label{eq:lim_infL}
\lim_{a\to \infty} \inf_{|\alpha^*-\alpha|+ \|\beta^*-\beta\|_{\itW^{1,\itH}} < \eta} L(a\,\theta^*) >L(\theta_0)\,,
\end{equation}
where we have denoted $\theta^*=(\alpha^*, \beta^*)$. The proof is an extension to the functional setting of that of Lemma 6.3 in \citet{Bianco:yohai:1996}.
 
 Take $\delta < (D-L(\theta_0))/(2\,M_{\phi})$ where $D=\esp\left(\lim_{a\to \infty}  \phi(y, a \; \ell(X, \theta)) w(X)\right)$,  the quantity  $D-L(\theta_0)$ is positive from \eqref{eq:Dpositivo} and $M_{\phi}=\sup_{y\in \{0,1\}, t\in \real} \phi(y,t)$.

From \textsl{Step 1} we have that, for any $\theta\in \itB$, there exist $j\in \{1,\dots, s\}$ such that \eqref{eq:cota1} holds.  Let $j_0$ be the index corresponding to the chosen $\theta   \in  \itB$ and define the set $\itE=\{X: |\ell(X,\theta)| > \varphi_{j_0}/{2}\cap \|X\|  \le K\}$. Then, from \eqref{eq:cota1},  we get that 
\begin{equation}
\label{eq:cotaitE}
\prob(X\notin \itE)< 2\delta <  \frac{ D-L(\theta_0)}{M_{\phi}}\,.
\end{equation} 
Take $X\in \itE$ and define $\eta=\min_{1\le j\le s} \varphi_{j} /(8(K+1))$. Then, using again that $\|f\|\le \|f\|_{W^{1,\itH}}$ and that the set $\{  \theta^*=(\alpha^*, \beta^*):  |\alpha^*-\alpha|+ \|\beta^*-\beta\|_{\itW^{1,\itH}} < \eta\} \subset\{ \theta^*:  |\alpha^*-\alpha|+ \|\beta^*-\beta\|_{\itW^{1,\itH}}\le \varphi_{{j_0}}/(4(K+1))\}$ which is compact in $\real\times \itH$, it is easy to see  that for any 
 $\theta^*=(\alpha^*, \beta^*)$ such that $|\alpha^*-\alpha|+ \|\beta^*-\beta\|_{\itW^{1,\itH}} \le \eta \le \varphi_{{j_0}}/(4(K+1))$, we have that 
 \begin{equation}
 \label{eq:signos}
\left\{ 
\begin{array}{c }
\mbox{sign}\left(\ell(X, \theta^*) \right)  = \mbox{sign}\left( \ell(X, \theta)\right)\,,\\
 |\ell(X, \theta^*)|   \ge \dst\frac{\varphi_{j_0}}{4}\,.
 \end{array}
 \right.
 \end{equation}
Using that the set $\itV_{\eta}=\{\theta^* \in \real\times \itW^{1,\itH}: |\alpha^*-\alpha|+ \|\beta^*-\beta\|_{\itW^{1,\itH}} \le \eta\}$ is compact in $\real\times \itH$ we conclude that there exists $\{\theta_m^*\}_{m\in \natu}\subset \itV_{\eta}$ (depending on $a$ and $X$) and a value $\wttheta^*=\wttheta^*_{a,X}\in \itV_{\eta}$ which also depends on $a$ and $X$, such that 
$d(\theta_m^*, \wttheta^*)\to 0$ and $\lim_{m\to \infty}  \phi(y, a \; \ell(X,\theta_m^*))= \inf_{\theta^*\in \itV_{\eta}} \phi(y, a\; \ell(X,\theta^*))$.
The continuity of $\phi$ together with the fact that $d(\theta_m^*, \wttheta^*)\to 0$ and the Cauchy--Schwartz inequality leads to  $\phi(y, a\, \ell(X,\wttheta^*))= \inf_{\theta^*\in \itV_{\eta}} \phi(y, a \,  \ell(X,\theta^*))$. Then, using \eqref{eq:signos}, we conclude that
$$ \liminf_{a\to \infty} \inf_{\theta^*\in \itV_{\eta}} \phi\left(y, a \; \ell(X,\theta^*)\right)= \lim_{a\to \infty}   \phi\left(y, a \;\ell\left(X,\wttheta^*_{a,X}\right)\right)\,.$$
Therefore, using Fatou's Lemma we obtain that
\begin{align*}
\liminf_{a\to \infty}   \inf_{|\alpha^*-\alpha|+ \|\beta^*-\beta\|_{\itW^{1,\itH}} < \eta} L(a\,\theta^*) & \ge \esp\left\{\liminf_{a\to \infty} \inf_{\theta^*\in \itV_{\eta}}  \phi(y, a \; \ell(X,\theta^*))w(X)\right\}\\
& \ge \esp\left\{\indica_{\itE}(X)\, w(X)\, \lim_{a\to \infty} \inf_{\theta^*\in \itV_{\eta}}  \phi(y, a \;\ell(X,\theta^*))\right\}\\
& \ge \esp\left\{\indica_{\itE}(X)\,w(X)\, \lim_{a\to \infty}   \phi\left(y, a \; \ell\left(X,\wttheta^*_{a,X}\right)\right)\right\} \\
& = D-  \esp\left\{\indica_{\itE^c}(X)\,w(X)\,  \lim_{a\to \infty} \phi\left(y, a \;\ell\left(X,\wttheta^*_{a,X}\right)\right)\right\}\,.
\end{align*}
where we denote $\itE^c$  the complement of $\itE$.  Using \eqref{eq:cotaitE} and assumption \ref{ass:funcionw}, we obtain that
$$ \esp\left\{\indica_{\itE^c}(X)\, w(X)\, \lim_{a\to \infty} \phi\left(y, a \;\ell\left(X,\wttheta^*_{a,X}\right)\right)\right\} \le M_{\phi} \prob\left(X\in \itE^c \right)<  D-L(\theta_0) \,,$$
so
$$\lim_{a\to \infty} \inf \inf_{|\alpha^*-\alpha|+ \|\beta^*-\beta\|_{\itW^{1,\itH}} < \eta} L(a\,\theta^*)   \ge D-  \esp\left\{\indica_{\itE^c}(X)\, w(X)\,  \lim_{a\to \infty} \phi\left(y, a \; \ell(X,\wttheta^*_{a,X})\right)\right\} >  L(\theta_0) \,,$$
which concludes the proof of \eqref{eq:lim_infL}. 

\noi \textsl{Step 4.} We will show that there exists $M>0$ and $\tau>0$
\begin{equation}\label{eq:cota_inf}
 \inf _{\theta=(\alpha,\beta): \|\beta\|_{\itW^{1,\itH}}+|\alpha| > M} L\left(\theta \right) > L(\theta_0)+\tau\;.
\end{equation}
Note that by proving \eqref{eq:cota_inf}, we may indeed conclude the proof. In fact, from Lemma \ref{lema:M_betahat}, we have that $L(\walfa, \wbeta ) \convpp L(\alpha_0, \beta_0)$, thus if $\itN=\{\lim_{n\to \infty}L(\walfa, \wbeta ) \ne L(\alpha_0, \beta_0)\}$, $\prob(\itN)=0$. Take $\omega \notin \itN$ and $n_0$ such that for all $n\ge n_0$, $|L(\walfa, \wbeta ) - L(\alpha_0, \beta_0)| <\tau$, then $L(\walfa, \wbeta ) < L(\alpha_0,\beta_0)+\tau <\inf _{|\alpha|+\|\beta\|_{\itW^{1,\itH}} > M} L\left(\alpha, \beta \right) $, which entails that for all $n\ge n_0$, $|\walfa|+\|\wbeta\|_{\itW^{1,\itH}} \le M$, as desired.

Let us show that \eqref{eq:cota_inf} holds. 
From \textsl{Step 3},   there exists $\eta$ such that, for any  $\theta=(\alpha,\beta) \in \itB$, we have that $D(\theta)>L(\theta_0)$ where
$$D(\theta)=\lim_{a\to \infty} \inf_{|\alpha^*-\alpha|+ \|\beta^*-\beta\|_{\itW^{1,\itH}} < \eta } L(a\,\theta^*)\,,$$
and $\theta^*= (\alpha^*,\beta^*)$. Define  $0<\tau_{\theta}< (D(\theta)-L(\theta_0))/2$, then $D(\theta)> L(\theta_0)+2\;\tau_{\theta} $, which implies that there exists $a_{\theta}>0$ such that
\begin{equation}\label{eq:Dtheta_inf}
\inf_{ a>a_{\theta}} \inf_{|\alpha^*-\alpha|+ \|\beta^*-\beta\|_{\itW^{1,\itH}} < \eta} L(a\,\theta^*) >L(\theta_0)+ \tau_{\theta}\,.
\end{equation}
Taking into account that   the open balls $\itB(\theta,\eta_{\theta})=\{\theta^* \in \real\times \itW^{1,\itH}: |\alpha^*-\alpha|+ \|\beta^*-\beta\|_{\itW^{1,\itH}} <\eta\}$ provide a covering of  $\itB$ which is compact in $\real\times \itH$, we get that there exist $\theta_1, \dots \theta_s$ such that $\theta_j \in \itB$ and $\itB\subset \cup_{j=1}^s \itB(\theta_j,\eta)$. Thus, if $a_j=a_{\theta_j}$, $A=\max_{1\le j\le s}(a_j)$, $\tau_j=\tau_{\theta_j}$ and $\tau=\min_{1\le j\le s}(\tau_j)$, from \eqref{eq:Dtheta_inf}, we obtain that for $j=1, \dots, s$ we have
\begin{equation}\label{eq:Djtheta_inf}
\inf_{ a>A} \inf_{\theta^* \in  \itB(\theta_j,\eta)} L(a\,\theta^*) >L(\theta_0)+ \tau\,.
\end{equation}
Let $\theta \in \real\times \itW^{1,\itH}$ be such that $d_{\theta}=| \alpha|+ \| \beta\|_{\itW^{1,\itH}}>A$, define $\wttheta= \theta/d_{\theta}$, then $\wttheta\in \itB$, so there exists $j_0$ such that $\wttheta\in\itB(\theta_{j_0},\eta)$ and $L(\theta)=L(d_{\theta}\; \wttheta)$. Therefore, using \eqref{eq:Djtheta_inf}, we obtain that $L(\theta)>L(\theta_0)+ \tau$, so
$$\inf _{|\alpha|+\|\beta\|_{\itW^{1,\itH}} > A} L\left(\alpha, \beta \right) > L(\theta_0)+ \tau\,,$$
concluding the proof.
 \end{proof}

\begin{proof}[Proof of Theorem \ref{teo:CONSIST}]
We will show only b), that is, that the result holds when $\beta_0\in  C([0,1])$, $B_j\in  C([0,1])$, $1\le j\le k_n$,   $B_j\in \itW^{1,2}$ and provide approximations in $L^2([0,1])$, that is, below $\itW^{1,\itH}$ is the  Sobolev space $\itW^{1,2}$. The situation where $\itH=L^2([0,1])$ follows similarly, replacing  the supremum norm below by the $L^2-$norm and using that in this case, the Sobolev space $\itW^{1,\itH}$  is compactly embedded in  $L^2([0,1])$ and that \ref{ass:probaX*}  holds for any $\beta \in \itW^{1,\itH}$. 

Assume that we have shown that
\begin{equation}
\label{eq:cotainfL}
\inf_{ \theta \in \itA_{\epsilon}}L(\theta) >L(\theta_0)
\end{equation}
where $\itA_{\epsilon}=\{\theta=(\alpha,\beta)\in \real\times \itW^{1,\itH}:  |\alpha |+\| \beta \|_{\itW^{1,\itH}}\le M,  
|\alpha-\alpha_0|+\| \beta-\beta_0\|_{\infty}>\epsilon\}$. Then, taking into account  that $\wbeta \in \itW^{1,\itH}\cap C([0,1])$, that Lemma  \ref{lema:M_betahat}  implies that $L(\wtheta) \convpp L(\theta_0)$ and  Lemma \ref{lema:betaultimate}, we conclude that $|\walfa-\alpha_0|+\| \wbeta-\beta_0\|_{\infty}\convpp 0$.

Let us derive \eqref{eq:cotainfL}. Let $\{\theta_m\}_{m\ge 1}$ be a sequence such that $\theta_m=(\alpha_m,\beta_m) \in \itA_\epsilon$ and $\lim_{m\to \infty} L(\theta_m) = \inf_{\theta \in \itA_{\epsilon}}L(\theta)$.  The Rellich--Kondrachov theorem entails that  $\itW^{1,\itH}$ is compactly embedded in  $C([0,1])$, hence the ball $\{\theta=(\alpha,\beta)\in \real\times \itW^{1,\itH}:  |\alpha|+\| \beta\|_{\itW^{1,\itH}}\le M\}$ is compact in  $\real \times C([0,1])$. Thus, there exist a subsequence  $\{ \theta_{m_j}\}_{j\ge 1}$ of $\{\theta_m\}_{m\ge 1}$ and a point $\theta^{\star}=(\alpha^{\star}, \beta^{\star})\in \real\times \itW^{1,\itH}$  with $|\alpha^{\star} |+\| \beta^{\star} \|_{\itW^{1,\itH}}\le M $ such that 
	$|\alpha_{m_j}-\alpha^{\star}|+\|\beta_{m_j}-\beta^{\star}\|_{ \infty}\to 0$. Then,  using that $\theta_m \in \itA_\epsilon$, we have $|\alpha_{m_j}-\alpha_0|+\| \beta_{m_j}-\beta_0\|_{\infty}>\epsilon$ which implies that $| \alpha^{\star}- \alpha_0|+\|\beta^{\star}-\beta_0\|_{\infty}\ge \epsilon $. 
	Using that $\|f\|\le \|f\|_{\infty}$ and the Cauchy-Schwartz inequality, we get  that for any $v\in L^2([0,1])$,   $\alpha_{m_j}+\langle v, \beta_{m_j} \rangle   \to \alpha^{\star}+\langle v, \beta^{\star} \rangle $. Thus, using  the Bounded Convergence Theorem, the continuity of $\phi\left(y, t \right)$ with respect to $t$ and its boundedness,   we get that $L(\theta_{m_j})\to L(\theta^{\star})$, which leads to   $\inf_{\theta\in \itA_{\epsilon}}L(\theta)= L(\theta^{\star})$. Using that \ref{ass:probaX*}  holds   and taking into account that $\beta^{\star}\in \itW^{1,\itH}$, Lemma \ref{lema:FC}(b) implies that $L(\theta^{\star})>L(\theta_0)$   concluding the proof.  
\end{proof} 

%%%%%%%%%%%%%%%%%%%%%%%%%%%%%
% DEMO TASAS
%%%%%%%%%%%%%%%%%%%%%%%%%%%%%%%%
 
\subsection{Proof of Theorem \ref{teo:RATES} and Proposition \ref{prop:cotainf}}
%%%%%%%%%%%%%%%%%%%%%%%%%%%%%%
% TASAS
%%%%%%%%%%%%%%%%%%%%%%%%%%%%%%%%%%%%%

\begin{proof}[Proof of Theorem \ref{teo:RATES}] 
From assumption \ref{ass:approxorder}, we have that there  exists an element $\wtbeta_{k}\in \itM_k$,  $\wtbeta_{k}=\sum_{j=1}^{k} \wtb_j  B_j(x)$ such that  $\|\wtbeta_{k}-\beta_{0}\|_{\itH}=O(k^{-r})$. Without loss of generality, we assume that $\|\wtbeta_k-\beta_0\|_{\itH} < \epsilon_0/2$ with $\epsilon_0$ defined in \ref{ass:cotainf}.   Recall that  $\Theta_n^{\star} =  \real\times \itM_{k_n} $.
 
 In order to get the convergence rate  of our estimator $\wtheta = (\walfa,\wbeta)$ we will apply Theorem 3.4.1 of  \citet{vanderVaart:wellner:1996}. According to the notation in that Theorem, let $d (\theta_1, \theta_2)= \wtpi_{\prob} (\theta_1, \theta_2)$  and $\Theta_n = \{ \theta\in \Theta_n^{\star}: |\alpha -\alpha_0|+\|\beta -\beta_0\|_{\itH}\le  \epsilon_0/2\}$,  where $\epsilon_0$ is given in assumption \ref{ass:cotainf}.  Furthermore, the function $m_{\theta}$ in that Theorem equals
$$m_{\theta}(y,X)= \,-\, \phi\left(y, \alpha +\langle X, \beta \rangle \right) w(X)\,.$$ 
First of all, note that  Theorem \ref{teo:CONSIST} implies that $\prob( \wtheta \in \Theta_n)\to 1$ as required. Secondly, to emphasize the dependence on $n$ denote     $\wttheta_{0,n}=(\alpha_0,\wtbeta)$ with $\wtbeta=\wtbeta_{k}$ defined in assumption \ref{ass:approxorder}. Assumption \ref{ass:wx2}, the fact that $\|f\|\le \|f\|_{\itH}$ and the Bounded Convergence Theorem implies that $\wtpi_{\prob}(\wttheta_{0,n}, \theta_0)\to 0$, as $n\to \infty$. Furthermore, from Theorem \ref{teo:CONSIST}, we get that  $\wtpi_{\prob}(\wtheta, \theta_0)\convpp 0$. Hence, using the triangular inequality, we immediately obtain that $\wtpi_{\prob}(\wtheta, \wttheta_{0,n}) \convprob 0$ as required in Theorem 3.4.1 of  \citet{vanderVaart:wellner:1996}. Moreover, we also have that  $L_n(\wtheta)\le L_n(\wttheta_{0,n})$, since $\wttheta_{0,n}\in \real\times\itM_k$, which is also a requirement to apply that result.

 Let $A=\sup_{y\in\{0,1\}, t\in \real} |{\partial^2}  \phi(y,t)/{\partial  t^2}|$ and denote $C^2= 16(A+C_0^{\star}) \esp \left\{w(X)\|X\|^2\right\}/C_0^{\star}$ with $C_0^{\star}$ the constant  given in assumption \ref{ass:cotainf}. Define $\delta_n= C \|\wtbeta-\beta_0\|_{\itH}$. Note that from \ref{ass:approxorder}, $\delta_n = O(n^{-\,r\varsigma})$. 

To make use of Theorem 3.4.1 of  \citet{vanderVaart:wellner:1996}, we have to show that there exists a function $\varphi_n:(0,\infty)\to \real$ such that $\varphi_n(\delta)/\delta^{\gamma}$ is decreasing on $(\delta_n, \infty)$, for some $\gamma<2$ and such that,  for any $\delta> \delta_n$, we have
\begin{align}
&\sup_{\theta\in \Theta_{n, \delta}} \esp \left( m_{\theta}(y,X)- m_{\wttheta_{0,n}}(y,X)\right)  =   L( \wttheta_{0,n})- \inf_{\theta\in \Theta_{n, \delta}} L( \theta)  \lesssim    -\delta^2   \,,
\label{eq:aprobar1}\\
&\esp^{*} \sup_{ \theta\in \Theta_{n, \delta}}   \left |\bbG_n 
\left(m_{\theta}-m_{\wttheta_{0,n}}\right) \right |  =
 \esp^{*} \sup_{ \theta \in \Theta_{n,\delta}} \sqrt{n} \biggl|
 \left(L_n( \theta)- L( \theta\right) - \left(L_n(\wttheta_{0,n})- L( \wttheta_{0,n})\right) \biggr|   \lesssim \varphi_n(\delta) \,,
\label{eq:aprobar2}
\end{align}
where $\bbG_n f= \sqrt{n} (P_n-P)f$,   $\esp^{*}$ stands for the outer expectation and $\Theta_{n,\delta}=\{\theta \in \Theta_n: \delta/2 < d(\theta,\wttheta_{0,n})\le \delta\}$.

We begin by showing \eqref{eq:aprobar1}. Assumption \ref{ass:cotainf} entails that,  for any $ \theta\in \Theta_n$,   
\begin{equation}
L( \theta)-L (\theta_0)\ge C_0^{\star}\,\wtpi_{\prob}^2( \theta, \theta_0)\,,
\label{eq:cota1probar}
\end{equation}
while from \eqref{eq:Lwttheta} in the proof of Theorem \ref{teo:RATES1}(c), we get that 
\begin{equation}
 L(\wttheta_{0,n}) - L(\theta_0) \le A \, \esp \left\{w(X)\|X\|^2\right\} \|\wtbeta- \beta_0\|_{\itH}^2= \frac{ A \, \esp \left\{w(X)\|X\|^2\right\}}{C^2} \delta_n^2\,.
 \label{eq:cota2probar}
\end{equation}
Moreover,  using the Cauchy Schwartz inequality and the fact that $\|f\|\le \|f\|_{\itH}$, we have
$$\wtpi_{\prob} ( \theta, \wttheta_{0,n}) \le \wtpi_{\prob} ( \theta, \theta_0) +\wtpi_{\prob} ( \theta_0, \wttheta_{0,n})\le \wtpi_{\prob} ( \theta, \theta_0)+ \left\{\esp \left\{w(X)\|X\|^2\right\} \|\wtbeta- \beta_0\|_{\itH}^2\right\}^{1/2}\,,$$
which  together with the inequality $(a+b)^2 \le 2(a^2+b^2)$, implies that
\begin{equation} 
\wtpi_{\prob}^2( \theta, \theta_0) \ge \frac 12 \wtpi_{\prob}^2 ( \theta, \wttheta_{0,n}) -  \esp \left\{w(X)\|X\|^2\right\} \|\wtbeta- \beta_0\|_{\itH}^2 =\frac 12 \wtpi_{\prob}^2 ( \theta, \wttheta_{0,n}) -\frac{\esp \left\{w(X)\|X\|^2\right\}}{C^2} \delta_n^2\,.
\label{eq:cota3probar}
\end{equation}
Then combining \eqref{eq:cota1probar}, \eqref{eq:cota2probar} and \eqref{eq:cota3probar}, we get that for any $\theta \in \Theta_n$, such that $ \delta/2 < d(\theta,\wttheta_{0,n})= \wtpi_{\prob}  ( \theta, \wttheta_{0,n})$, we have  
\begin{align*}
L( \theta)- L(\wttheta_{0,n}) & = \left\{L( \theta)-L (\theta_0)\right\}- \left\{L(\wttheta_{0,n}) - L(\theta_0)\right\}  \ge C_0^{\star}\,\wtpi_{\prob}^2( \theta, \theta_0) - \frac{ A \, \esp \left\{w(X)\|X\|^2\right\}}{C^2} \delta_n^2\\
& \ge  \frac{C_0^{\star}}{2} \wtpi_{\prob}^2 ( \theta, \wttheta_{0,n}) -\frac{(A+C_0^{\star})\esp \left\{w(X)\|X\|^2\right\}}{C^2} \delta_n^2
 \ge  \frac{C_0^{\star}}{8} \delta^2  - \frac{C_0}{16}\delta_n^2\ge \frac{C_0^{\star}}{16} \delta^2 \,,
\end{align*}
concluding the proof of \eqref{eq:aprobar1}.

We have now to  find $\varphi_n(\delta)$ such that $\varphi_n(\delta)/\delta^\gamma$ is decreasing in $\delta$, for some $\gamma<2$ and \eqref{eq:aprobar2} holds.   Define the class of functions
$$\itF_{n,\delta} = \left\{V_{  \theta}-V_{\wttheta_{0,n}}:    \theta\in \Theta_{n,\delta} \right\}\,,$$ 
with  $V_{  \theta}=\,-\, m_{\theta}=\phi\left(y, \alpha+ \langle X, \beta \rangle\right) w(X) $.
Inequality \eqref{eq:aprobar2} involves an empirical process indexed by $\itF_{n,\delta}$, since
$$\esp^{*} \sup_{ \stackrel{ \theta \in \Theta_n }{ \delta/2< d (\theta, \wttheta_{0,n})\le \delta} }   \left |\bbG_n 
\left(m_{ \theta }-m_{\wttheta_{0,n}}\right) \right |  =\esp^{*} \sup_{f\in \itF_{n,\delta}} \sqrt{n} |(P_n-P) f|\,.$$
For any $f\in \itF_{n,\delta} $ we have that $\|f\|_{\infty} \le A_1 = 2 \sup_{y\in \{0,1\}; t\in \real} \phi(y,t)= 2 M_{\phi} $. Furthermore, if $A_2= 2\, M_{\psi}$ using  the inequality
\begin{align*}
|V_{  \theta}-V_{\wttheta_{0,n}}| & =  w(X) \left|\phi\left(y, \alpha+ \langle X, \beta \rangle\right) -\phi\left(y, \alpha_0+ \langle X, \wtbeta \rangle\right) \right| \\
& \le 2\, M_{\psi}   \left|\alpha-\alpha_0 +  \langle X, \beta-\wtbeta  \rangle\right| w(X)\,,
\end{align*}
and the fact that $\|w\|_{\infty}=1$ and $\wtpi_{\prob}^2(\theta, \wttheta_{0,n}) \leq \delta^2$,  we get that
$$P f^2\le  4\,M_{\psi}^2 \;\esp\left(  \left[\alpha-\alpha_0 +  \langle X, \beta-\wtbeta \rangle\right]^2 w(X)\right)  \le   A_2^2 \, \delta^2\,.$$
  Lemma 3.4.2 in \citet{vanderVaart:wellner:1996} leads to
$$\esp^{*} \sup_{f\in \itF_{n,\delta}} \sqrt{n} |(P_n-P) f|\le J_{[\;]}\left( A_2 \delta,\itF_{n,\delta}, L_2(P)\right) \left ( 1+ A_1 \frac{J_{[\;]}(A_2 \,\delta,\itF_{n,\delta}, L_2(P))}{A_2^2 \delta^2 \; \sqrt{n}}   \right ) \,,$$ 
where $J_{[\;]}(\delta, \itF, L_2(P)) =\int_0^\delta \sqrt{1+ \log N_{[\;]}(\epsilon, \itF, L_2(P)) } d\epsilon$ is the bracketing integral of the class $\itF$. 
 
Recall that $\|\wtbeta_k-\beta_0\|_{\itH}= O(n^{-\,r\varsigma})$, so for $n$ large enough, $\|\wtbeta_k-\beta_0\|_{\itH}< \epsilon_0/2$, so that, for any $\theta=(\alpha,\beta)\in \Theta_n$, we have    $|\alpha-\alpha_0|+\|\wtbeta_k-\beta\|_{\itH}  < \epsilon_0$. Therefore, $\itF_{n,\delta} \subset   \itG_{n,c, \wttheta_{0,n},\wttheta_{0,n}}$ where $ \itG_{n,c, \wttheta_0,\theta_0^*}$ is defined in Lemma  \ref{lema:bracketing} and we take $\wttheta_{0}=\wttheta_{0,n}=(\alpha_0,\wtbeta_k)$,  $\theta_0^*=\wttheta_{0,n}$ and $c= \epsilon_0$. Hence,  the bound given in Lemma \ref{lema:bracketing}  ensures that leads to 
$$
N_{[\;]}\left( \epsilon,\itF_{n,\delta}, L_2(P)\right)\le \left(\frac{ B_1 }{\epsilon}+1\right)^{ k+1}\,,
$$
for some positive constant  $B_1$   independent of $n$, $\wttheta_{0,n}$ and $\epsilon$. Therefore, for $\delta < B_1/A_2$, we have
\begin{align*}
J_{[\;]}\left( A_2 \delta,\itF_{n,\delta}, L_2(P)\right) &\le  \int_{0}^{A_2\,\delta}  \sqrt{1+\log\left(  \left(\frac{ B_1 }{\epsilon}+1\right)^{k+1}\right)} d\epsilon\\
& \le \int_{0}^{A_2\,\delta}  \sqrt{1+ (k+1)\log\left(2\,\frac{B_1 }{\epsilon}\right)} d\epsilon\\
& \le 2\,  (k+1)^{1/2}\int_{0}^{A_2\,\delta}  \sqrt{1+ \log\left(2\,\frac{B_1 }{\epsilon}\right)} d\epsilon \\
& = 4\, B_1\,  (k+1)^{1/2}  \int_{0}^{B_2\;\delta}  \sqrt{1+ \log\left(\frac{1}{\epsilon}\right)} d\epsilon  \,,
\end{align*}
where $B_2= {A_2}/({2\,B_1})$. Note that  $\int_{0}^{\delta} \sqrt{1+\log(1/\epsilon)}\, d\epsilon=O(\delta  \sqrt{\log(1/\delta)})$ as $\delta\to 0$, hence there exists $\delta_0>0$ and a constant $C>0$ such that for any $\delta<\delta_0$,  $\int_{0}^{\delta} \sqrt{1+\log(1/\epsilon)}\, d\epsilon\le C \, \delta \,\sqrt{\log(1/\delta)}$. This implies that for $\delta<\delta_0/B_2$
$$J_{[\;]}( A_2 \delta,\itF_{n,\delta}, L_2(P)) \lesssim \delta\, \sqrt{\log\left(\frac{1}{\delta}\right)}   \sqrt{k+1}\,.$$
If we denote $q_n = k+1$, we obtain that for some constant $A_3$ independent of $n$ and $\delta$,
$$\esp^{*} \sup_{\theta \in \Theta_{n,\delta} }  \left |\bbG_n 
\left(m_{ \theta }-m_{\wttheta_{0,n} }\right) \right | \leq A_3\,\left[\delta \, q_n^{1/2}  \sqrt{\log\left(\frac{1}{\delta}\right)}    + \frac{ q_n  }{ \sqrt{n}}\; \log\left(\frac{1}{\delta}\right)\right]\,.  $$
Choosing
$$\varphi_n(\delta)=A_3\,\left[\delta \, q_n^{1/2}  \sqrt{\log\left(\frac{1}{\delta}\right)}     + \frac{ q_n  }{ \sqrt{n}} \; \log\left(\frac{1}{\delta}\right)\right]\,,$$
we have that $\varphi_n(\delta)/\delta$ is decreasing in $\delta$, concluding the proof of \eqref{eq:aprobar2}.

 To apply Theorem 3.4.1 of  \citet{vanderVaart:wellner:1996}, it remains to show that $\gamma_n \lesssim \delta_n^{-1}$ and 
 \begin{equation}
 \gamma_n^2\, \varphi_n\left(\frac{1}{\gamma_n}\right) \lesssim  \sqrt{n}\,,
 \label{eq:gamman}
 \end{equation}
 since $\varphi_n(c\delta)\le c \,\phi_n(\delta)$, for $c>1$. 
First note that  $\gamma_n =O( n^{r\varsigma})$ and $\delta_n= C \|\wtbeta-\beta_0\|_{\itH}= O(n^{-\,r\varsigma})$, then  $\gamma_n \lesssim \delta_n^{-1}$.

To derive \eqref{eq:gamman}, observe that 
$$
\gamma_n^2\varphi_n \left(\frac{1}{\gamma_n}\right)=A_3\,\left[\gamma_n  q_n^{1/2} \, \sqrt{\log(\gamma_n)} + 
\gamma_n^2\, \log(\gamma_n)\; \frac{ q_n }{\sqrt{n}}\right] =A_3\,\left[\sqrt{n}\; a_n(1+a_n) \right]\, ,
$$ 
where $a_n=\gamma_n \, \sqrt{\log(\gamma_n)}\;   q_n^{1/2}/\sqrt{n}$.
Hence, to derive that $\gamma_n^2\varphi_n \left(1/{\gamma_n}\right)\lesssim \sqrt{n}$, it is 
enough to show that $a_n=O(1)$, which follows easily since    $q_n=O(n^{\varsigma})$ and
$\gamma_n \sqrt{\log(\gamma_n)} = O(n^{(1-\varsigma)/2})$, concluding the proof of  \eqref{eq:gamman}.  

Hence, from Theorem 3.4.1 of  \citet{vanderVaart:wellner:1996}, we get that $\gamma_n  \wtpi_{\prob} (\wtheta, \wttheta_{0,n}) =O_{\prob}(1)$. As noticed above, 
$$\wtpi_{\prob} ( \theta_0, \wttheta_{0,n})\le  \left\{\esp \left\{w(X)\|X\|^2\right\} \|\wtbeta- \beta_0\|_{\itH}^2\right\}^{1/2}=O(n^{-\, r\varsigma})\,.$$
Then, using that $\gamma_n =O( n^{r\varsigma})$, we get that $\gamma_n  \wtpi_{\prob} ( \theta_0, \wttheta_{0,n}) =O_{\prob}(1)$ and from the triangular inequality we obtain that   $\gamma_n  \wtpi_{\prob} (\wtheta, \theta_{0}) =O_{\prob}(1)$, as desired.
\end{proof}

\begin{proof}[Proof of Proposition \ref{prop:cotainf}]
From Lemma \ref{lemma:L_function},  we have that there exists a constant $ C_0>0$   independent from $n$ such that, for any $\theta$,
$$ L(\theta) - L(\theta_0) \ge C_0\,\pi_\prob^2(\theta,\theta_0)\,,$$
then to show that  \ref{ass:cotainf} holds, it will be enough to show that there exists a constant $C_1>0$   such that,  for any $\theta=(\alpha,\beta) \in \real\times  \itH$ with $|\alpha-\alpha_0|+\|\beta-\beta_0\|<1$, 
$$\pi_\prob^2(\theta,\theta_0)\ge C_1\,\wtpi_\prob^2(\theta,\theta_0)\,,$$
and then take $\epsilon_0=1$ and $C_0^{\star}=C_0 C_1$.

Since $\prob(\|X\|\le C)=1$,  we have that, with probability one, for any $|\alpha-\alpha_0|+\|\beta-\beta_0\|<1$,
$$\left|\alpha+\langle X,  \beta \rangle\right|\le  |\alpha_0|+1 + C\left(\|\beta_0\|+1\right)=C^{\star}\,.$$
Thus, using that $F$ is strictly increasing we get that 
\begin{equation}\label{eq:acotoFalfa}
 0<A_1=F\left(-C^{\star}\right)\le  F\left(\alpha+\langle X,  \beta \rangle\right) \le F\left(C^{\star}\right)=A_2<1\,.
 \end{equation}
The Mean Value Theorem implies that given $\theta=(\alpha,\beta) \in \real\times  \itH$ with $|\alpha-\alpha_0|+\|\beta-\beta_0\|<1$ there exists $(\alpha^{\star}_X,\beta^{\star}_X)=(1-\omega_X) \theta + \omega_X \theta_0$, $0\le \omega_X\le 1$, such that
\begin{align*}
\pi_\prob^2(\theta ,\theta_0 ) &= \esp\left\{w(X) \left[F(\alpha+\langle X,  \beta \rangle)- F(\alpha_0+\langle X,  \beta_0 \rangle )\right]^2\right\}\\
&=  \esp\left\{w(X) \left\{ F\left(\alpha_X^{\star}+\langle X,  \beta_X^{\star} \rangle\right)\left[1-F\left(\alpha_X^{\star}+\langle X,  \beta_X^{\star} \rangle\right)\right]\left(\alpha-\alpha_0+\langle X, \beta- \beta_0 \rangle \right)\right\}^2\right\}\\
& \ge A_1^2 (1-A_2)^2  \esp\left\{w(X)\left[\alpha-\alpha_0+\langle X, \beta- \beta_0 \rangle \right]^2\right\}= A_1^2 (1-A_2)^2\wtpi_\prob^2(\theta,\theta_0)\,,
 \end{align*}
 where the last inequality follows from \eqref{eq:acotoFalfa}, since $|\alpha^{\star}_X-\alpha_0|+\|\beta^{\star}_X-\beta_0\|<1$ and the proof is concluded taking $C_1=A_1^2 (1-A_2)^2$.
\end{proof}
%%%%%%%%%%%%%%%%%%%%%%%%%%%%%%%%%%%
 %REFERENCES
%%%%%%%%%%%%%%%%%%%%%%%%%%%%%%
\small
%\nocite{*}

%\nocite{*}
%\bibliographystyle{apalike}
%\bibliography{referencias2}

\newcommand{\noopsort}[1]{}

\end{document}

%% file: graficos-poda0-sinC6.tex
\begin{figure}[tp]
 \begin{center}
 \footnotesize
 \renewcommand{\arraystretch}{0.2}
 \newcolumntype{M}{>{\centering\arraybackslash}m{\dimexpr.01\linewidth-1\tabcolsep}}
   \newcolumntype{G}{>{\centering\arraybackslash}m{\dimexpr.45\linewidth-1\tabcolsep}}
%\begin{tabular}{MGG}
\begin{tabular}{GG}
  $\wbeta_{\clas}$ & $\wbeta_{\eme}$   \\[-3ex]
     \includegraphics[scale=0.40]{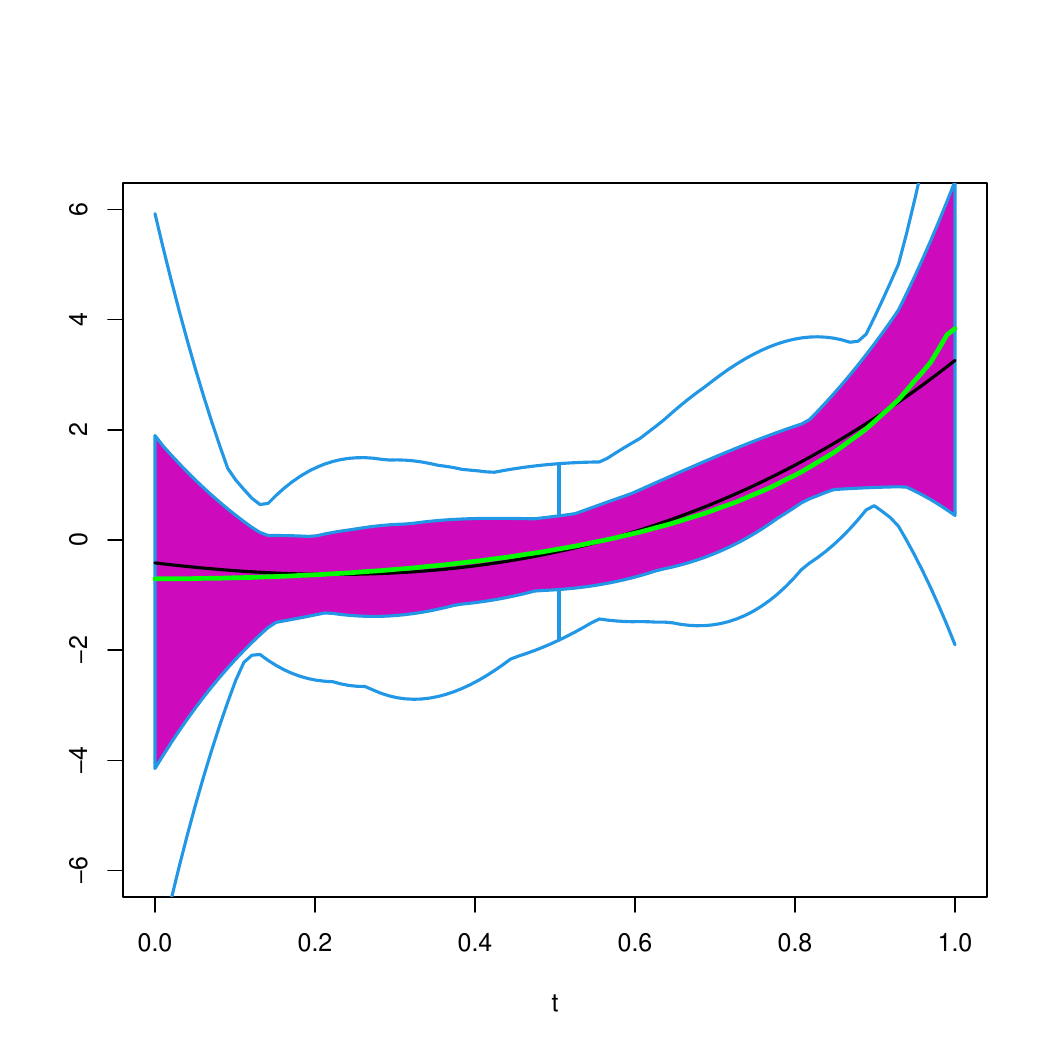} 
&  \includegraphics[scale=0.40]{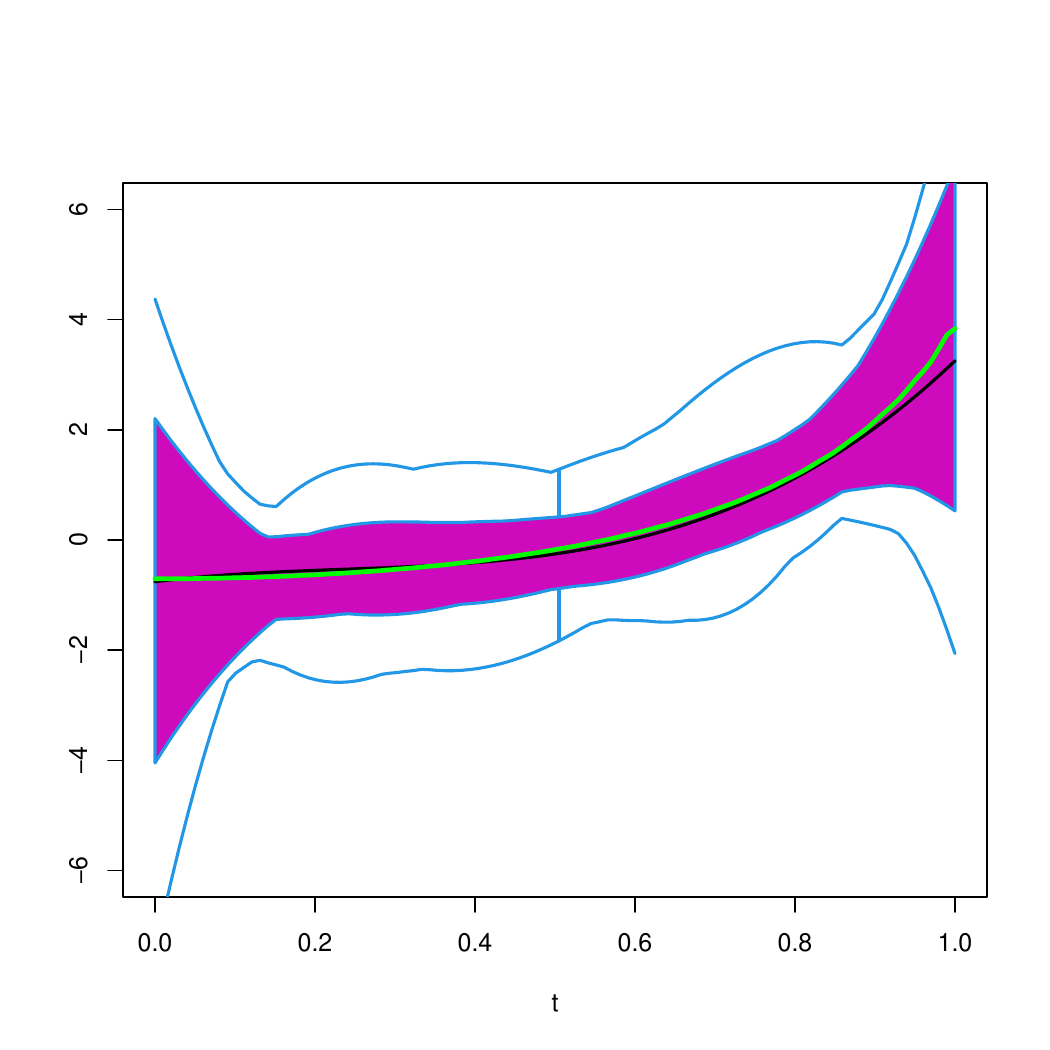} \\
   $\wbeta_{\wclHR}$ & $\wbeta_{\wemeHR}$ \\[-3ex] 
   \includegraphics[scale=0.40]{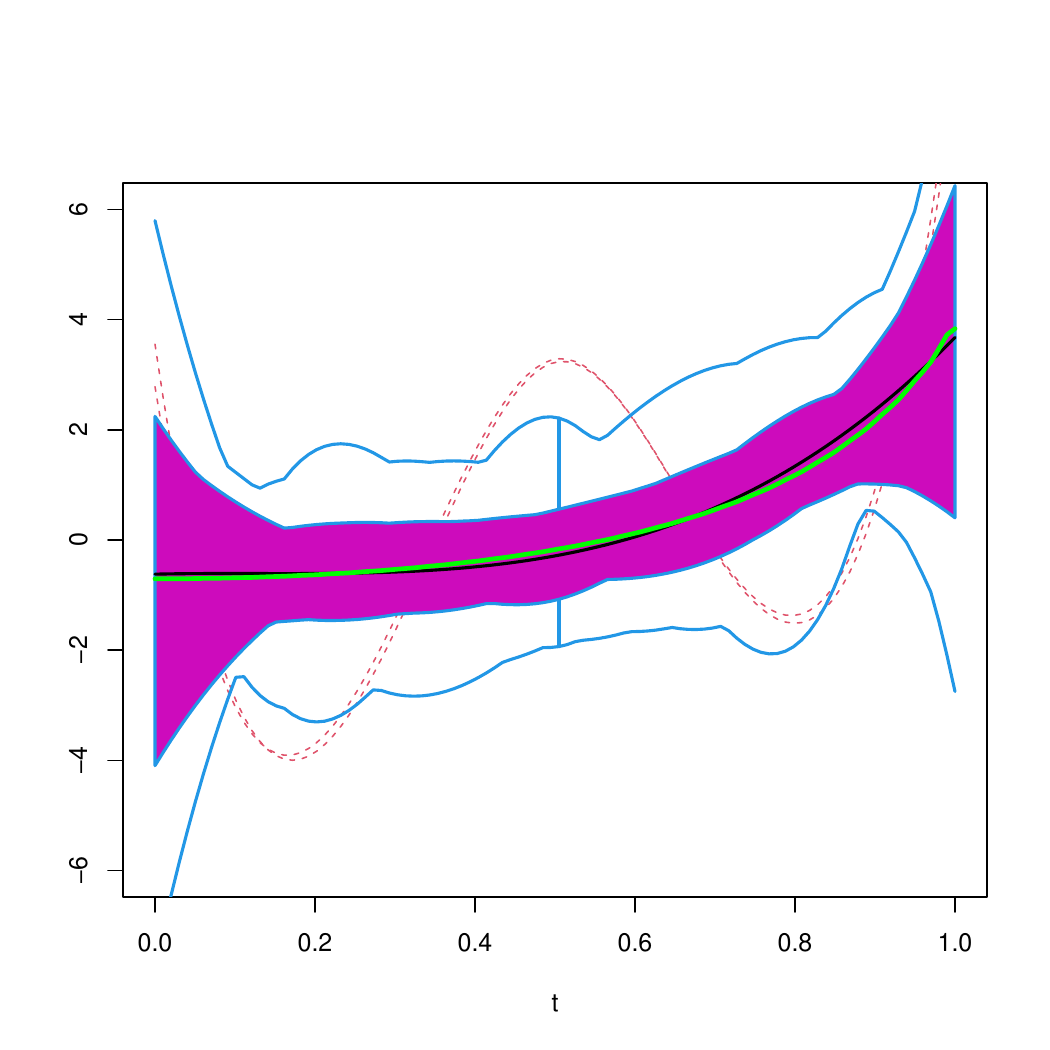}
&  \includegraphics[scale=0.40]{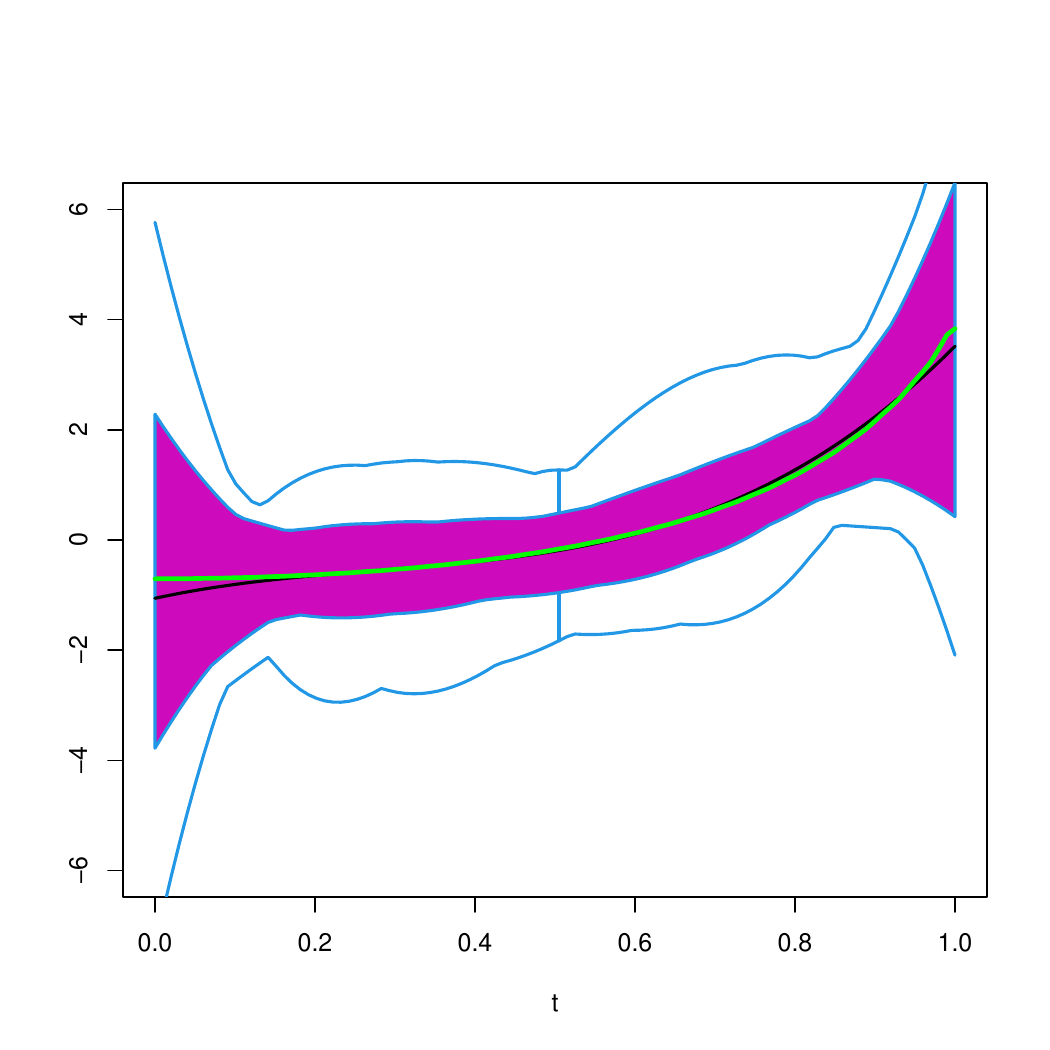} 
  \\
   $\wbeta_{\wclBOX}$ & $\wbeta_{\wemeBOX}$ \\[-3ex]
\includegraphics[scale=0.40]{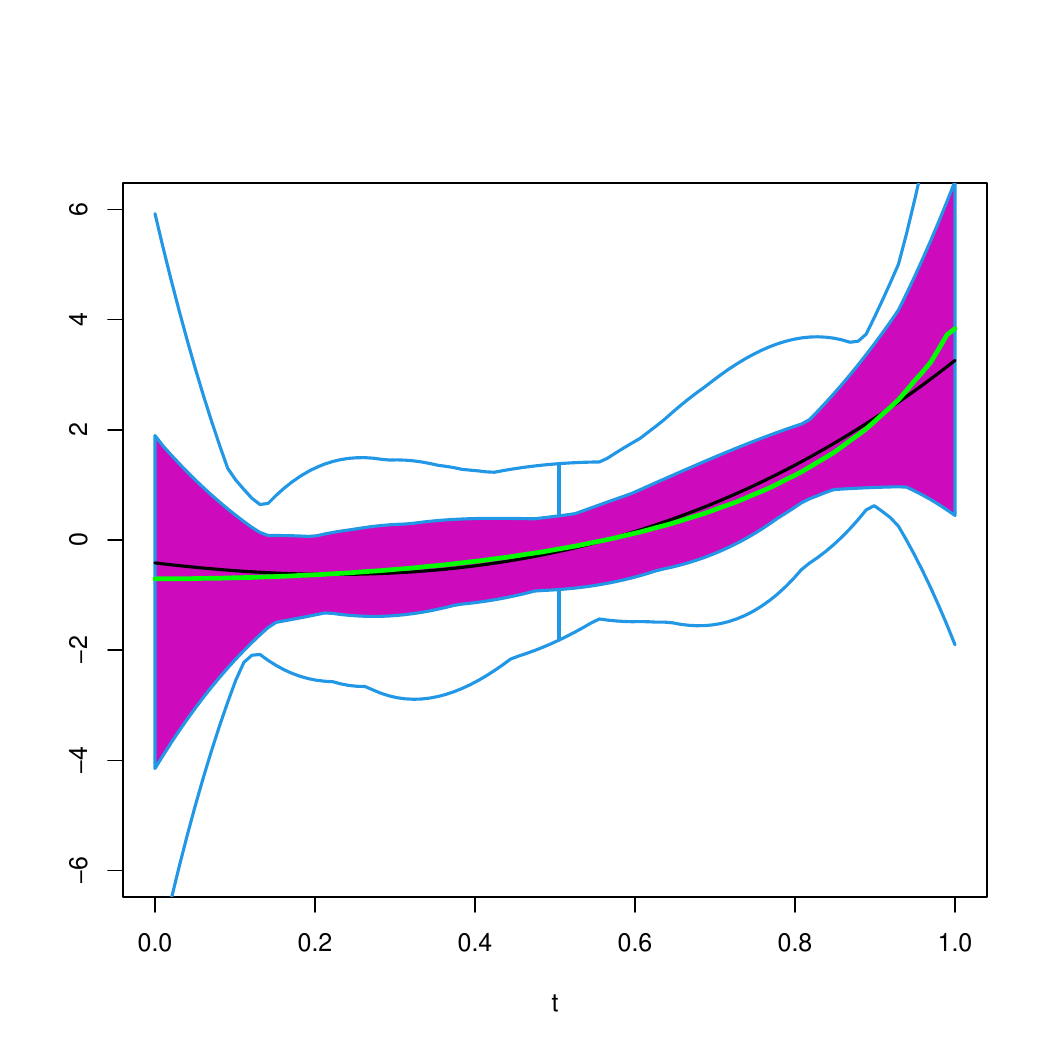}
&   \includegraphics[scale=0.40]{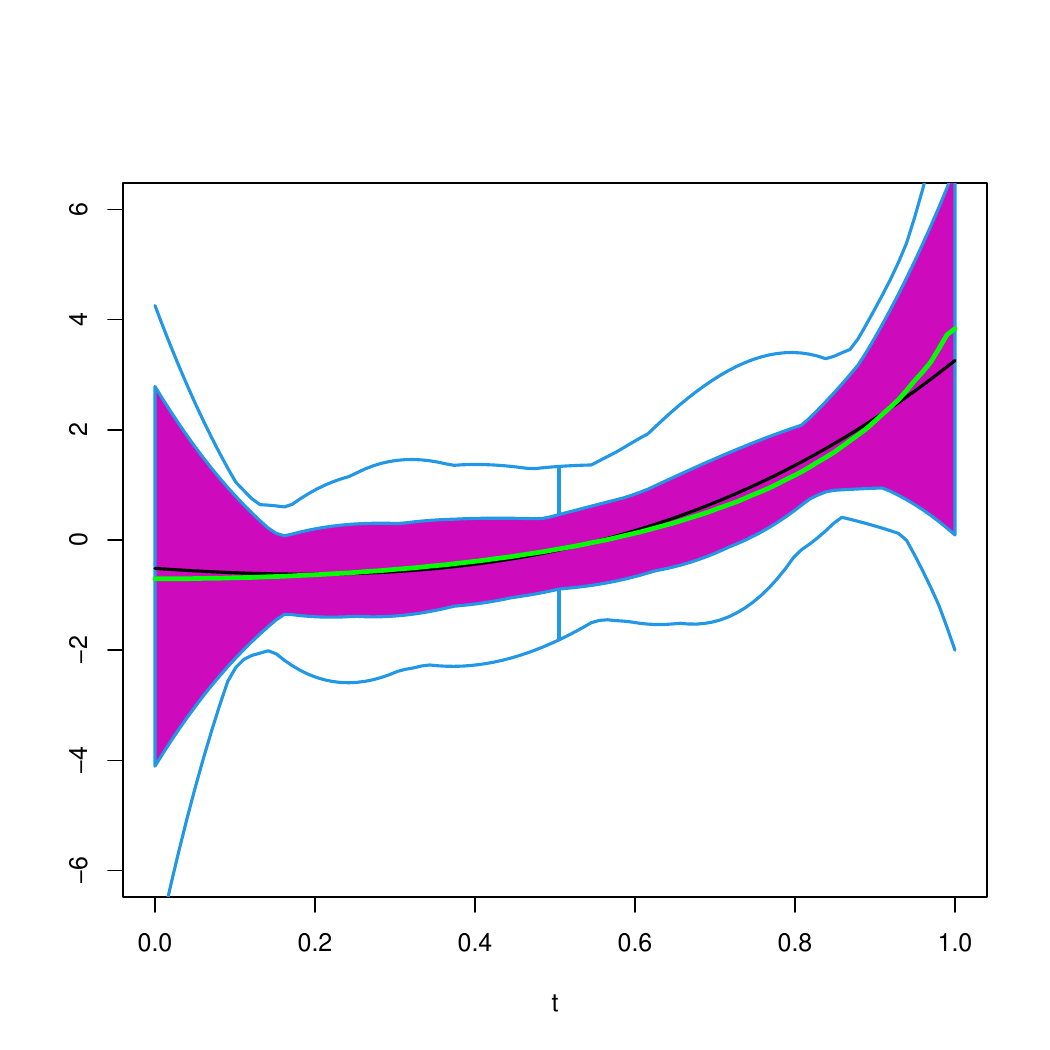} 
 \end{tabular}
\caption{\small \label{fig:wbeta-C0-poda0}  Functional boxplot of the estimators for $\beta_0$ under $C_{0}$ within the interval $[0,1]$. 
The true function is shown with a green dashed line, while the black solid one is the central 
curve of the $n_R = 1000$ estimates $\wbeta$. }
\end{center} 
\end{figure} 

\begin{figure}[tp]
 \begin{center}
 \footnotesize
 \renewcommand{\arraystretch}{0.2}
 \newcolumntype{M}{>{\centering\arraybackslash}m{\dimexpr.01\linewidth-1\tabcolsep}}
   \newcolumntype{G}{>{\centering\arraybackslash}m{\dimexpr.45\linewidth-1\tabcolsep}}
%\begin{tabular}{MGG}
\begin{tabular}{GG}
  $\wbeta_{\clas}$ & $\wbeta_{\eme}$   \\[-3ex]   
 
\includegraphics[scale=0.40]{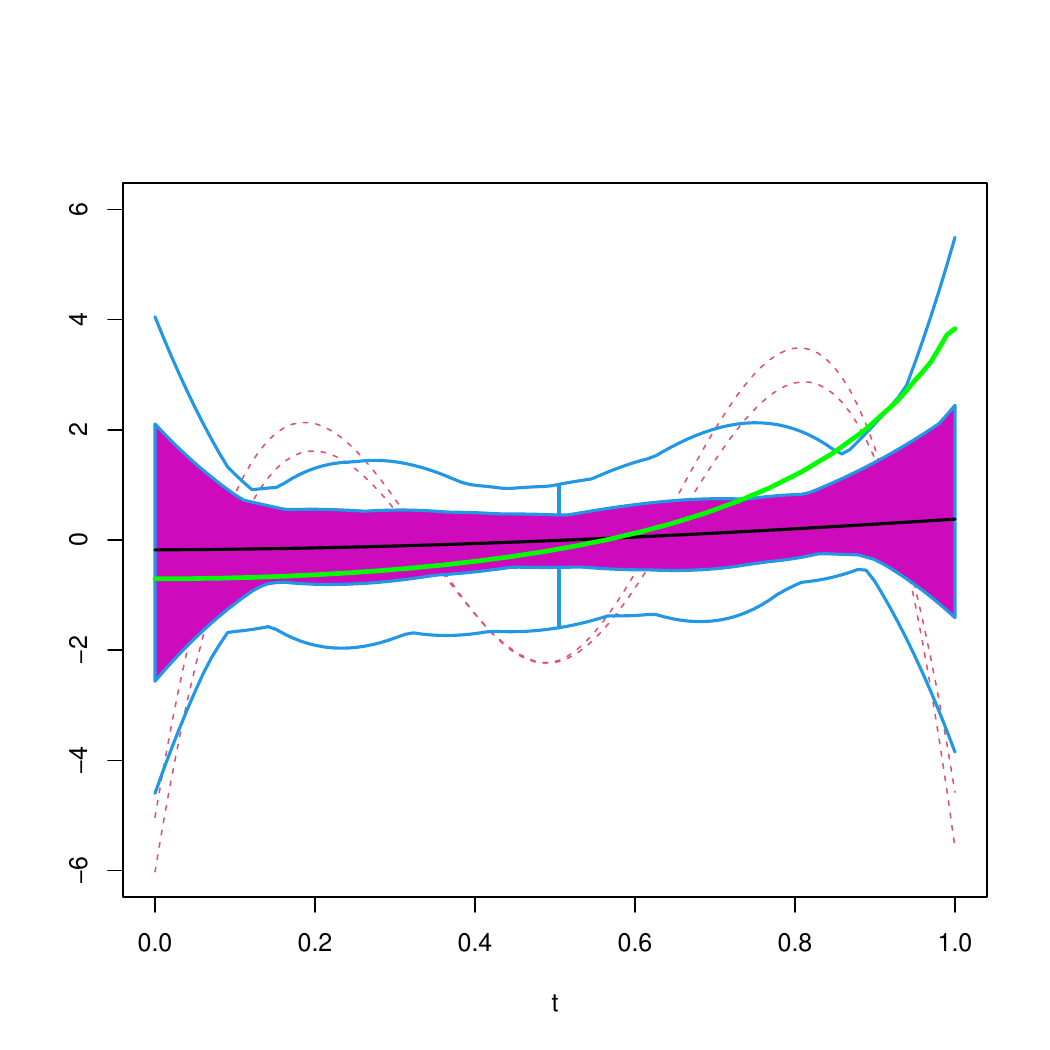}
 &  \includegraphics[scale=0.40]{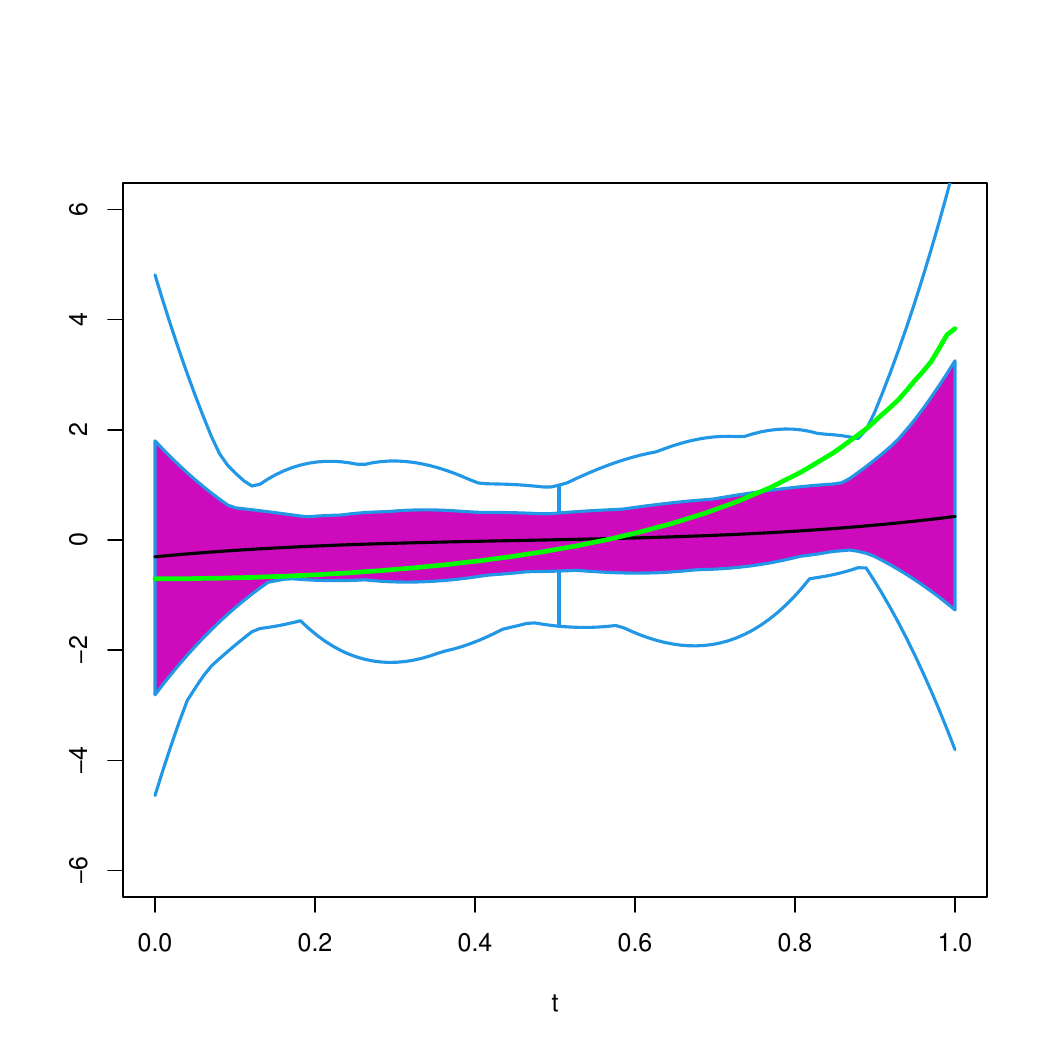}\\
   $\wbeta_{\wclHR}$ & $\wbeta_{\wemeHR}$ \\[-3ex] 
    \includegraphics[scale=0.40]{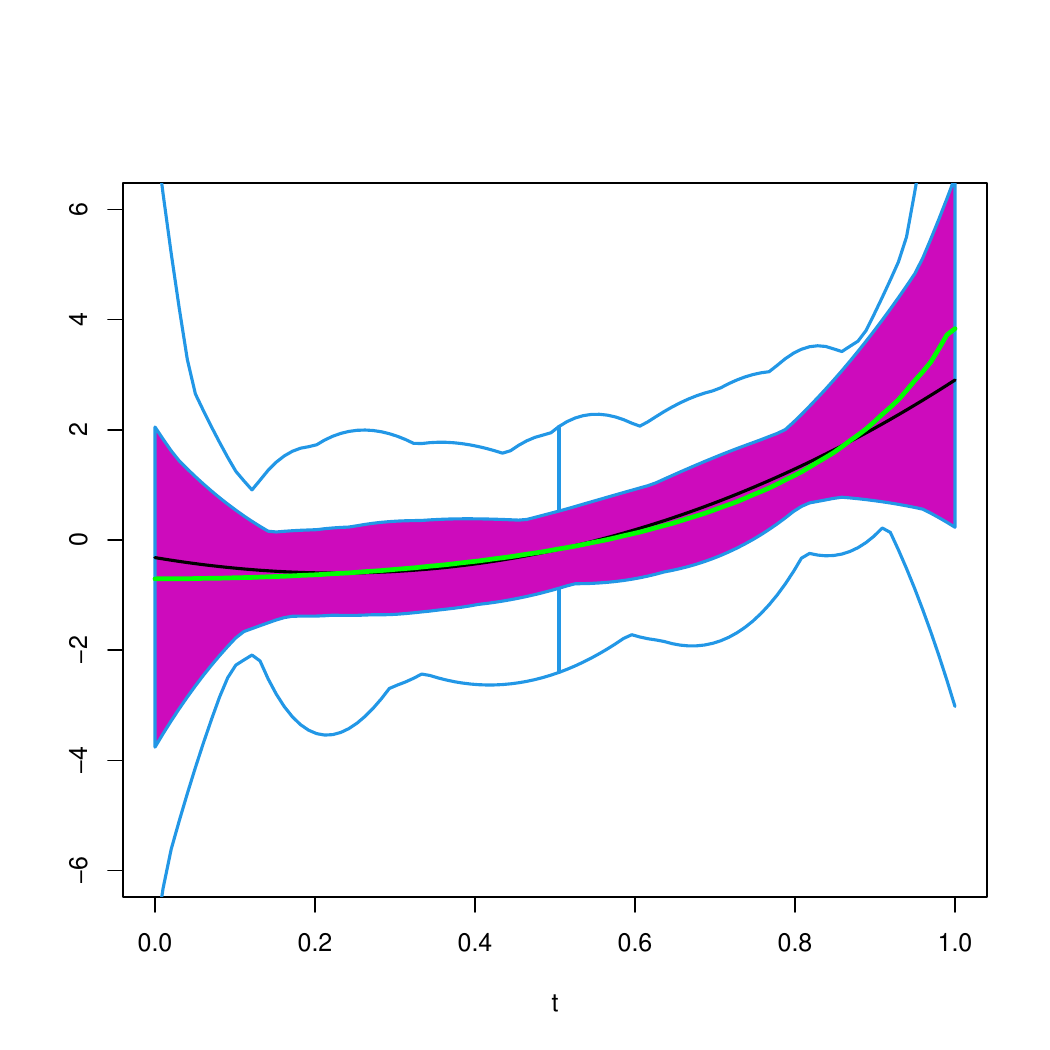}
  &  \includegraphics[scale=0.40]{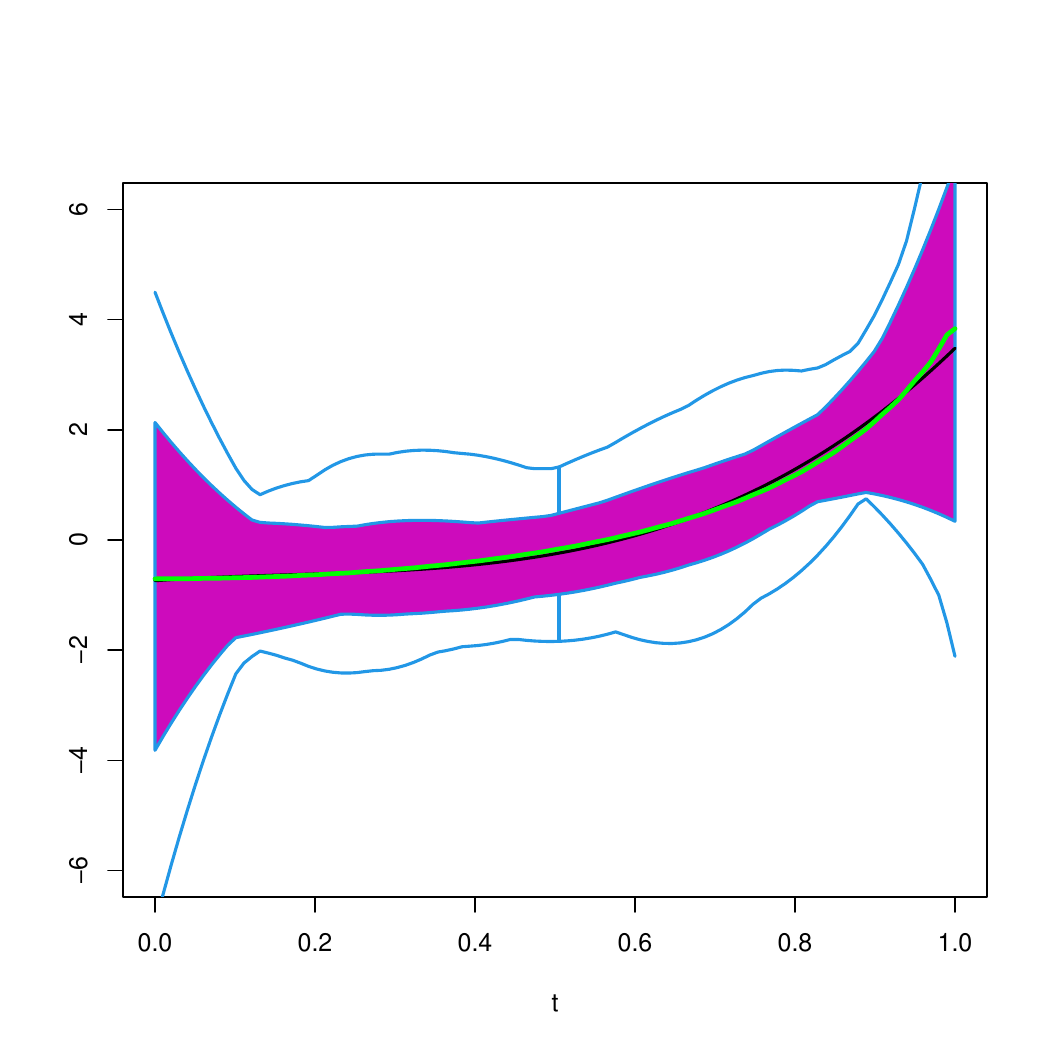}
   \\
   $\wbeta_{\wclBOX}$ & $\wbeta_{\wemeBOX}$ \\[-3ex]
  \includegraphics[scale=0.40]{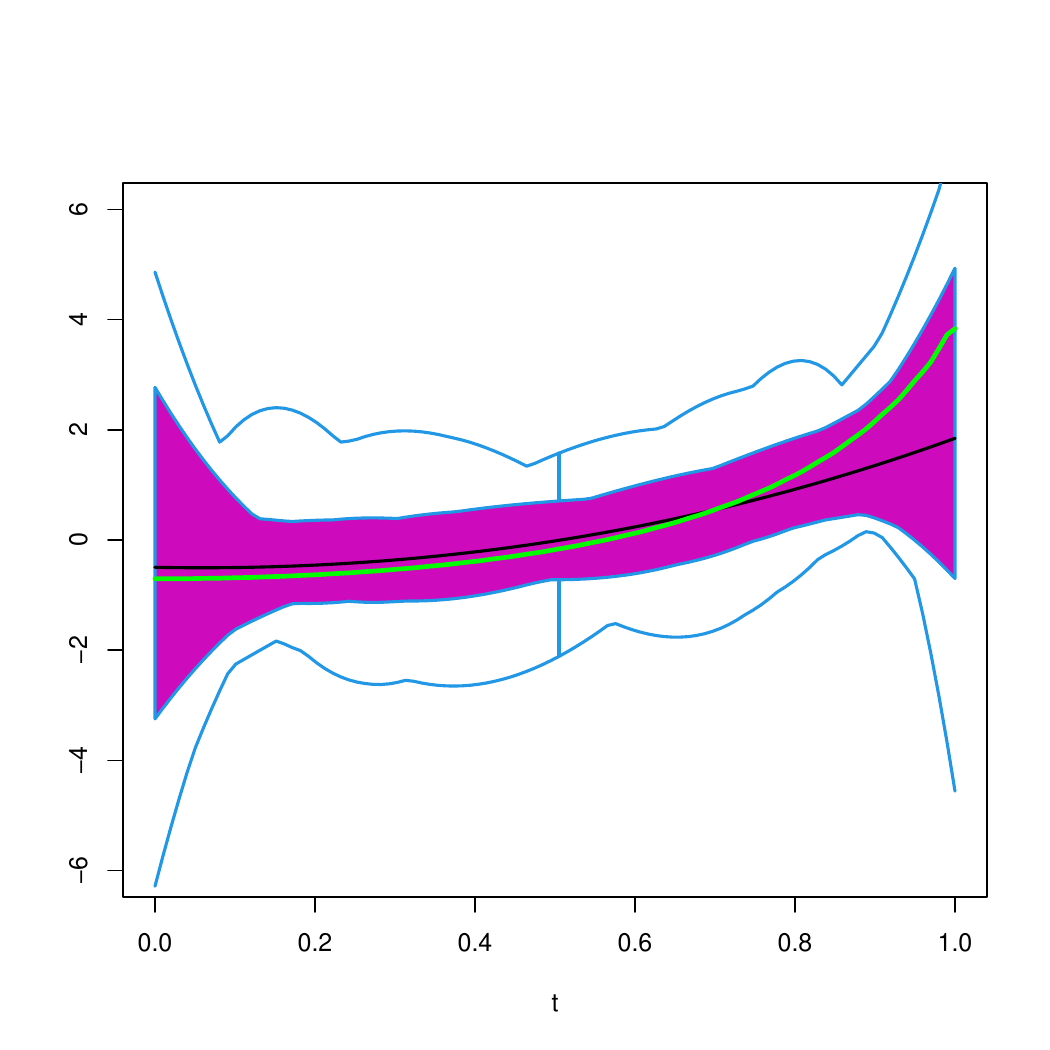}
  &  \includegraphics[scale=0.40]{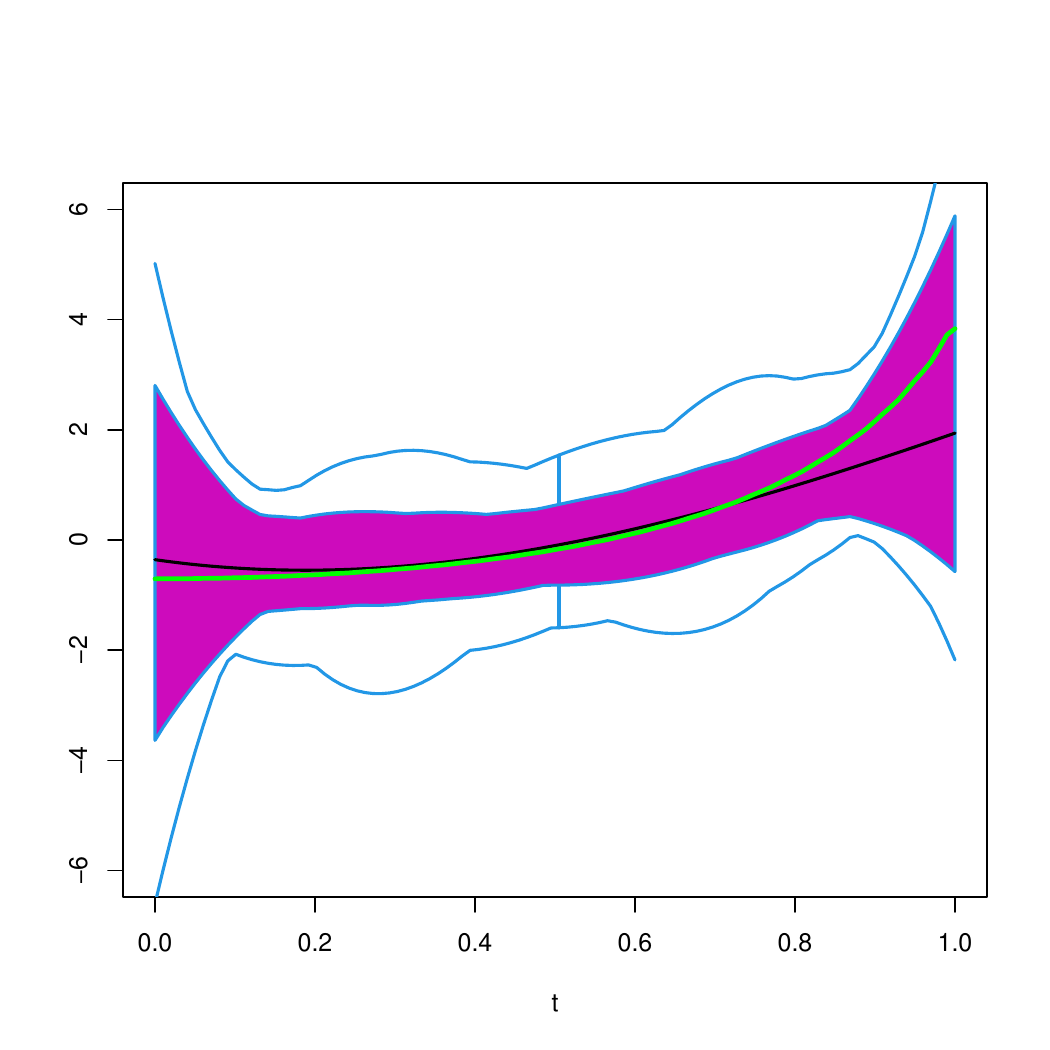}

\end{tabular}
\caption{\small \label{fig:wbeta-C15-poda0}  Functional boxplot of the estimators for $\beta_0$ under $C_{1,0.05}$  within the interval $[0,1]$. 
The true function is shown with a green dashed line, while the black solid one is the central 
curve of the $n_R = 1000$ estimates $\wbeta$.  }
\end{center} 
\end{figure}

\begin{figure}[tp]
 \begin{center}
 \footnotesize
 \renewcommand{\arraystretch}{0.2}
 \newcolumntype{M}{>{\centering\arraybackslash}m{\dimexpr.01\linewidth-1\tabcolsep}}
   \newcolumntype{G}{>{\centering\arraybackslash}m{\dimexpr.45\linewidth-1\tabcolsep}}
%\begin{tabular}{MGG}
\begin{tabular}{GG}
  $\wbeta_{\clas}$ & $\wbeta_{\eme}$   \\[-3ex]   
\includegraphics[scale=0.40]{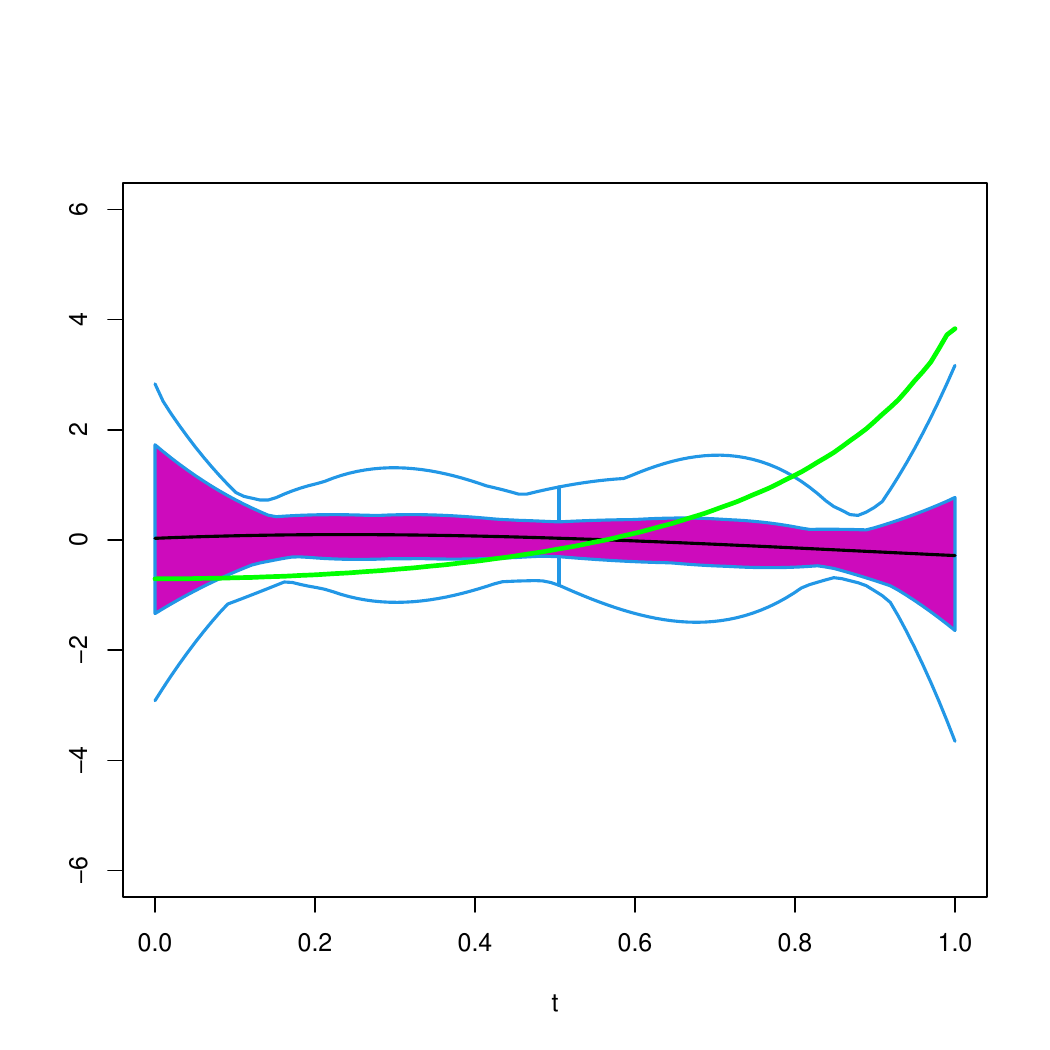}
&   \includegraphics[scale=0.40]{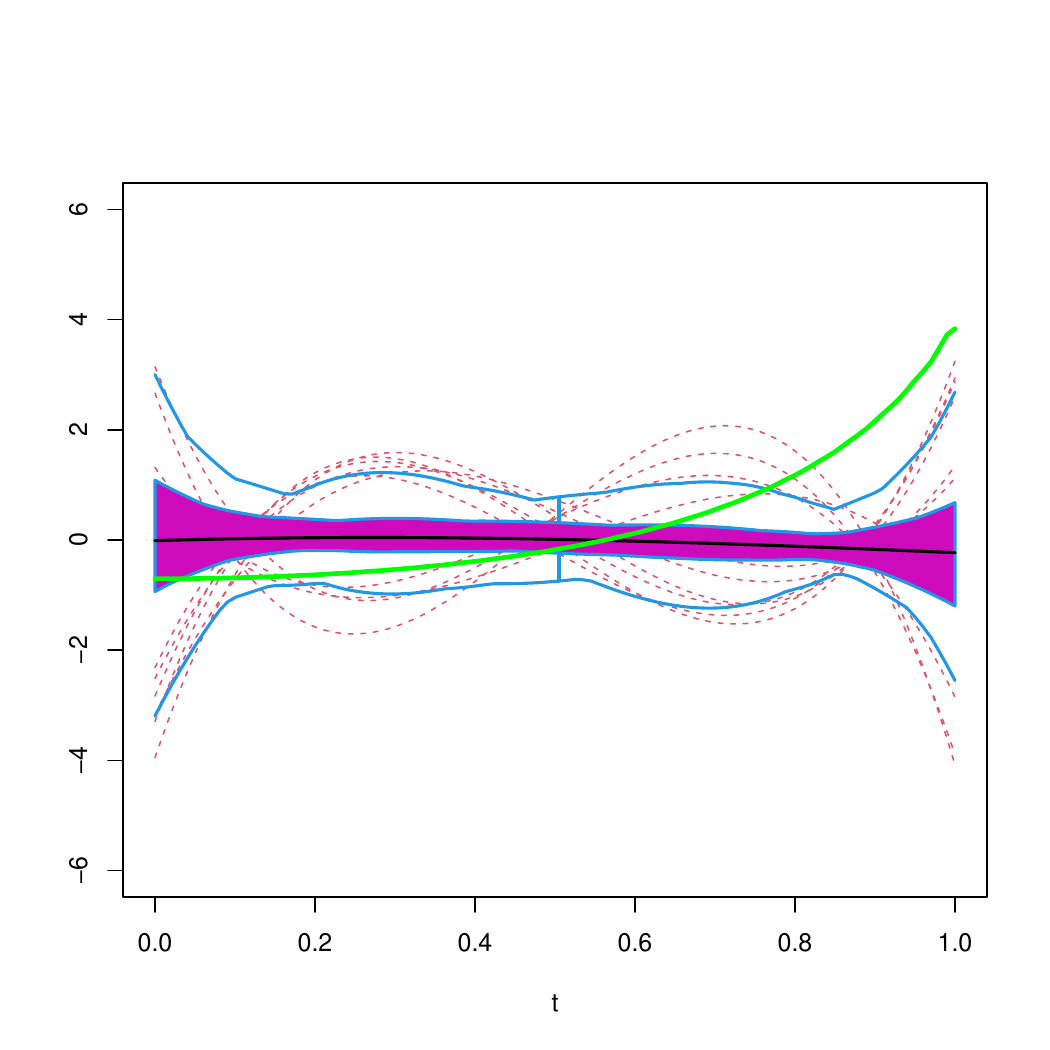}\\
   $\wbeta_{\wclHR}$ & $\wbeta_{\wemeHR}$ \\[-3ex] 
    \includegraphics[scale=0.40]{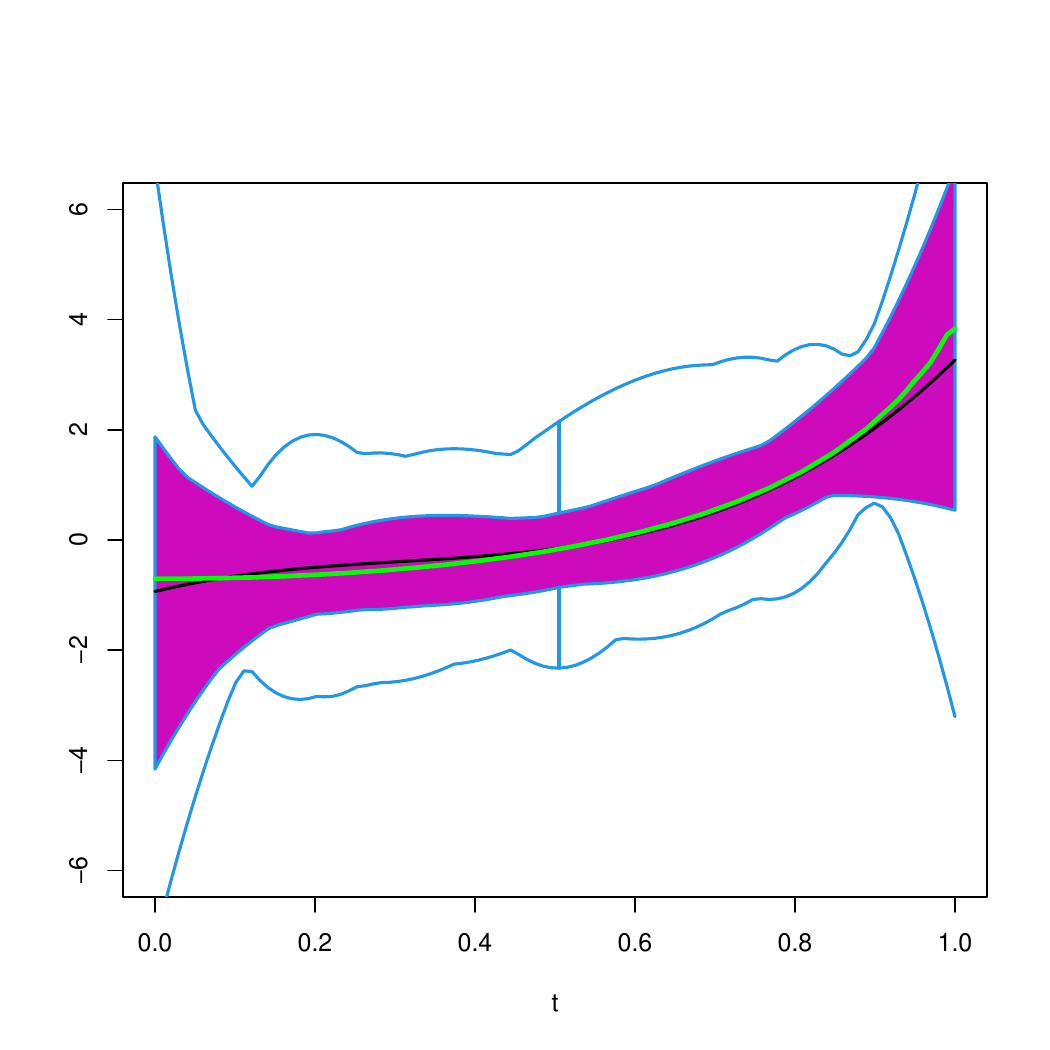}
  &  \includegraphics[scale=0.40]{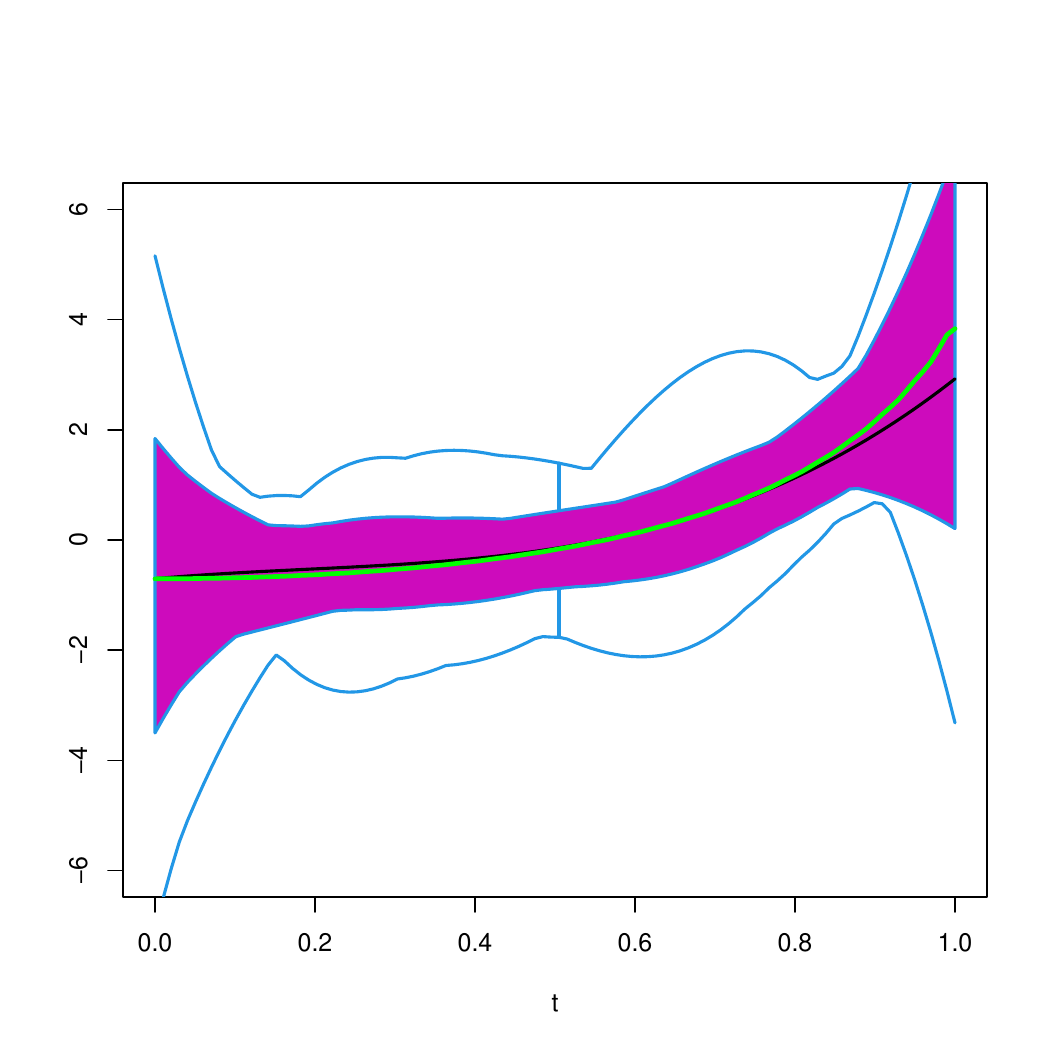}
   \\
   $\wbeta_{\wclBOX}$ & $\wbeta_{\wemeBOX}$ \\[-3ex]
    \includegraphics[scale=0.40]{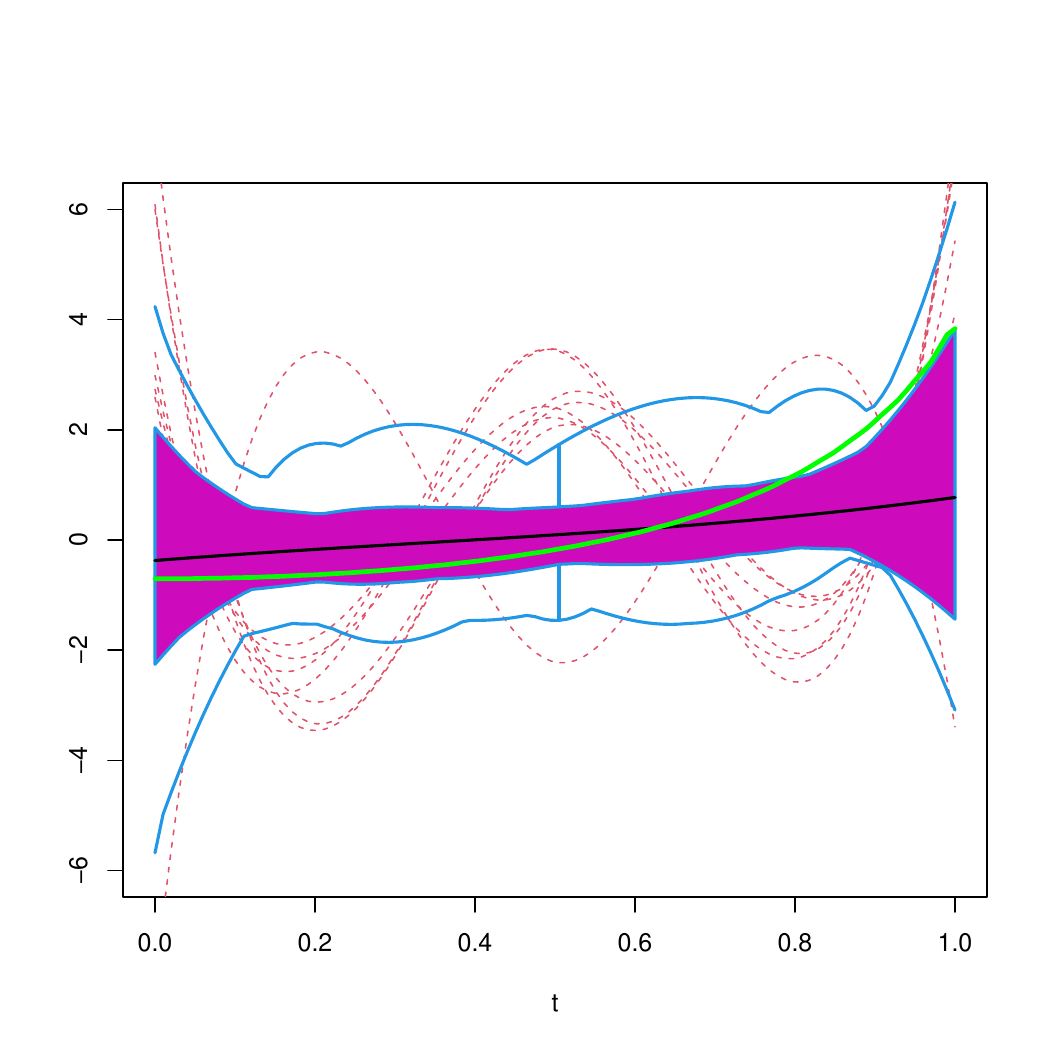}
  &  \includegraphics[scale=0.40]{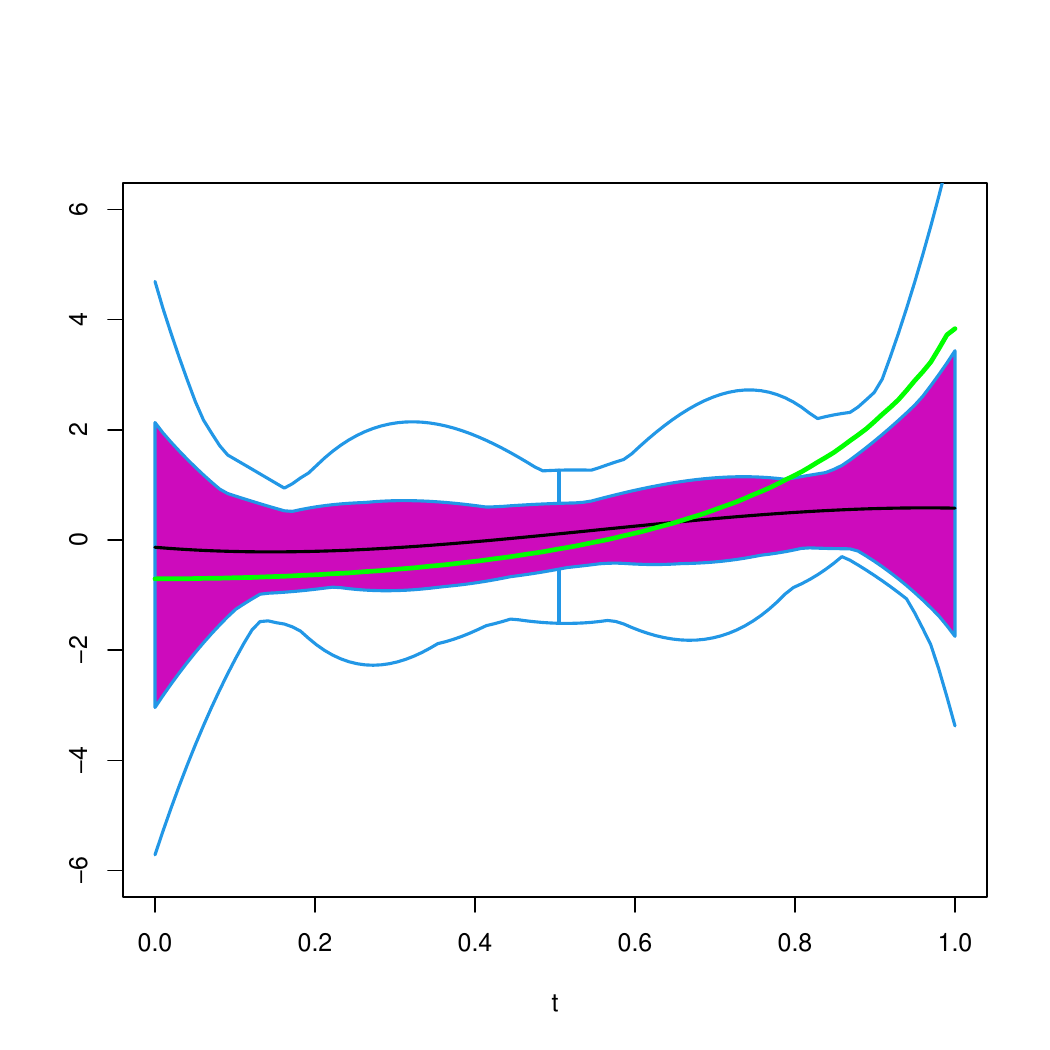}

\end{tabular}
\caption{\small \label{fig:wbeta-C110-poda0}  Functional boxplot of the estimators for $\beta_0$ under $C_{1,0.10}$  within the interval $[0,1]$. 
The true function is shown with a green dashed line, while the black solid one is the central 
curve of the $n_R = 1000$ estimates $\wbeta$. }
\end{center} 
\end{figure}  

\begin{figure}[tp]
 \begin{center}
 \footnotesize
 \renewcommand{\arraystretch}{0.2}
 \newcolumntype{M}{>{\centering\arraybackslash}m{\dimexpr.01\linewidth-1\tabcolsep}}
   \newcolumntype{G}{>{\centering\arraybackslash}m{\dimexpr.45\linewidth-1\tabcolsep}}
%\begin{tabular}{MGG}
\begin{tabular}{GG}
  $\wbeta_{\clas}$ & $\wbeta_{\eme}$   \\[-3ex]      
 \includegraphics[scale=0.40]{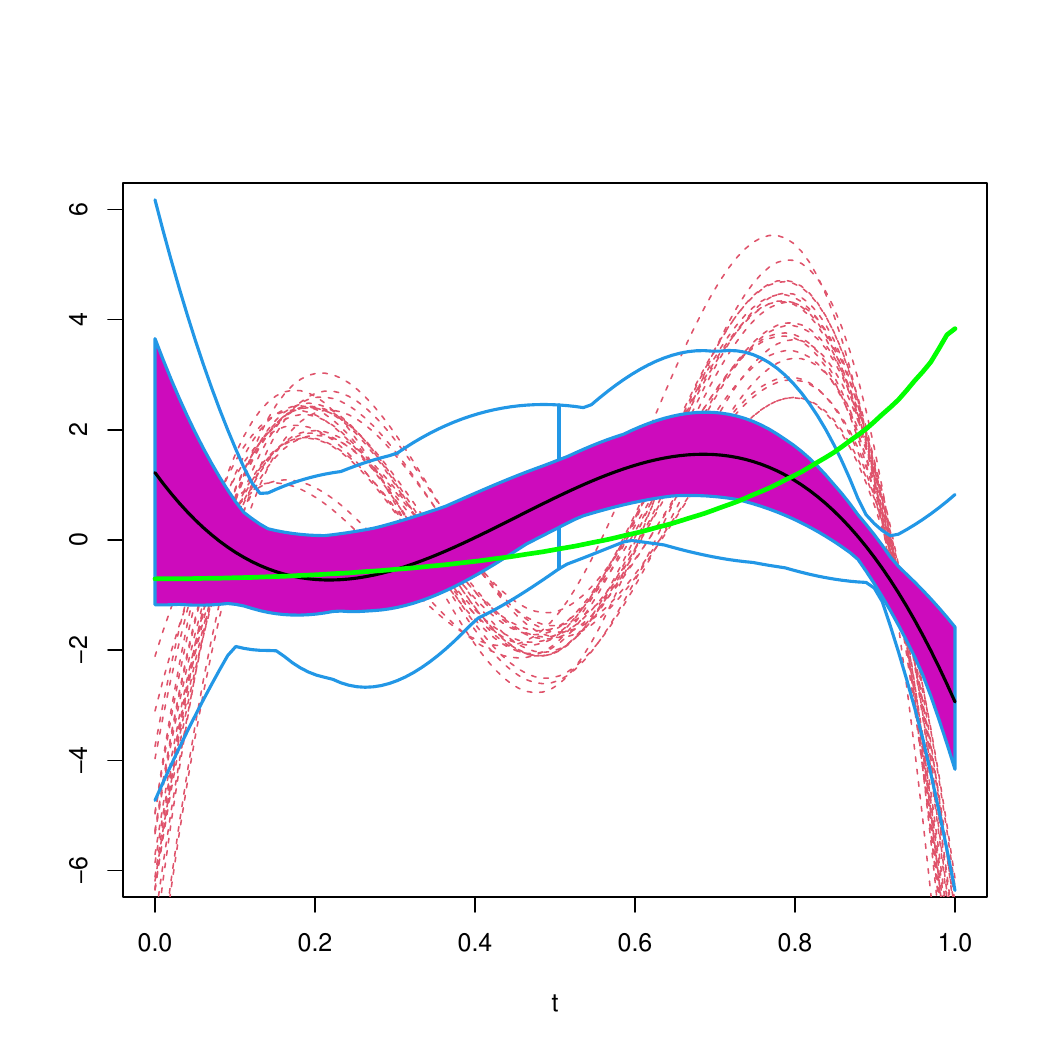}
 &  \includegraphics[scale=0.40]{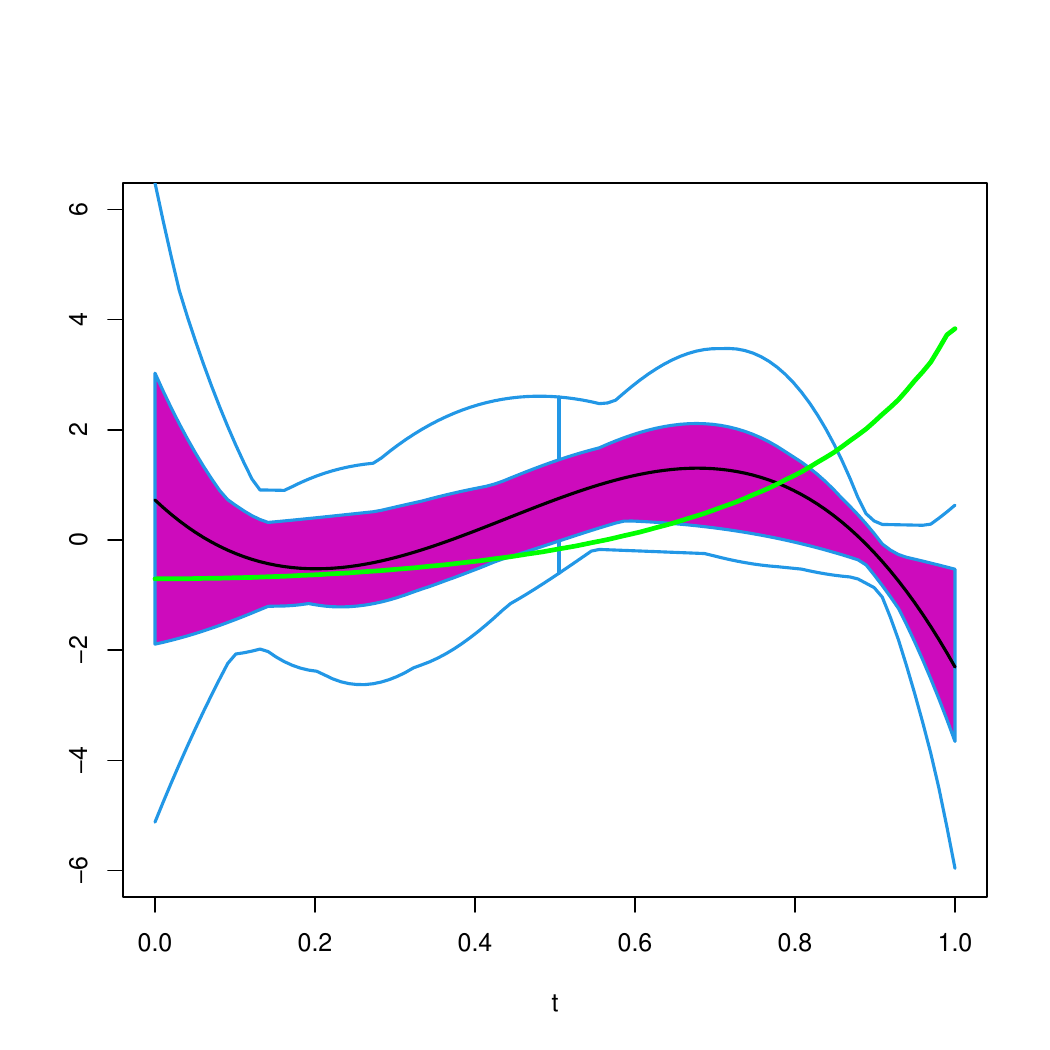}\\
   $\wbeta_{\wclHR}$ & $\wbeta_{\wemeHR}$ \\[-3ex] 
    \includegraphics[scale=0.40]{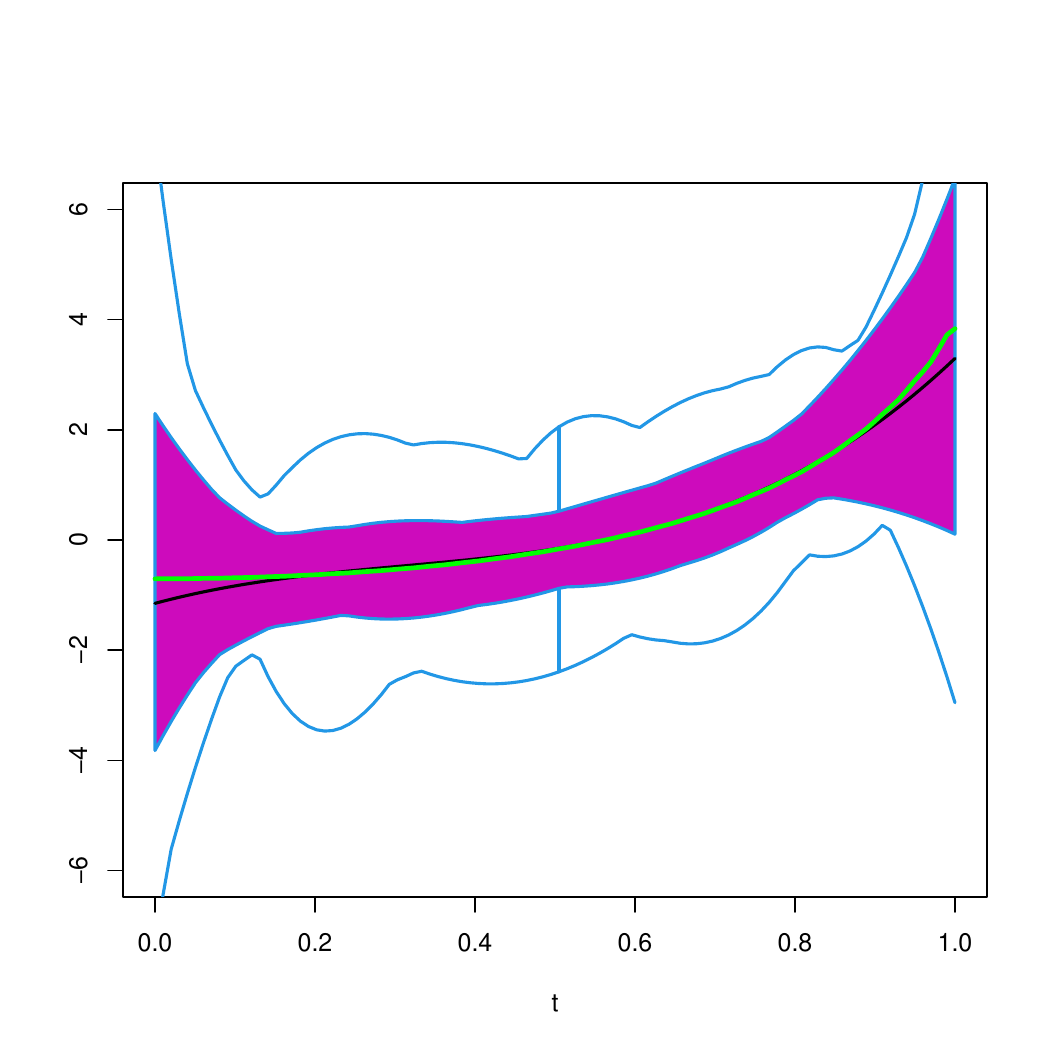}
  &  \includegraphics[scale=0.40]{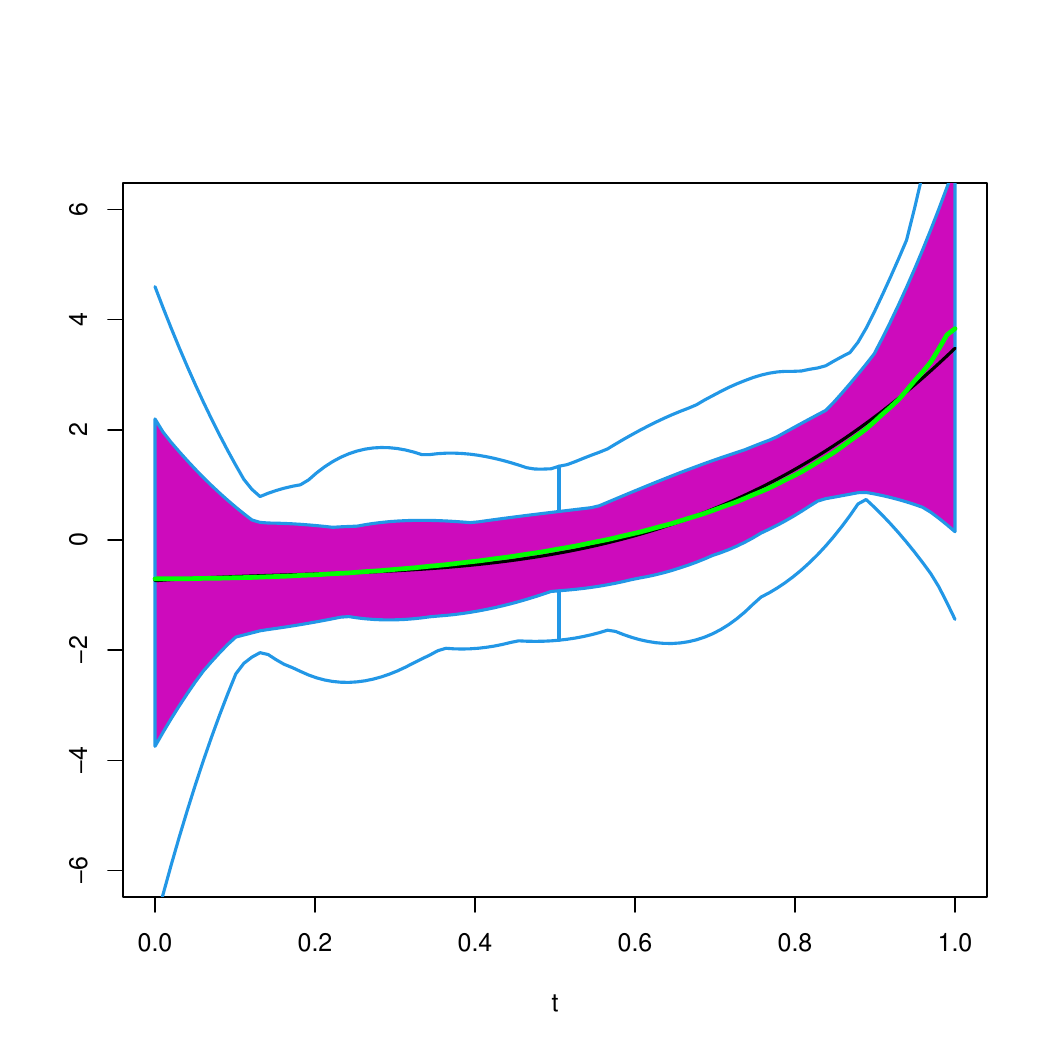}
   \\
   $\wbeta_{\wclBOX}$ & $\wbeta_{\wemeBOX}$ \\[-3ex]
 \includegraphics[scale=0.40]{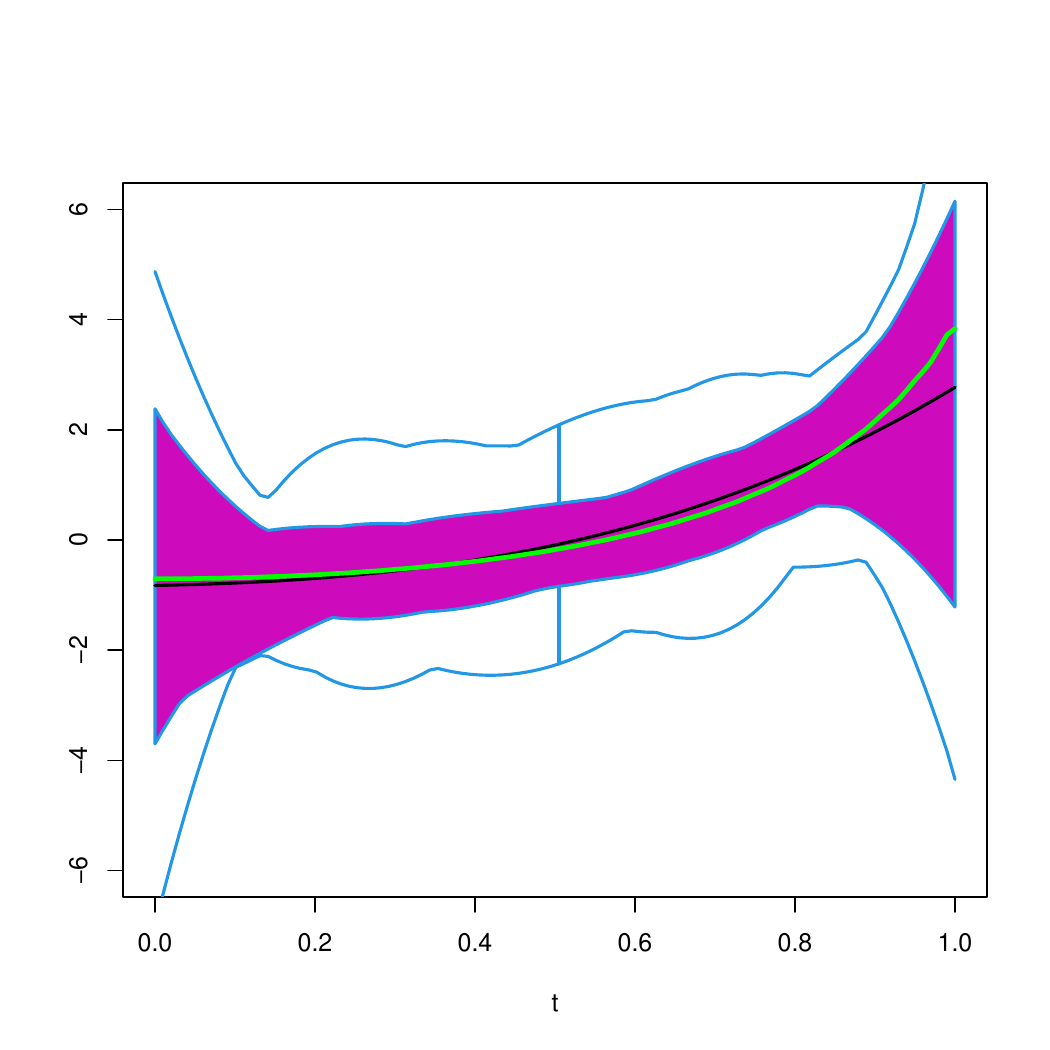}
  &  \includegraphics[scale=0.40]{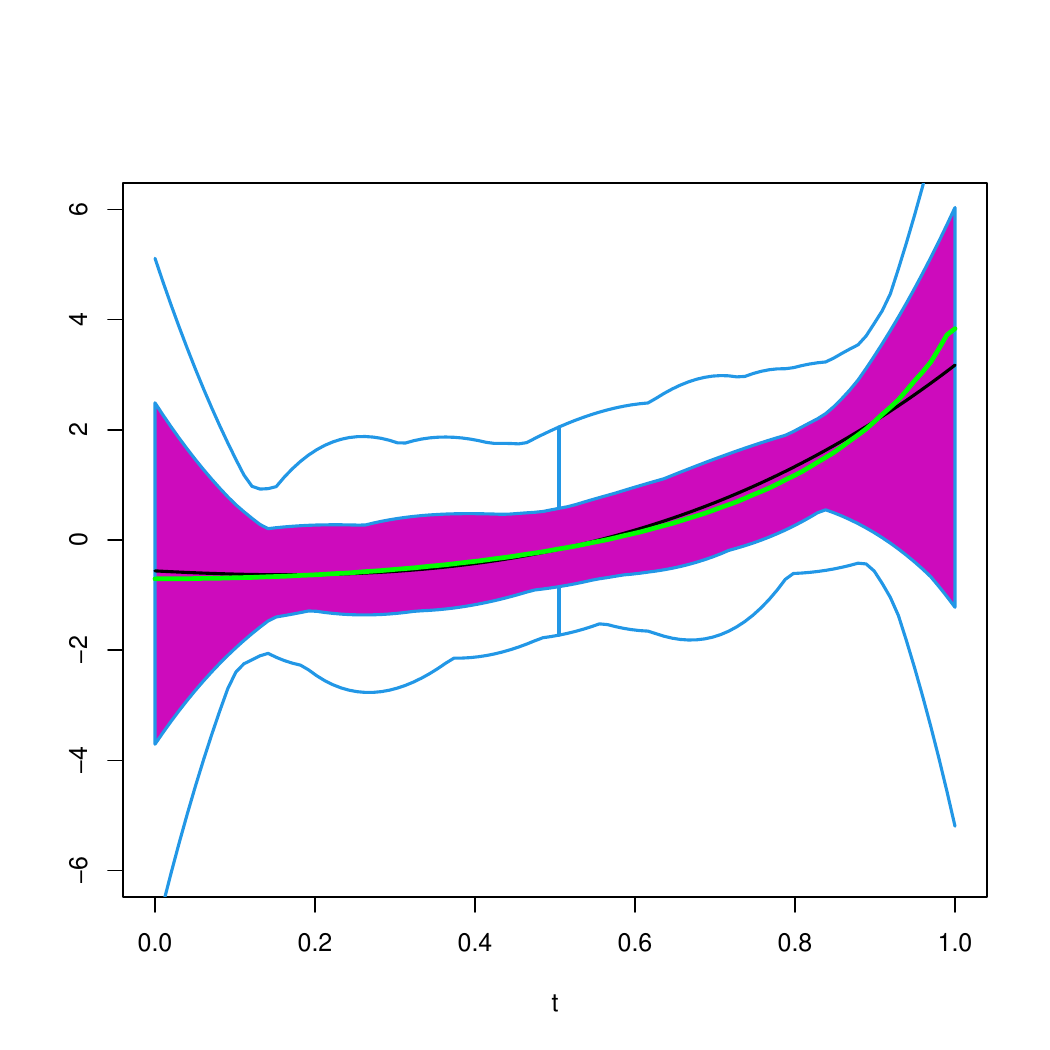}
  
\end{tabular}
\caption{\small \label{fig:wbeta-C25-poda0}  Functional boxplot of the estimators for $\beta_0$ under $C_{2,0.05}$  within the interval $[0,1]$. 
The true function is shown with a green dashed line, while the black solid one is the central 
curve of the $n_R = 1000$ estimates $\wbeta$.  }
\end{center} 
\end{figure}

\begin{figure}[tp]
 \begin{center}
 \footnotesize
 \renewcommand{\arraystretch}{0.2}
 \newcolumntype{M}{>{\centering\arraybackslash}m{\dimexpr.01\linewidth-1\tabcolsep}}
   \newcolumntype{G}{>{\centering\arraybackslash}m{\dimexpr.45\linewidth-1\tabcolsep}}
%\begin{tabular}{MGG}
\begin{tabular}{GG}
  $\wbeta_{\clas}$ & $\wbeta_{\eme}$   \\[-3ex]      
\includegraphics[scale=0.40]{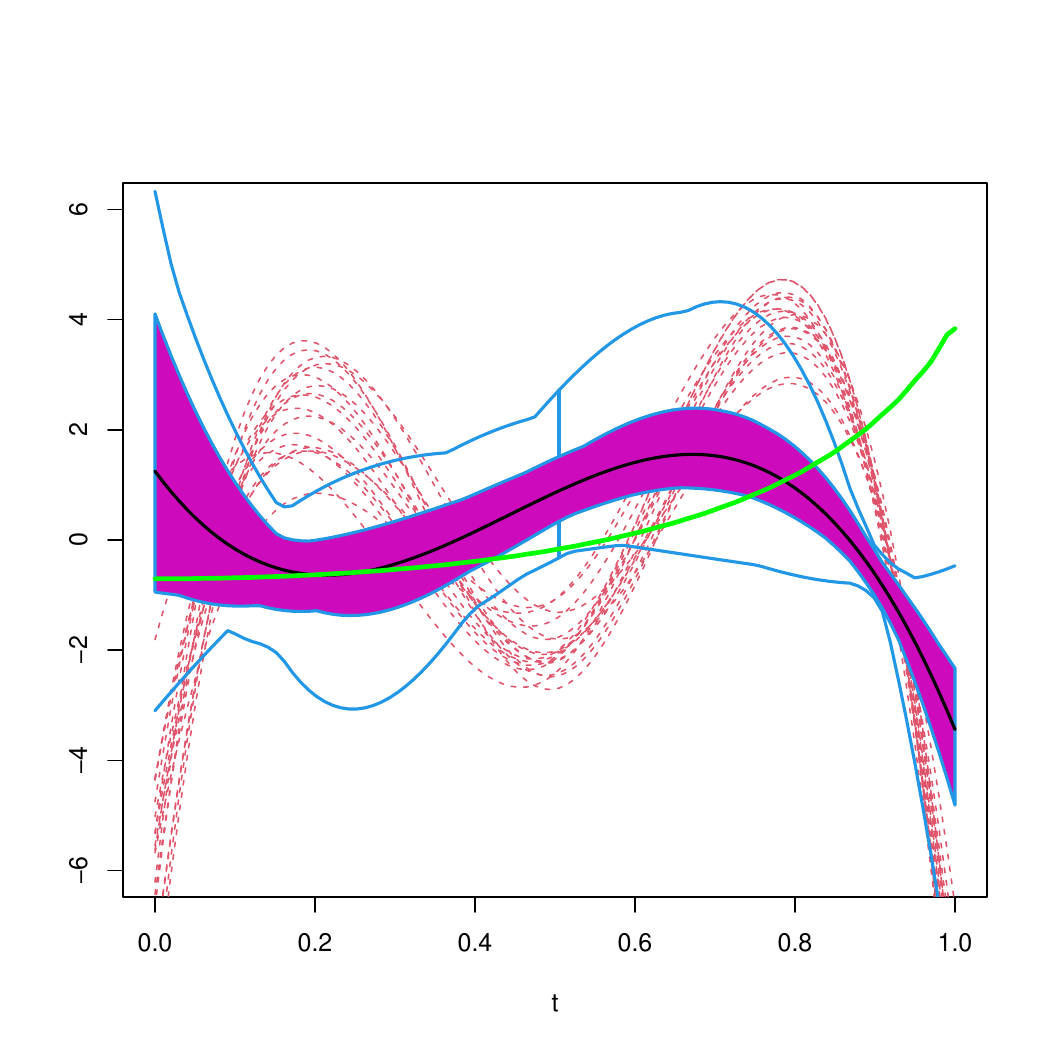}
 &  \includegraphics[scale=0.40]{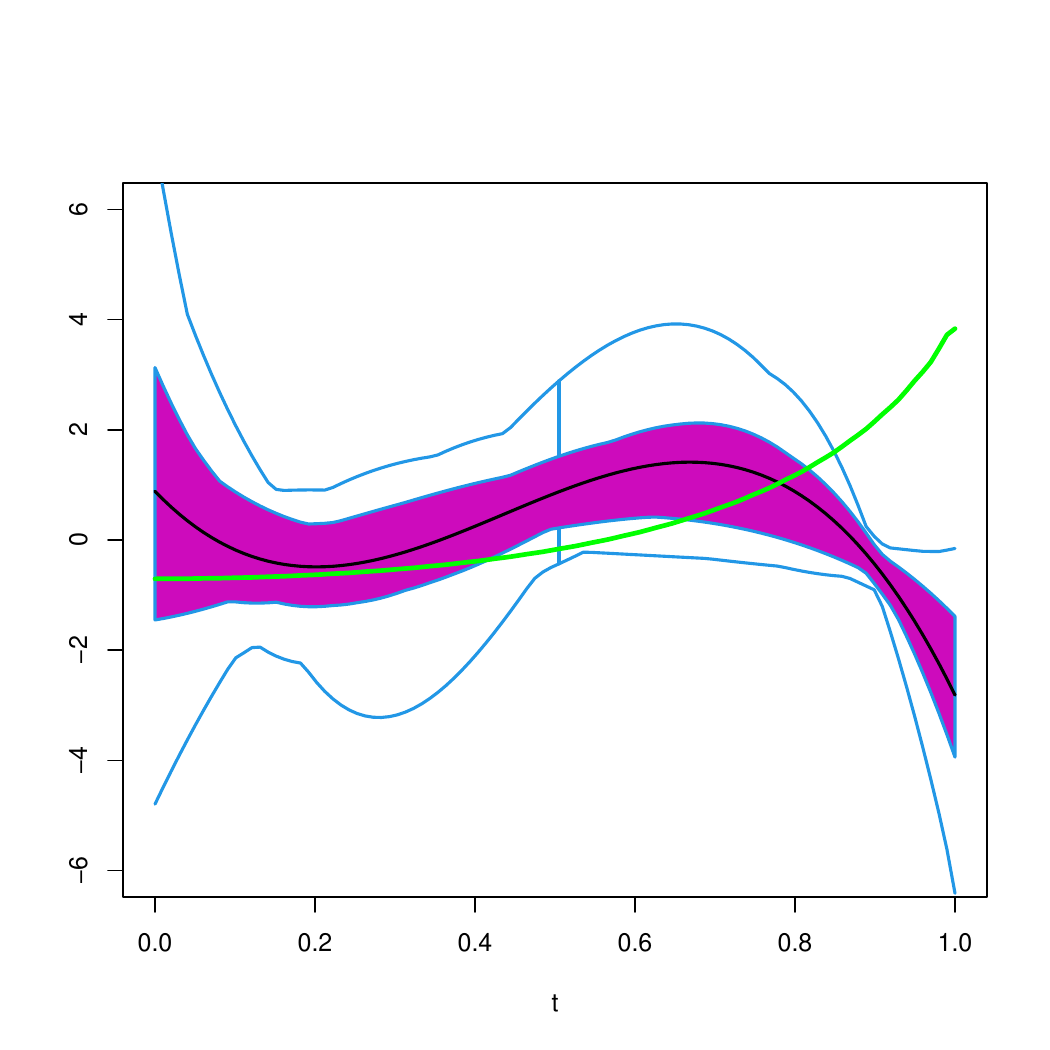}\\
   $\wbeta_{\wclHR}$ & $\wbeta_{\wemeHR}$ \\[-3ex] 
     \includegraphics[scale=0.40]{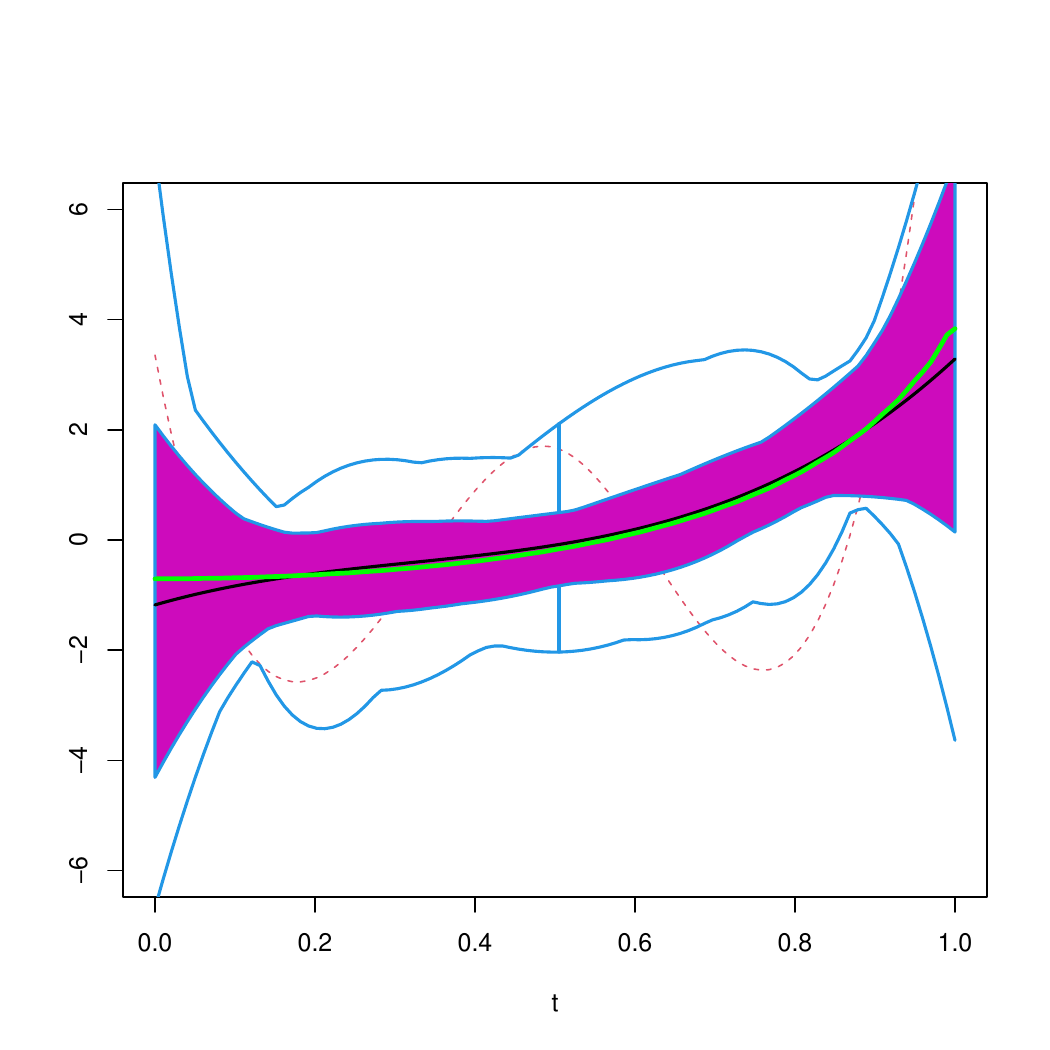}
  &  \includegraphics[scale=0.40]{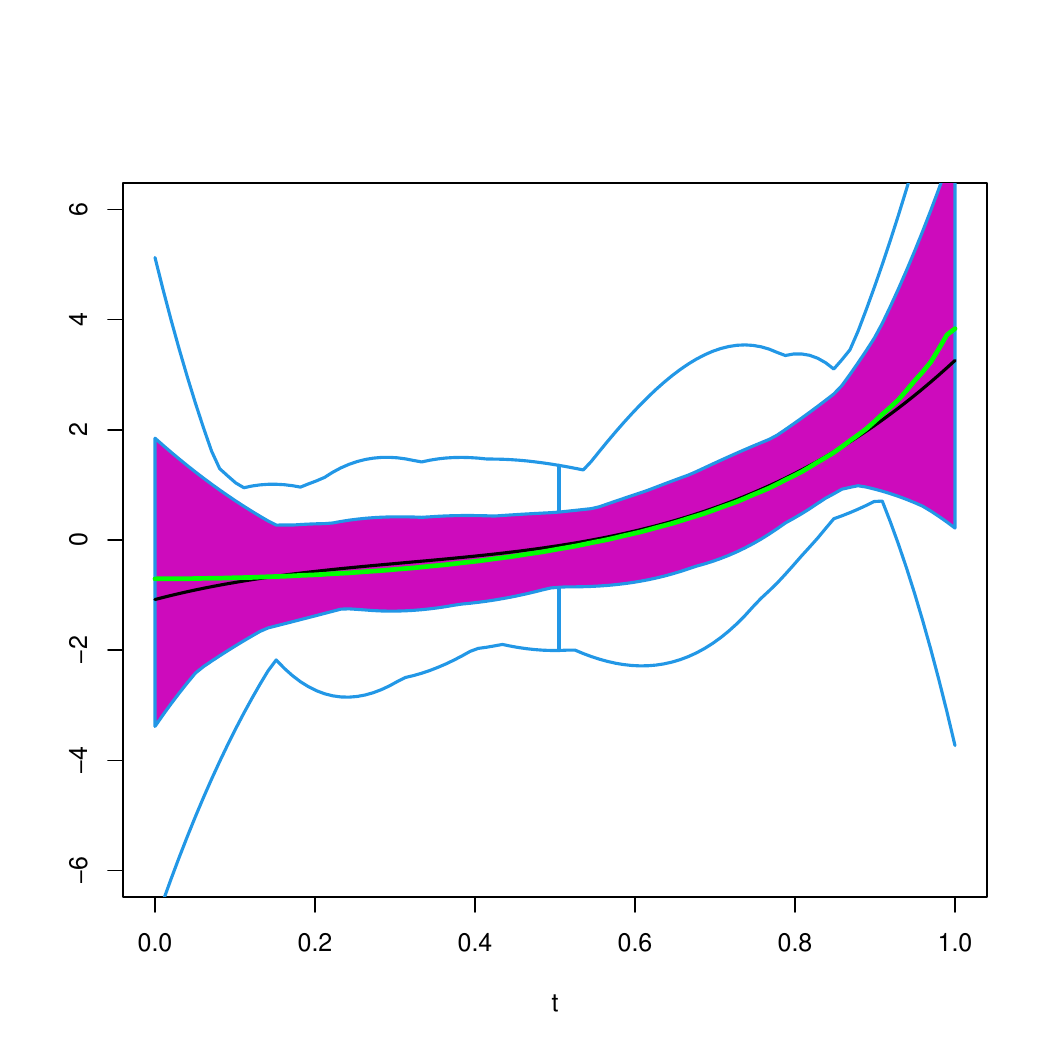}
     \\
   $\wbeta_{\wclBOX}$ & $\wbeta_{\wemeBOX}$ \\[-3ex]
    \includegraphics[scale=0.40]{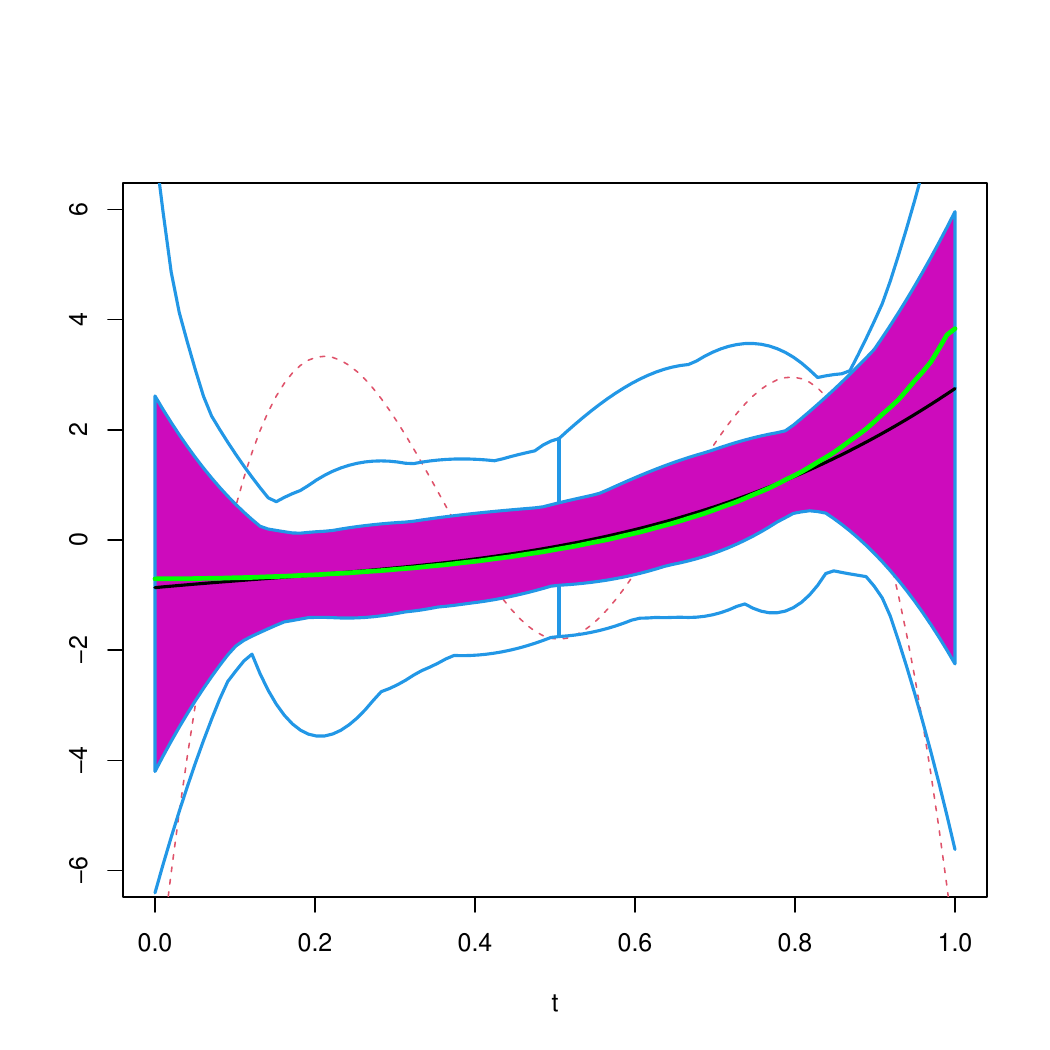}
  &  \includegraphics[scale=0.40]{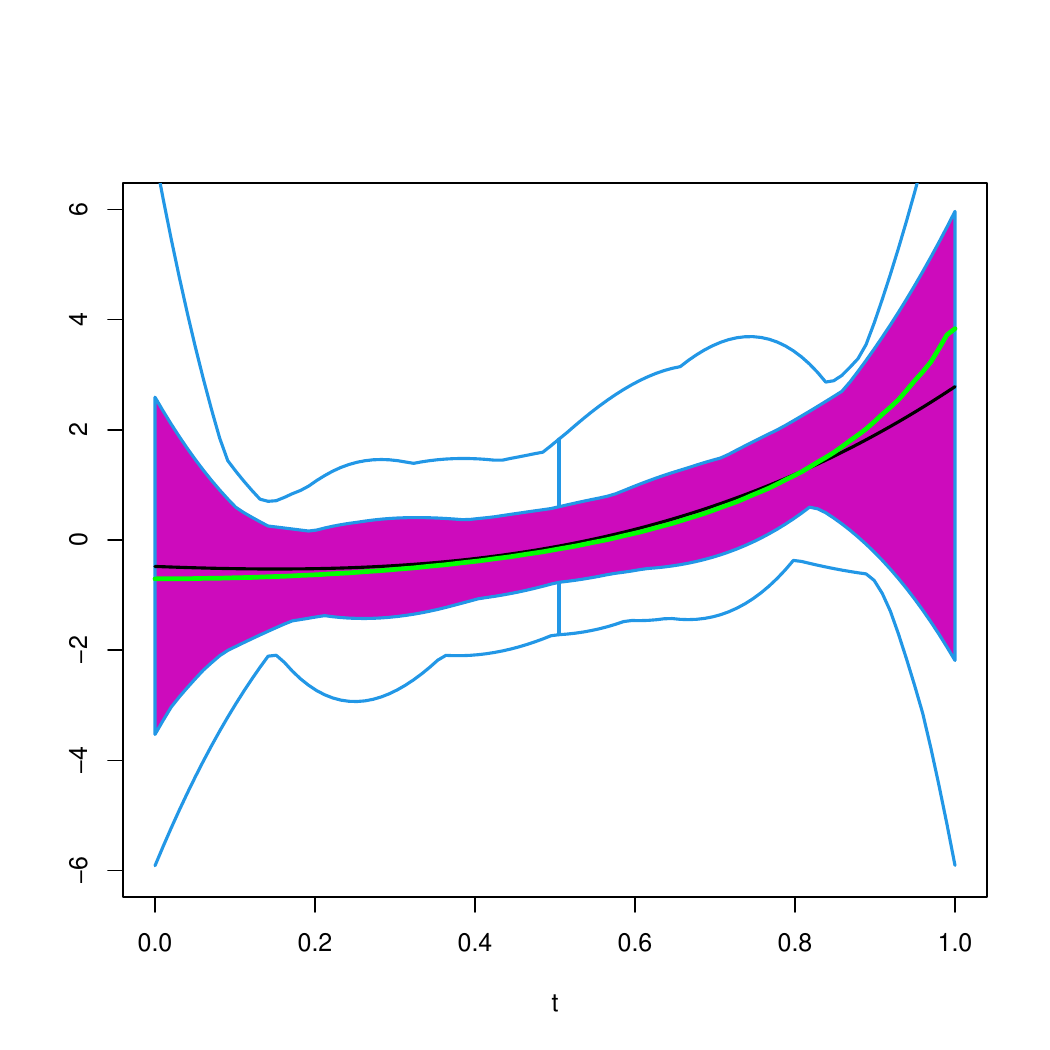}
 
\end{tabular}
\caption{\small \label{fig:wbeta-C210-poda0}  Functional boxplot of the estimators for $\beta_0$ under $C_{2,0.10}$  within the interval $[0,1]$. 
The true function is shown with a green dashed line, while the black solid one is the central 
curve of the $n_R = 1000$ estimates $\wbeta$.  }
\end{center} 
\end{figure}

\begin{figure}[tp]
 \begin{center}
 \footnotesize
 \renewcommand{\arraystretch}{0.2}
 \newcolumntype{M}{>{\centering\arraybackslash}m{\dimexpr.01\linewidth-1\tabcolsep}}
   \newcolumntype{G}{>{\centering\arraybackslash}m{\dimexpr.45\linewidth-1\tabcolsep}}
%\begin{tabular}{MGG}
\begin{tabular}{GG}
  $\wbeta_{\clas}$ & $\wbeta_{\eme}$   \\[-3ex]    
 
\includegraphics[scale=0.40]{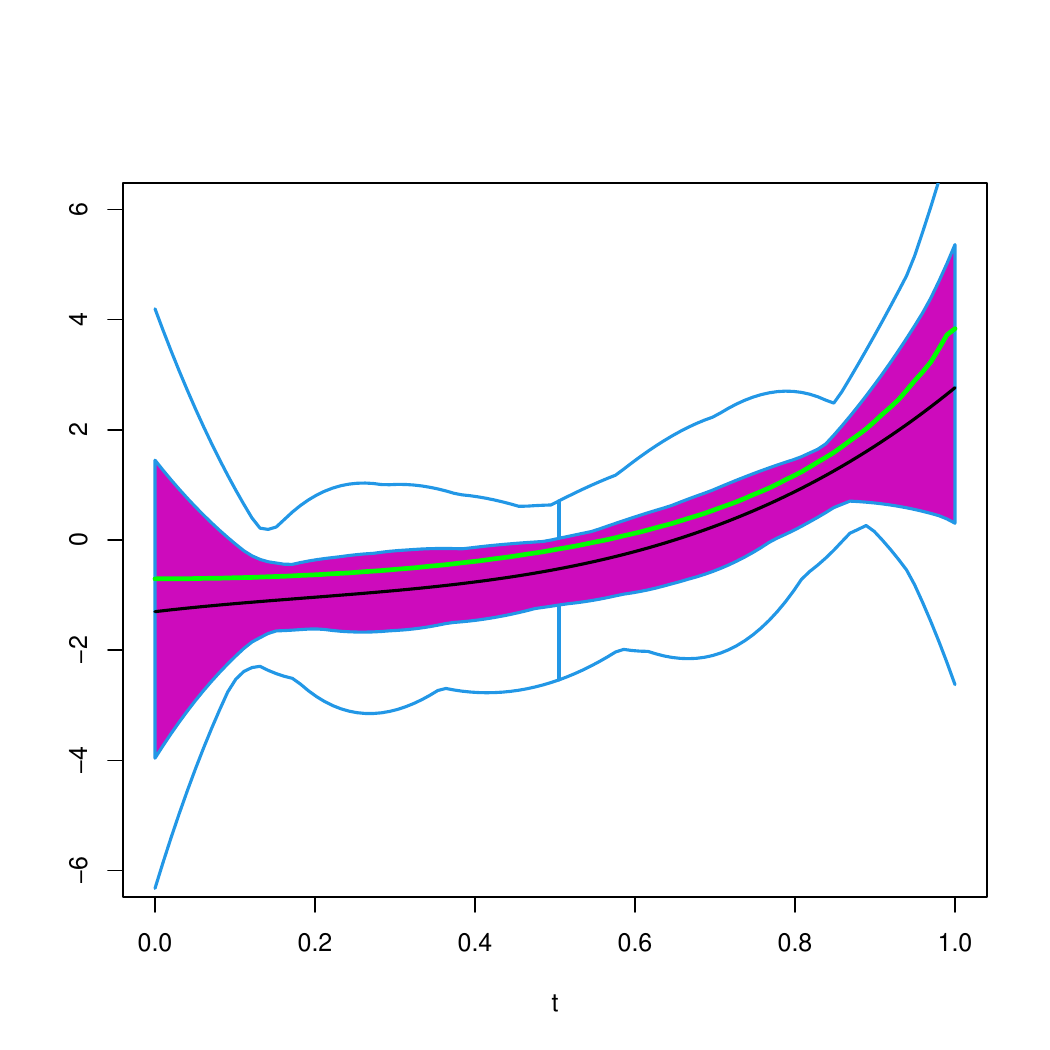}
 &  \includegraphics[scale=0.40]{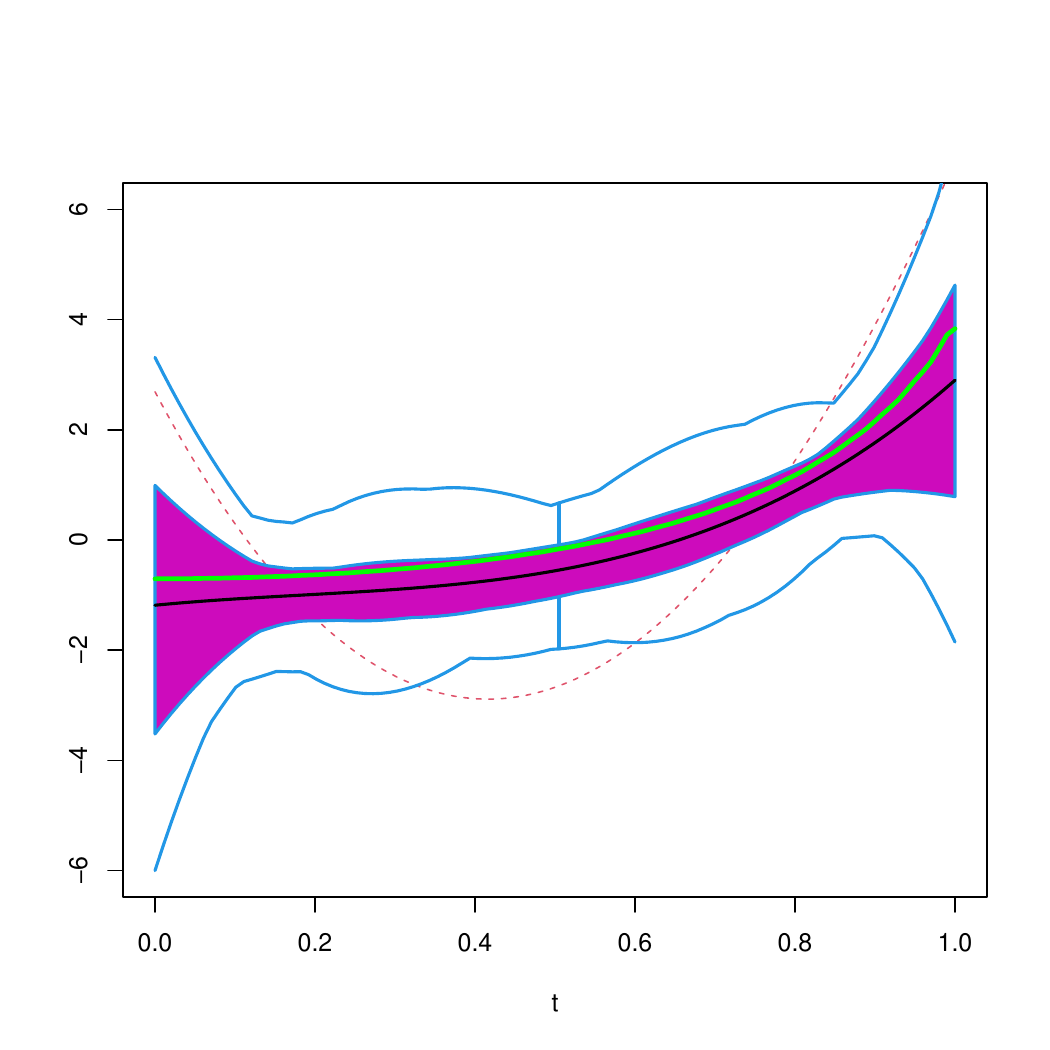}\\
   $\wbeta_{\wclHR}$ & $\wbeta_{\wemeHR}$ \\[-3ex] 
    \includegraphics[scale=0.40]{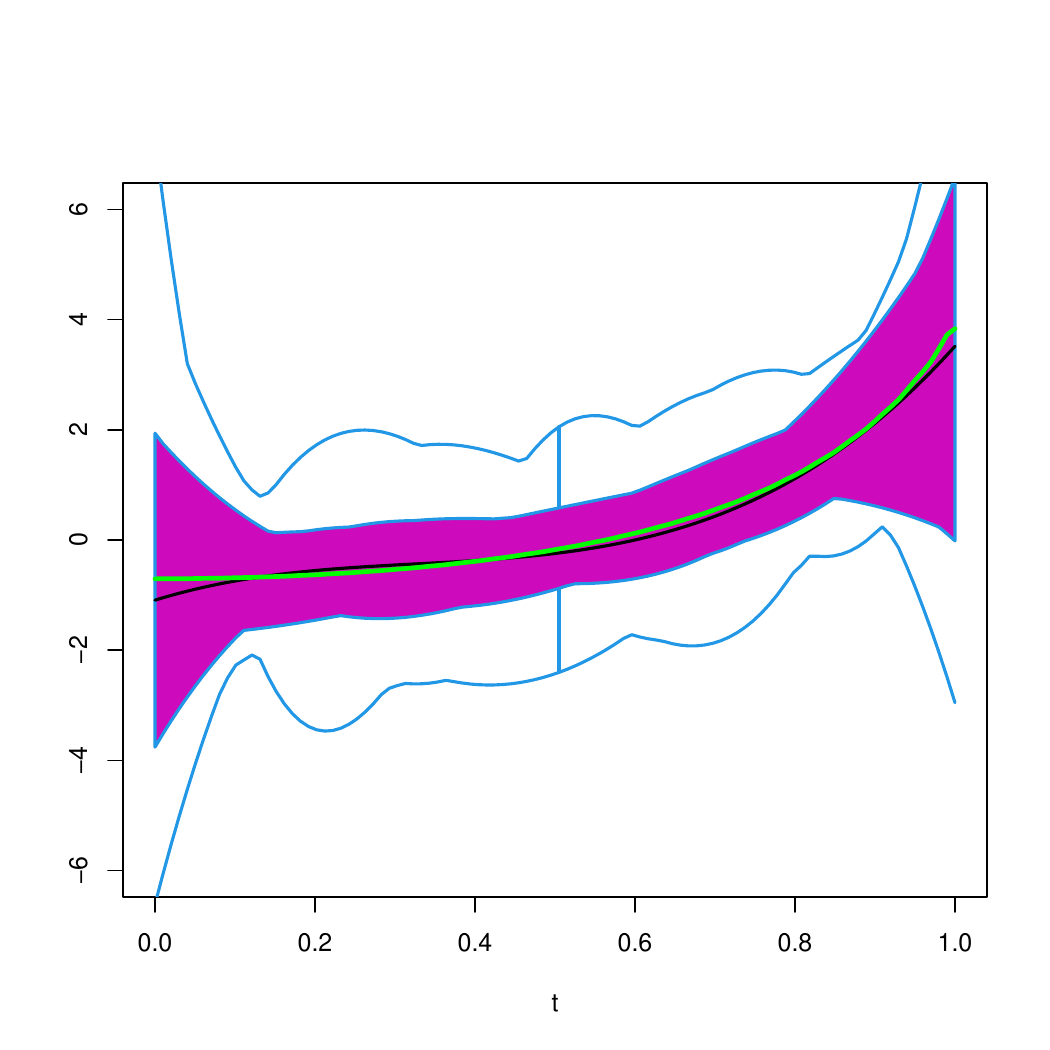}
  &  \includegraphics[scale=0.40]{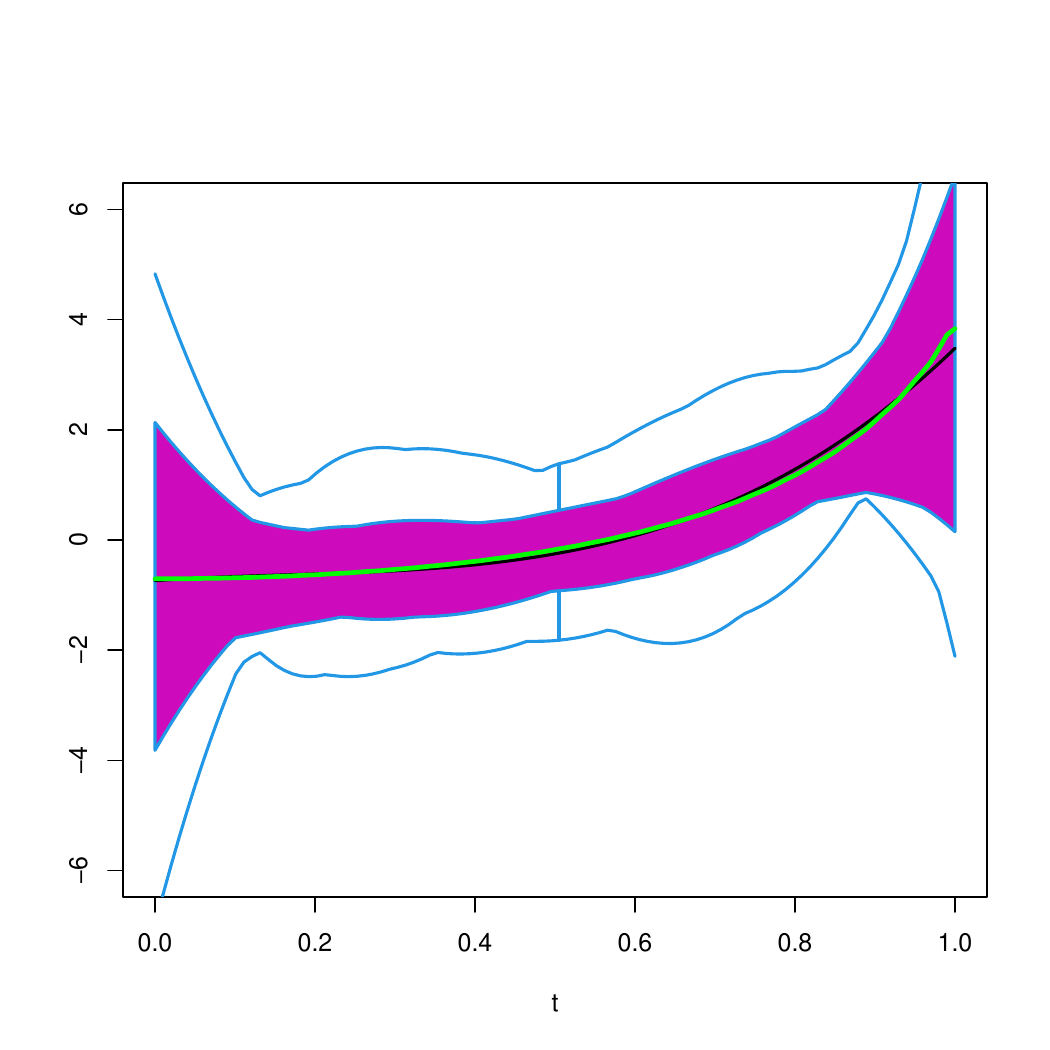}
   \\
   $\wbeta_{\wclBOX}$ & $\wbeta_{\wemeBOX}$ \\[-3ex]
  \includegraphics[scale=0.40]{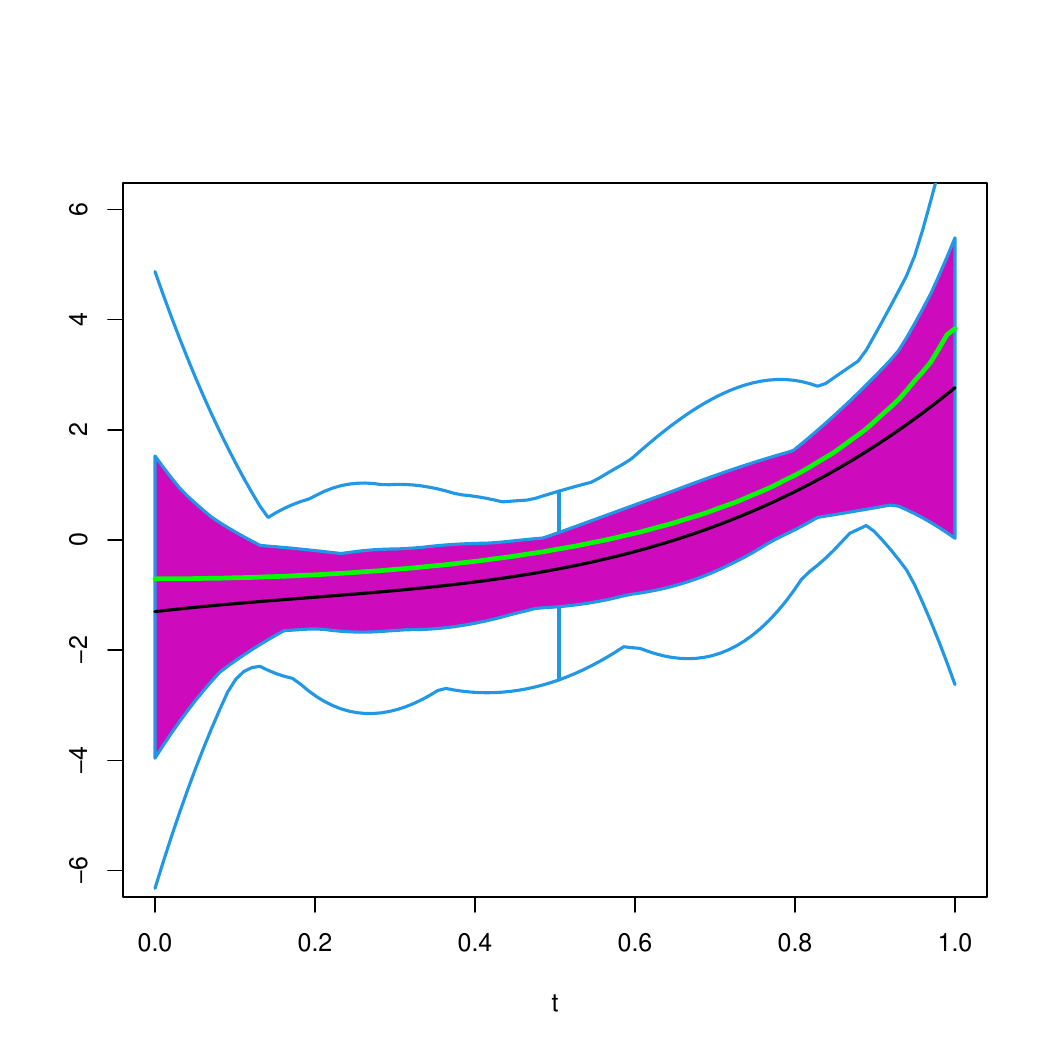}
  &  \includegraphics[scale=0.40]{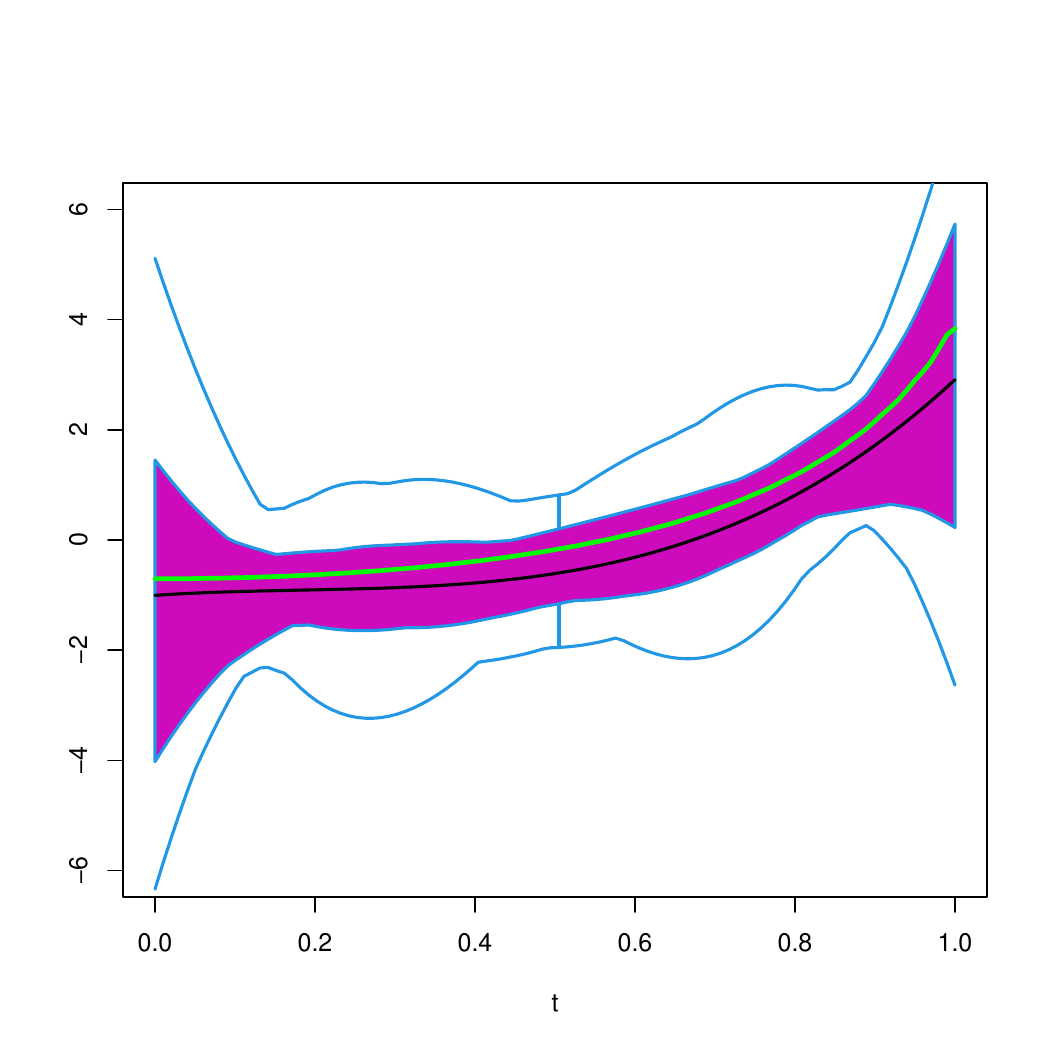}

\end{tabular}
\caption{\small \label{fig:wbeta-C35-poda0}  Functional boxplot of the estimators for $\beta_0$ under $C_{3,0.05}$  within the interval $[0,1]$. 
The true function is shown with a green dashed line, while the black solid one is the central 
curve of the $n_R = 1000$ estimates $\wbeta$.  }
\end{center} 
\end{figure}

\begin{figure}[tp]
 \begin{center}
 \footnotesize
 \renewcommand{\arraystretch}{0.2}
 \newcolumntype{M}{>{\centering\arraybackslash}m{\dimexpr.01\linewidth-1\tabcolsep}}
   \newcolumntype{G}{>{\centering\arraybackslash}m{\dimexpr.45\linewidth-1\tabcolsep}}
%\begin{tabular}{MGG}
\begin{tabular}{GG}
  $\wbeta_{\clas}$ & $\wbeta_{\eme}$   \\[-3ex]    
 
\includegraphics[scale=0.40]{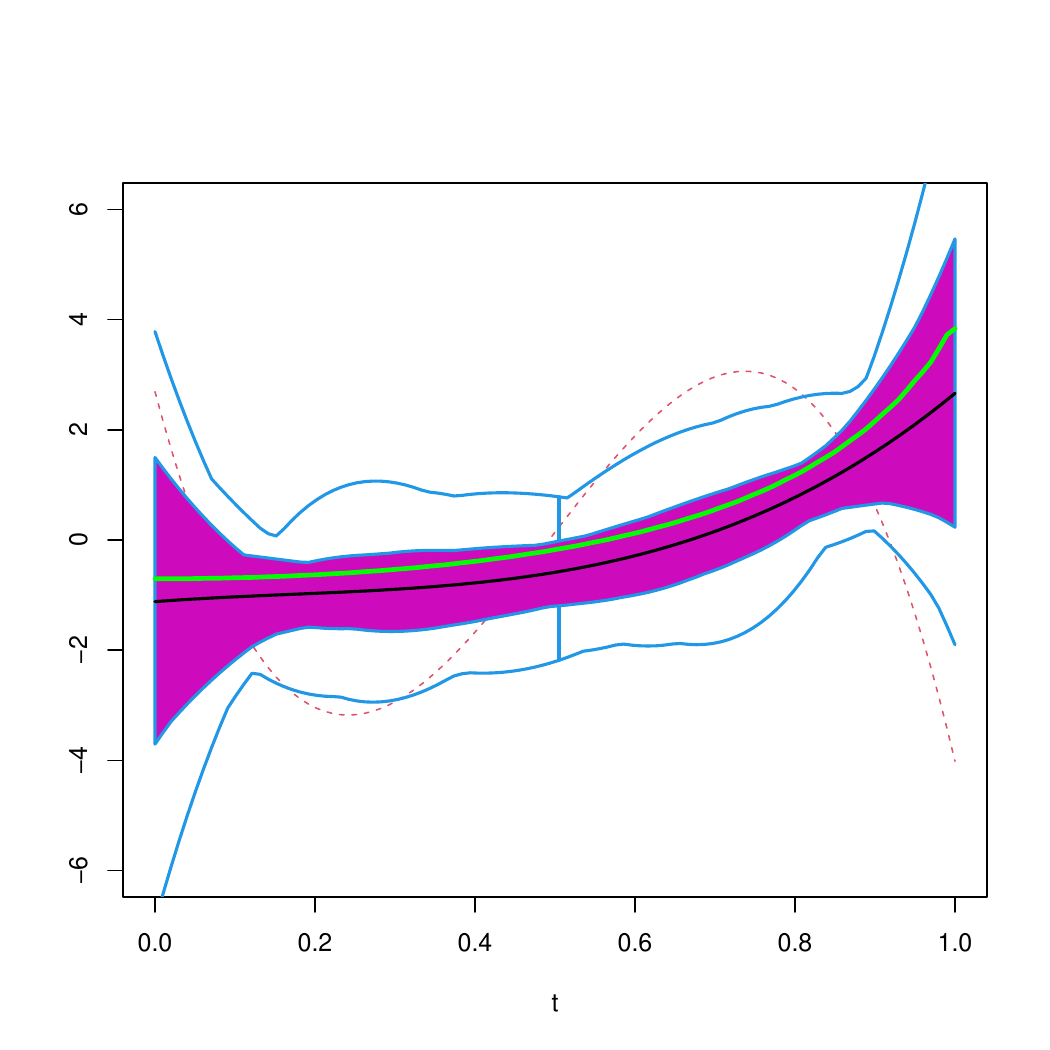}
 &  \includegraphics[scale=0.40]{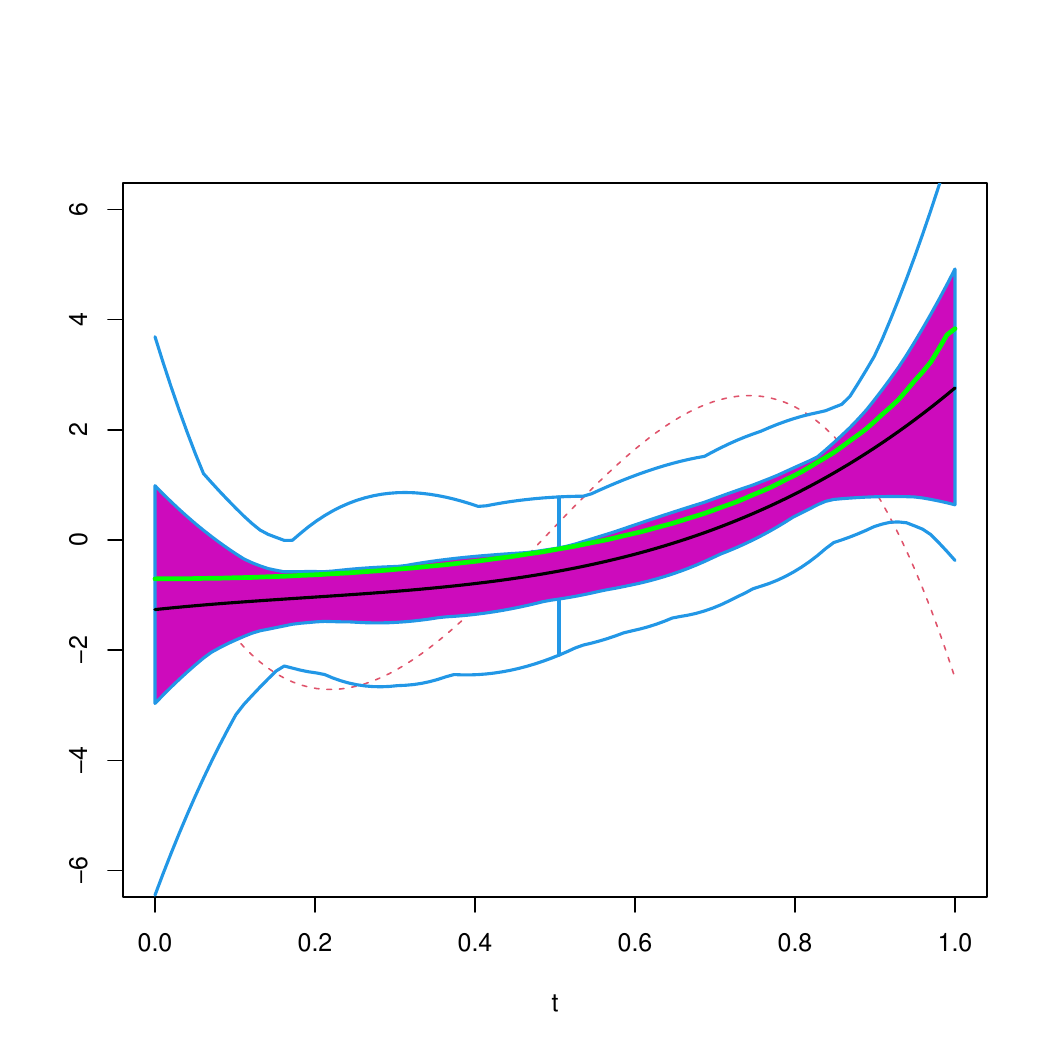}\\
   $\wbeta_{\wclHR}$ & $\wbeta_{\wemeHR}$ \\[-3ex] 
    \includegraphics[scale=0.40]{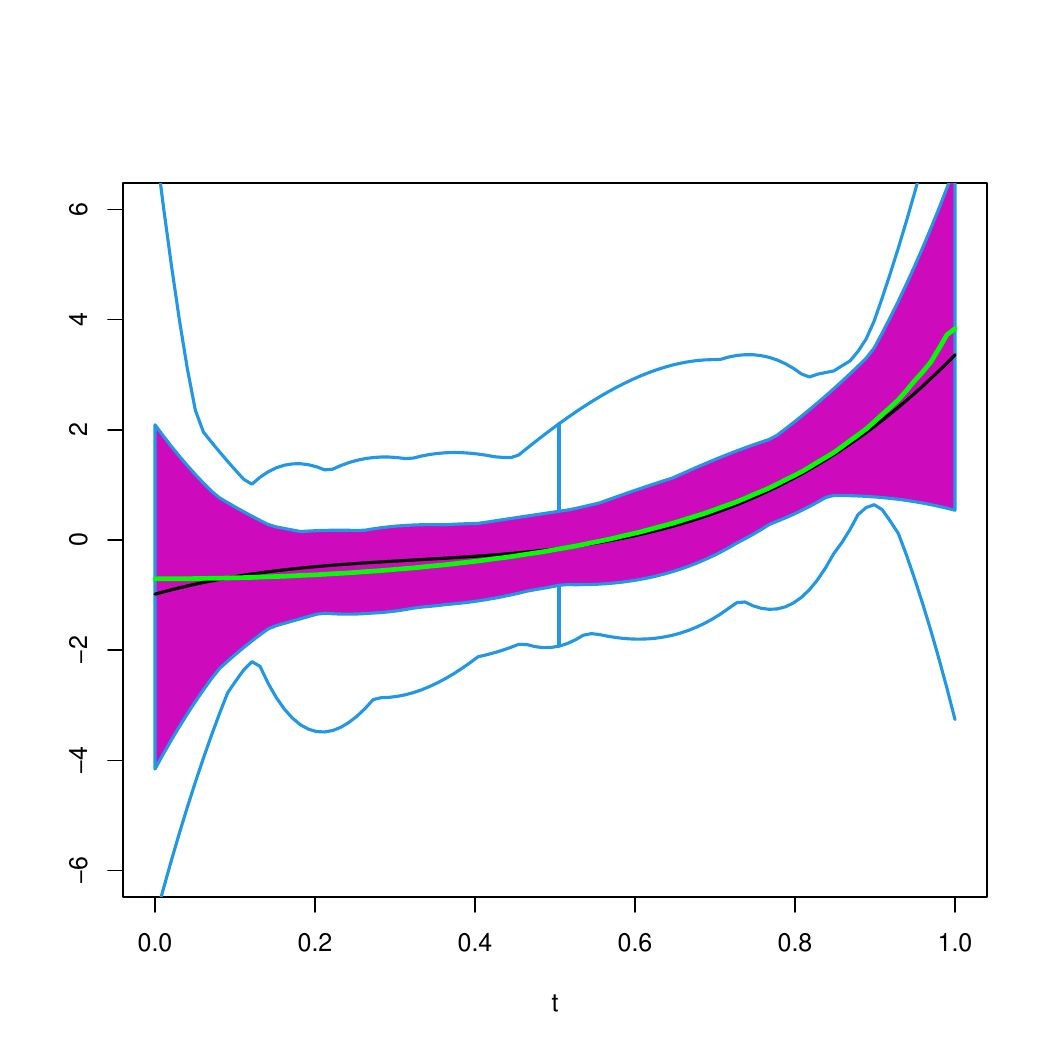}
  &  \includegraphics[scale=0.40]{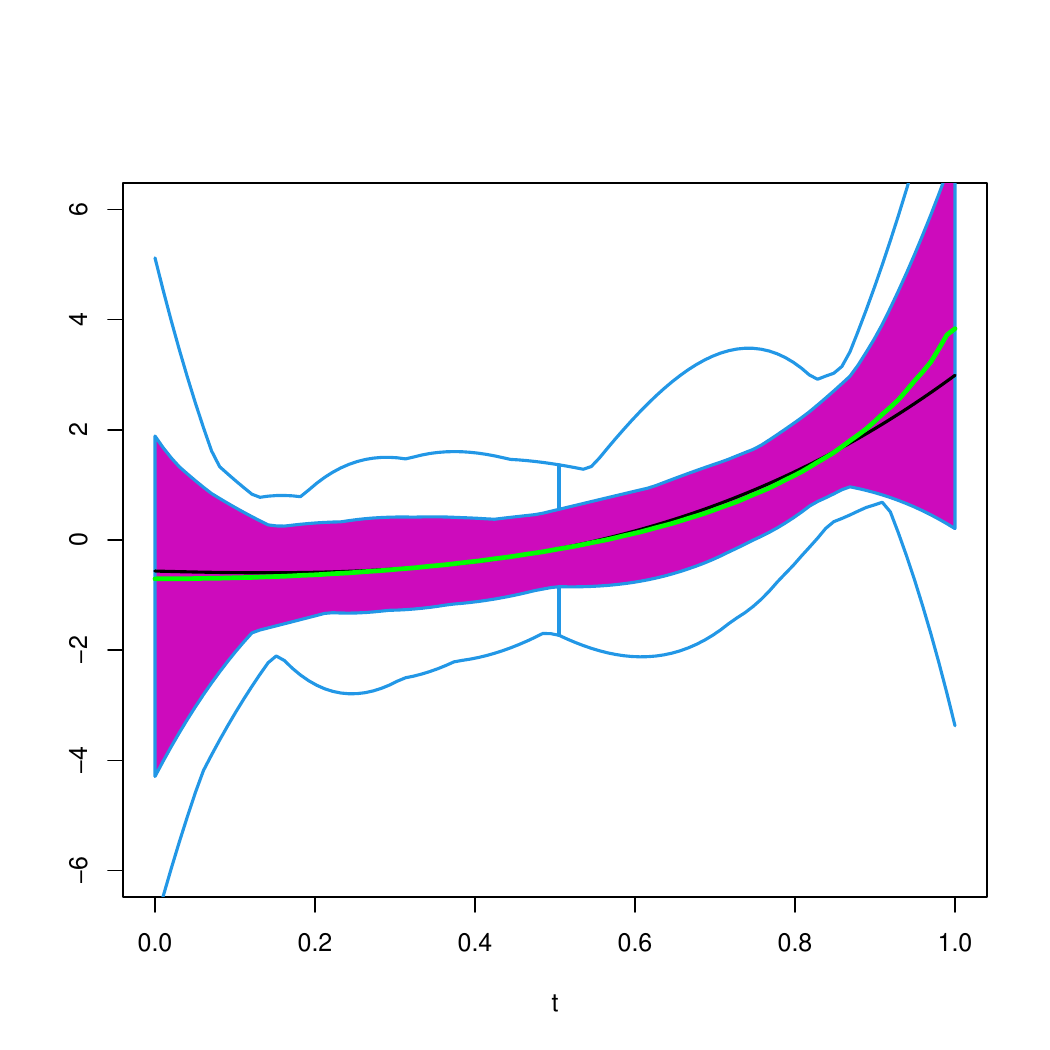}
   \\
   $\wbeta_{\wclBOX}$ & $\wbeta_{\wemeBOX}$ \\[-3ex]
  \includegraphics[scale=0.40]{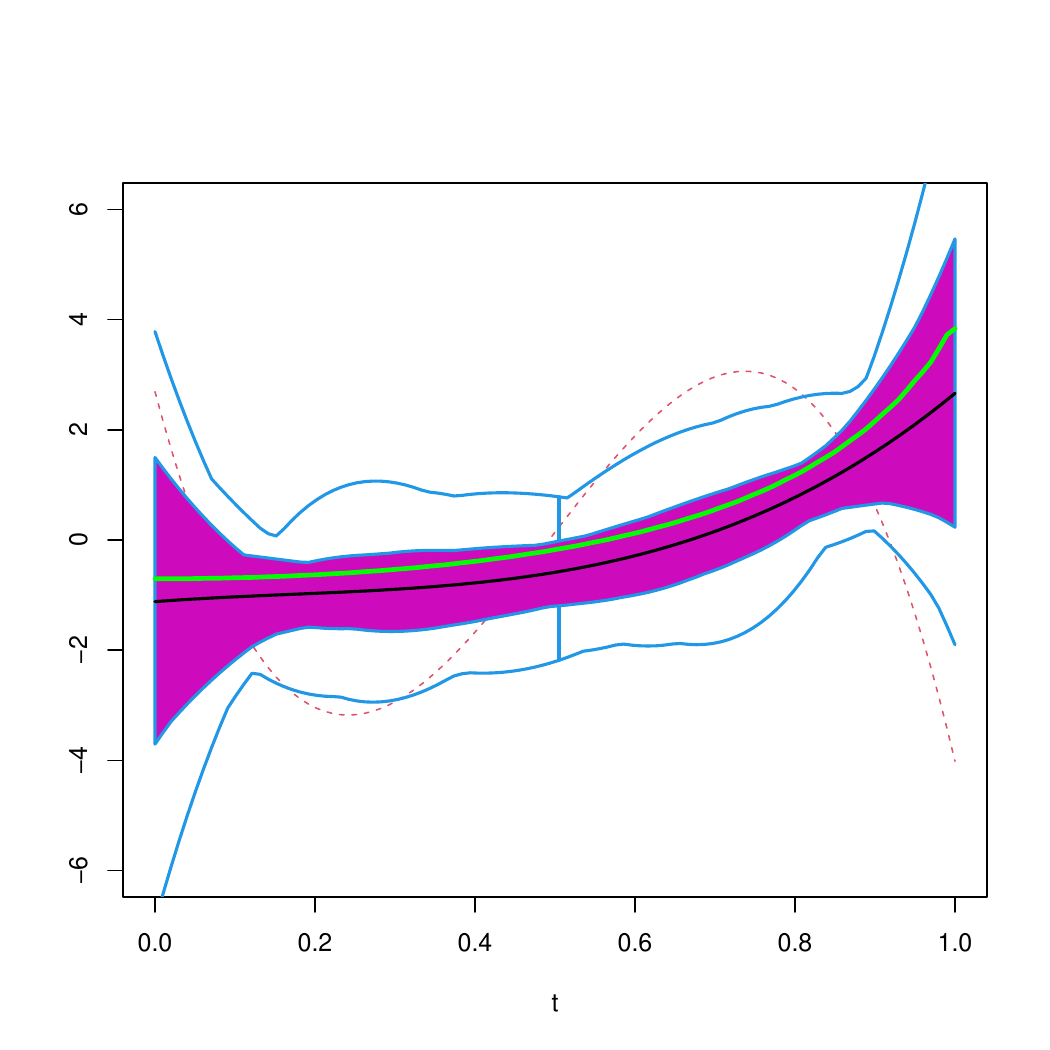}
  &  \includegraphics[scale=0.40]{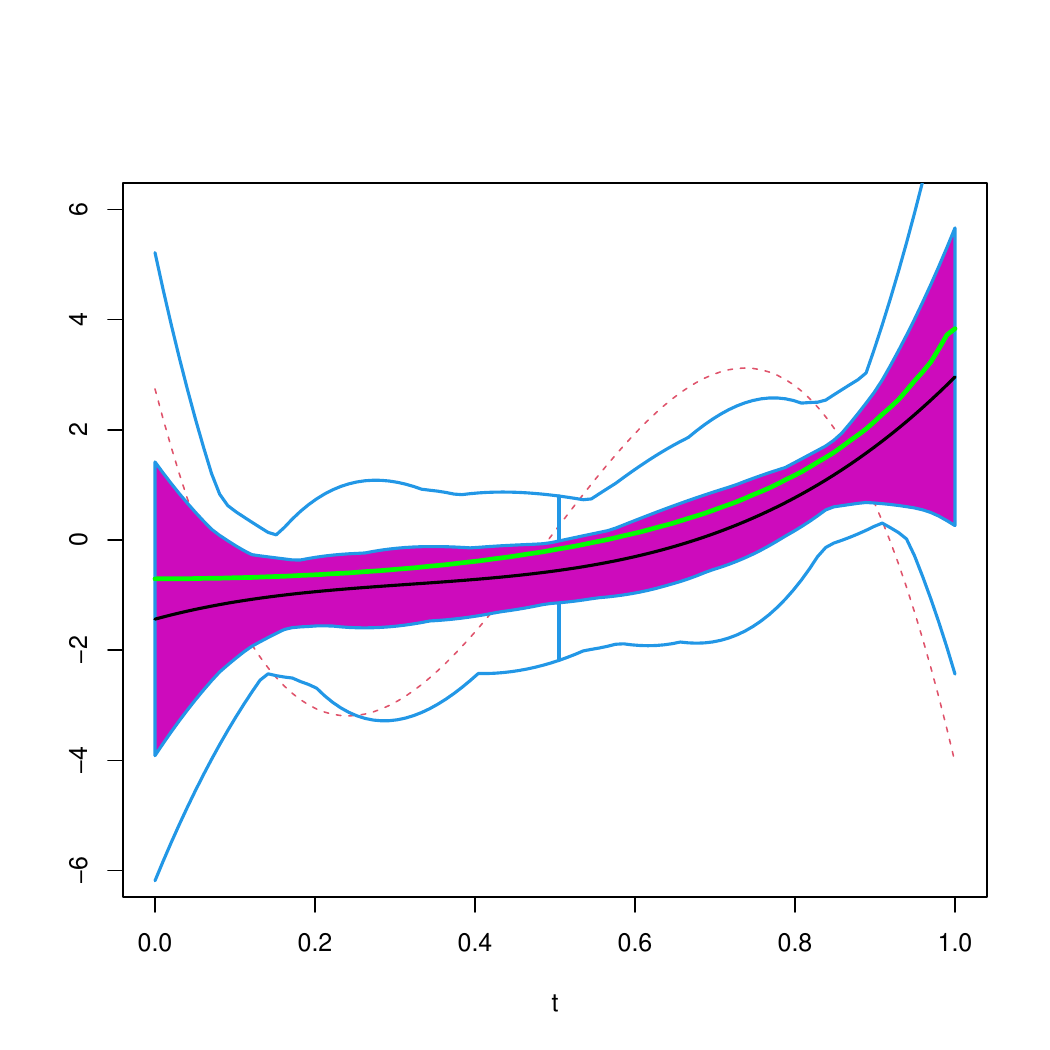}

\end{tabular}
\caption{\small \label{fig:wbeta-C310-poda0}  Functional boxplot of the estimators for $\beta_0$ under $C_{3,0.10}$  within the interval $[0,1]$. 
The true function is shown with a green dashed line, while the black solid one is the central 
curve of the $n_R = 1000$ estimates $\wbeta$.  }
\end{center} 
\end{figure}

\begin{figure}[tp]
 \begin{center}
 \footnotesize
 \renewcommand{\arraystretch}{0.2}
 \newcolumntype{M}{>{\centering\arraybackslash}m{\dimexpr.01\linewidth-1\tabcolsep}}
   \newcolumntype{G}{>{\centering\arraybackslash}m{\dimexpr.45\linewidth-1\tabcolsep}}
%\begin{tabular}{MGG}
\begin{tabular}{GG}
  $\wbeta_{\clas}$ & $\wbeta_{\eme}$   \\[-3ex]    
 
\includegraphics[scale=0.40]{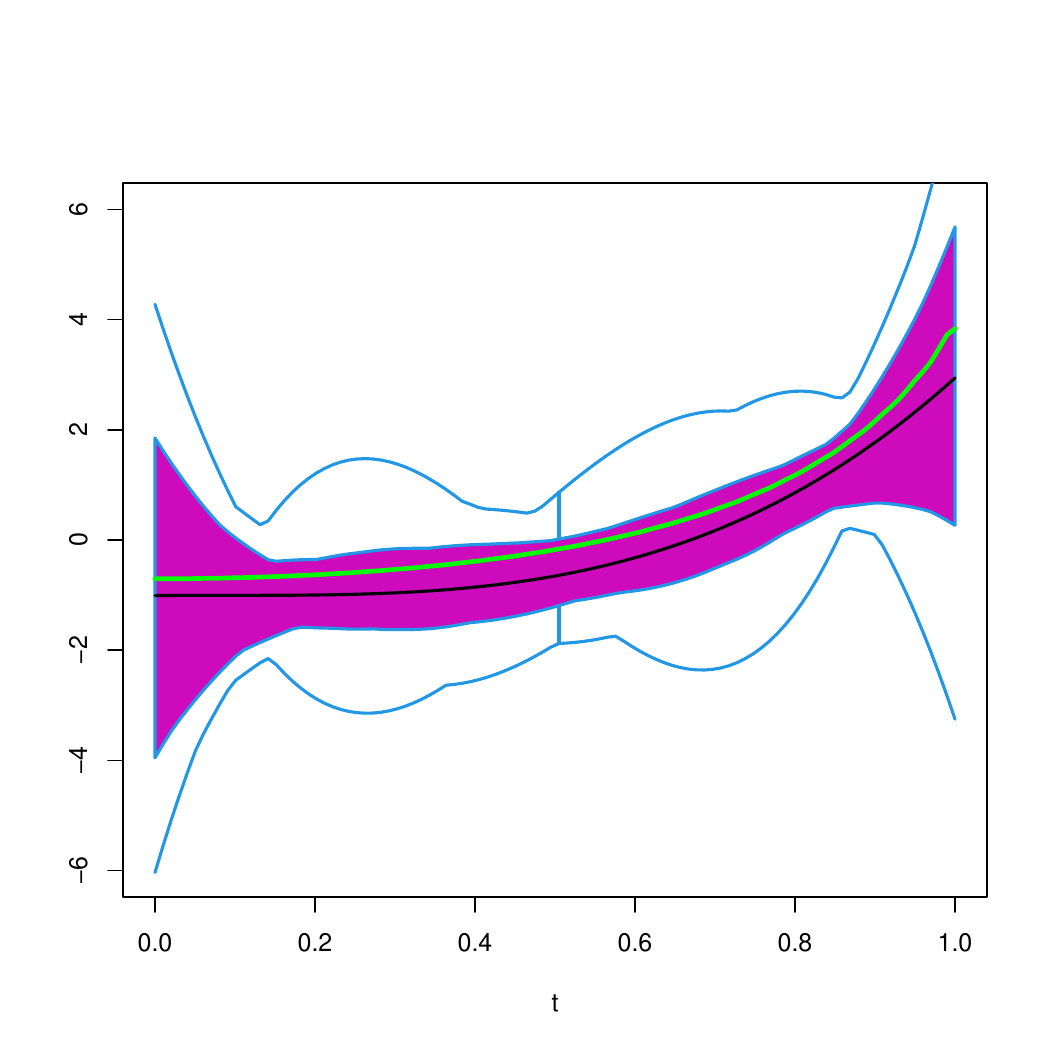}
 &  \includegraphics[scale=0.40]{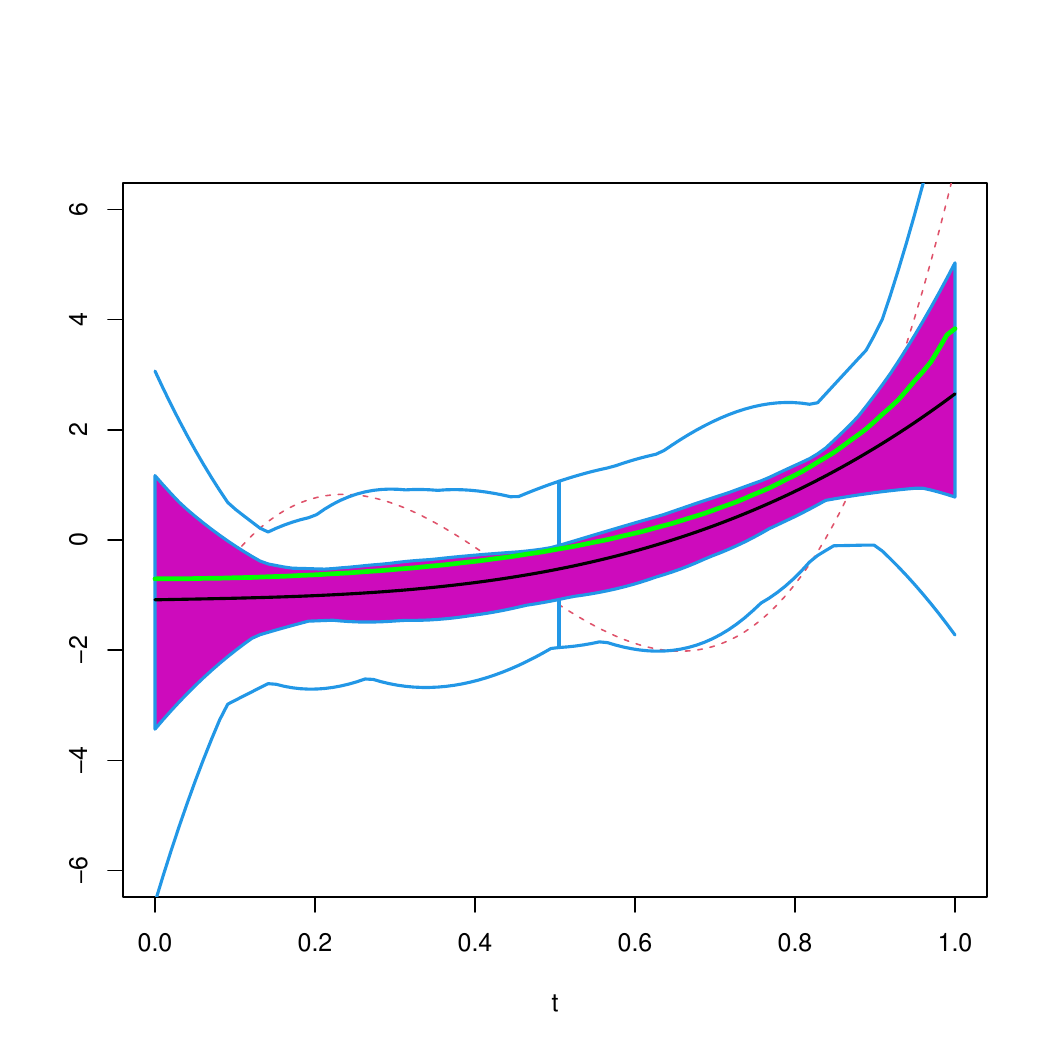}\\
   $\wbeta_{\wclHR}$ & $\wbeta_{\wemeHR}$ \\[-3ex] 
    \includegraphics[scale=0.40]{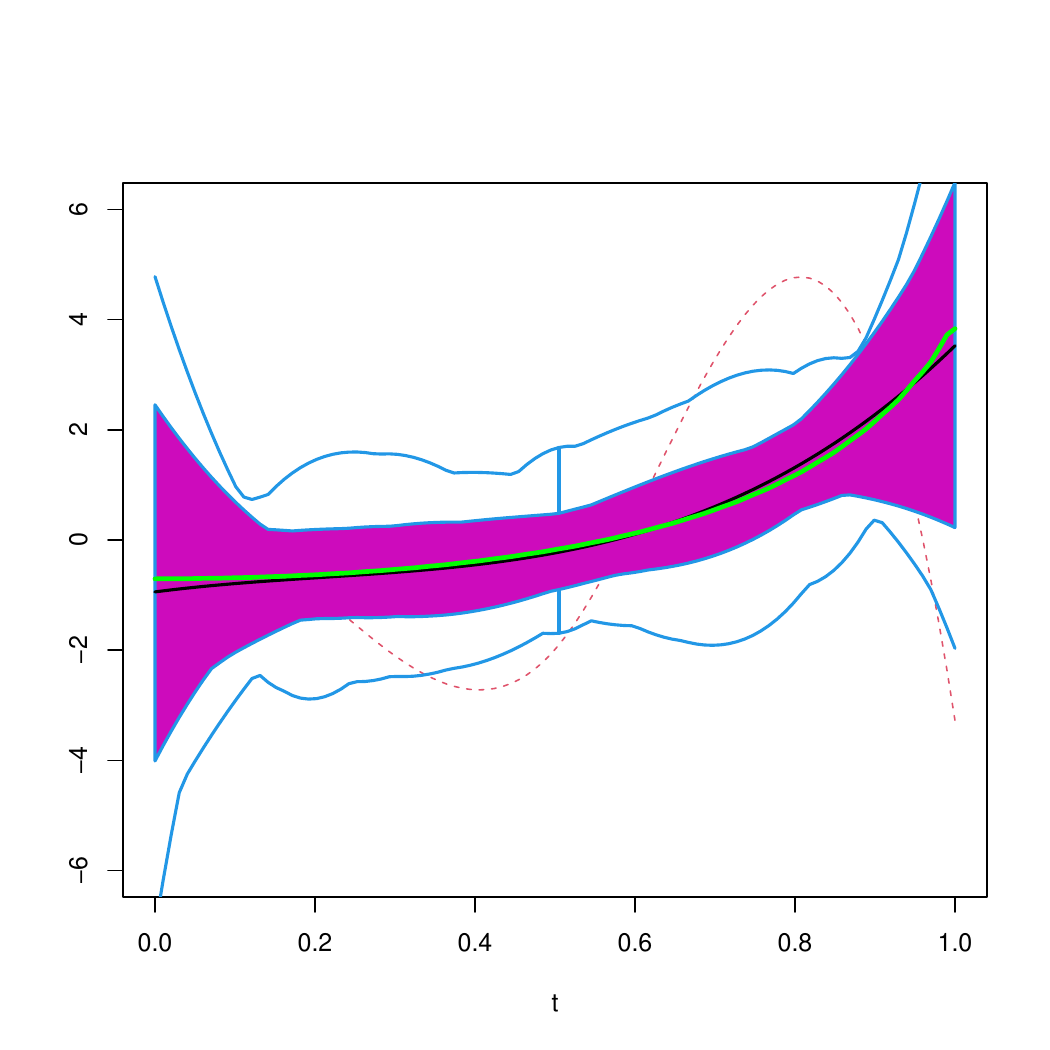}
  &  \includegraphics[scale=0.40]{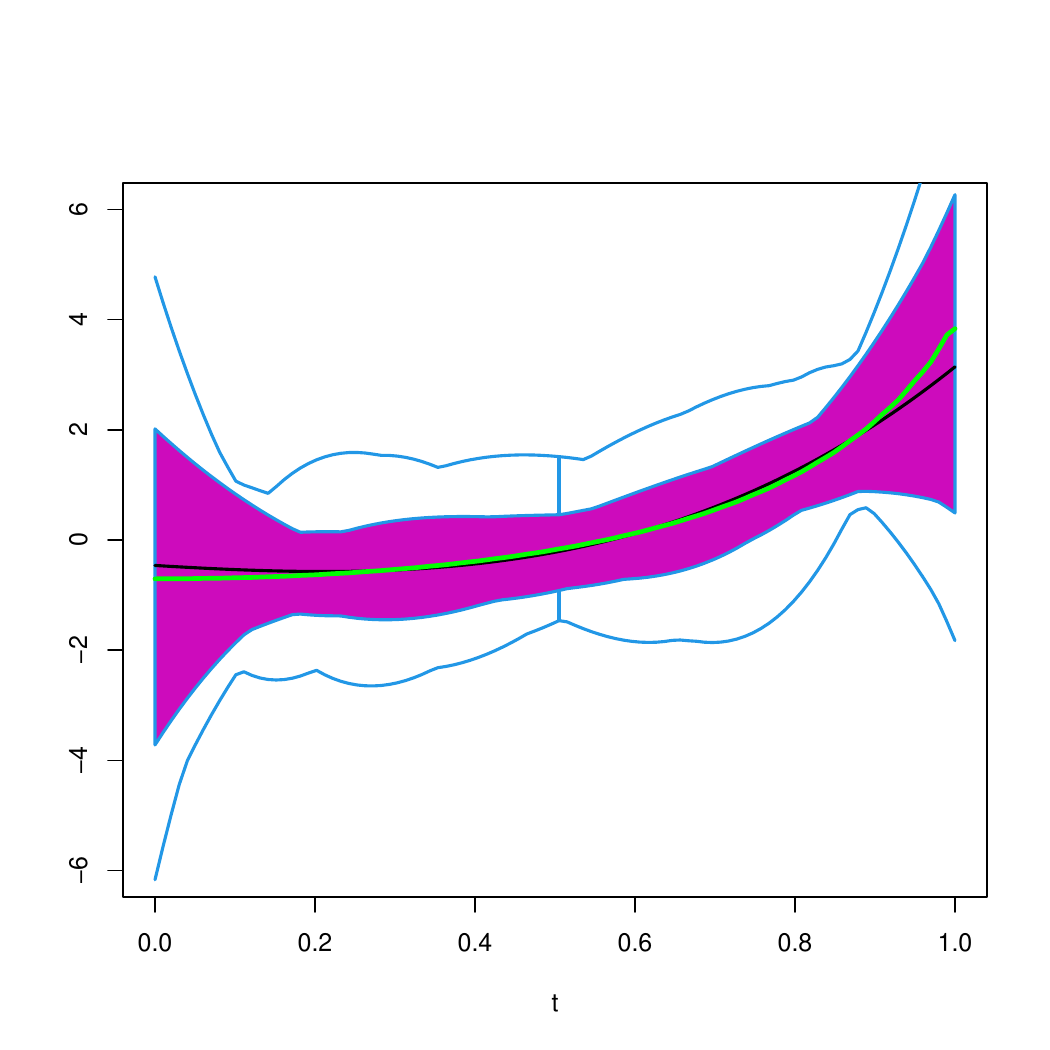}
   \\
   $\wbeta_{\wclBOX}$ & $\wbeta_{\wemeBOX}$ \\[-3ex]
  \includegraphics[scale=0.40]{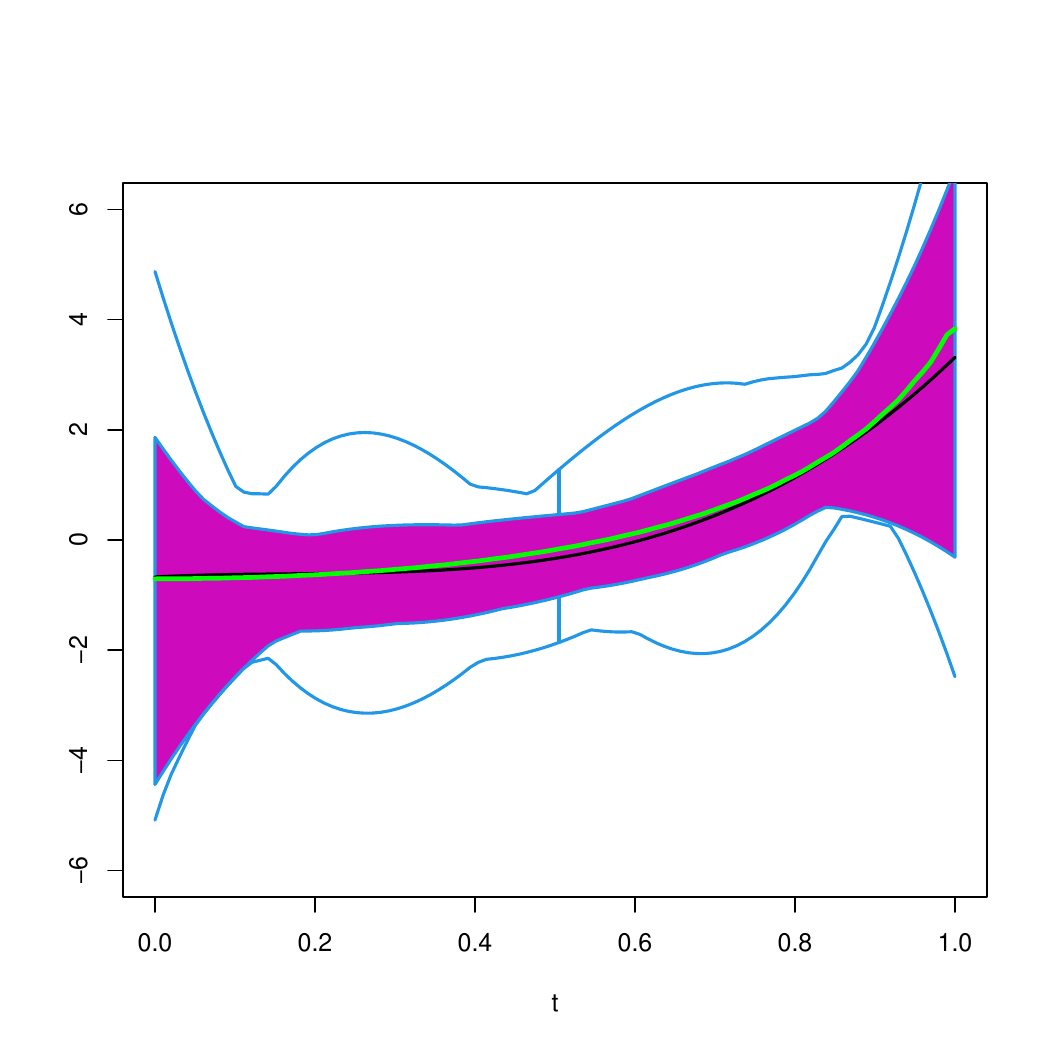}
  &  \includegraphics[scale=0.40]{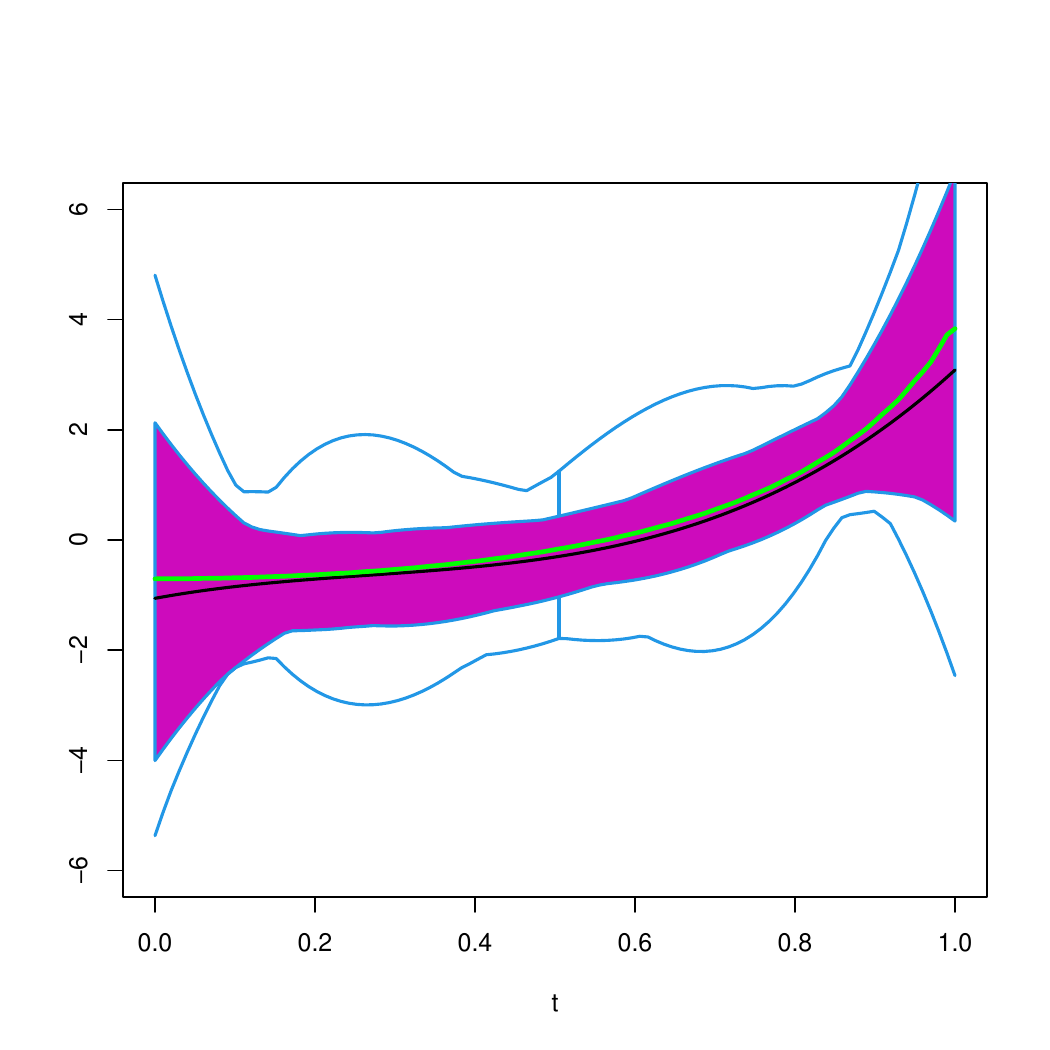}

\end{tabular}
\caption{\small \label{fig:wbeta-C45-poda0}  Functional boxplot of the estimators for $\beta_0$ under $C_{4,0.05}$  within the interval $[0,1]$. 
The true function is shown with a green dashed line, while the black solid one is the central 
curve of the $n_R = 1000$ estimates $\wbeta$.  }
\end{center} 
\end{figure}

\begin{figure}[tp]
 \begin{center}
 \footnotesize
 \renewcommand{\arraystretch}{0.2}
 \newcolumntype{M}{>{\centering\arraybackslash}m{\dimexpr.01\linewidth-1\tabcolsep}}
   \newcolumntype{G}{>{\centering\arraybackslash}m{\dimexpr.45\linewidth-1\tabcolsep}}
%\begin{tabular}{MGG}
\begin{tabular}{GG}
  $\wbeta_{\clas}$ & $\wbeta_{\eme}$   \\[-3ex]    
 
\includegraphics[scale=0.40]{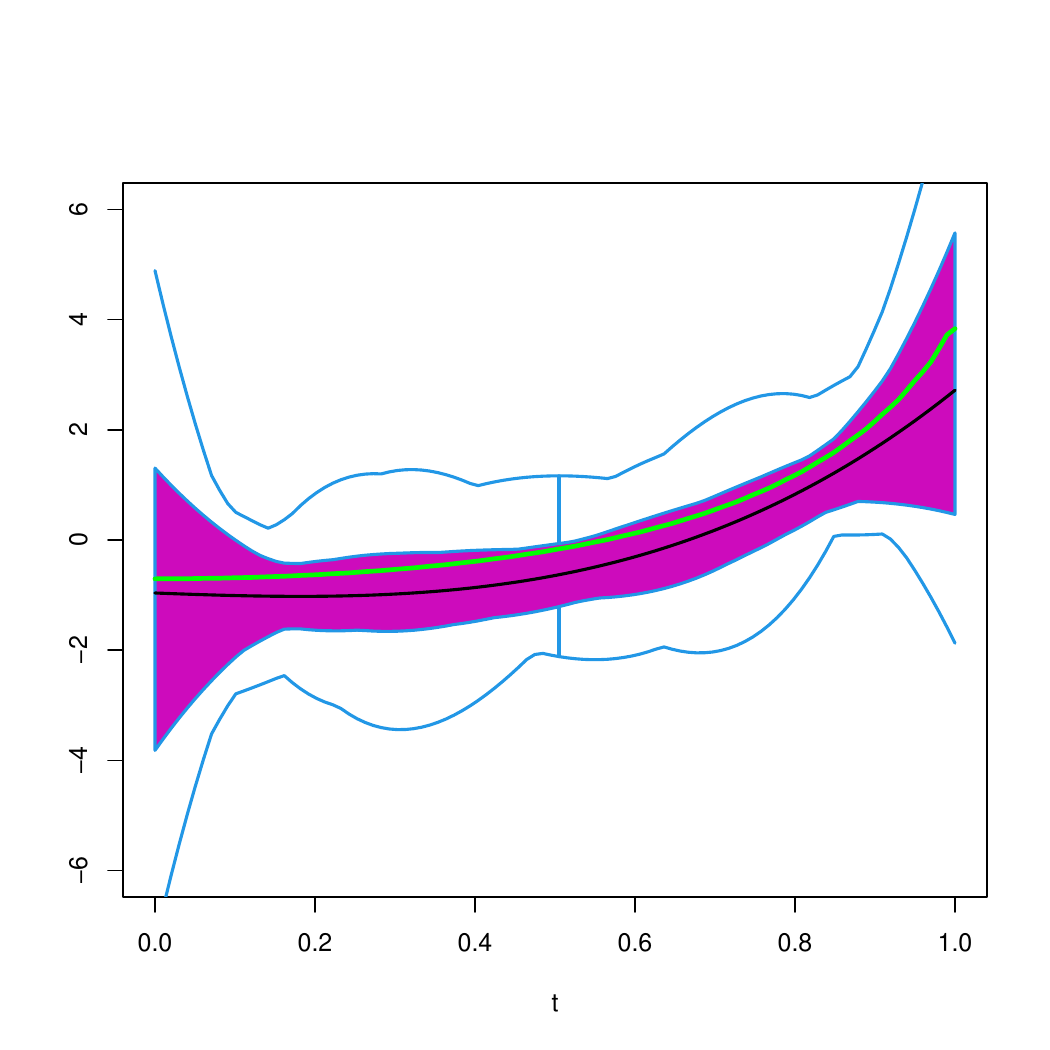}
 &  \includegraphics[scale=0.40]{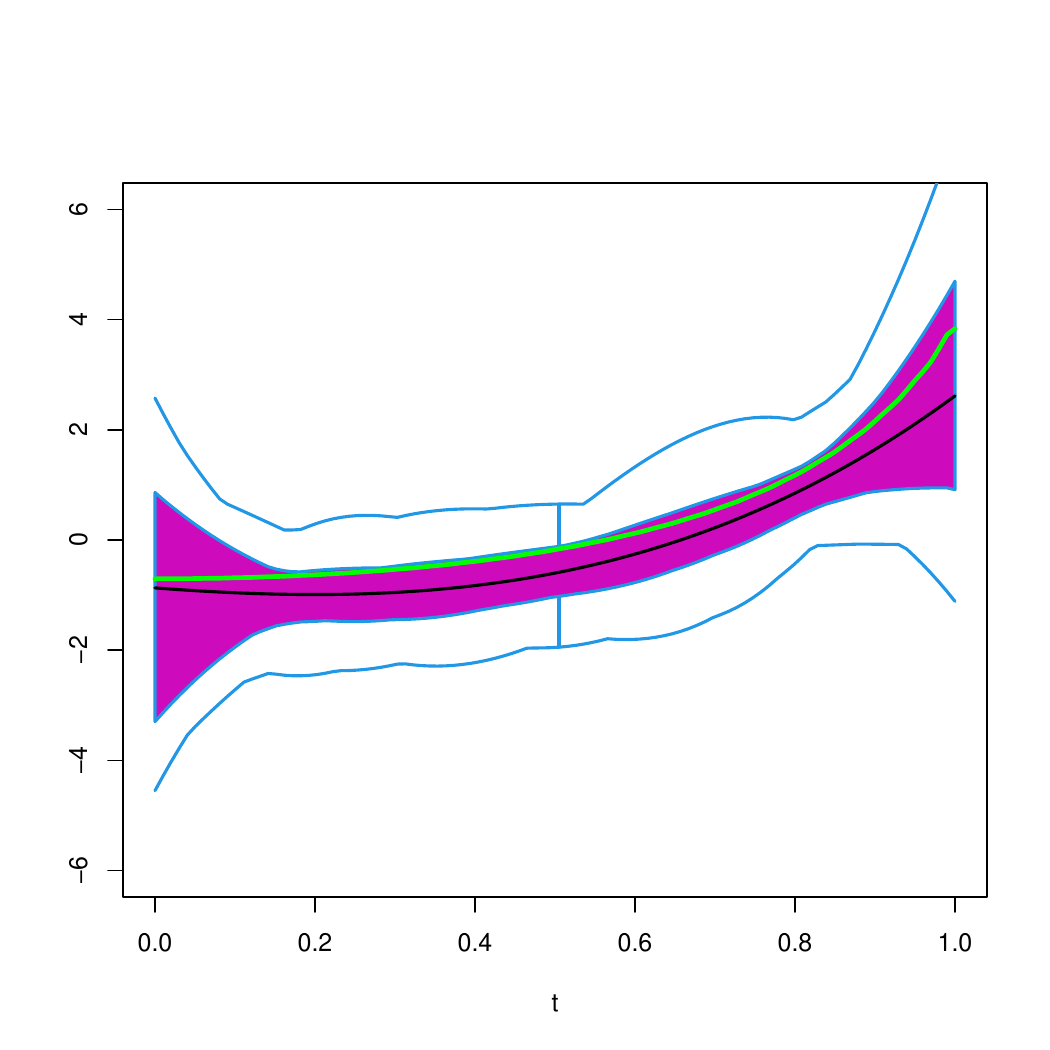}\\
   $\wbeta_{\wclHR}$ & $\wbeta_{\wemeHR}$ \\[-3ex] 
    \includegraphics[scale=0.40]{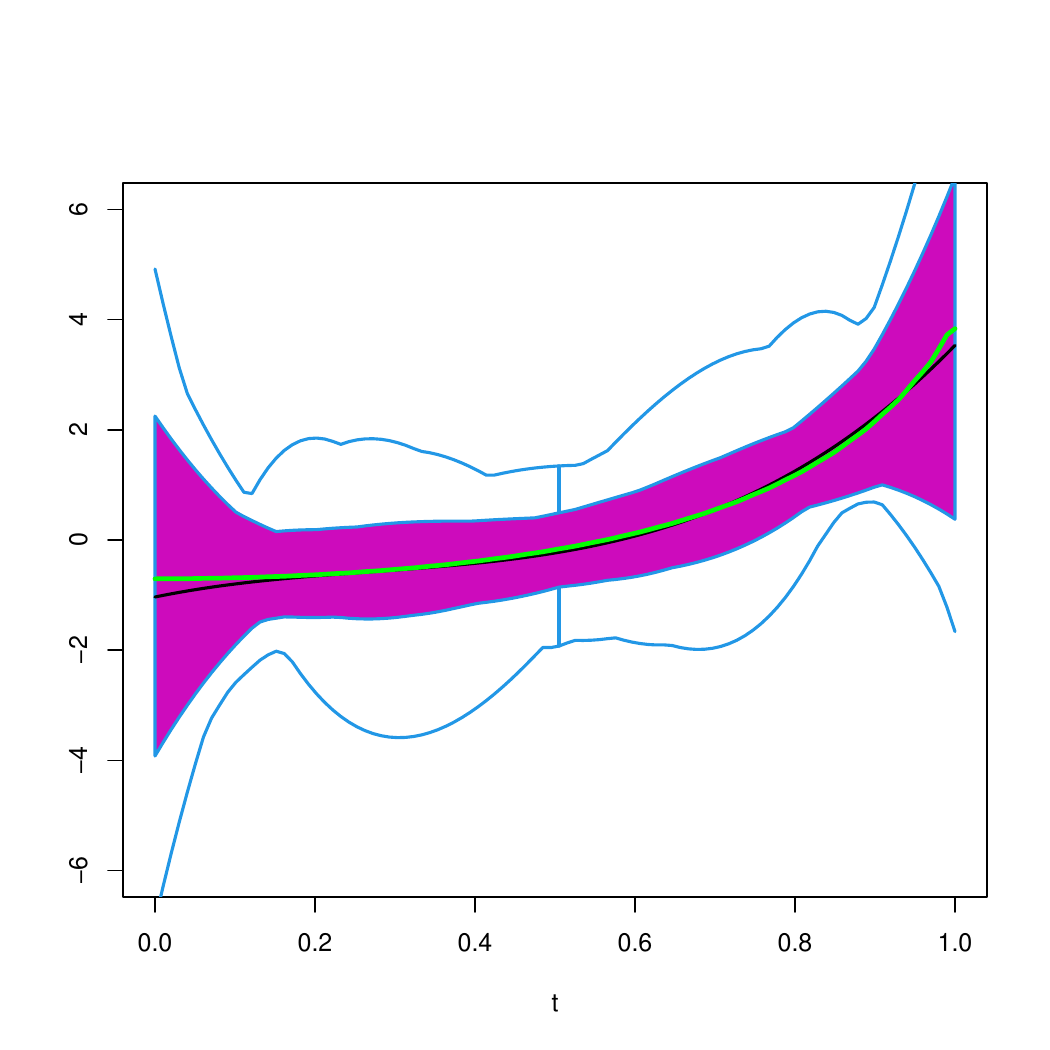}
  &  \includegraphics[scale=0.40]{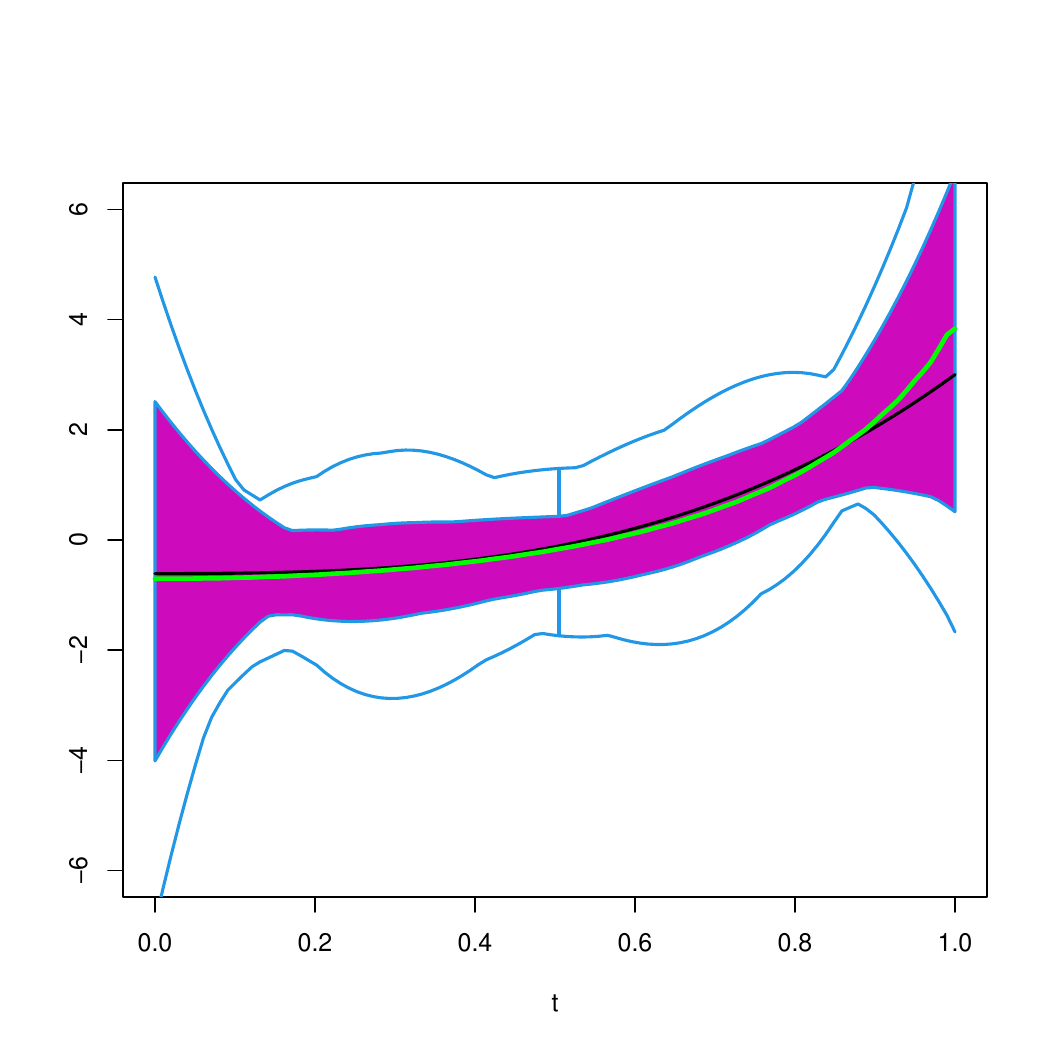}
   \\
   $\wbeta_{\wclBOX}$ & $\wbeta_{\wemeBOX}$ \\[-3ex]
  \includegraphics[scale=0.40]{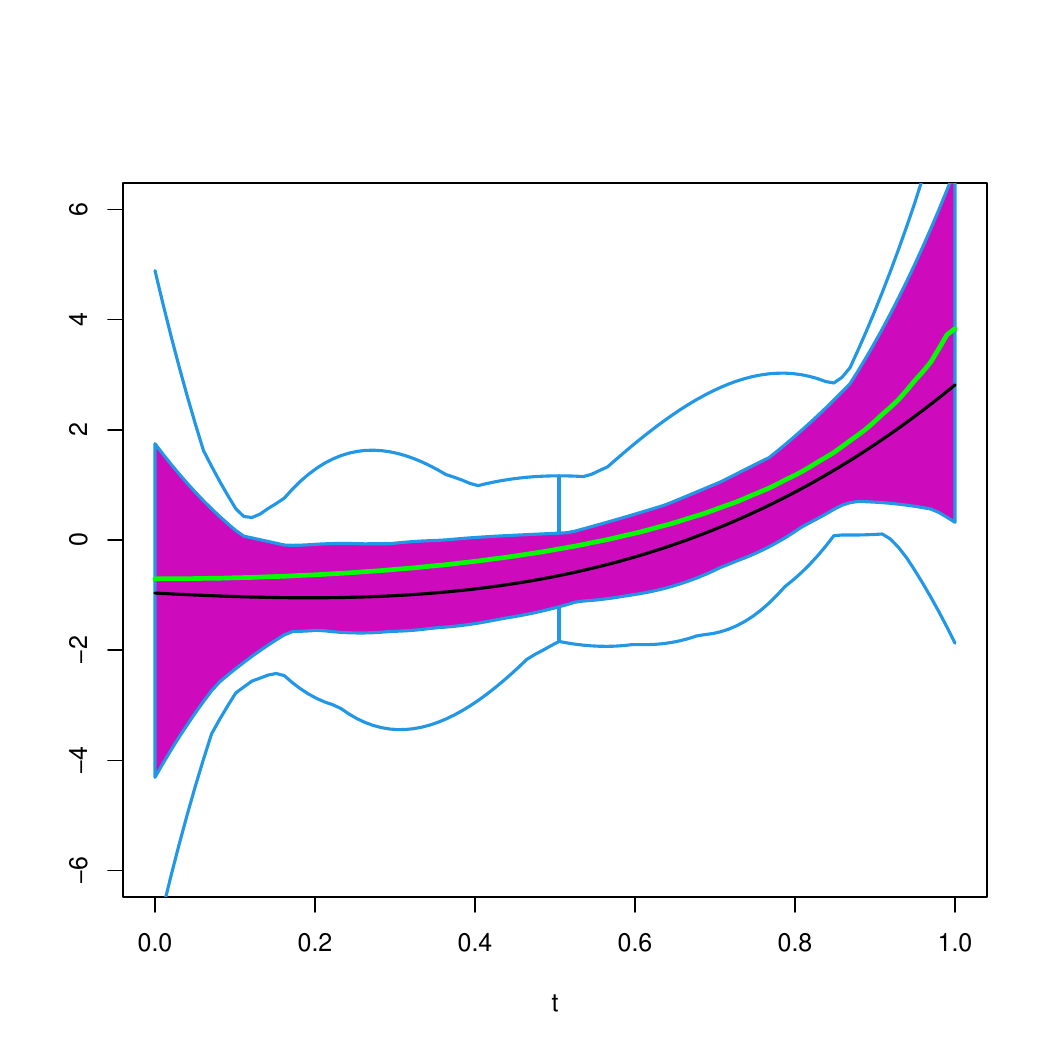}
  &  \includegraphics[scale=0.40]{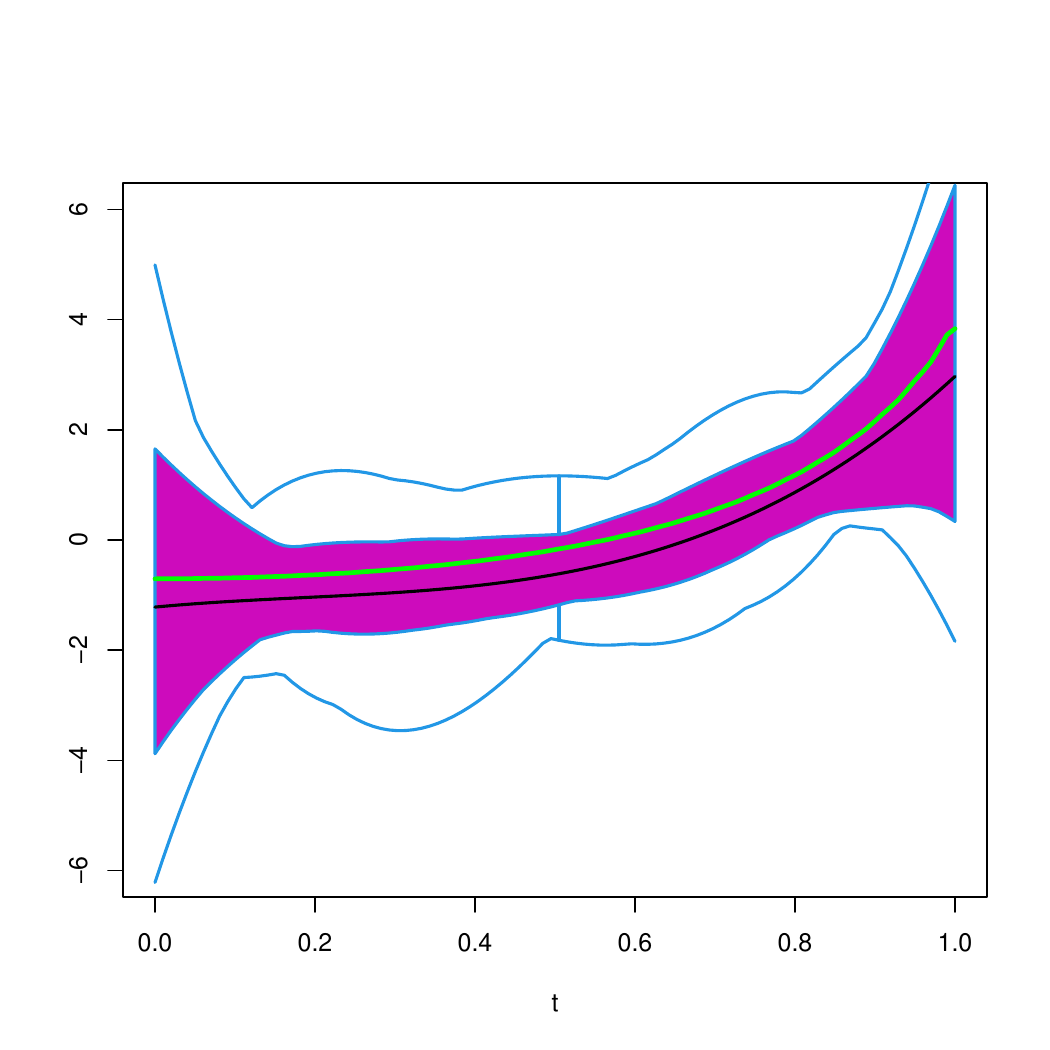}

\end{tabular}
\caption{\small \label{fig:wbeta-C410-poda0}  Functional boxplot of the estimators for $\beta_0$ under $C_{4,0.10}$  within the interval $[0,1]$. 
The true function is shown with a green dashed line, while the black solid one is the central 
curve of the $n_R = 1000$ estimates $\wbeta$.  }
\end{center} 
\end{figure}

\begin{figure}[tp]
 \begin{center}
 \footnotesize
 \renewcommand{\arraystretch}{0.2}
 \newcolumntype{M}{>{\centering\arraybackslash}m{\dimexpr.01\linewidth-1\tabcolsep}}
   \newcolumntype{G}{>{\centering\arraybackslash}m{\dimexpr.45\linewidth-1\tabcolsep}}
%\begin{tabular}{MGG}
\begin{tabular}{GG}
  $\wbeta_{\clas}$ & $\wbeta_{\eme}$   \\[-3ex]    
 
\includegraphics[scale=0.40]{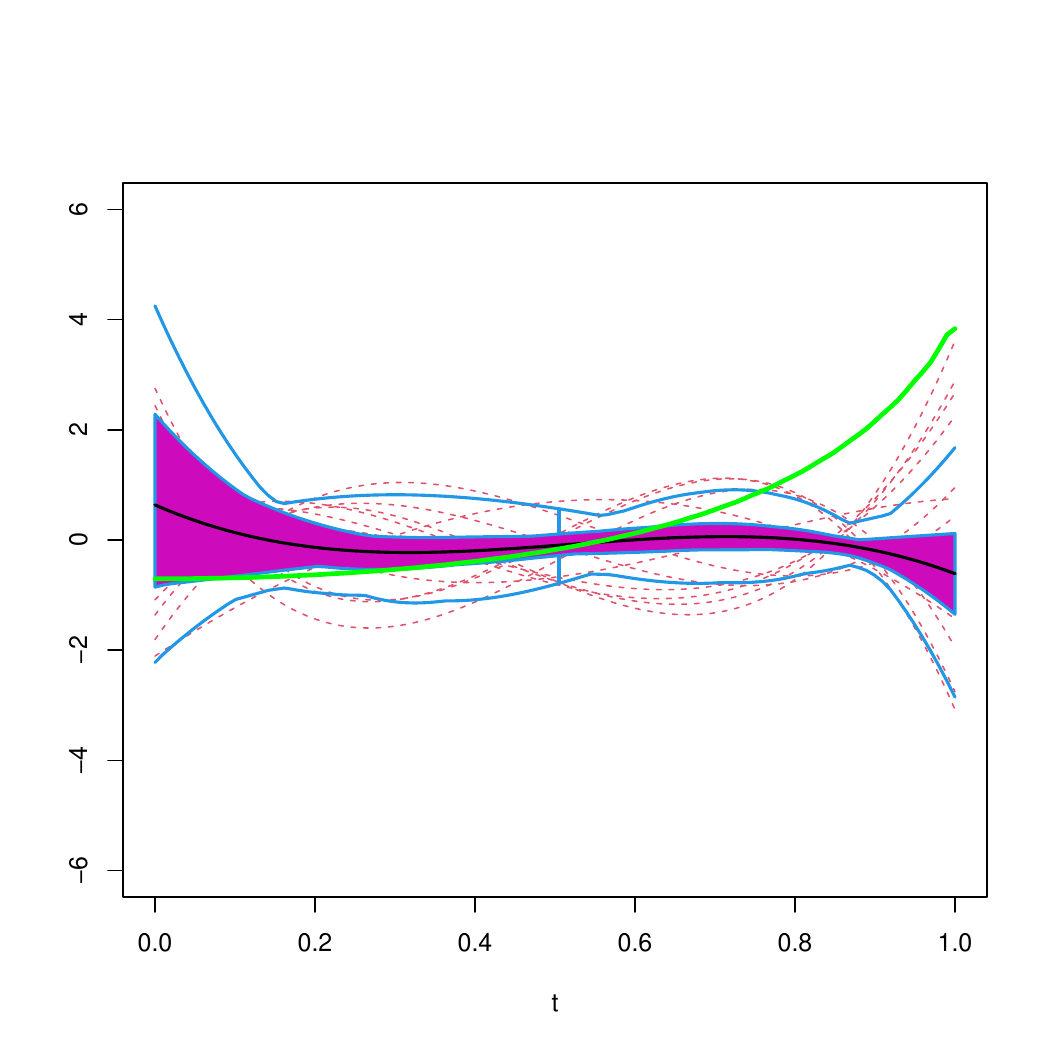}
 &  \includegraphics[scale=0.40]{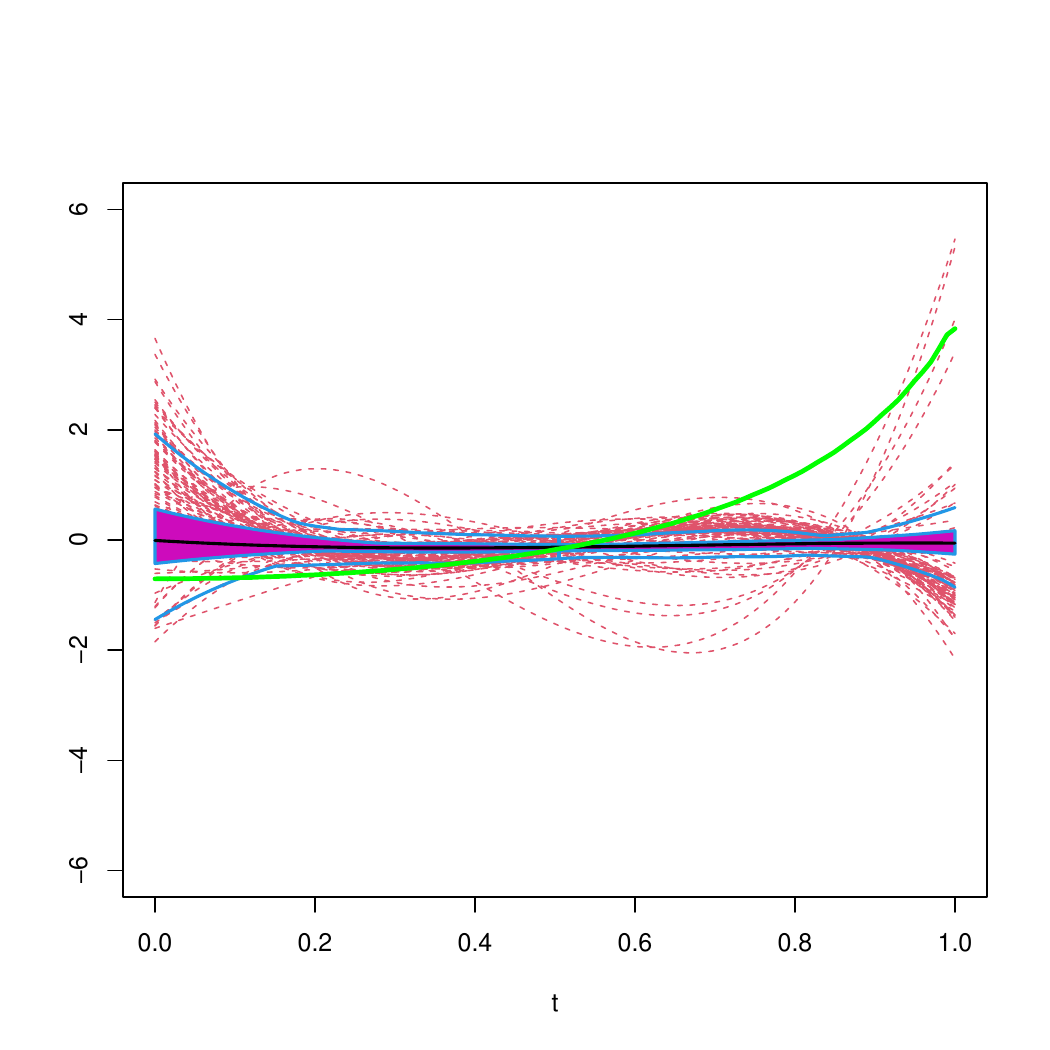}\\
   $\wbeta_{\wclHR}$ & $\wbeta_{\wemeHR}$ \\[-3ex] 
    \includegraphics[scale=0.40]{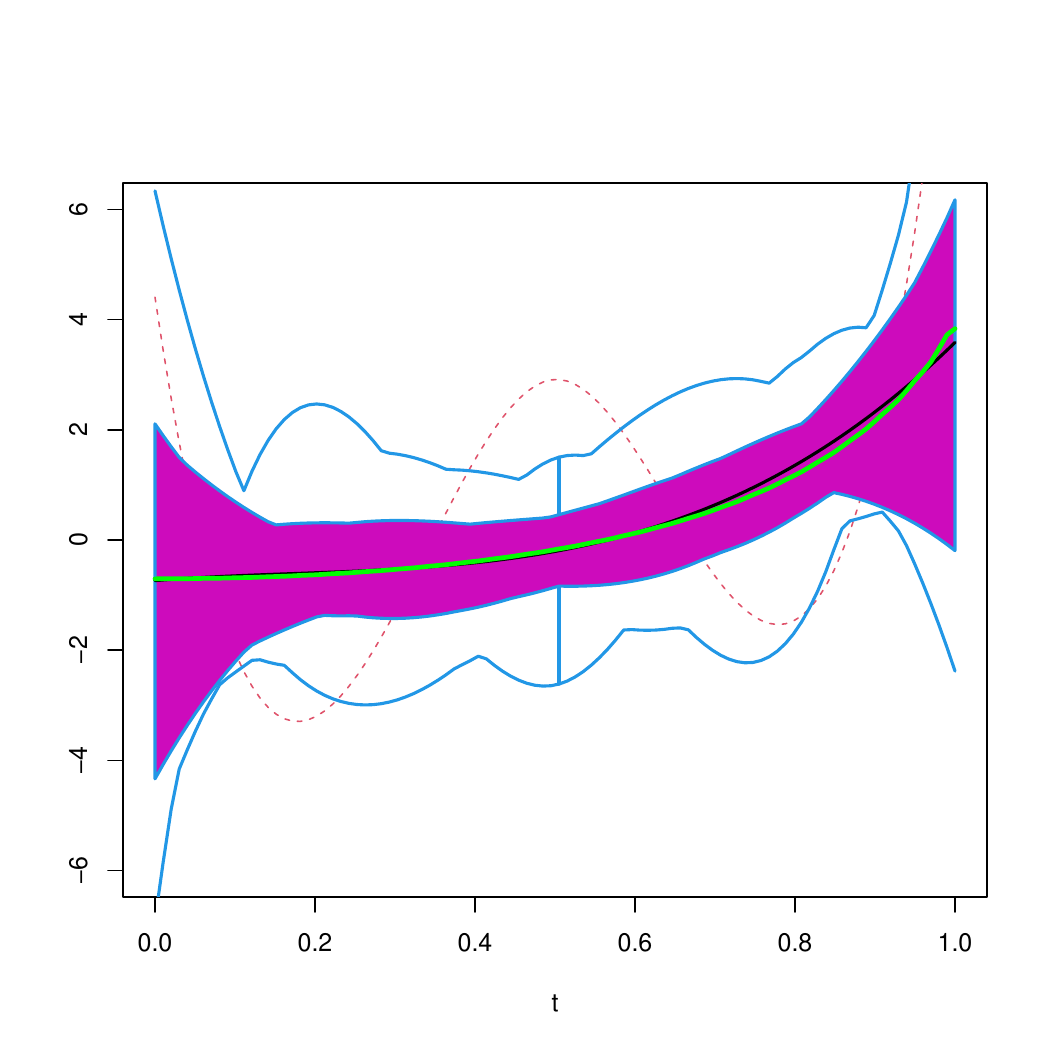}
  &  \includegraphics[scale=0.40]{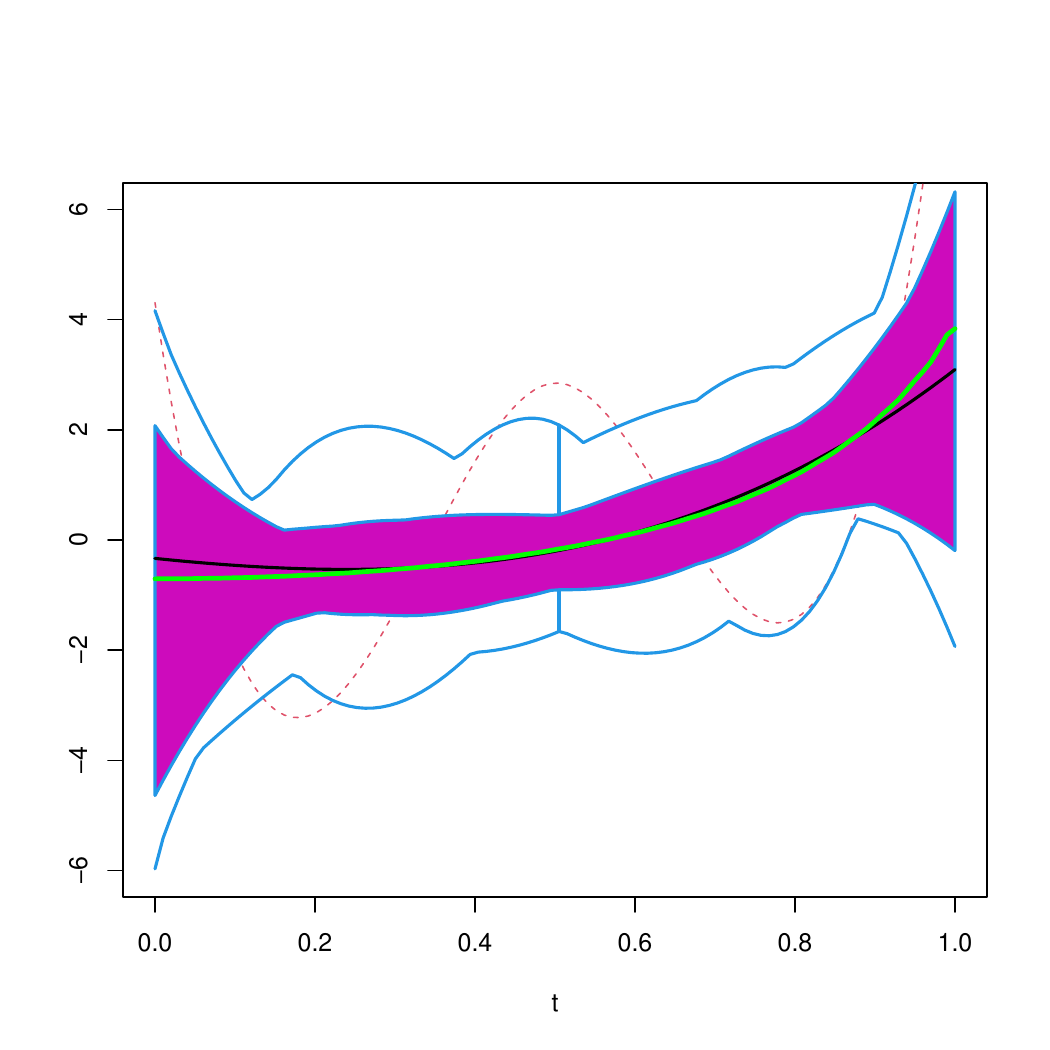}
   \\
   $\wbeta_{\wclBOX}$ & $\wbeta_{\wemeBOX}$ \\[-3ex]
  \includegraphics[scale=0.40]{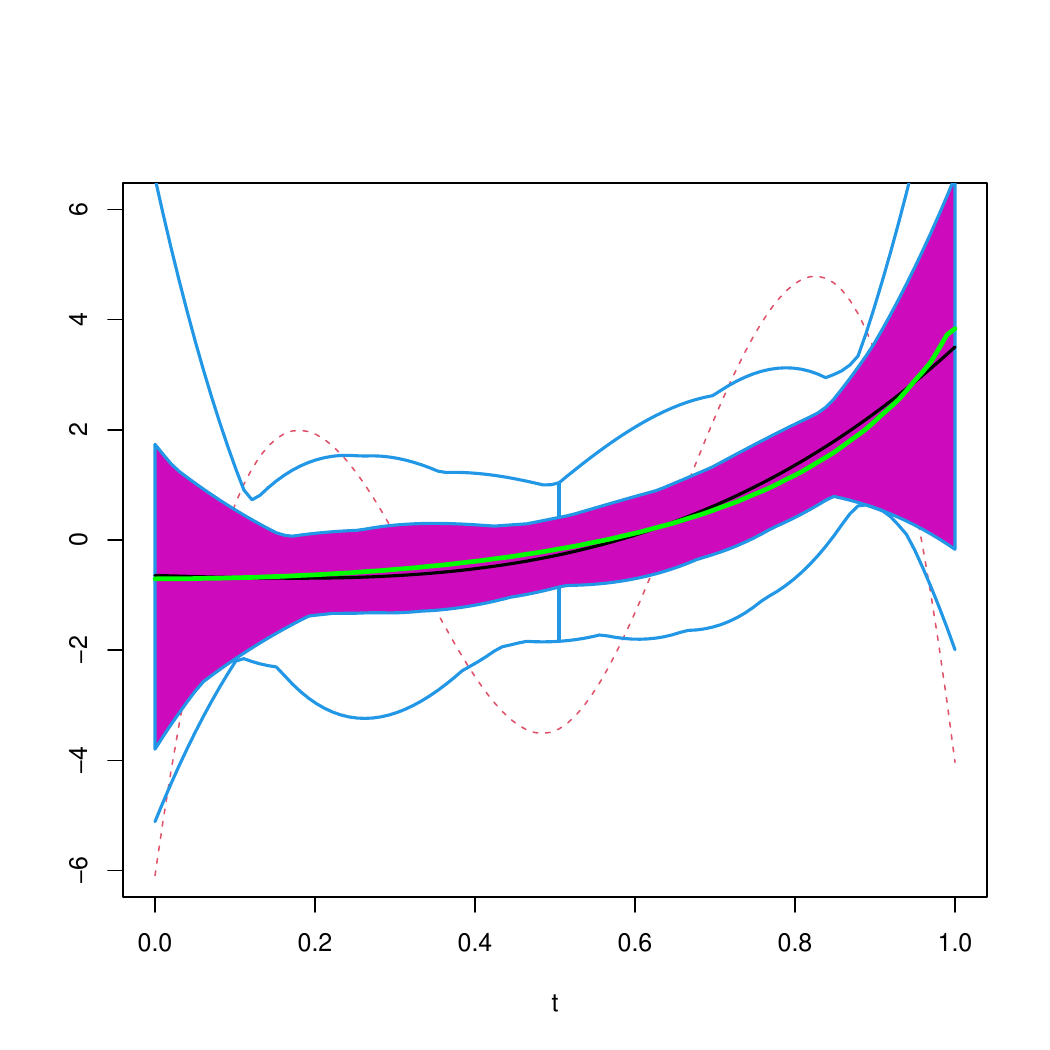}
  &  \includegraphics[scale=0.40]{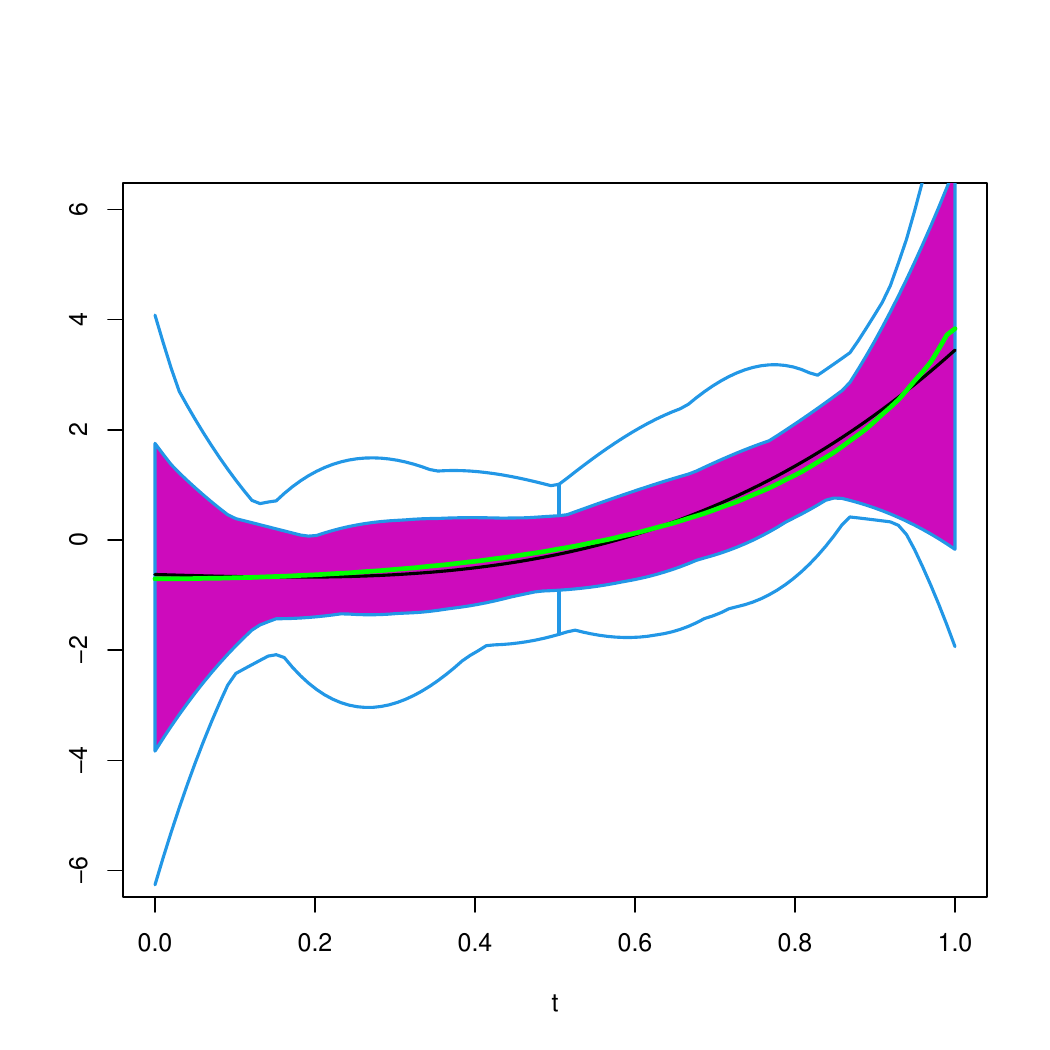}

\end{tabular}
\caption{\small \label{fig:wbeta-C55-poda0}  Functional boxplot of the estimators for $\beta_0$ under $C_{5,0.05}$  within the interval $[0,1]$. 
The true function is shown with a green dashed line, while the black solid one is the central 
curve of the $n_R = 1000$ estimates $\wbeta$.  }
\end{center} 
\end{figure}

\begin{figure}[tp]
 \begin{center}
 \footnotesize
 \renewcommand{\arraystretch}{0.2}
 \newcolumntype{M}{>{\centering\arraybackslash}m{\dimexpr.01\linewidth-1\tabcolsep}}
   \newcolumntype{G}{>{\centering\arraybackslash}m{\dimexpr.45\linewidth-1\tabcolsep}}
%\begin{tabular}{MGG}
\begin{tabular}{GG}
  $\wbeta_{\clas}$ & $\wbeta_{\eme}$   \\[-3ex]    
 
\includegraphics[scale=0.40]{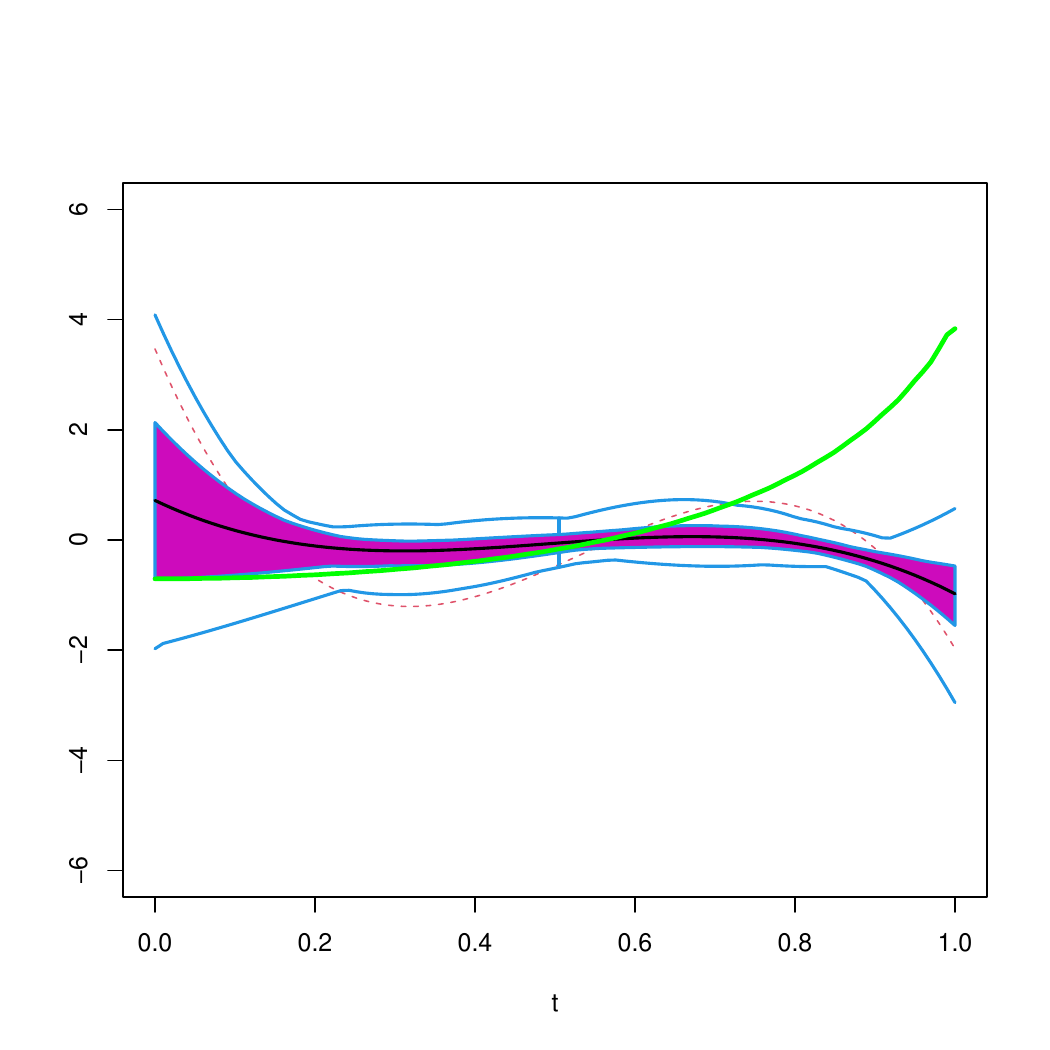}
 &  \includegraphics[scale=0.40]{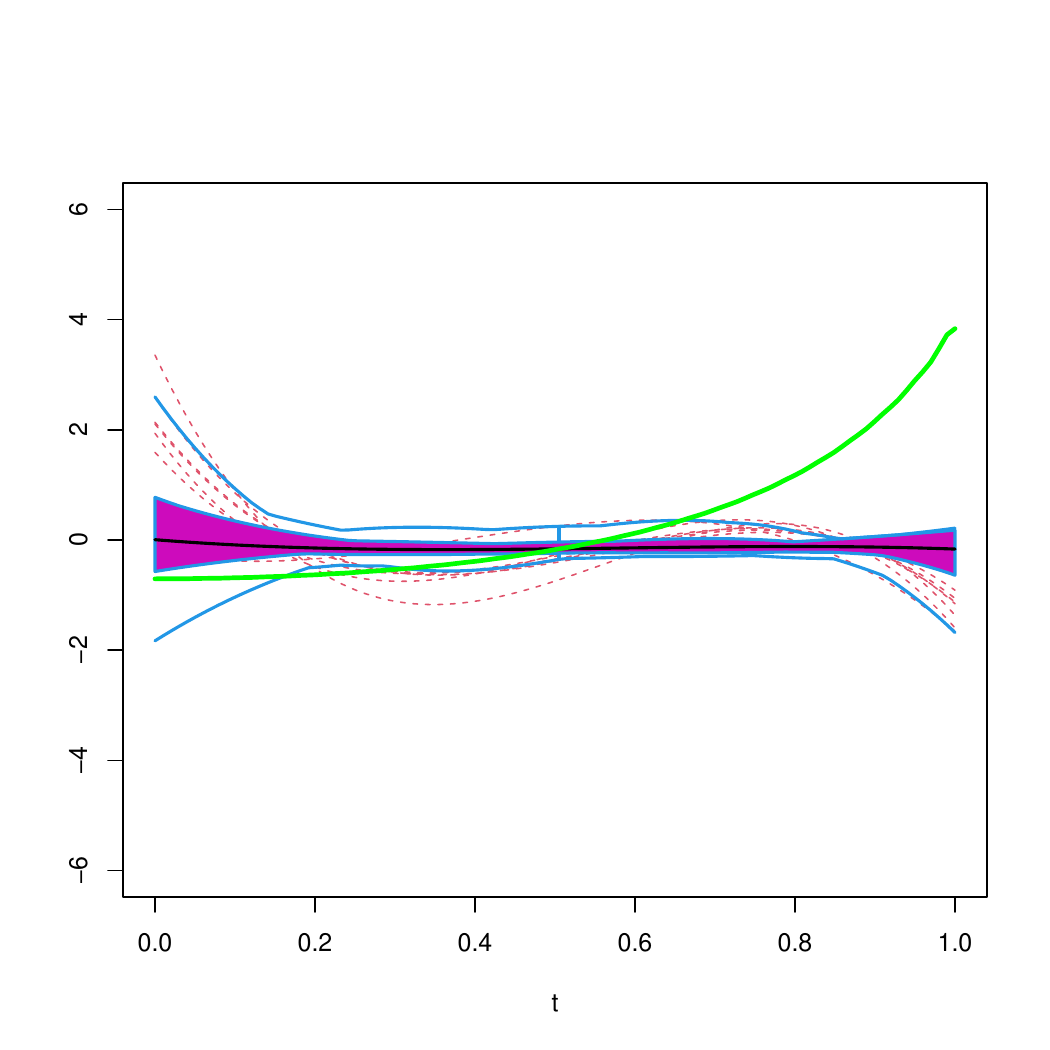}\\
   $\wbeta_{\wclHR}$ & $\wbeta_{\wemeHR}$ \\[-3ex] 
    \includegraphics[scale=0.40]{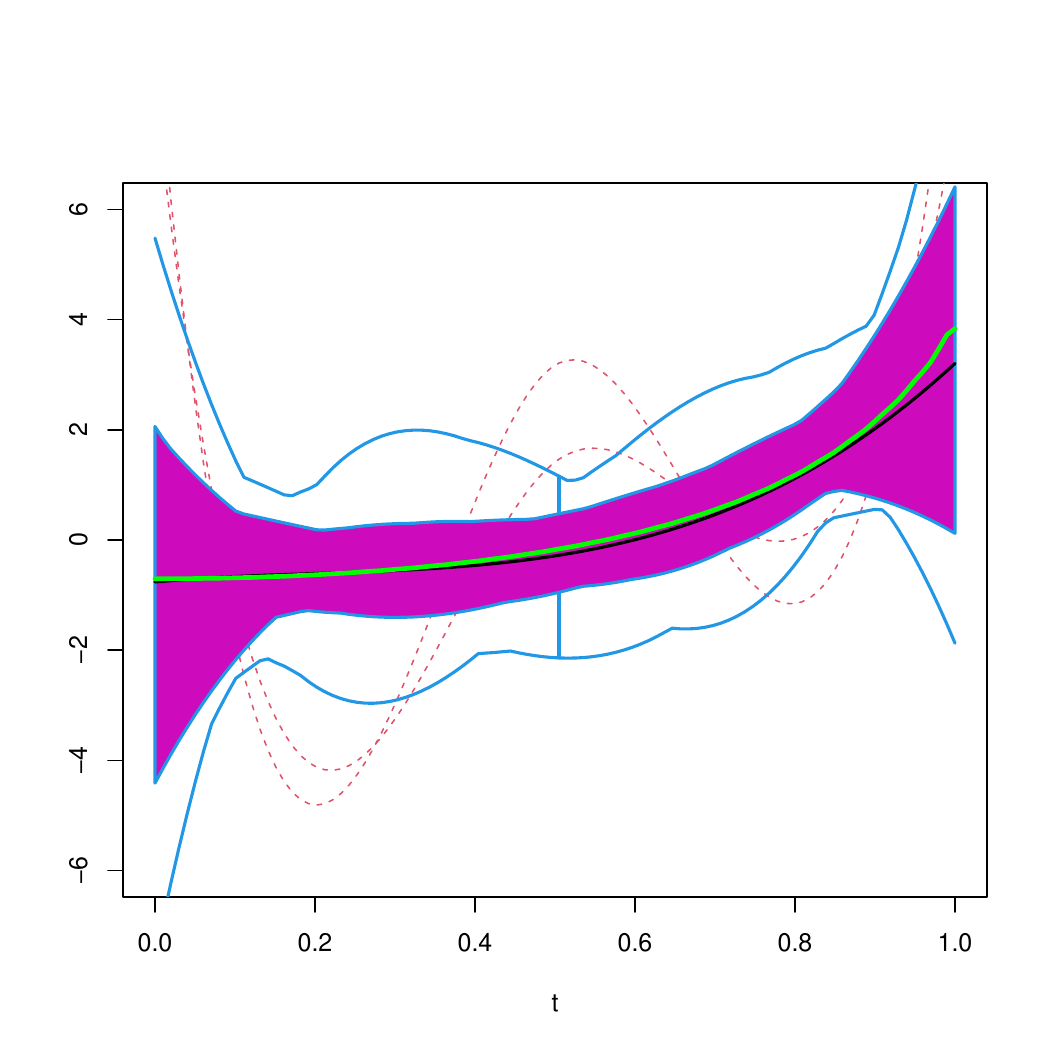}
  &  \includegraphics[scale=0.40]{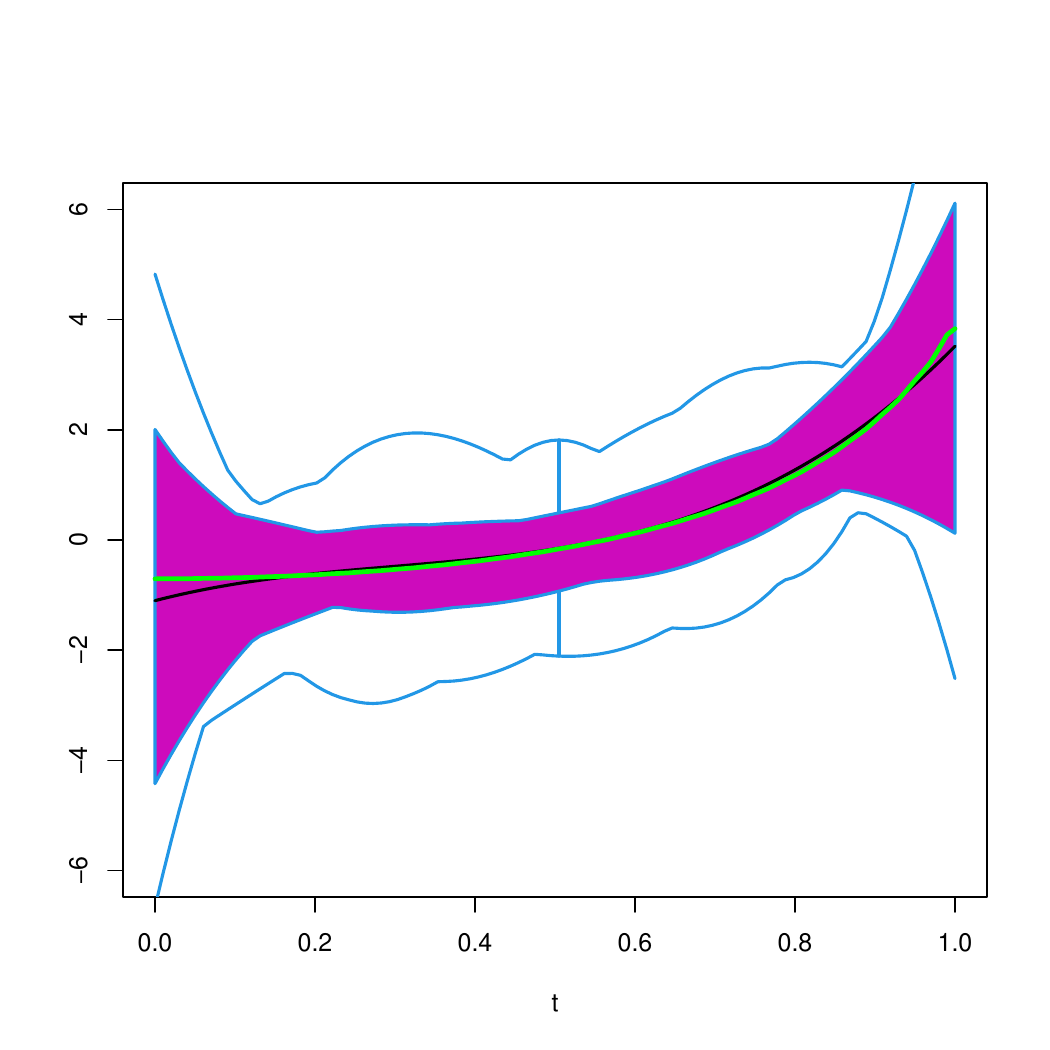}
   \\
   $\wbeta_{\wclBOX}$ & $\wbeta_{\wemeBOX}$ \\[-3ex]
  \includegraphics[scale=0.40]{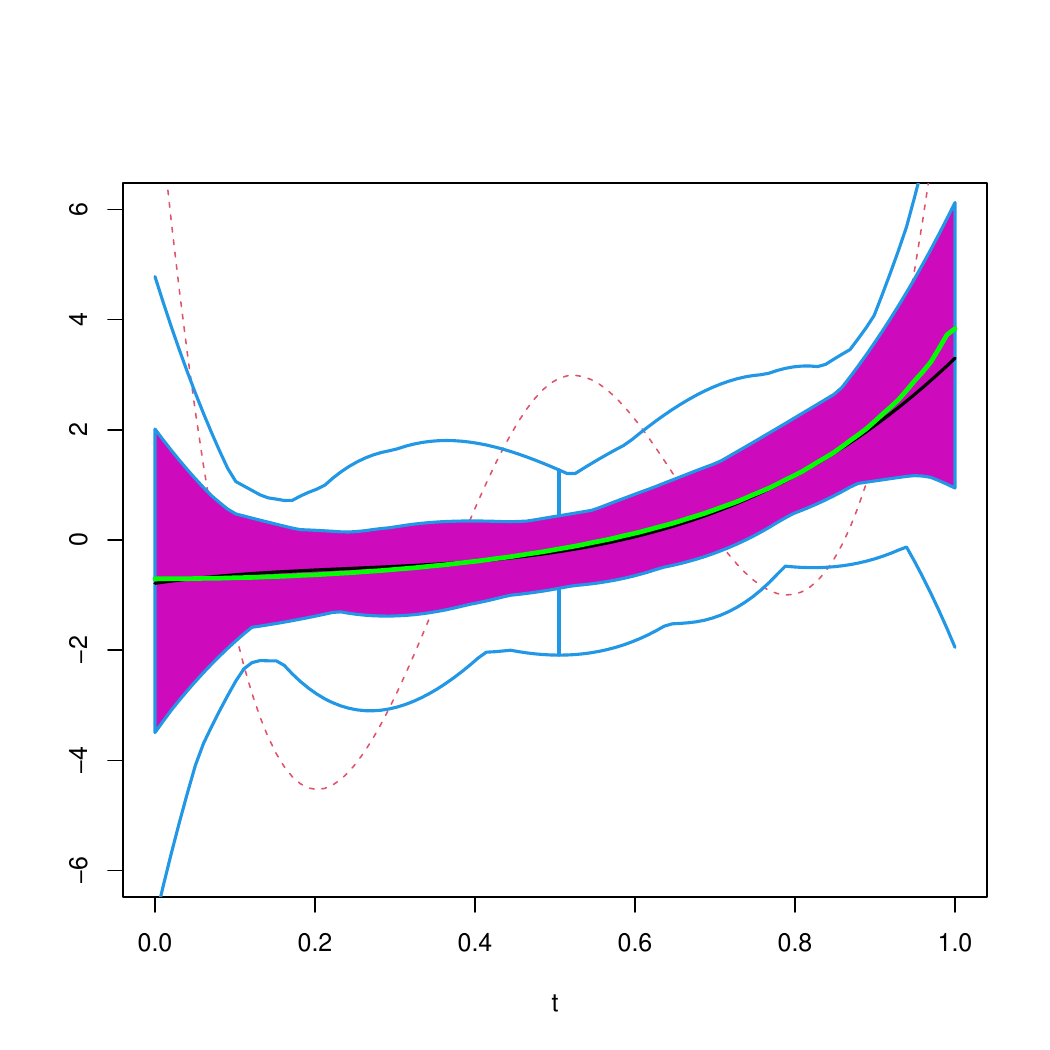}
  &  \includegraphics[scale=0.40]{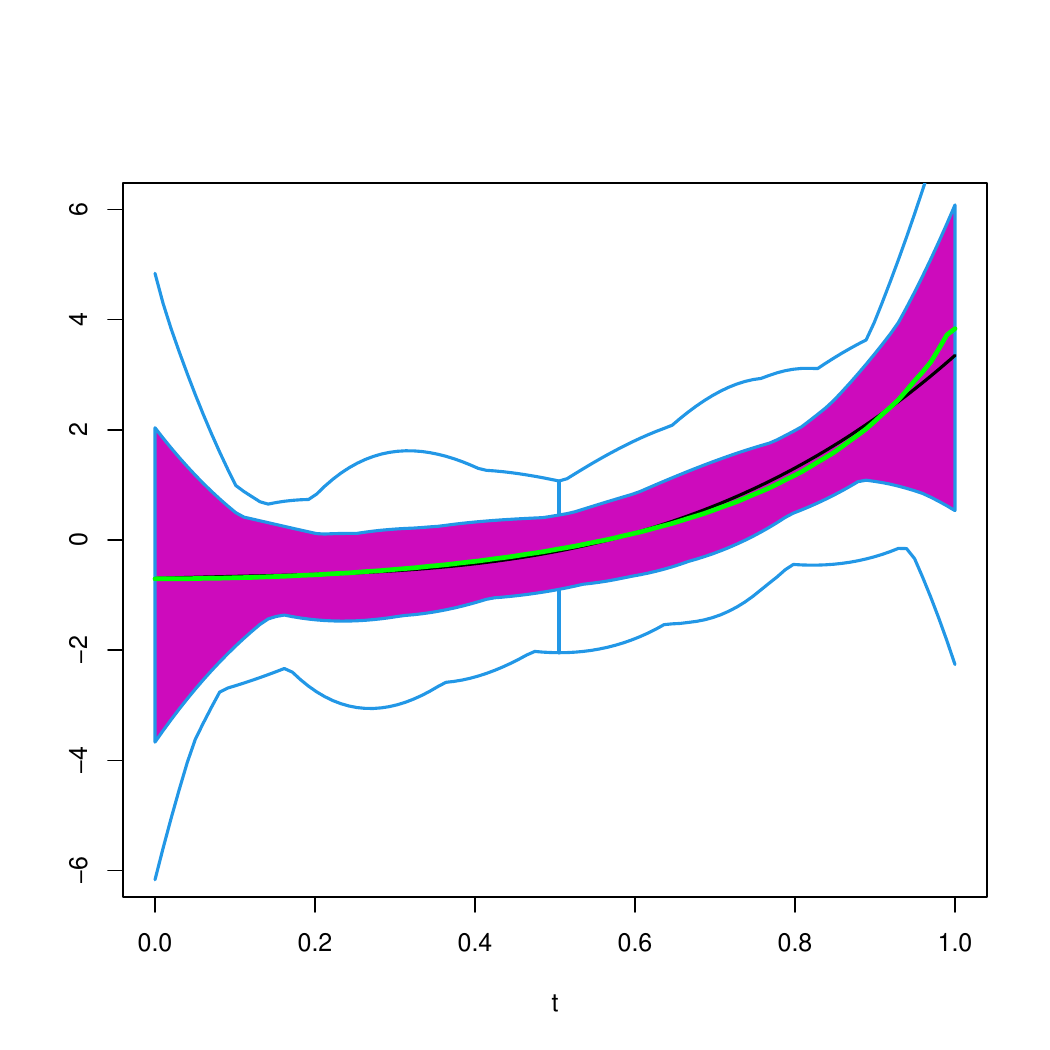}

\end{tabular}
\caption{\small \label{fig:wbeta-C510-poda0}  Functional boxplot of the estimators for $\beta_0$ under $C_{5,0.10}$  within the interval $[0,1]$. 
The true function is shown with a green dashed line, while the black solid one is the central 
curve of the $n_R = 1000$ estimates $\wbeta$.  }
\end{center} 
\end{figure}

%% file: graficos-poda5-sinC6.tex
\begin{figure}[tp]
 \begin{center}
 \footnotesize
 \renewcommand{\arraystretch}{0.2}
 \newcolumntype{M}{>{\centering\arraybackslash}m{\dimexpr.01\linewidth-1\tabcolsep}}
   \newcolumntype{G}{>{\centering\arraybackslash}m{\dimexpr.45\linewidth-1\tabcolsep}}
%\begin{tabular}{MGG}
\begin{tabular}{GG}
  $\wbeta_{\clas}$ & $\wbeta_{\eme}$   \\[-3ex]
     \includegraphics[scale=0.40]{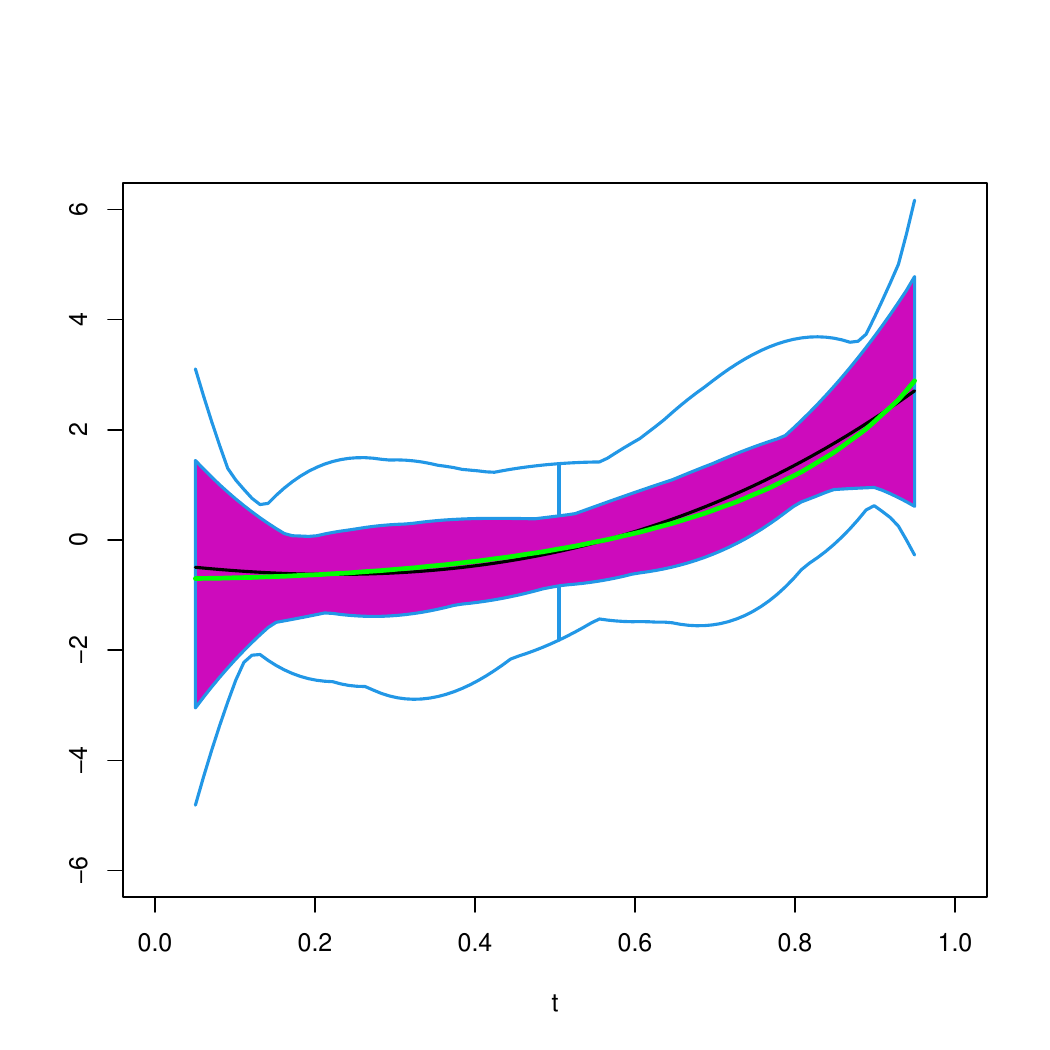} 
&  \includegraphics[scale=0.40]{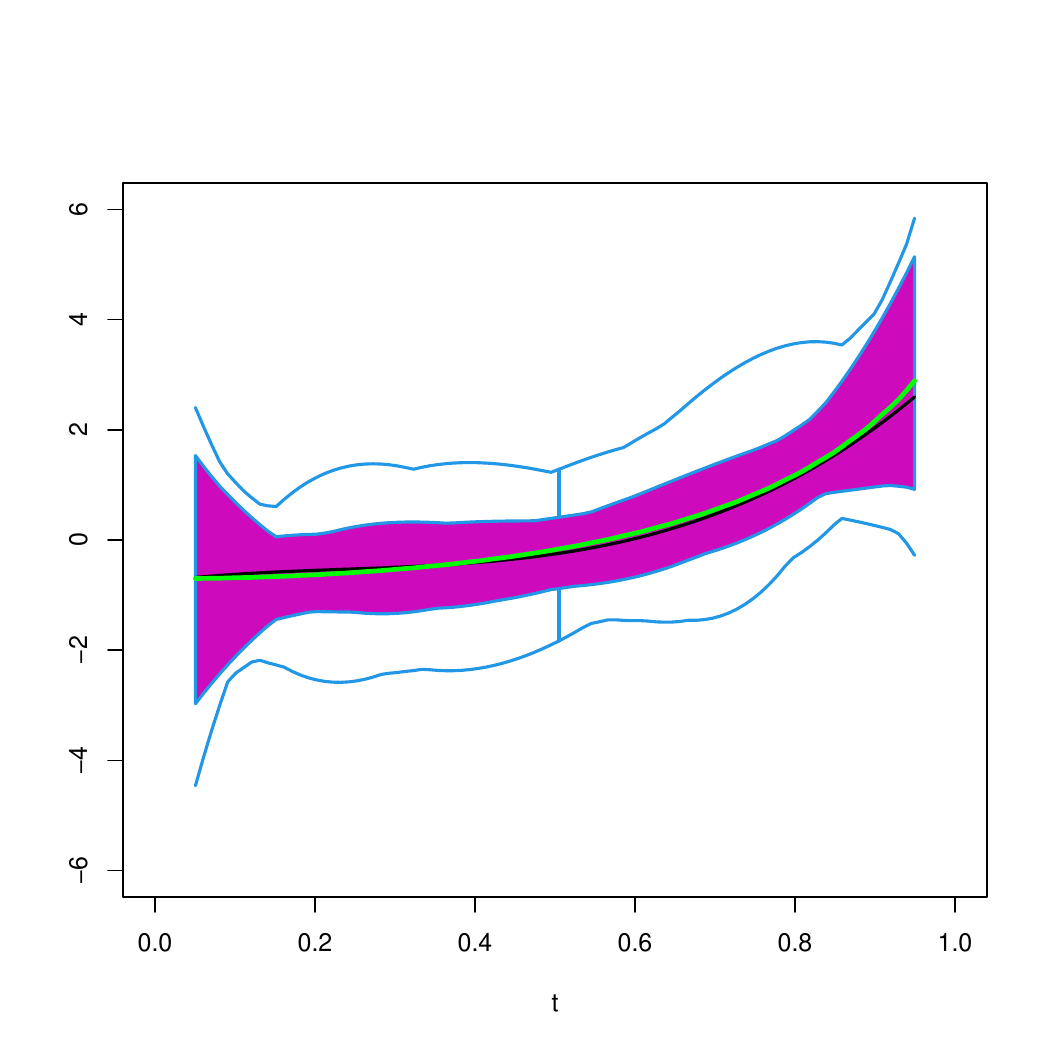} \\
   $\wbeta_{\wclHR}$ & $\wbeta_{\wemeHR}$ \\ [-3ex]
   \includegraphics[scale=0.40]{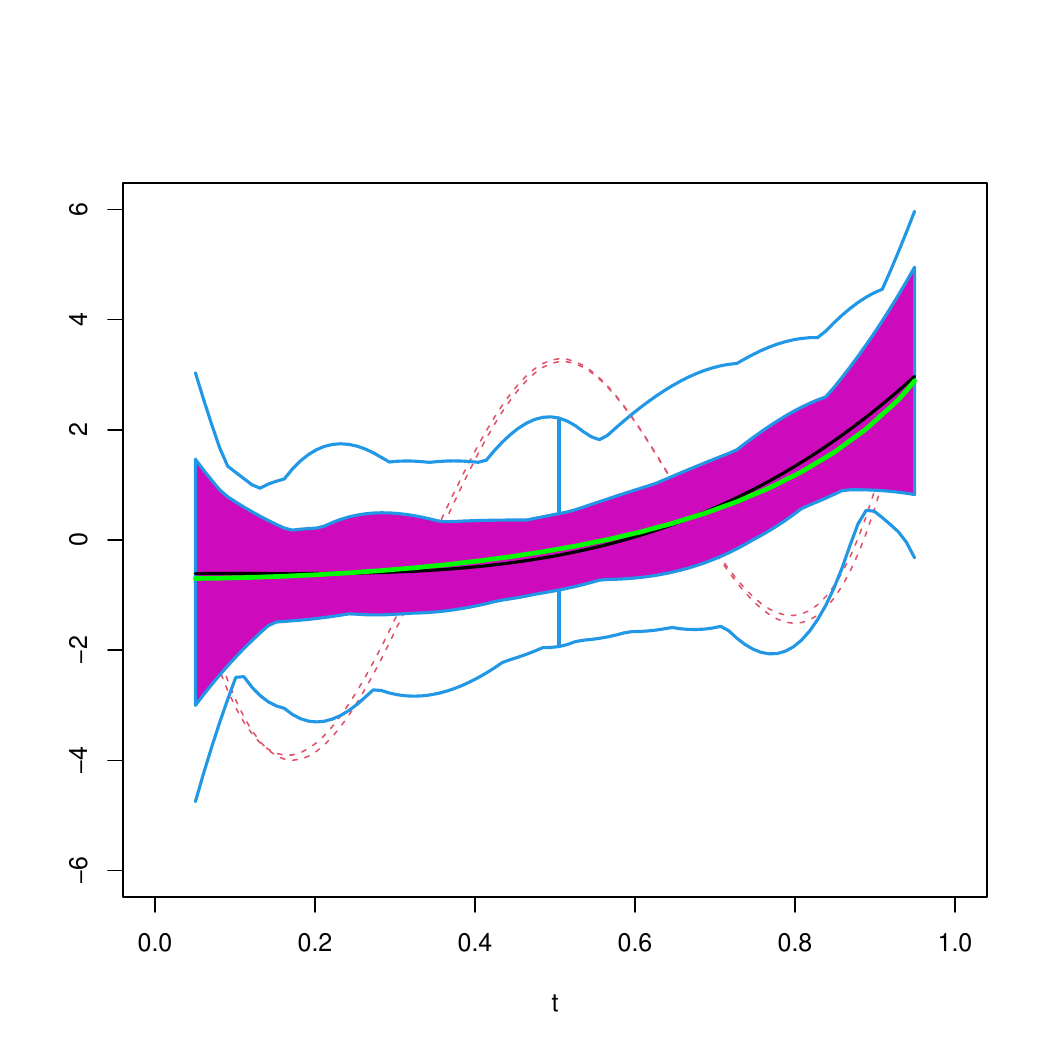}
&  \includegraphics[scale=0.40]{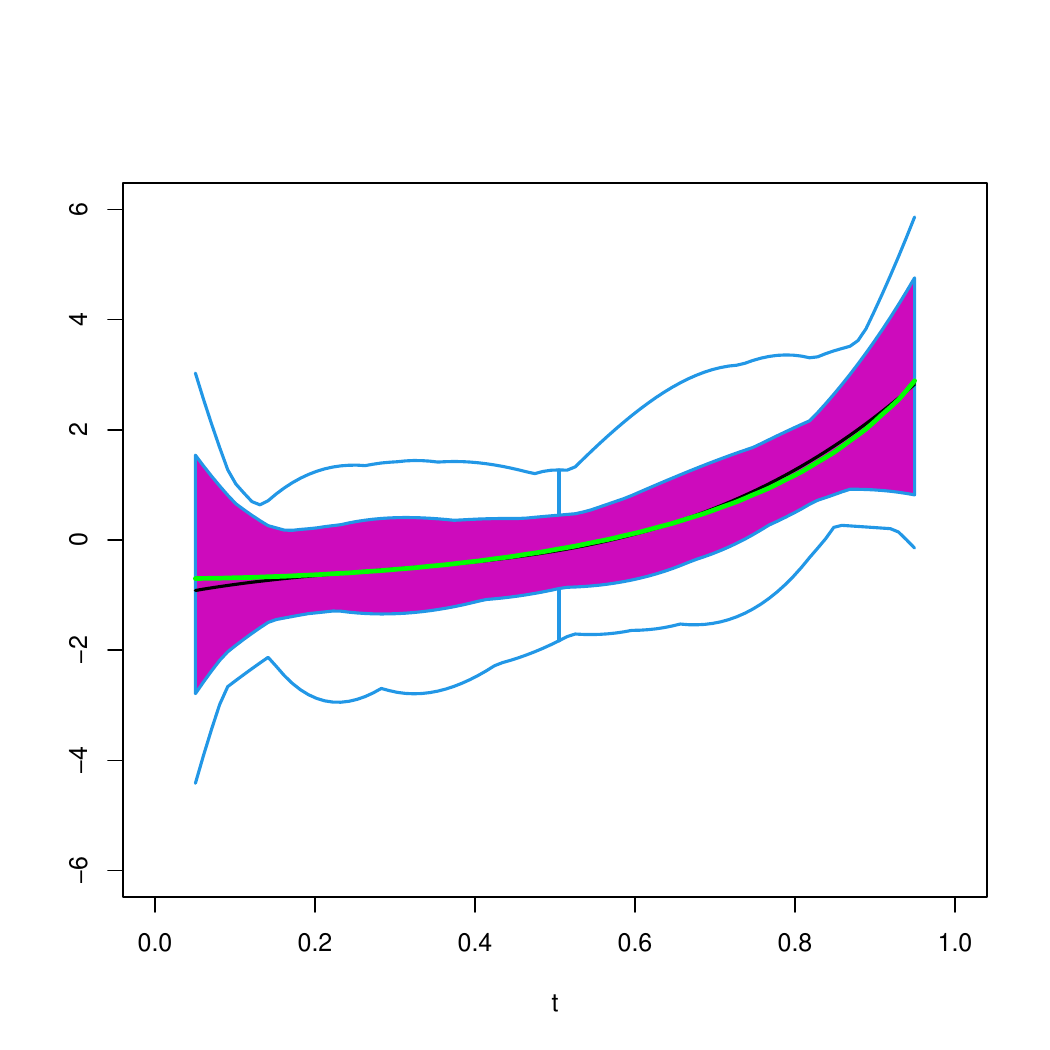} 
  \\
   $\wbeta_{\wclBOX}$ & $\wbeta_{\wemeBOX}$ \\[-3ex]
\includegraphics[scale=0.40]{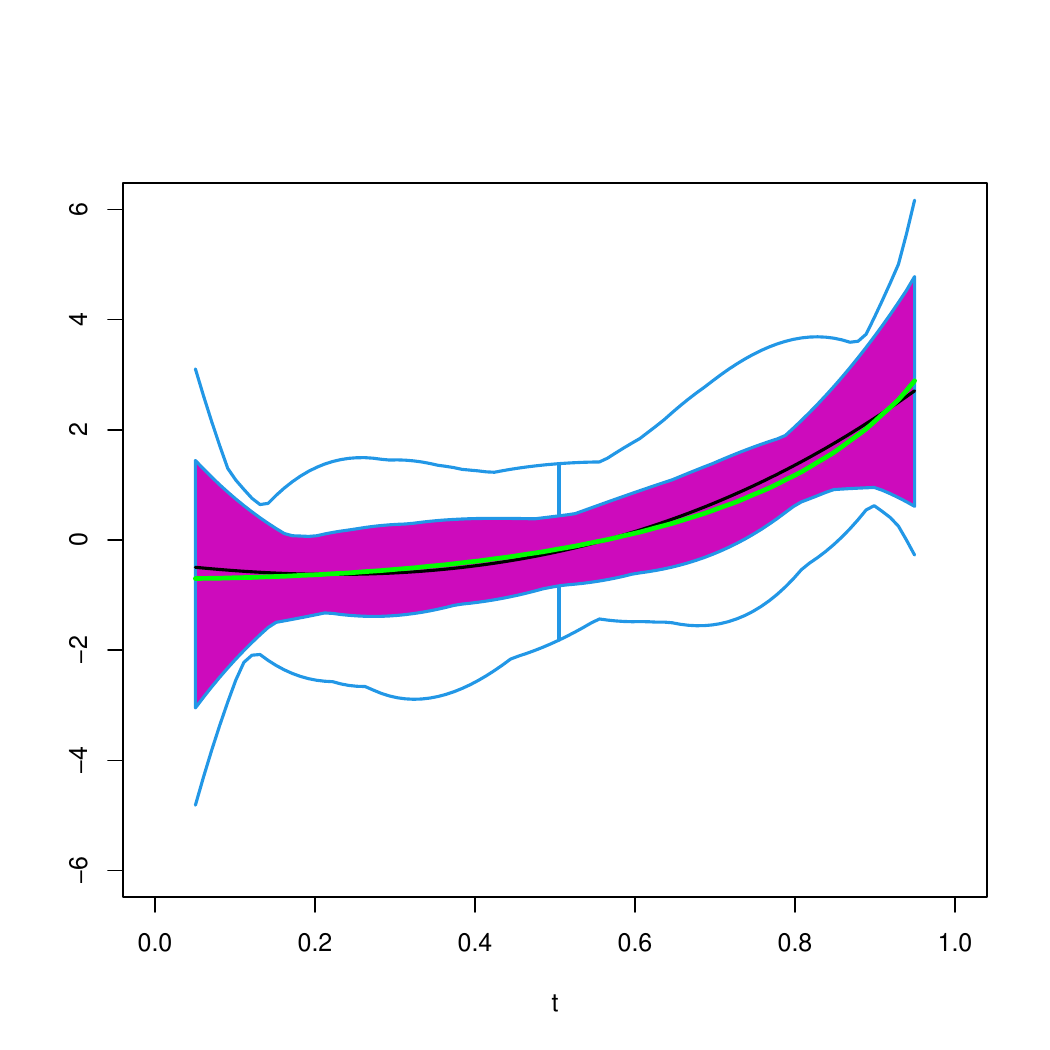}
&   \includegraphics[scale=0.40]{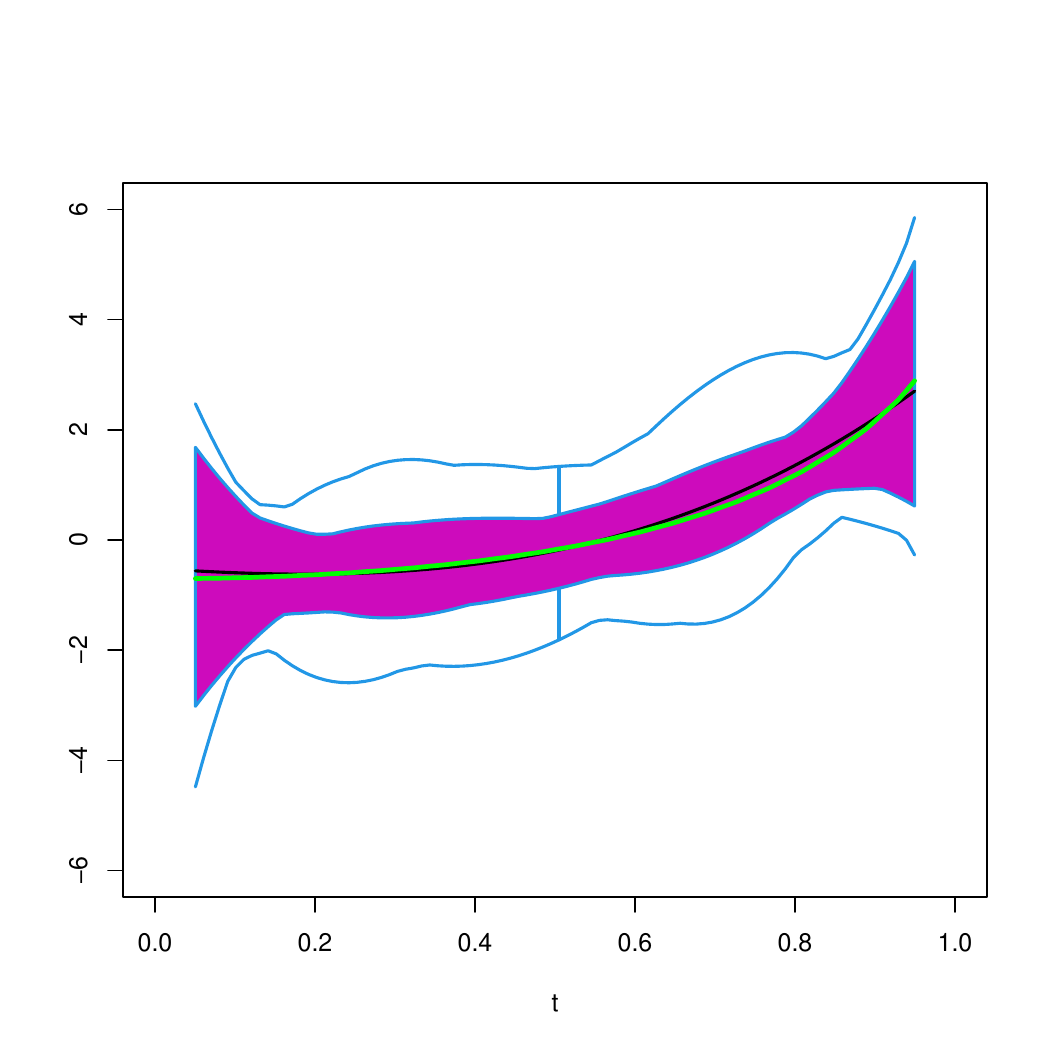} 
 \end{tabular}
\caption{\small \label{fig:wbeta-C0-poda5}  Functional boxplot of the estimators for $\beta_0$ under $C_{0}$ within the interval $[0.05,0.95]$. 
The true function is shown with a green dashed line, while the black solid one is the central 
curve of the $n_R = 1000$ estimates $\wbeta$. }
\end{center} 
\end{figure}

\begin{figure}[tp]
 \begin{center}
 \footnotesize
 \renewcommand{\arraystretch}{0.2}
 \newcolumntype{M}{>{\centering\arraybackslash}m{\dimexpr.01\linewidth-1\tabcolsep}}
   \newcolumntype{G}{>{\centering\arraybackslash}m{\dimexpr.45\linewidth-1\tabcolsep}}
%\begin{tabular}{MGG}
\begin{tabular}{GG}
  $\wbeta_{\clas}$ & $\wbeta_{\eme}$   \\[-3ex] 
 
\includegraphics[scale=0.40]{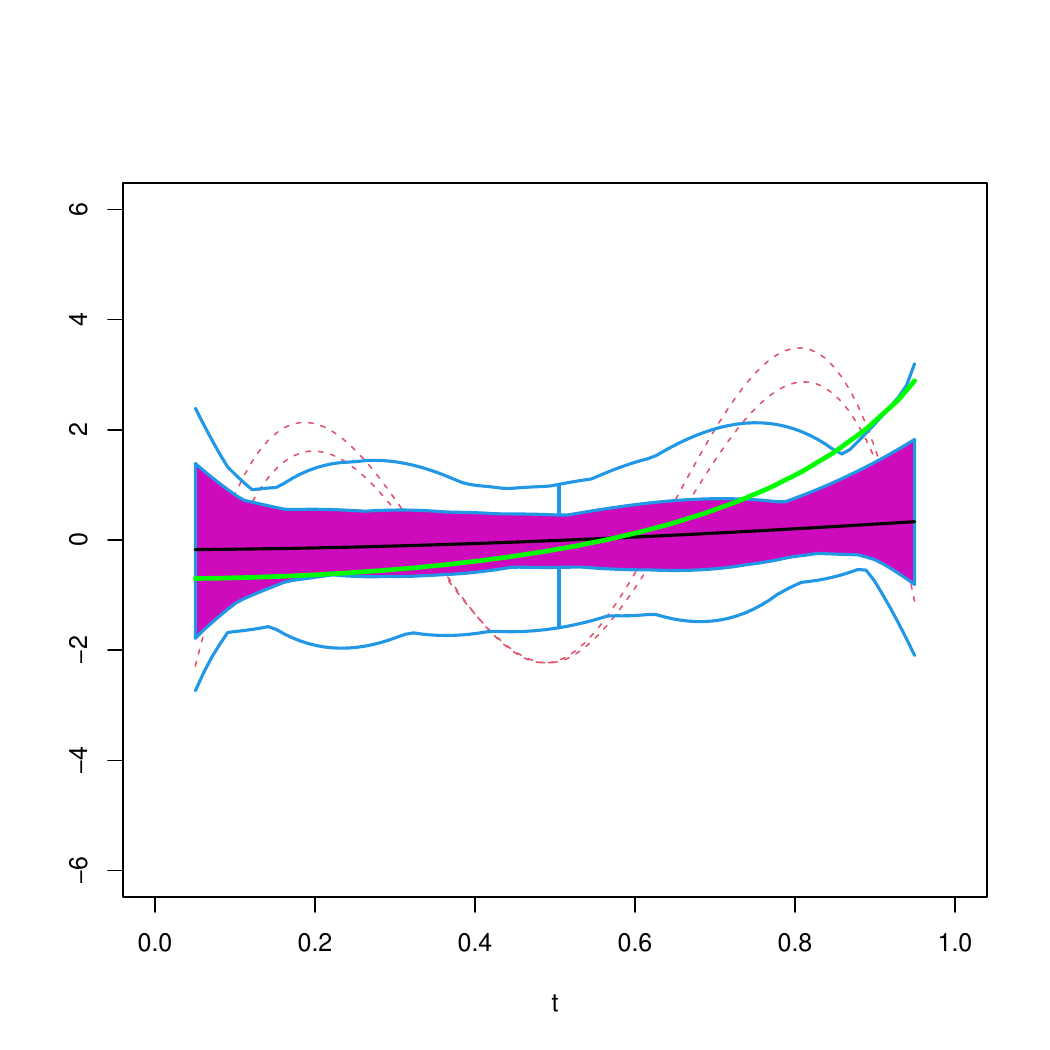}
 &  \includegraphics[scale=0.40]{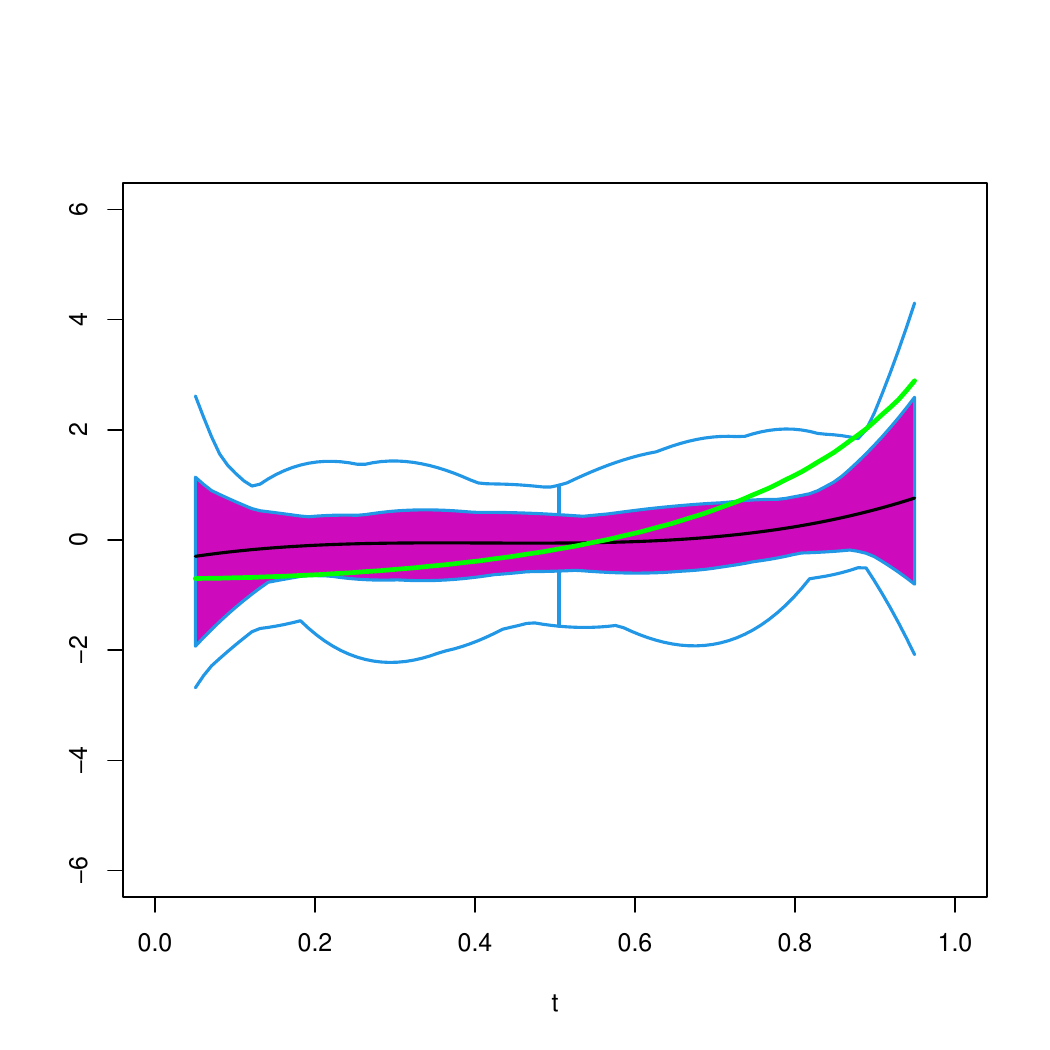}\\
   $\wbeta_{\wclHR}$ & $\wbeta_{\wemeHR}$ \\[-3ex] 
    \includegraphics[scale=0.40]{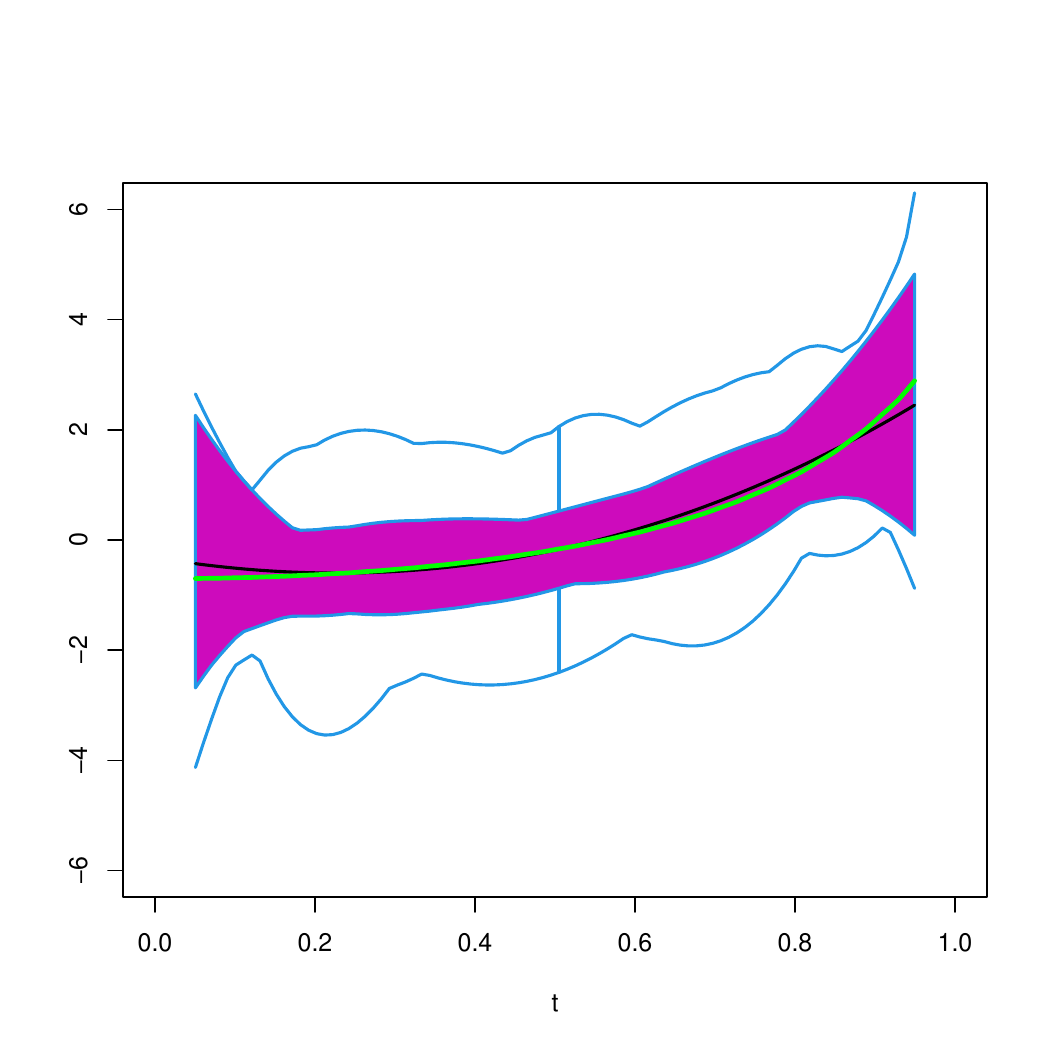}
  &  \includegraphics[scale=0.40]{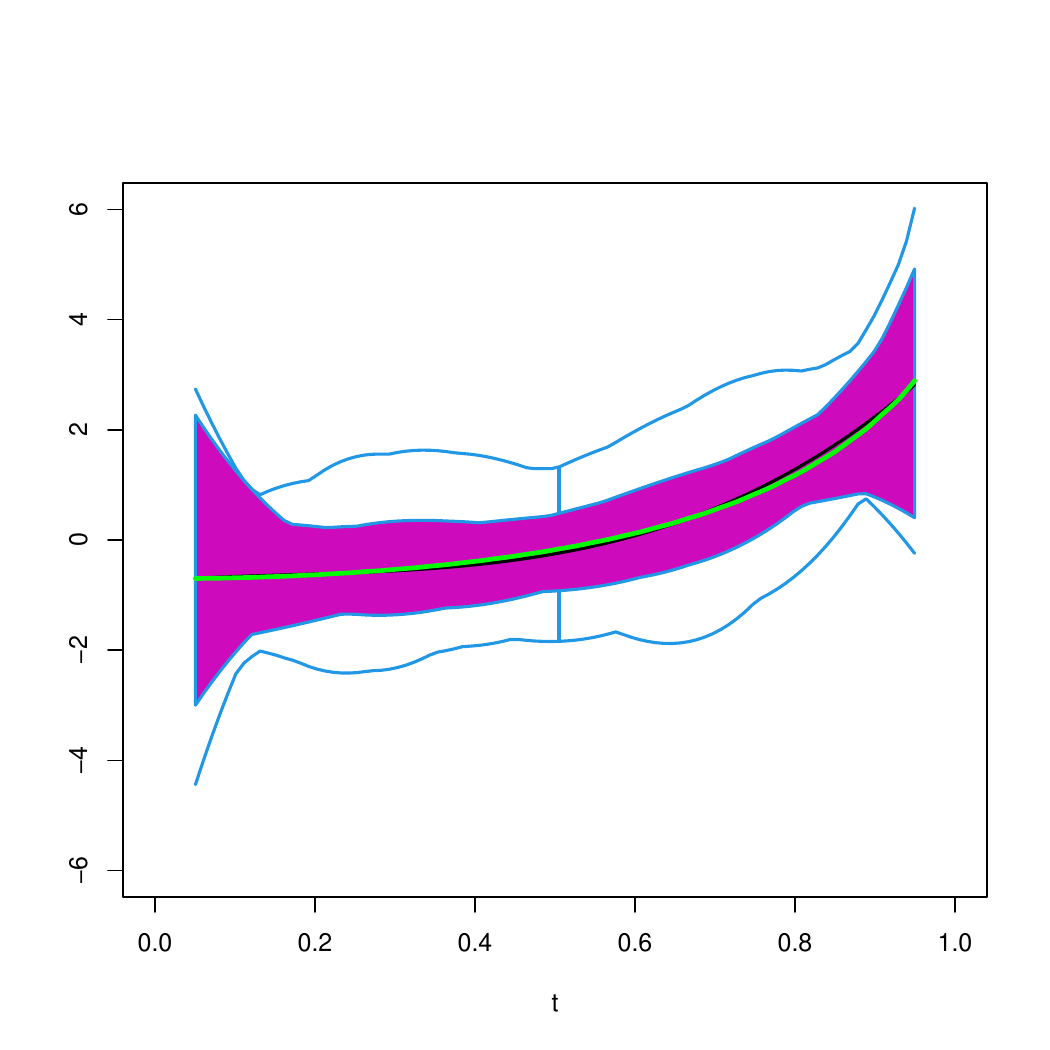}
   \\
   $\wbeta_{\wclBOX}$ & $\wbeta_{\wemeBOX}$ \\[-3ex]
  \includegraphics[scale=0.40]{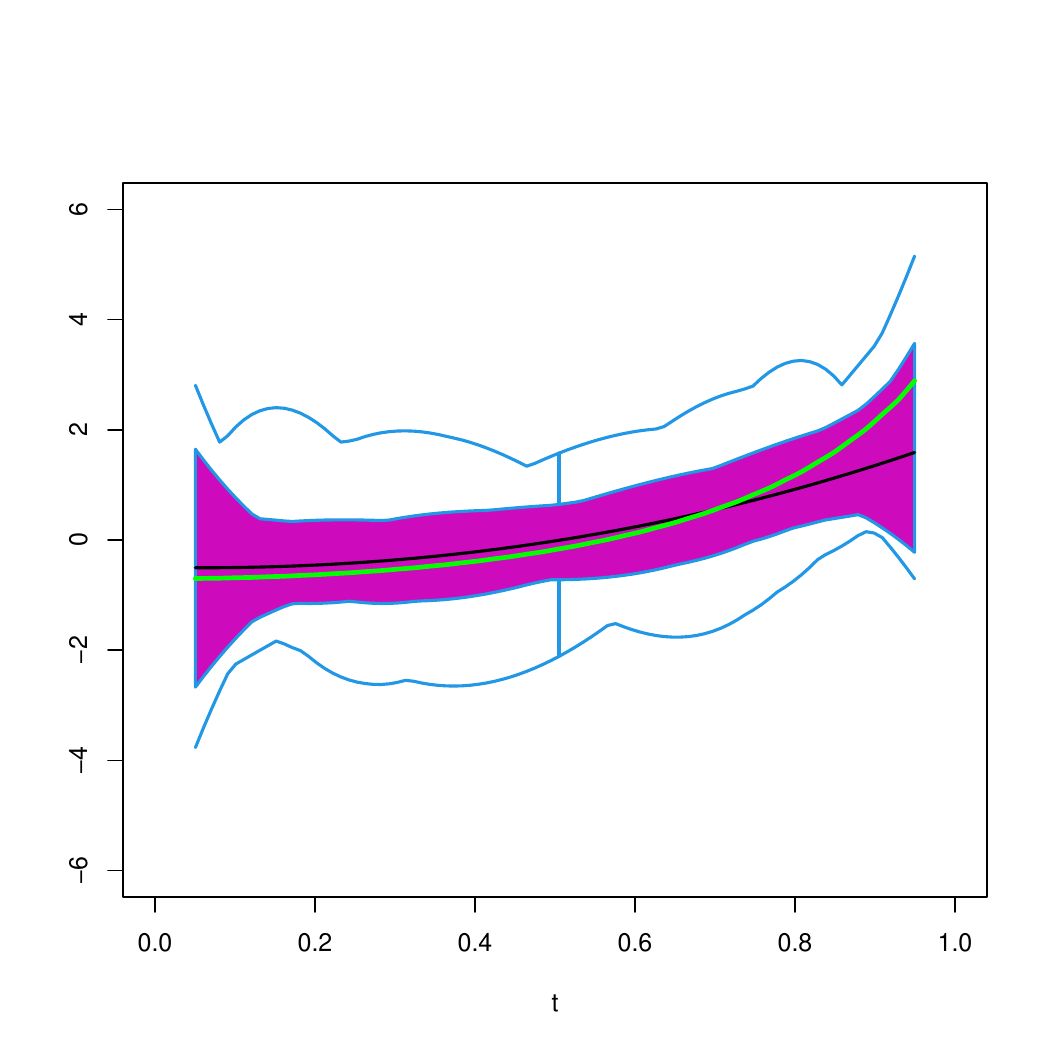}
  &  \includegraphics[scale=0.40]{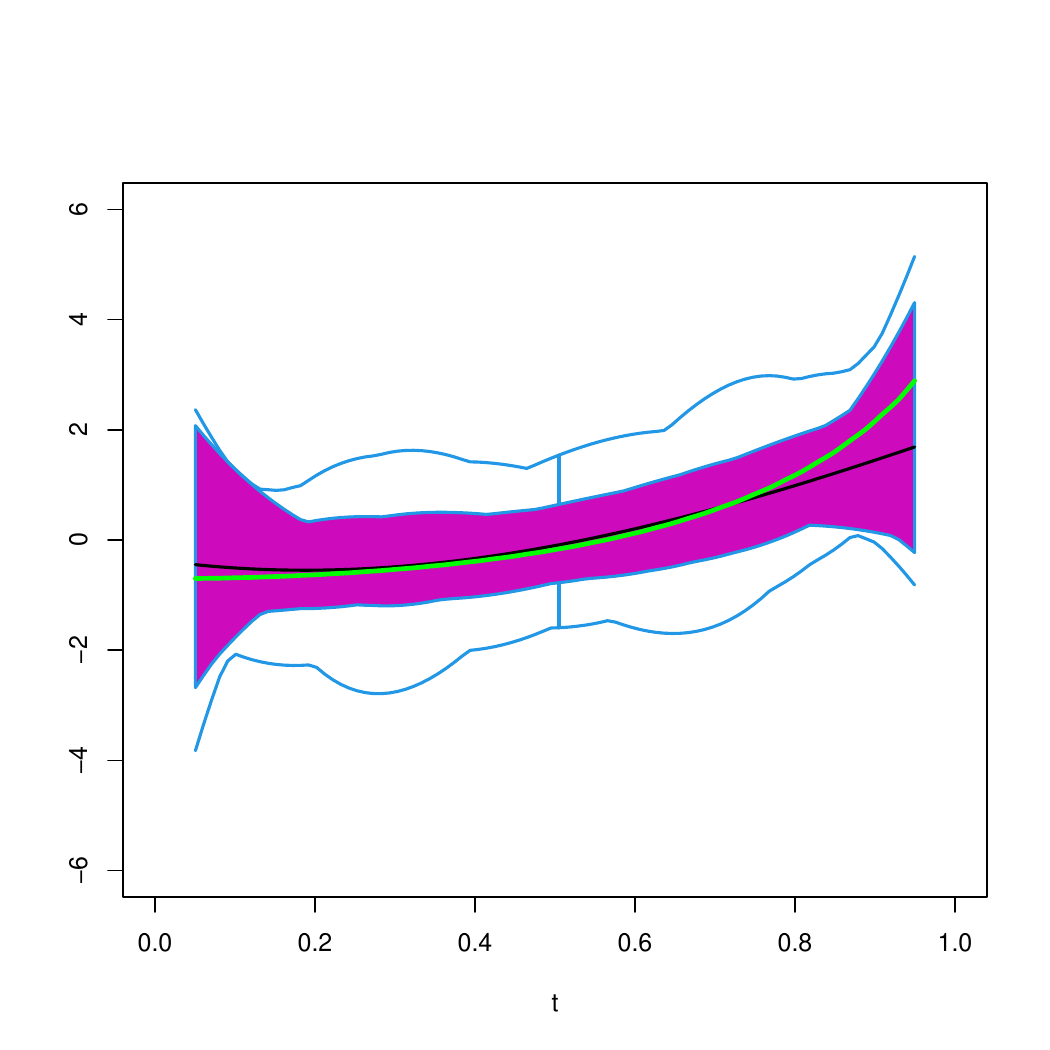}

\end{tabular}
\caption{\small \label{fig:wbeta-C15-poda5}  Functional boxplot of the estimators for $\beta_0$ under $C_{1,0.05}$  within the interval $[0.05,0.95]$. 
The true function is shown with a green dashed line, while the black solid one is the central 
curve of the $n_R = 1000$ estimates $\wbeta$.  }
\end{center} 
\end{figure}

\begin{figure}[tp]
 \begin{center}
 \footnotesize
 \renewcommand{\arraystretch}{0.2}
 \newcolumntype{M}{>{\centering\arraybackslash}m{\dimexpr.01\linewidth-1\tabcolsep}}
   \newcolumntype{G}{>{\centering\arraybackslash}m{\dimexpr.45\linewidth-1\tabcolsep}}
%\begin{tabular}{MGG}
\begin{tabular}{GG}
  $\wbeta_{\clas}$ & $\wbeta_{\eme}$   \\[-3ex]   
\includegraphics[scale=0.40]{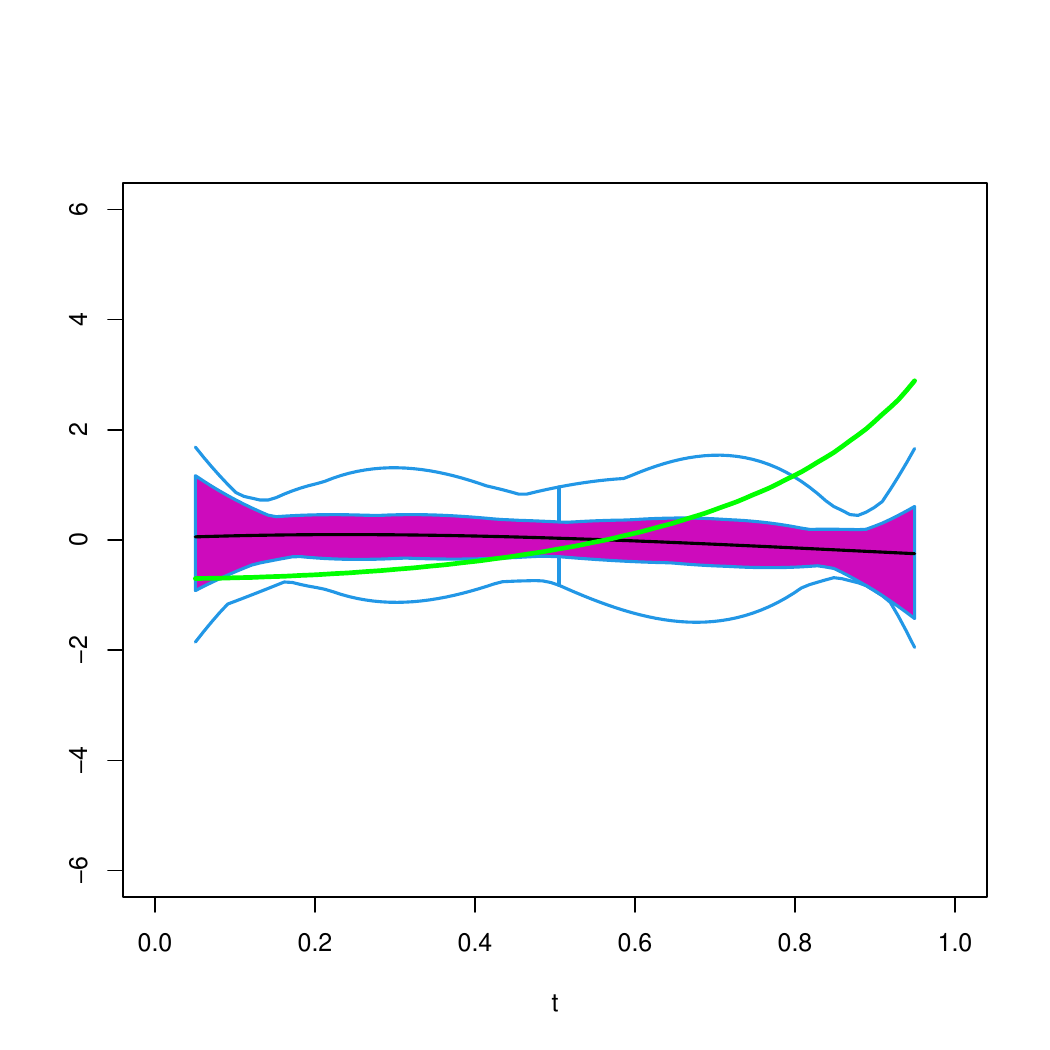}
&   \includegraphics[scale=0.40]{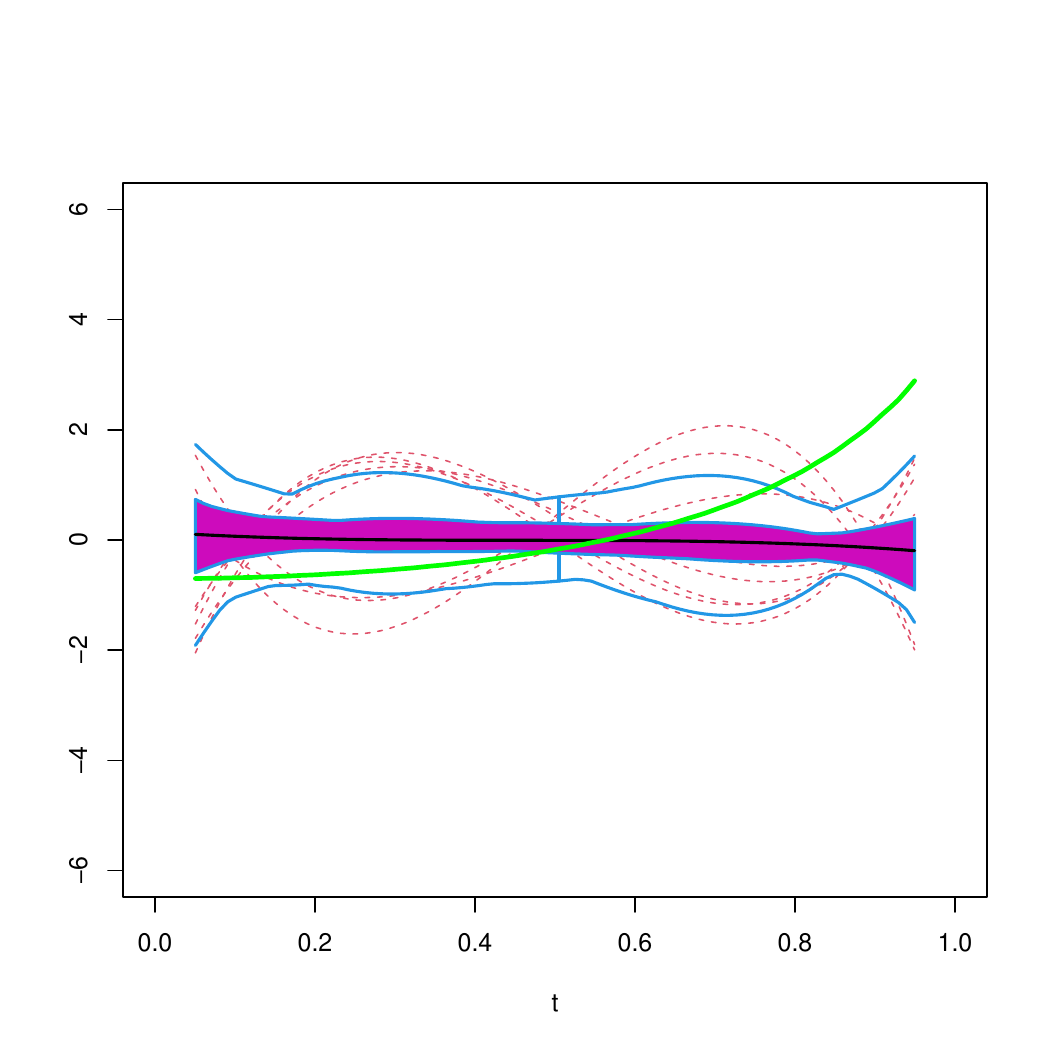}\\
   $\wbeta_{\wclHR}$ & $\wbeta_{\wemeHR}$ \\[-3ex] 
    \includegraphics[scale=0.40]{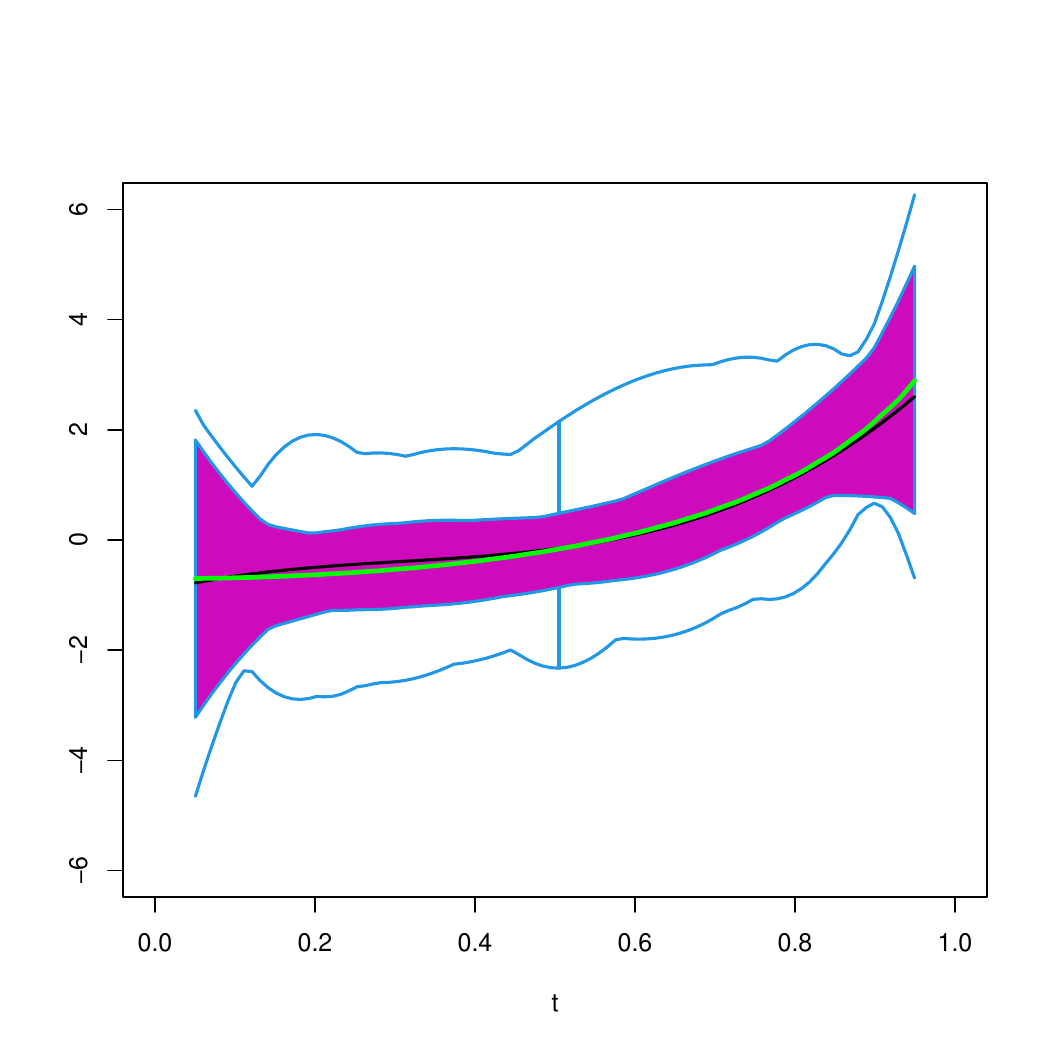}
  &  \includegraphics[scale=0.40]{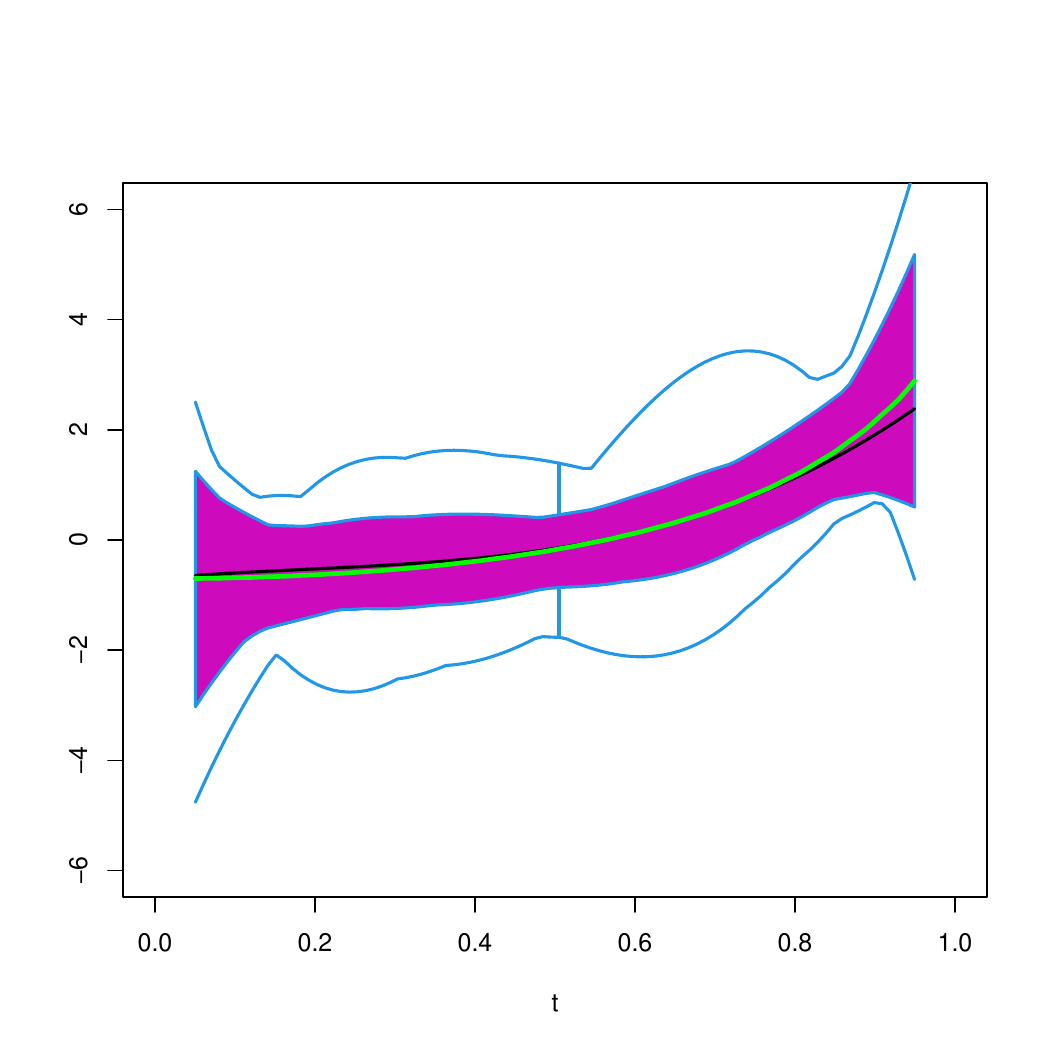}
   \\
   $\wbeta_{\wclBOX}$ & $\wbeta_{\wemeBOX}$ \\[-3ex]
    \includegraphics[scale=0.40]{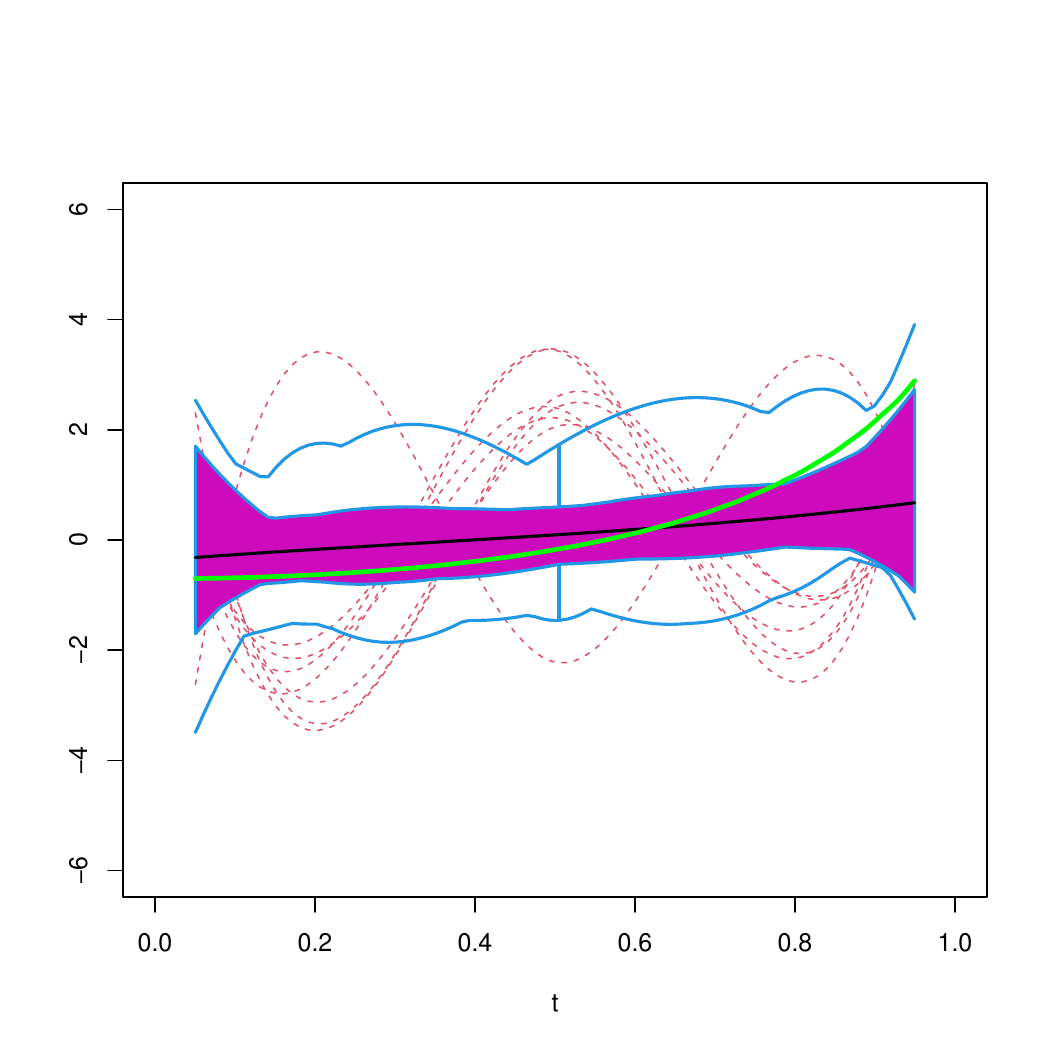}
  &  \includegraphics[scale=0.40]{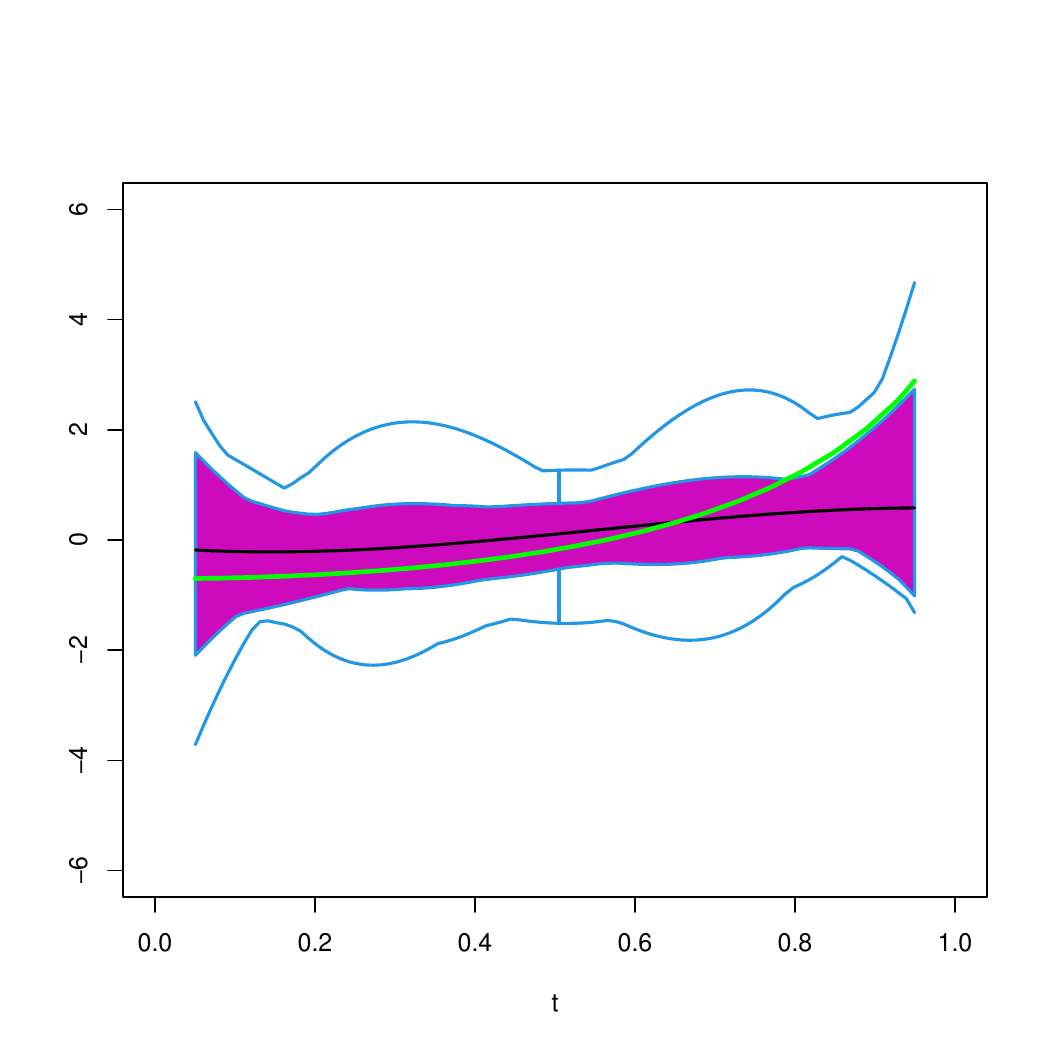}

\end{tabular}
\caption{\small \label{fig:wbeta-C110-poda5}  Functional boxplot of the estimators for $\beta_0$ under $C_{1,0.10}$  within the interval $[0.05,0.95]$. 
The true function is shown with a green dashed line, while the black solid one is the central 
curve of the $n_R = 1000$ estimates $\wbeta$. }
\end{center} 
\end{figure}  

\begin{figure}[tp]
 \begin{center}
 \footnotesize
 \renewcommand{\arraystretch}{0.2}
 \newcolumntype{M}{>{\centering\arraybackslash}m{\dimexpr.01\linewidth-1\tabcolsep}}
   \newcolumntype{G}{>{\centering\arraybackslash}m{\dimexpr.45\linewidth-1\tabcolsep}}
%\begin{tabular}{MGG}
\begin{tabular}{GG}
  $\wbeta_{\clas}$ & $\wbeta_{\eme}$   \\[-3ex]      
 \includegraphics[scale=0.40]{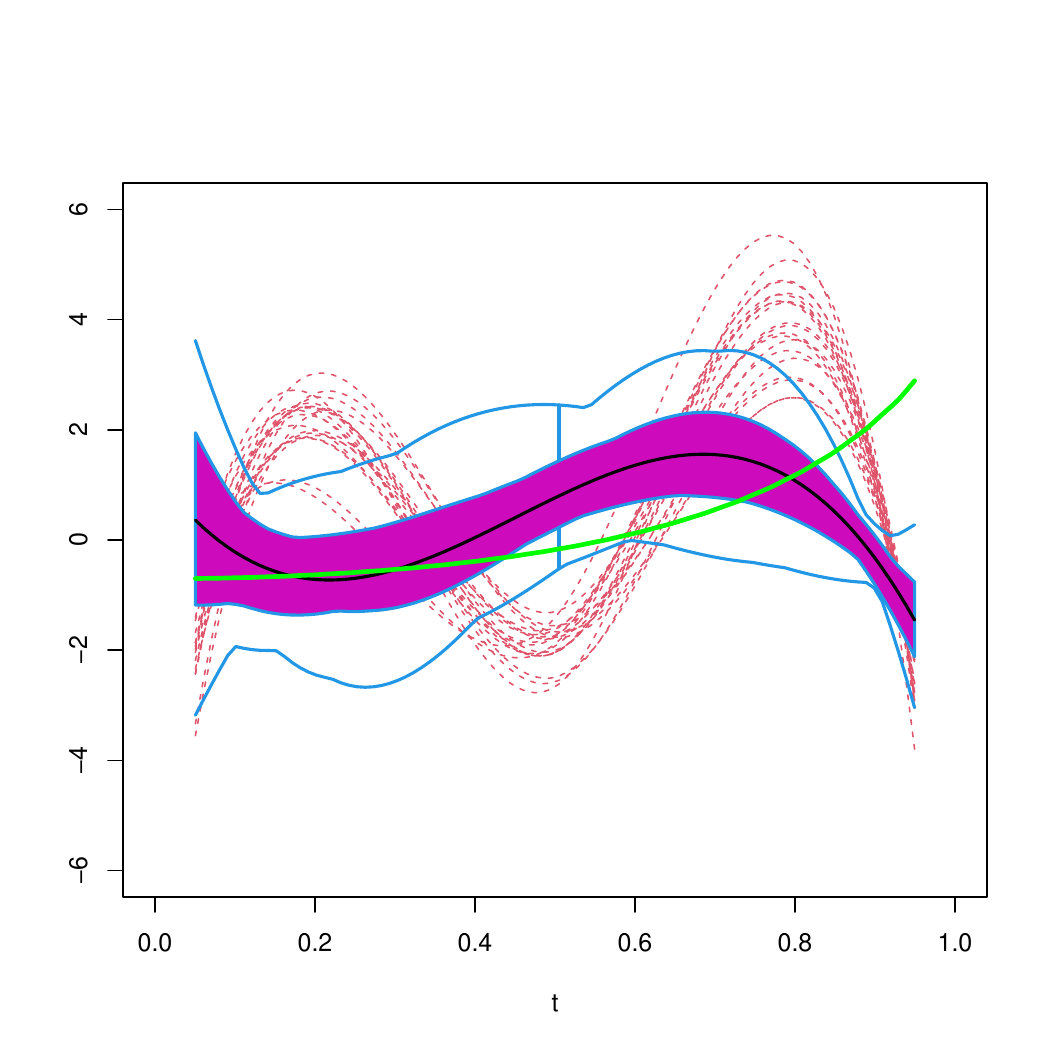}
 &  \includegraphics[scale=0.40]{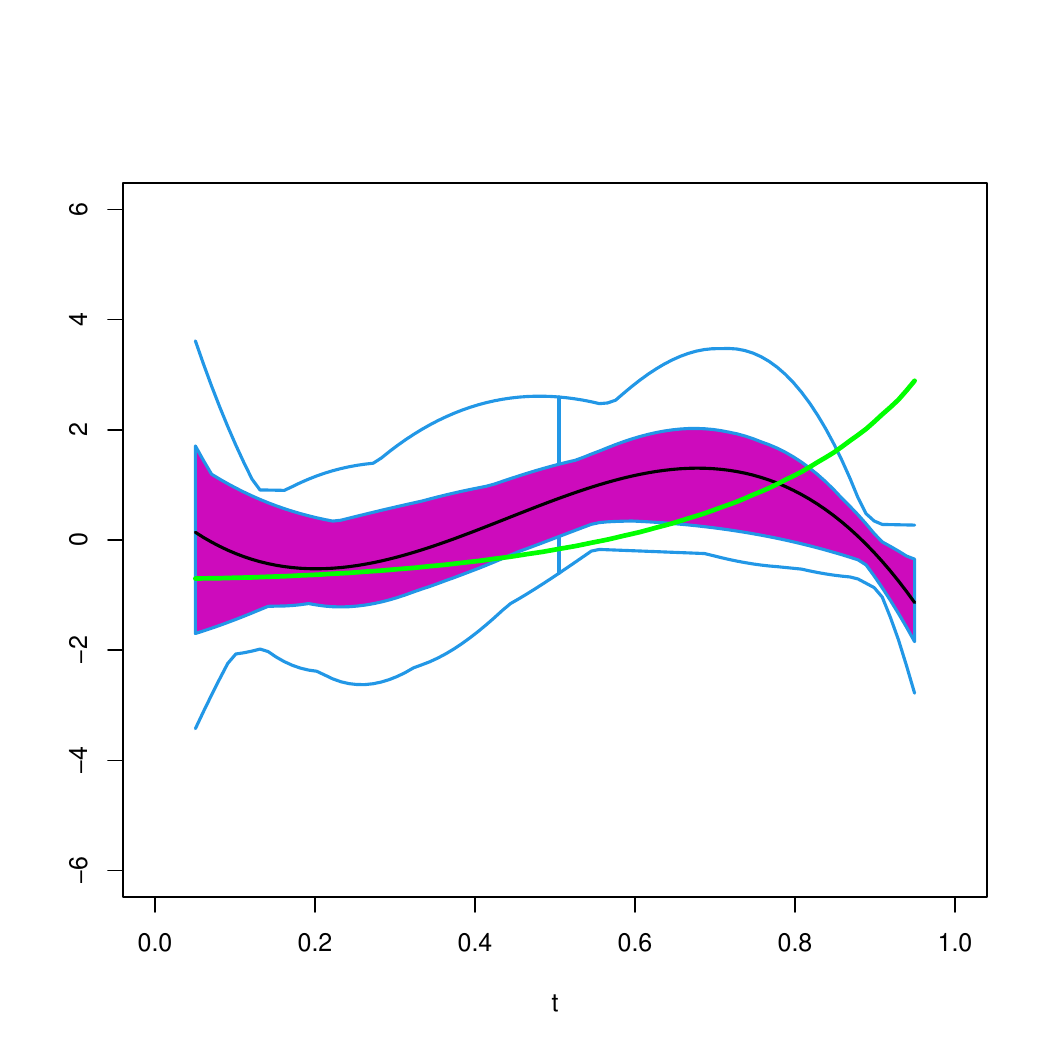}\\
   $\wbeta_{\wclHR}$ & $\wbeta_{\wemeHR}$ \\[-3ex] 
    \includegraphics[scale=0.40]{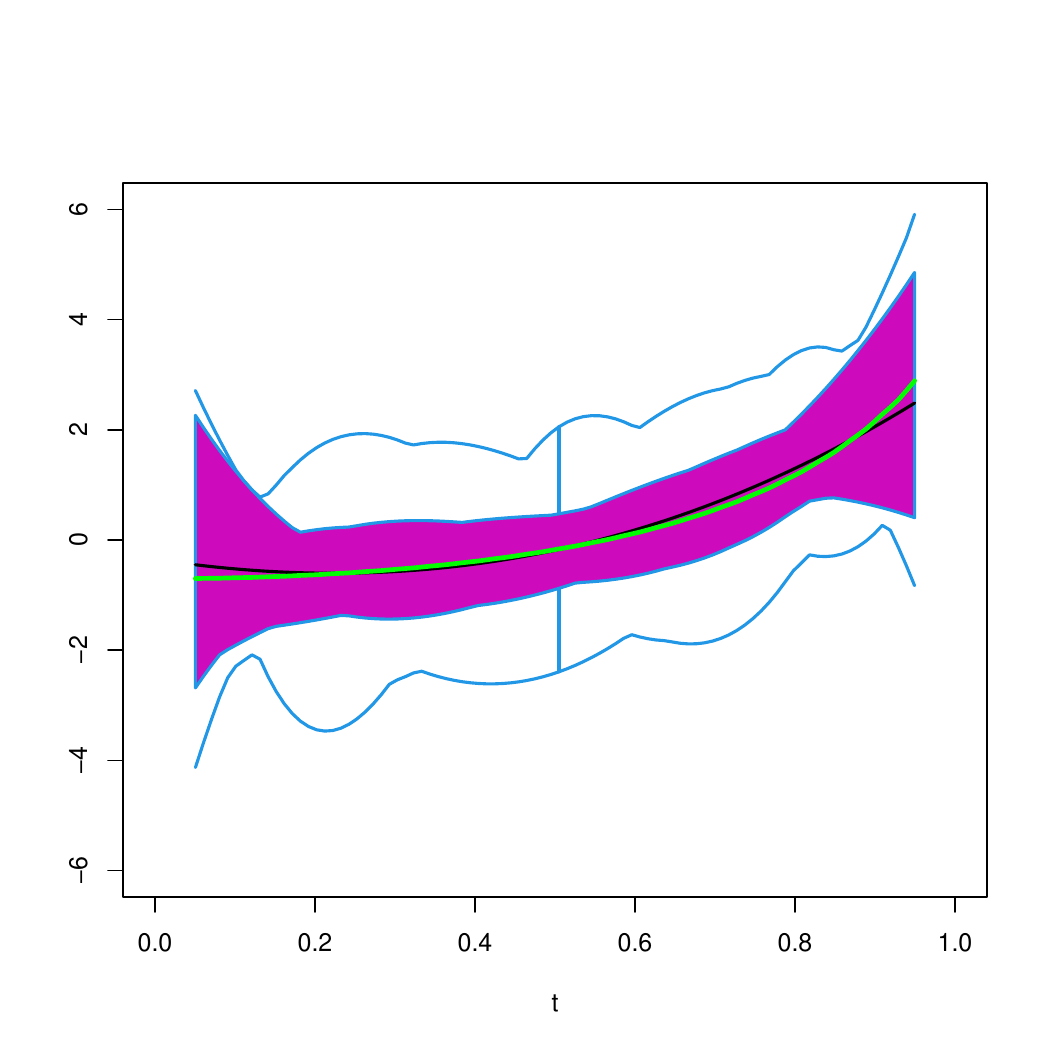}
  &  \includegraphics[scale=0.40]{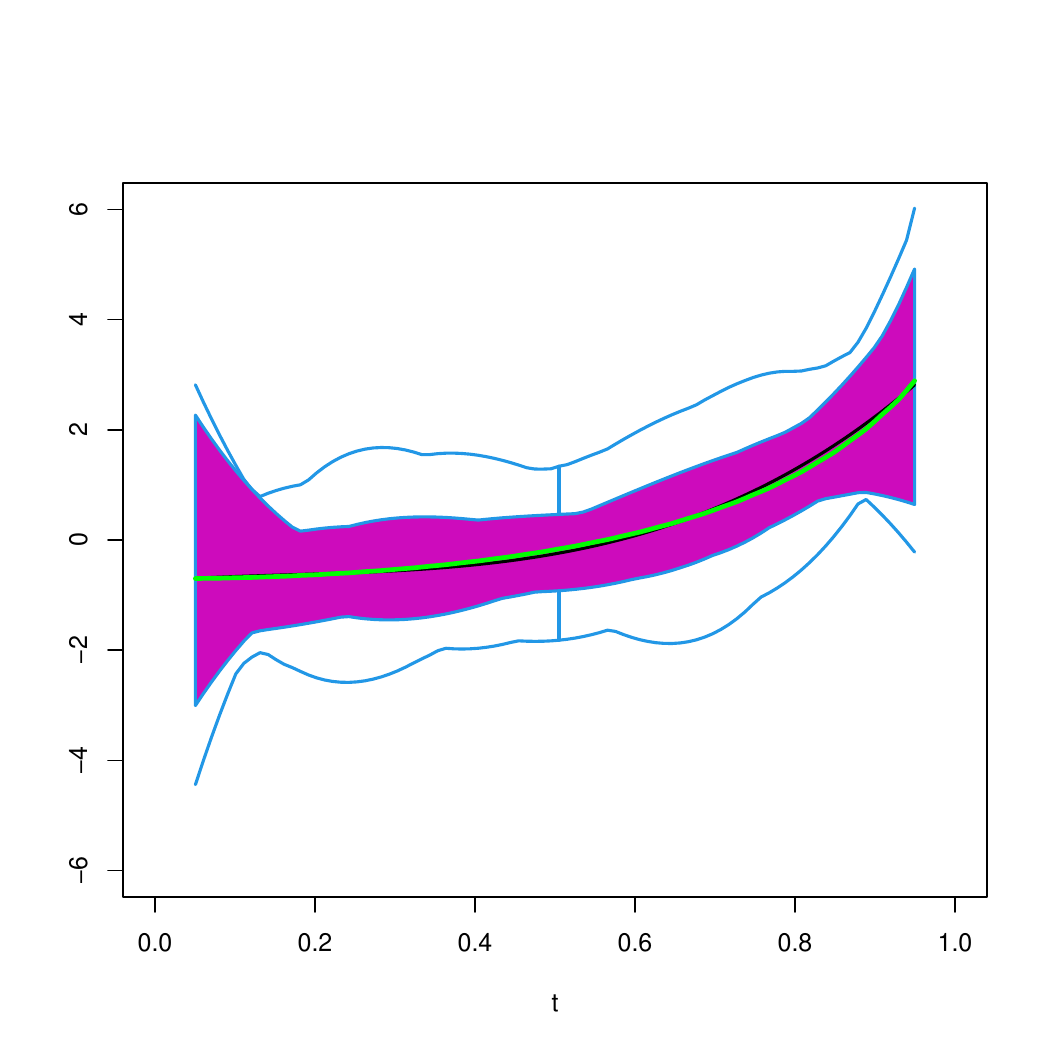}
   \\
   $\wbeta_{\wclBOX}$ & $\wbeta_{\wemeBOX}$ \\[-3ex]
 \includegraphics[scale=0.40]{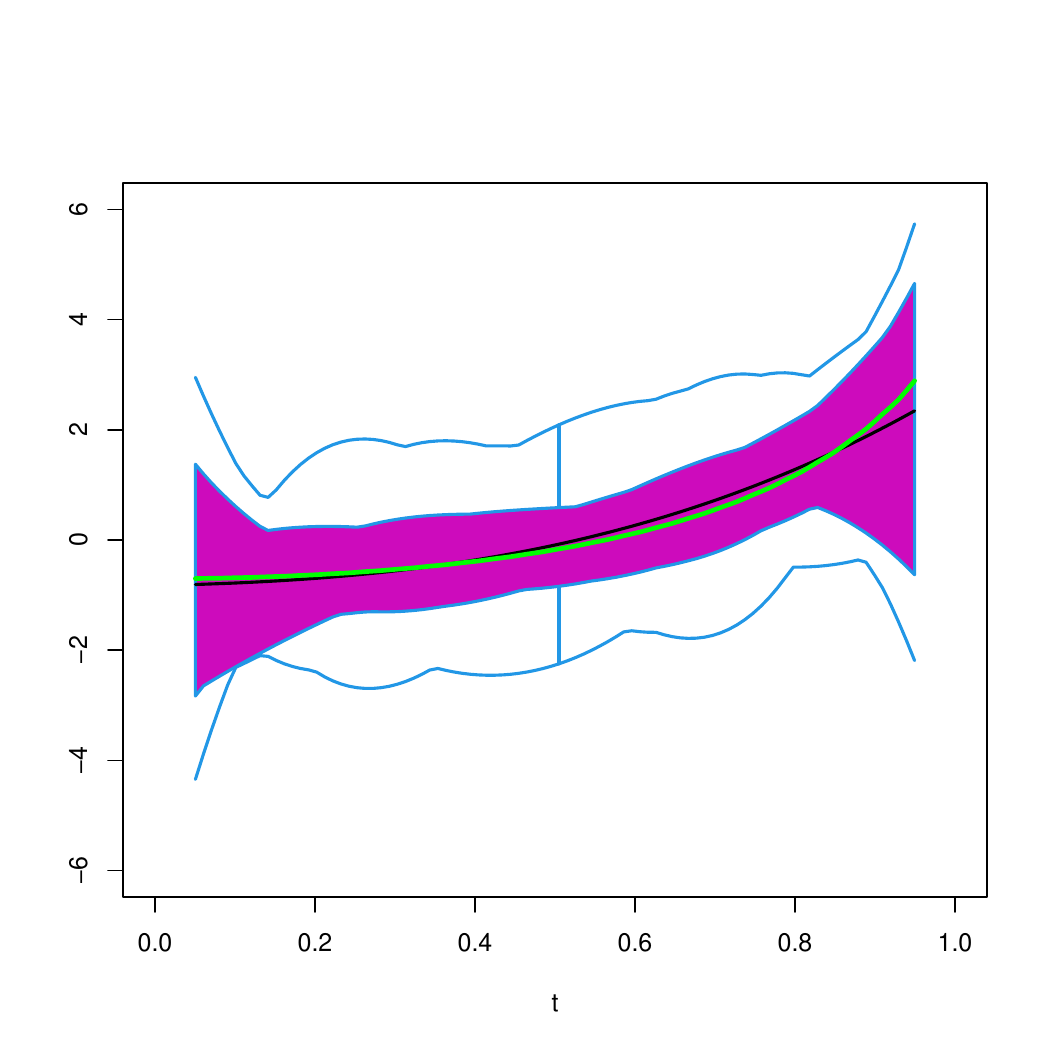}
  &  \includegraphics[scale=0.40]{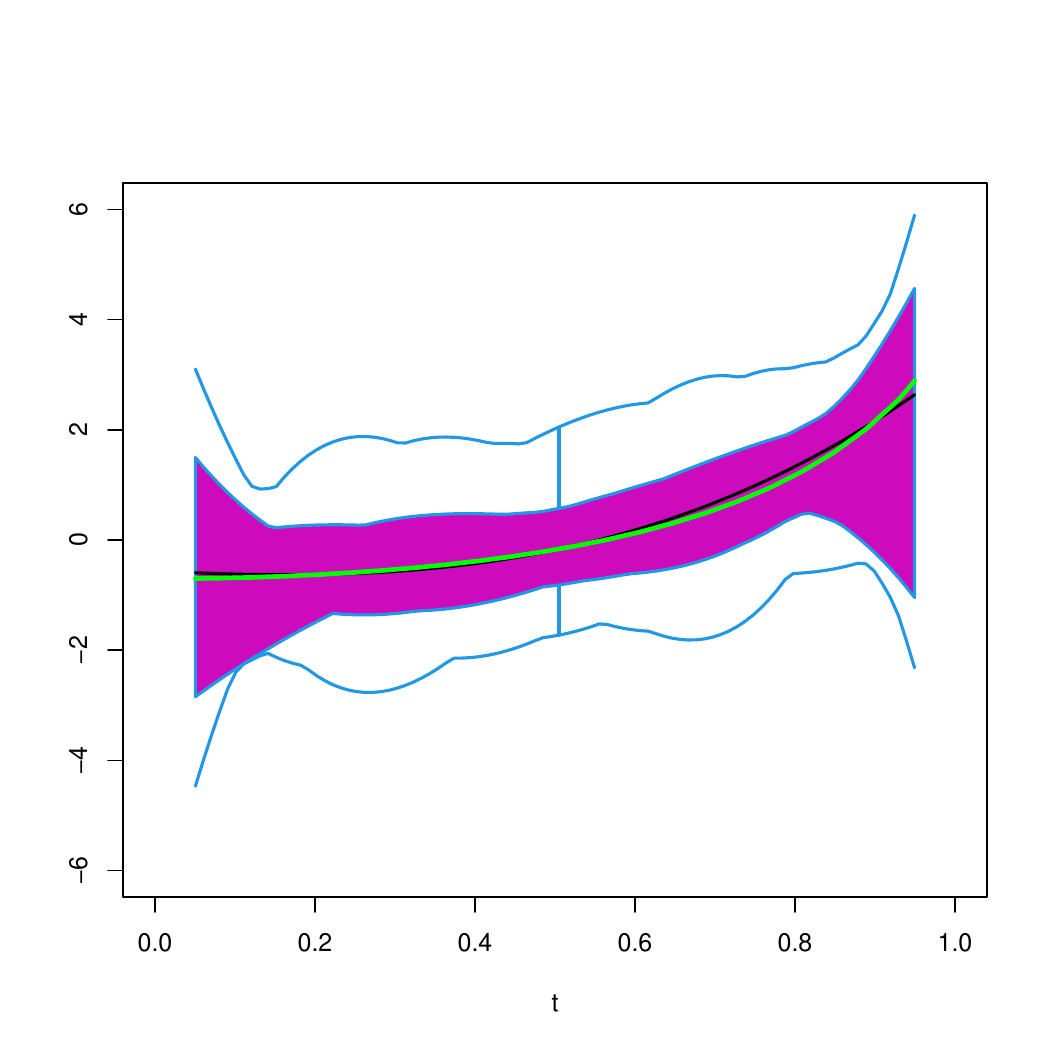}
  
\end{tabular}
\caption{\small \label{fig:wbeta-C25-poda5}  Functional boxplot of the estimators for $\beta_0$ under $C_{2,0.05}$  within the interval $[0.05,0.95]$. 
The true function is shown with a green dashed line, while the black solid one is the central 
curve of the $n_R = 1000$ estimates $\wbeta$.  }
\end{center} 
\end{figure}

\begin{figure}[tp]
 \begin{center}
 \footnotesize
 \renewcommand{\arraystretch}{0.2}
 \newcolumntype{M}{>{\centering\arraybackslash}m{\dimexpr.01\linewidth-1\tabcolsep}}
   \newcolumntype{G}{>{\centering\arraybackslash}m{\dimexpr.45\linewidth-1\tabcolsep}}
%\begin{tabular}{MGG}
\begin{tabular}{GG}
  $\wbeta_{\clas}$ & $\wbeta_{\eme}$   \\[-3ex]      
\includegraphics[scale=0.40]{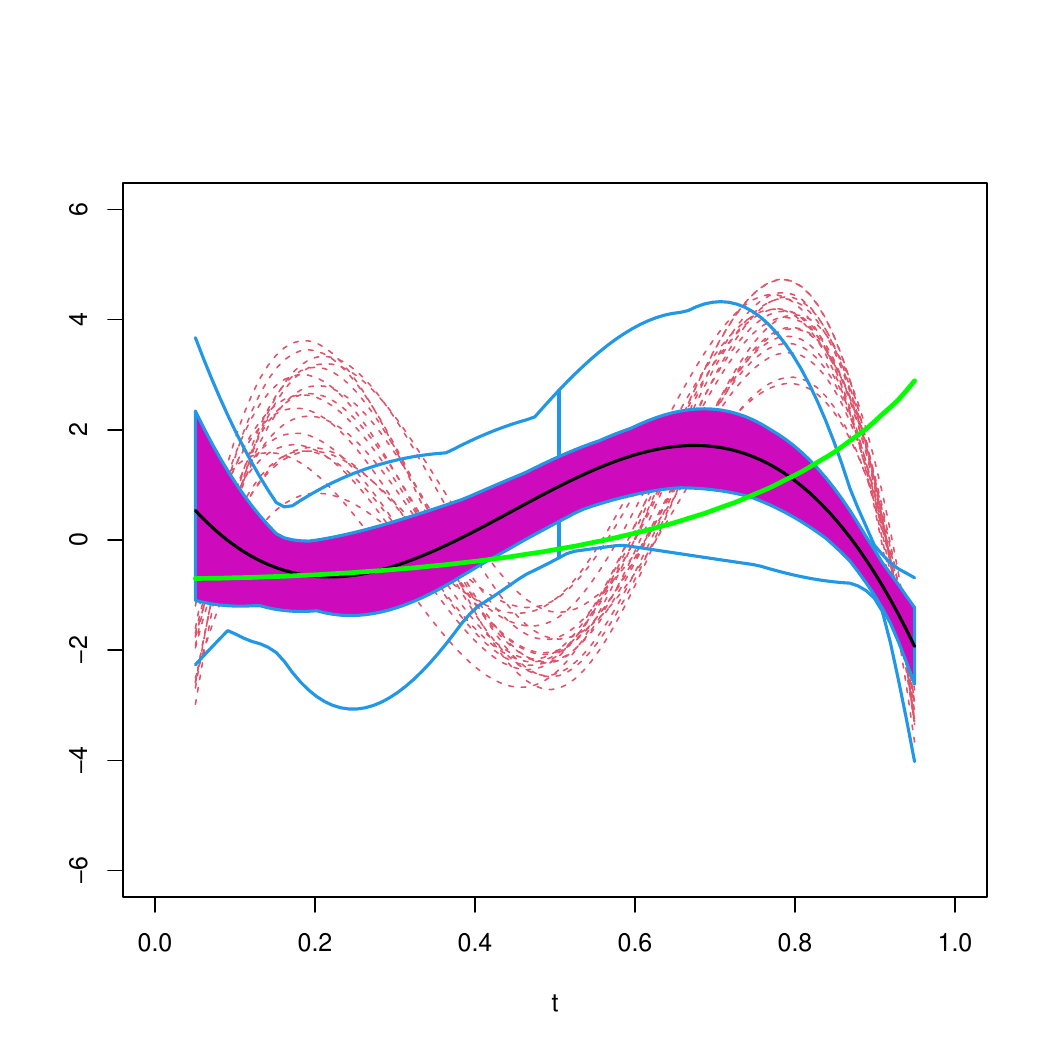}
 &  \includegraphics[scale=0.40]{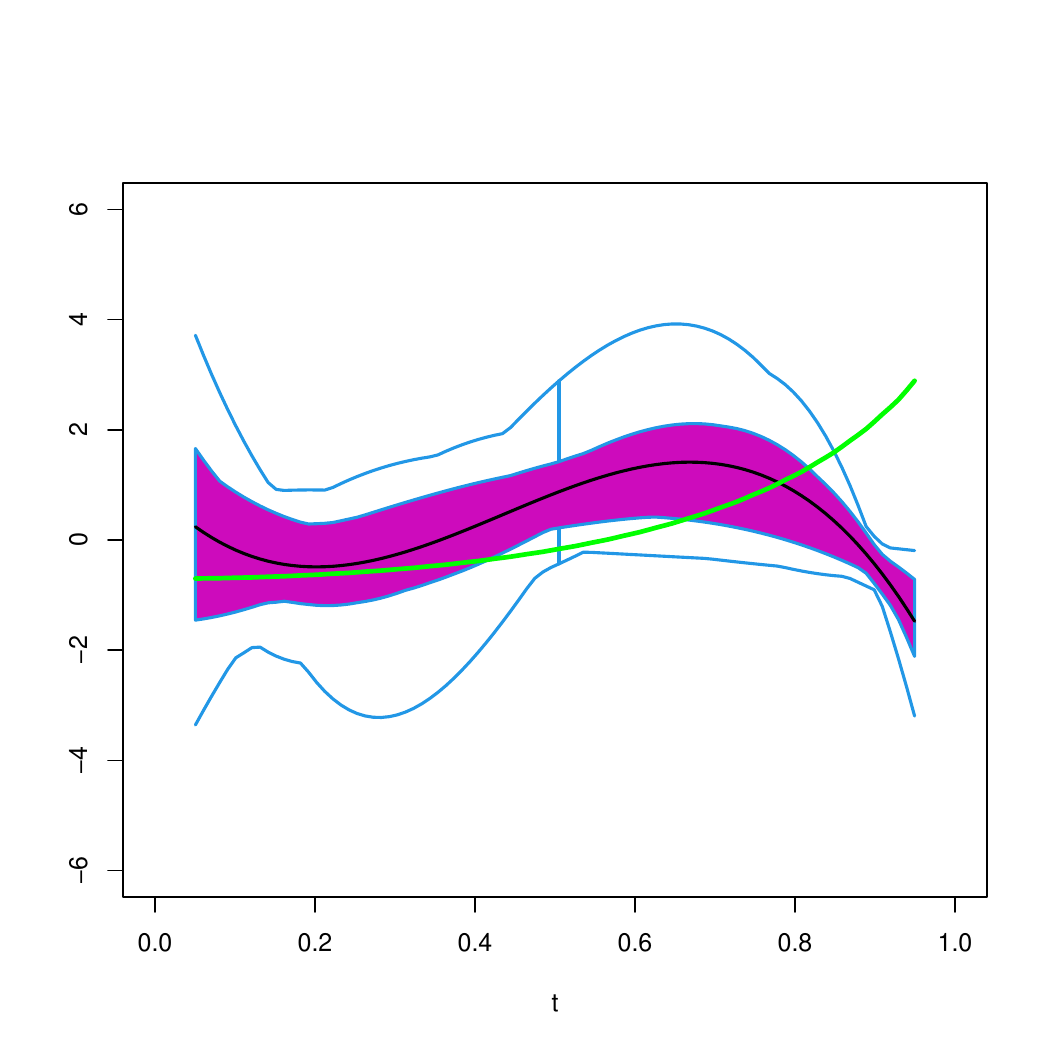}\\
   $\wbeta_{\wclHR}$ & $\wbeta_{\wemeHR}$ \\[-3ex] 
     \includegraphics[scale=0.40]{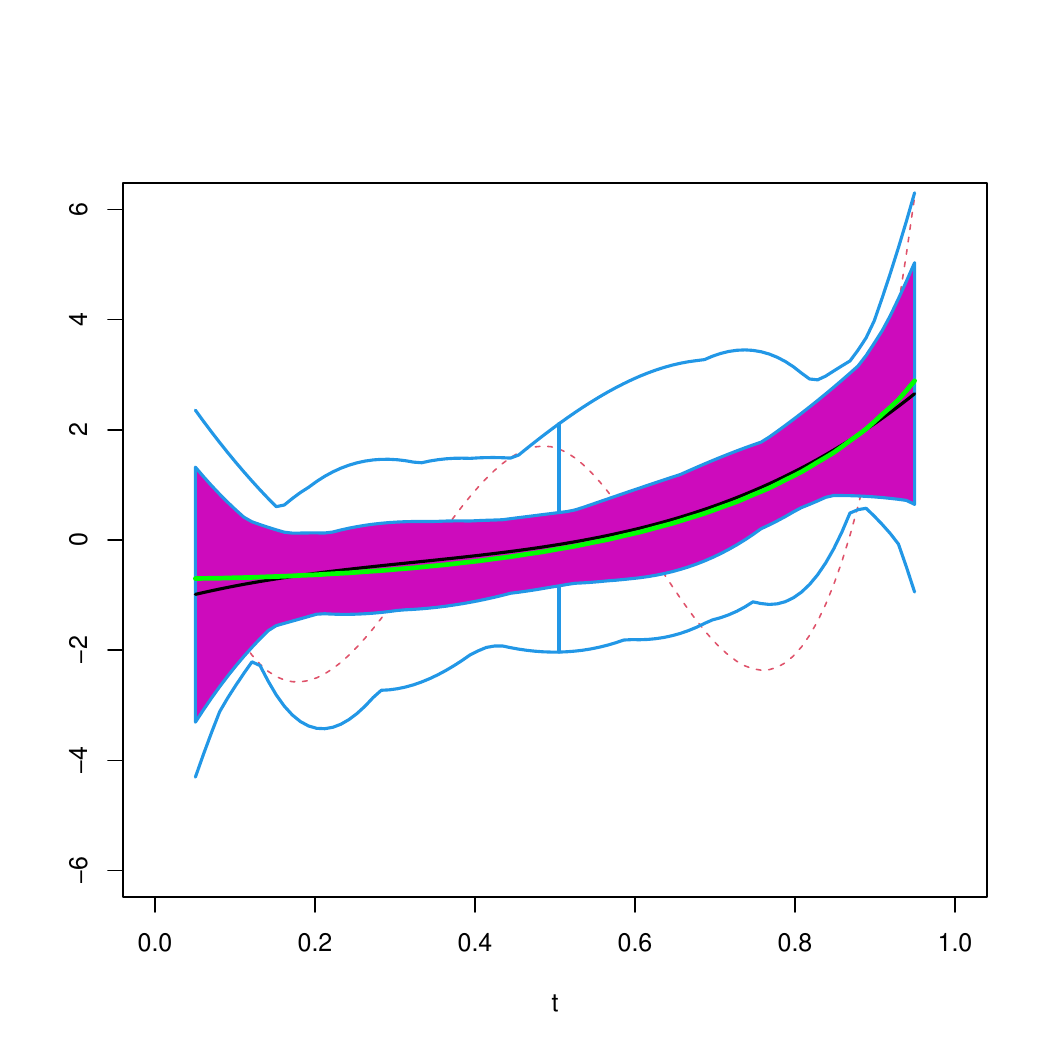}
  &  \includegraphics[scale=0.40]{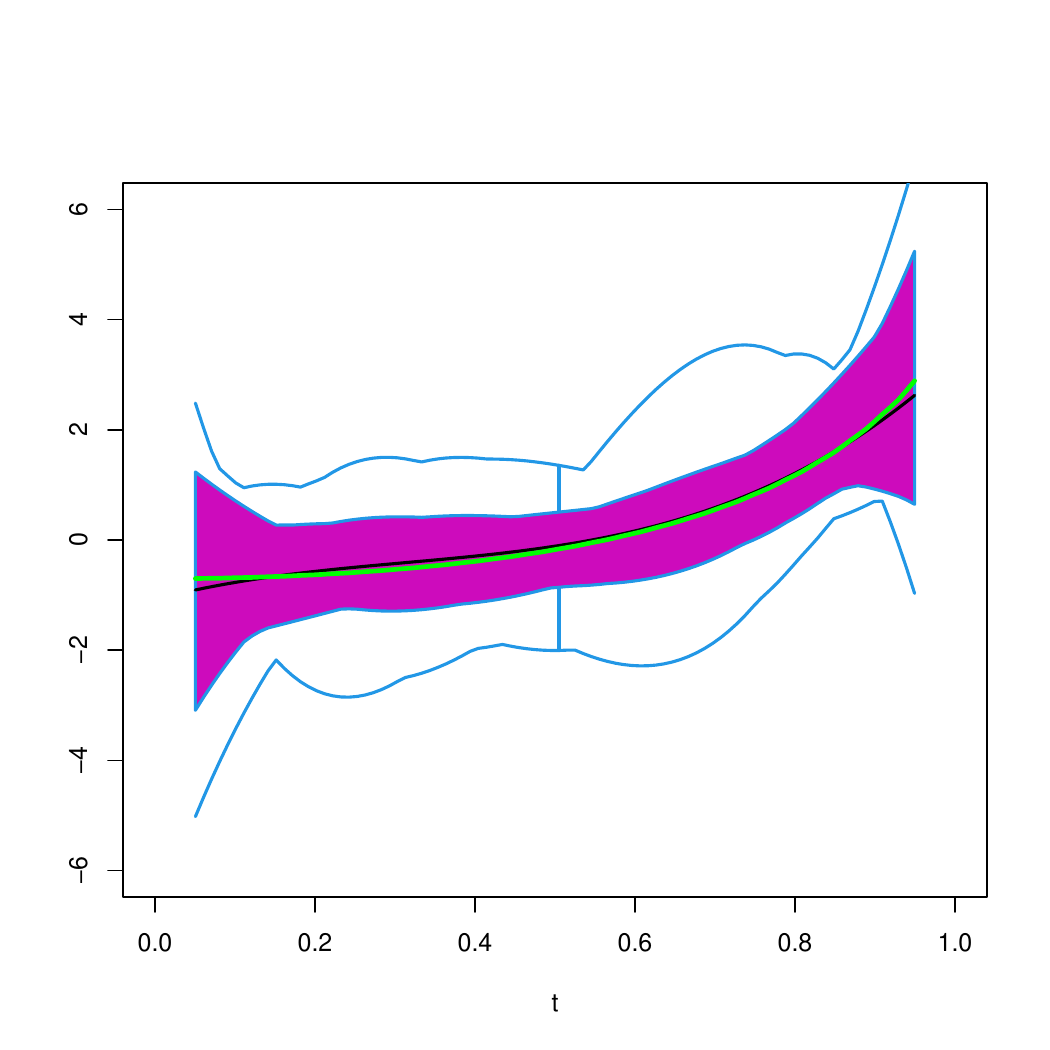}
     \\
   $\wbeta_{\wclBOX}$ & $\wbeta_{\wemeBOX}$ \\[-3ex]
    \includegraphics[scale=0.40]{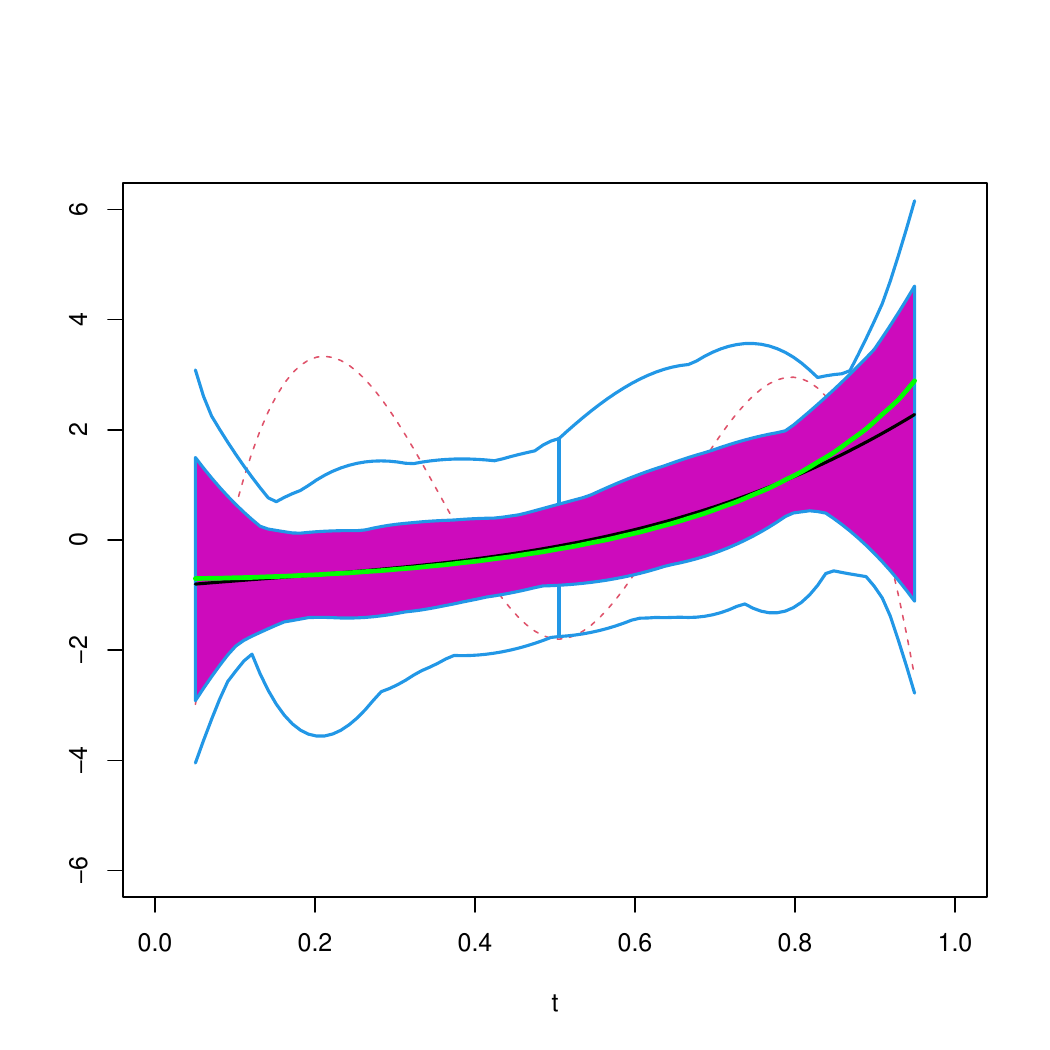}
  &  \includegraphics[scale=0.40]{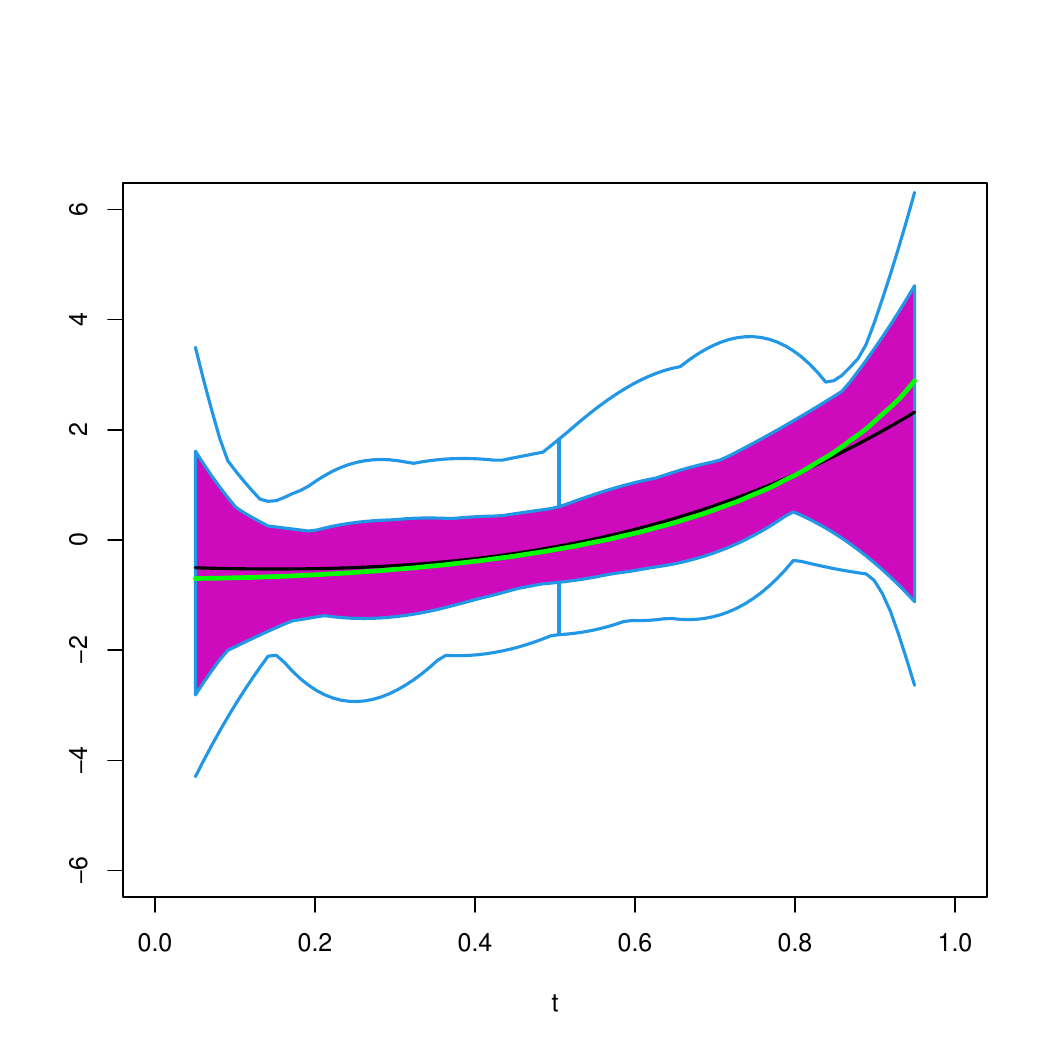}
 
\end{tabular}
\caption{\small \label{fig:wbeta-C210-poda5}  Functional boxplot of the estimators for $\beta_0$ under $C_{2,0.10}$  within the interval $[0.05,0.95]$. 
The true function is shown with a green dashed line, while the black solid one is the central 
curve of the $n_R = 1000$ estimates $\wbeta$.  }
\end{center} 
\end{figure}

\begin{figure}[tp]
 \begin{center}
 \footnotesize
 \renewcommand{\arraystretch}{0.2}
 \newcolumntype{M}{>{\centering\arraybackslash}m{\dimexpr.01\linewidth-1\tabcolsep}}
   \newcolumntype{G}{>{\centering\arraybackslash}m{\dimexpr.45\linewidth-1\tabcolsep}}
%\begin{tabular}{MGG}
\begin{tabular}{GG}
  $\wbeta_{\clas}$ & $\wbeta_{\eme}$   \\[-3ex] 
\includegraphics[scale=0.40]{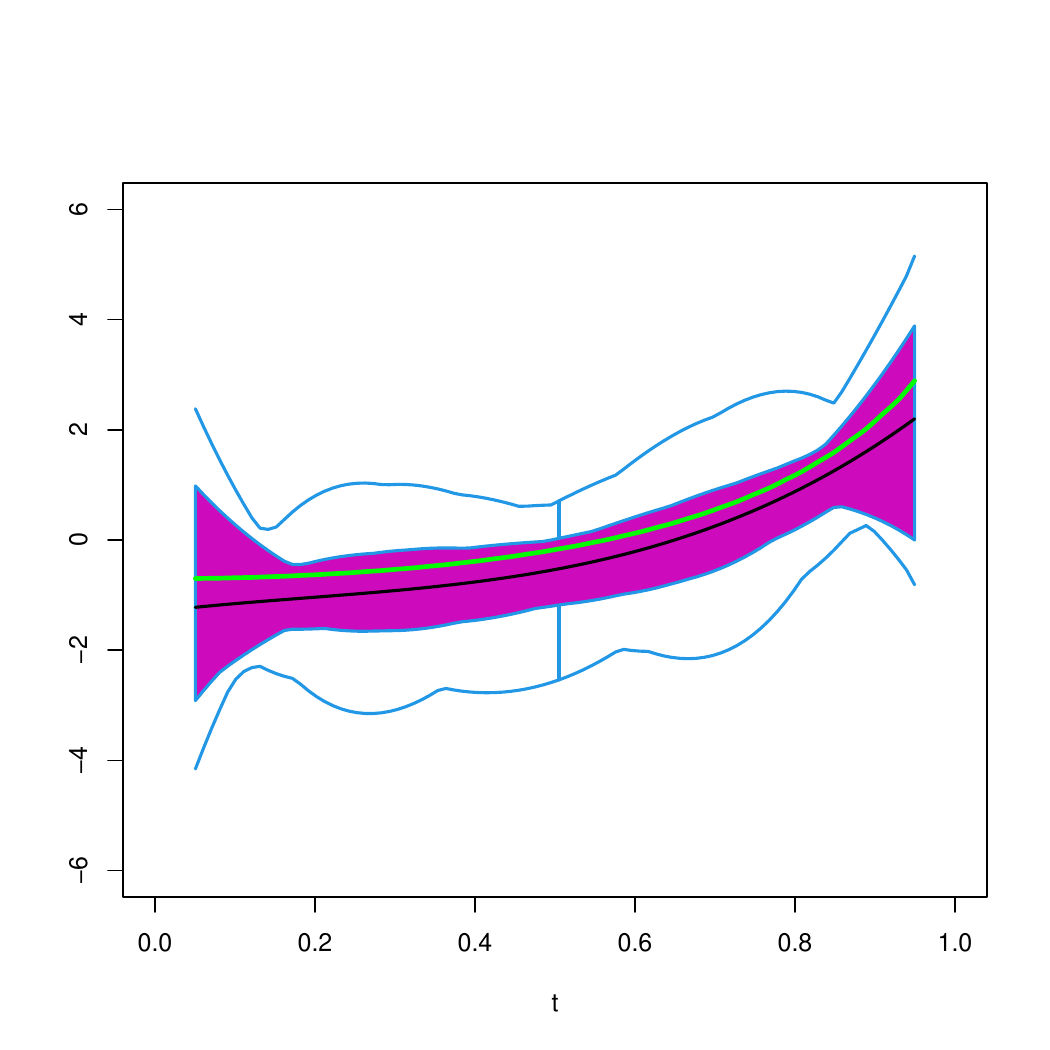}
 &  \includegraphics[scale=0.40]{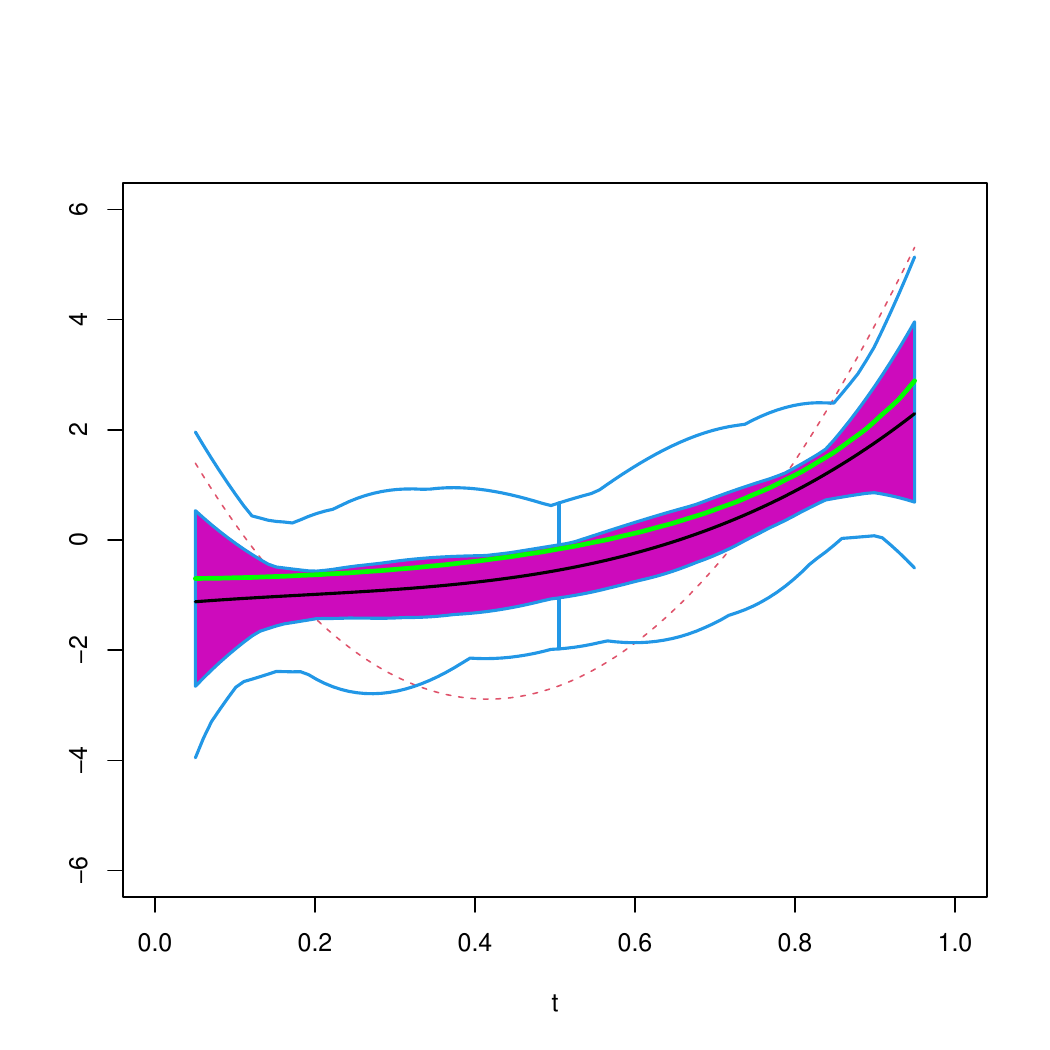}\\
   $\wbeta_{\wclHR}$ & $\wbeta_{\wemeHR}$ \\[-3ex] 
    \includegraphics[scale=0.40]{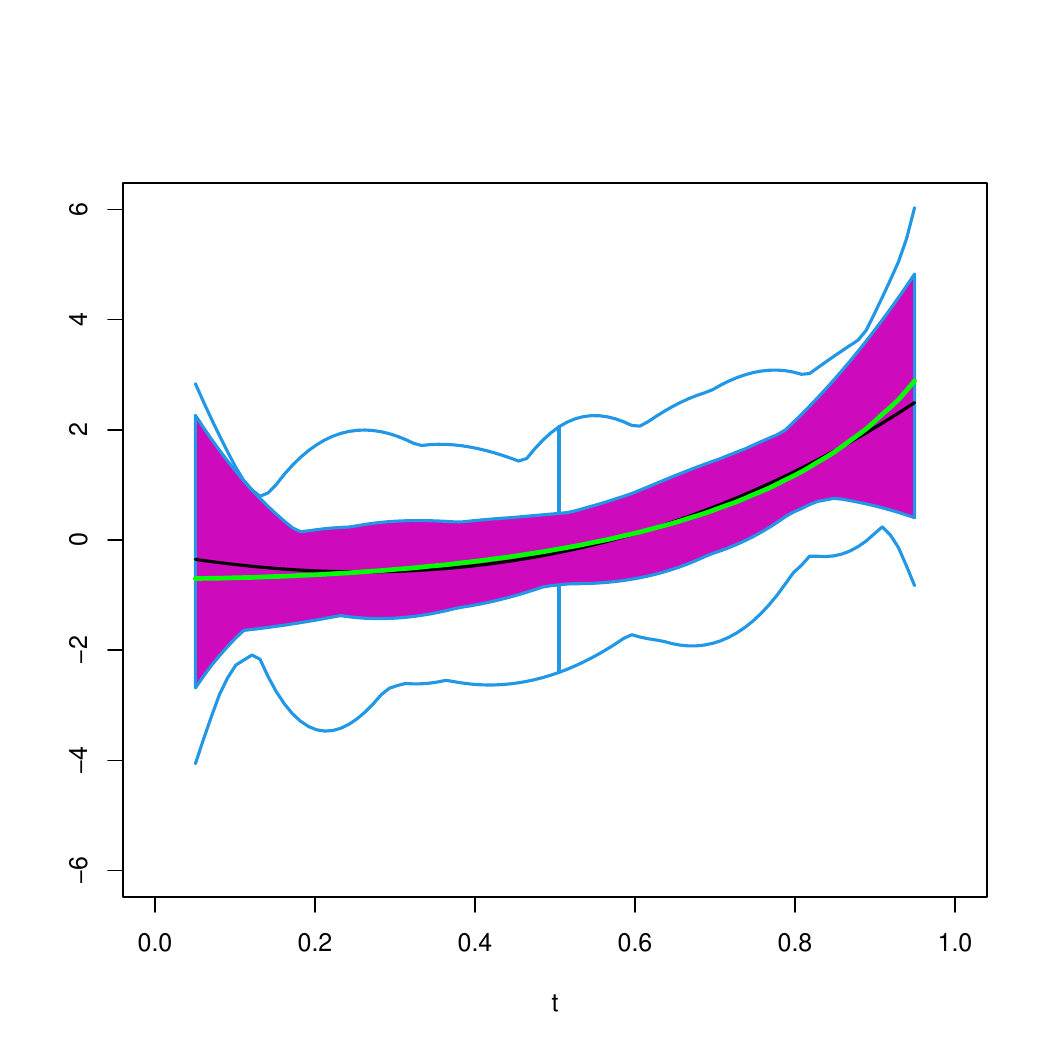}
  &  \includegraphics[scale=0.40]{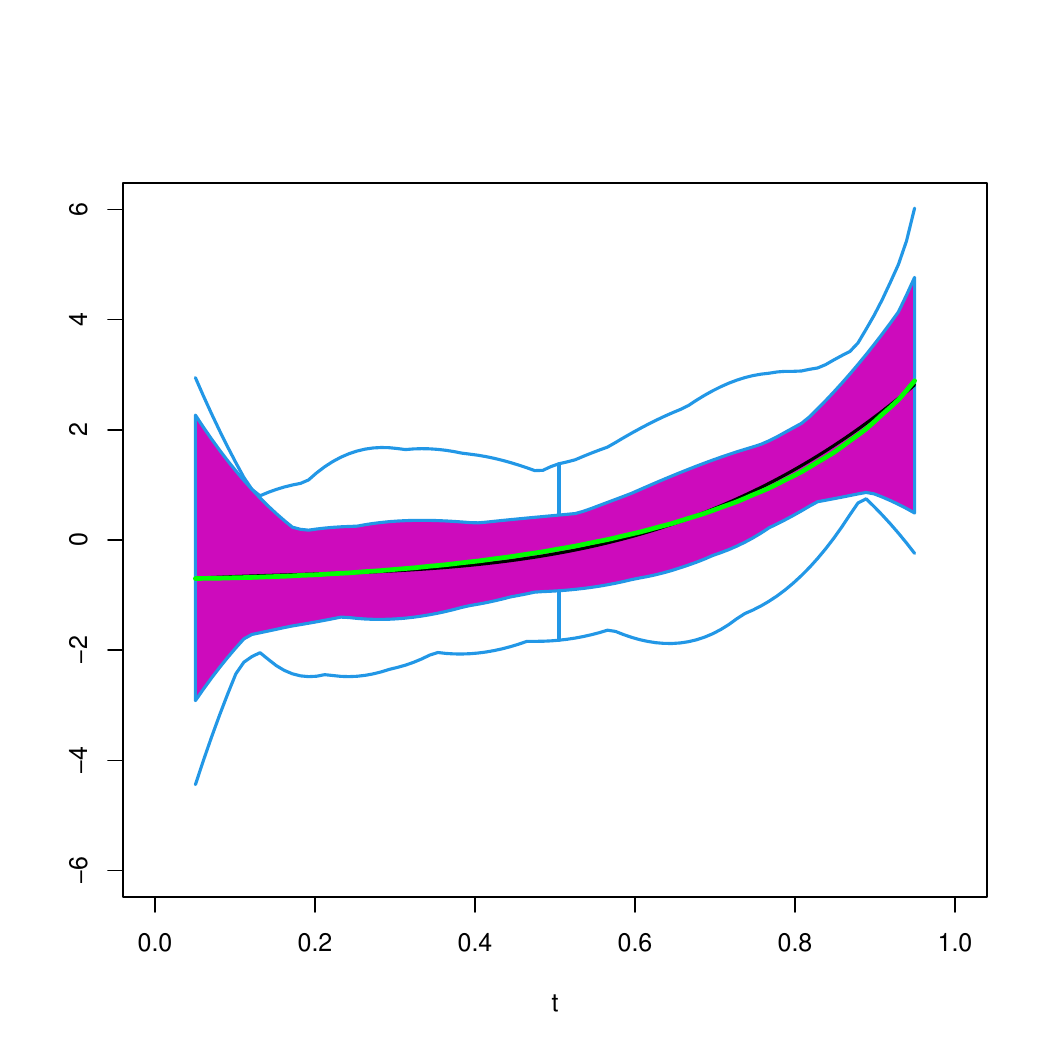}
   \\
   $\wbeta_{\wclBOX}$ & $\wbeta_{\wemeBOX}$ \\[-3ex]
  \includegraphics[scale=0.40]{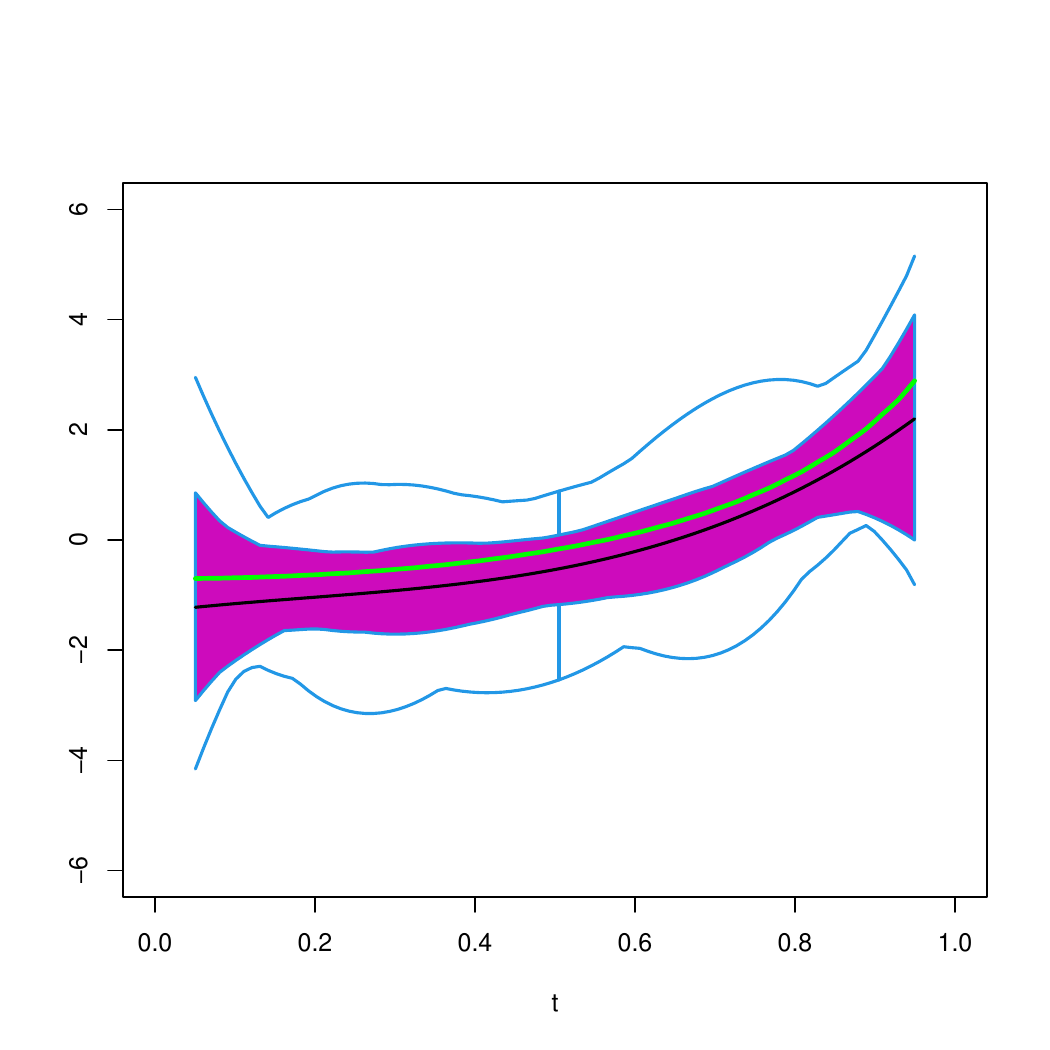}
  &  \includegraphics[scale=0.40]{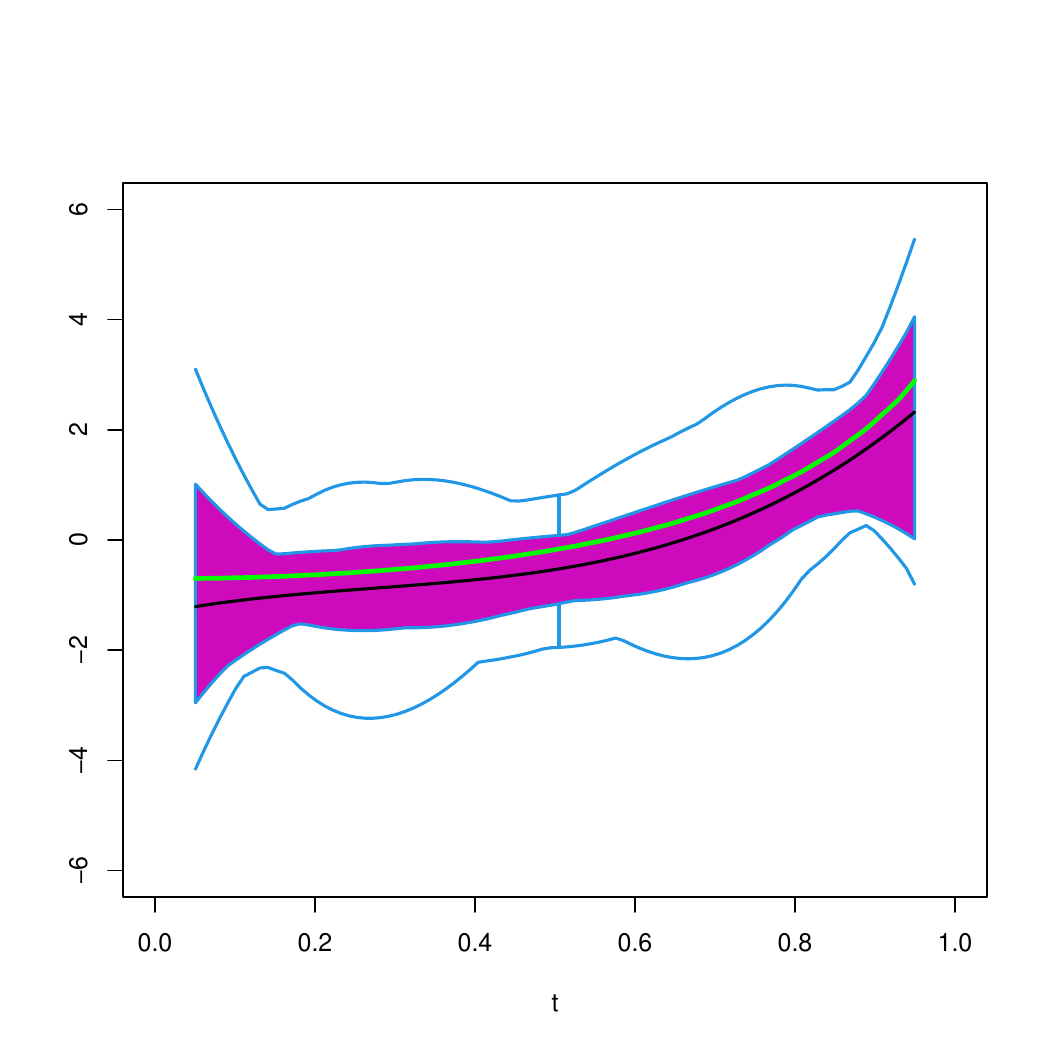}

\end{tabular}
\caption{\small \label{fig:wbeta-C35-poda5}  Functional boxplot of the estimators for $\beta_0$ under $C_{3,0.05}$  within the interval $[0.05,0.95]$. 
The true function is shown with a green dashed line, while the black solid one is the central 
curve of the $n_R = 1000$ estimates $\wbeta$.  }
\end{center} 
\end{figure}

\begin{figure}[tp]
 \begin{center}
 \footnotesize
 \renewcommand{\arraystretch}{0.2}
 \newcolumntype{M}{>{\centering\arraybackslash}m{\dimexpr.01\linewidth-1\tabcolsep}}
   \newcolumntype{G}{>{\centering\arraybackslash}m{\dimexpr.45\linewidth-1\tabcolsep}}
%\begin{tabular}{MGG}
\begin{tabular}{GG}
  $\wbeta_{\clas}$ & $\wbeta_{\eme}$   \\[-3ex] 
\includegraphics[scale=0.40]{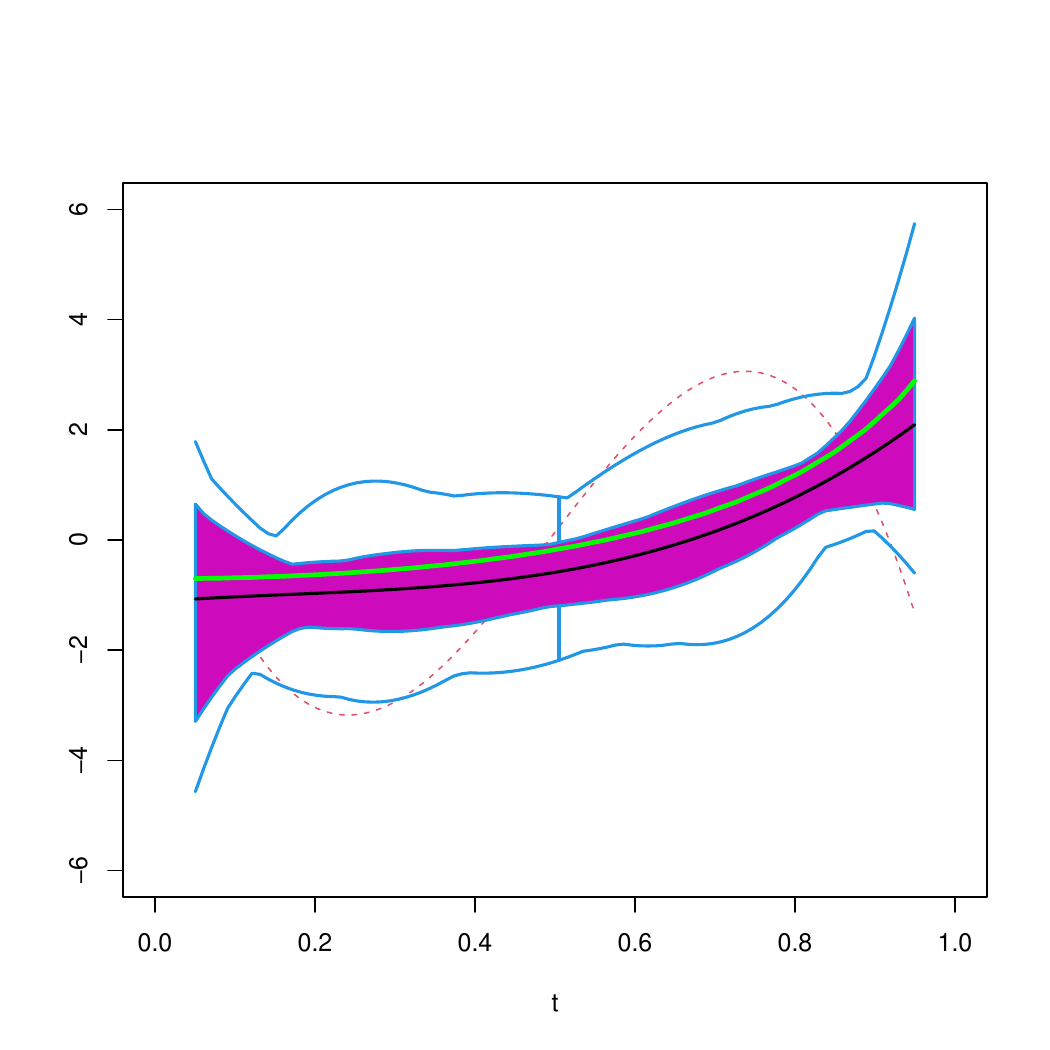}
 &  \includegraphics[scale=0.40]{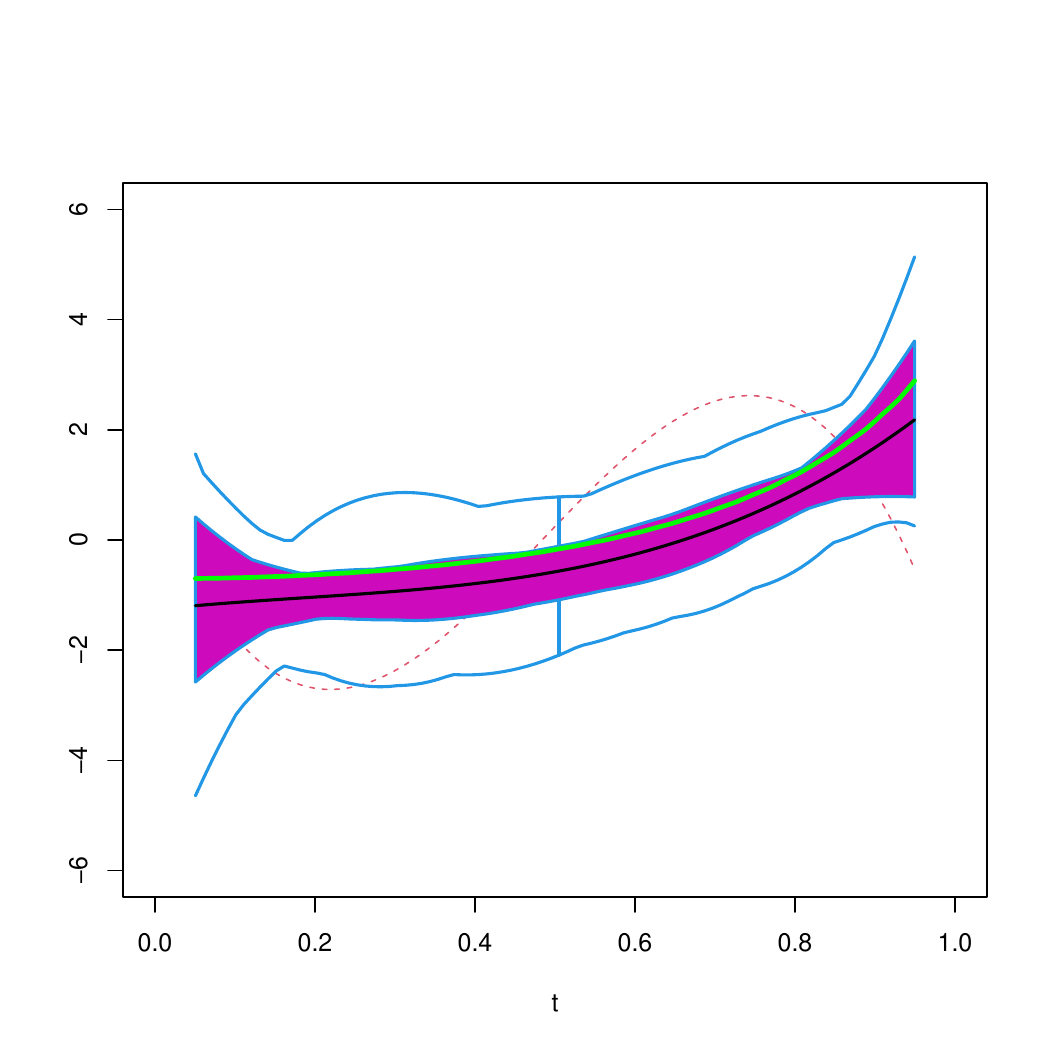}\\
   $\wbeta_{\wclHR}$ & $\wbeta_{\wemeHR}$ \\[-3ex] 
    \includegraphics[scale=0.40]{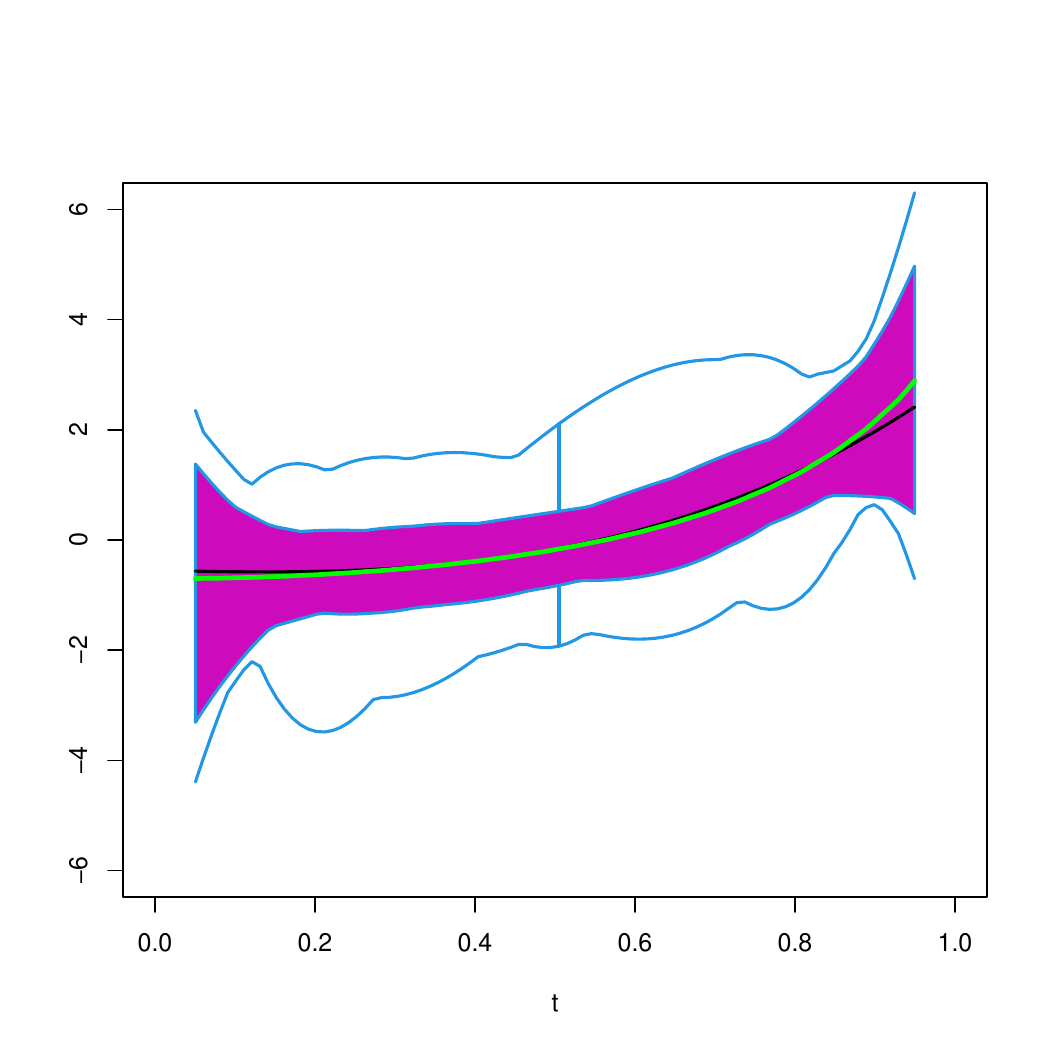}
  &  \includegraphics[scale=0.40]{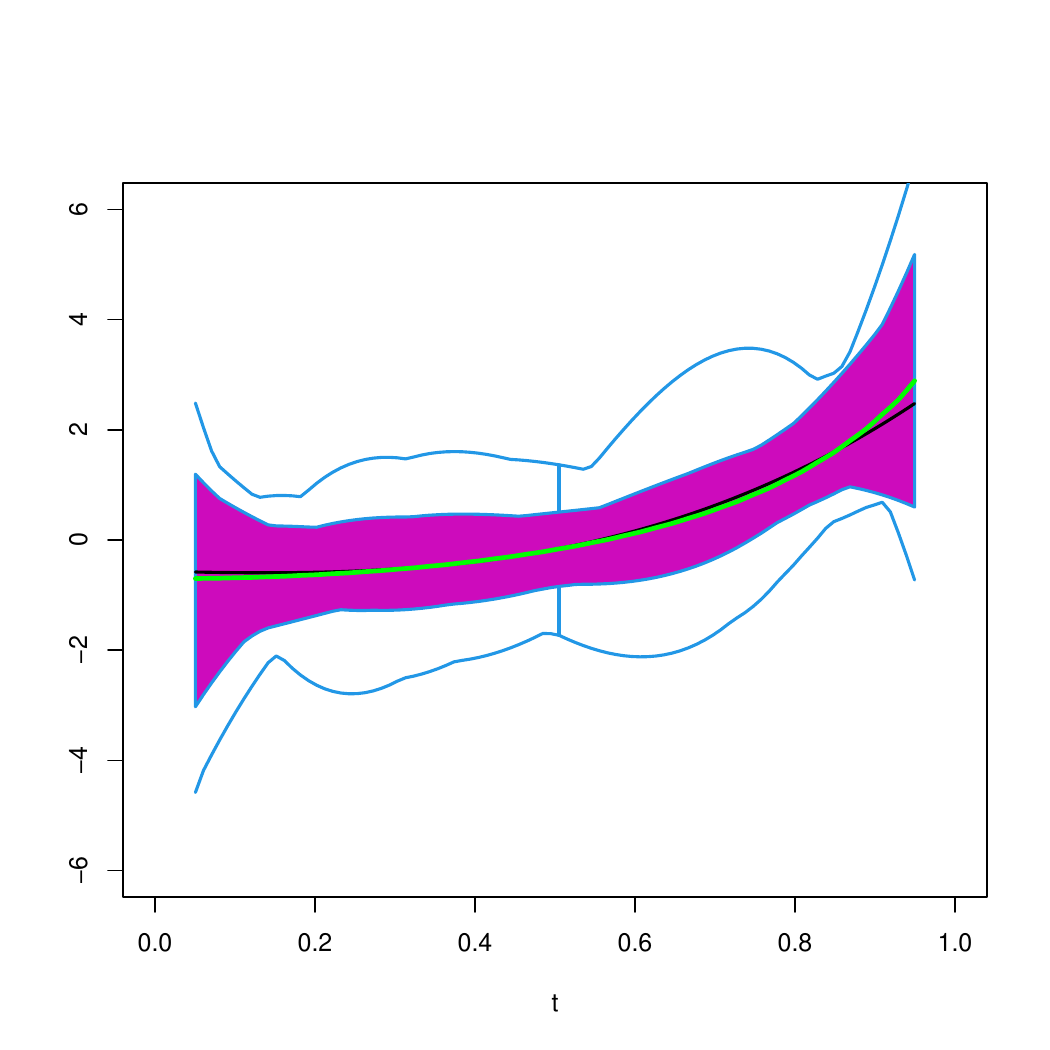}
   \\
   $\wbeta_{\wclBOX}$ & $\wbeta_{\wemeBOX}$ \\[-3ex]
  \includegraphics[scale=0.40]{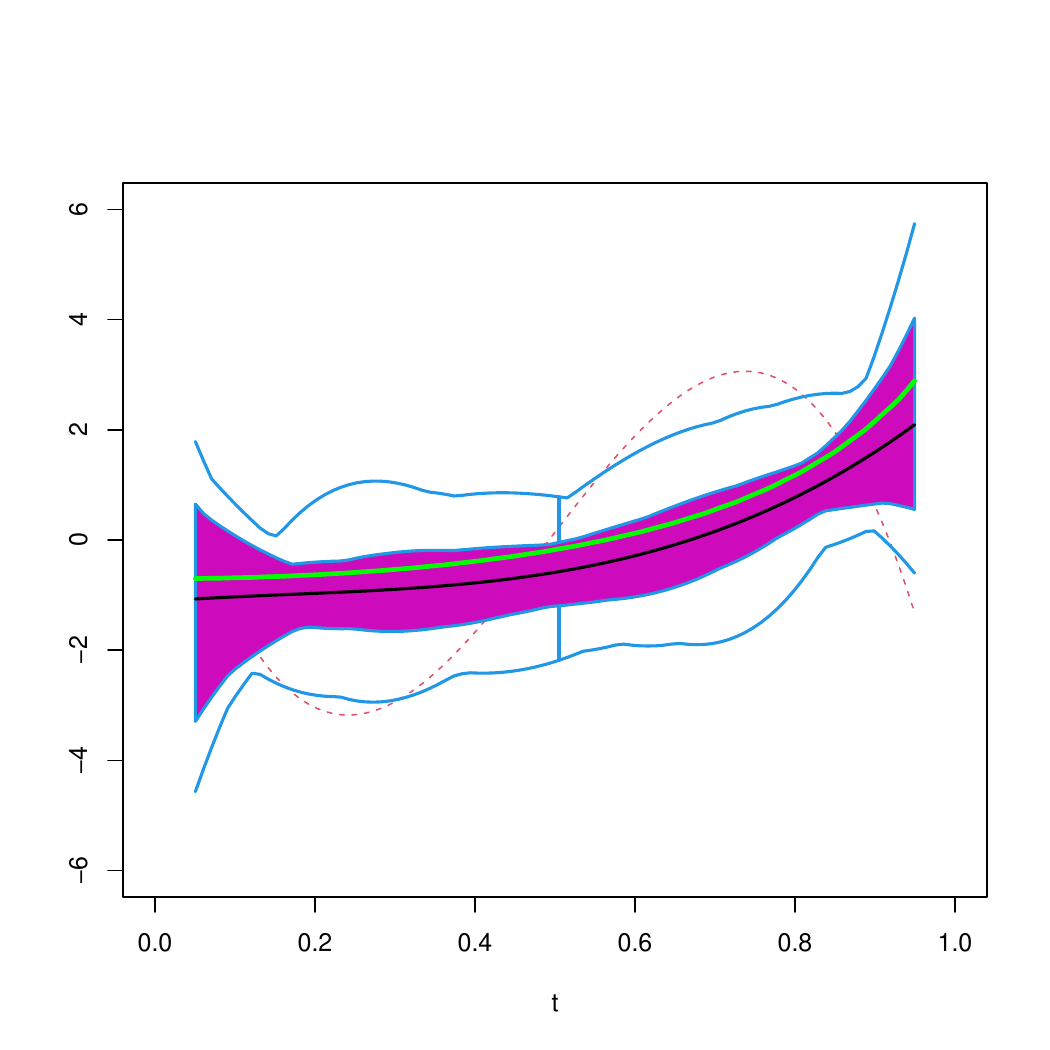}
  &  \includegraphics[scale=0.40]{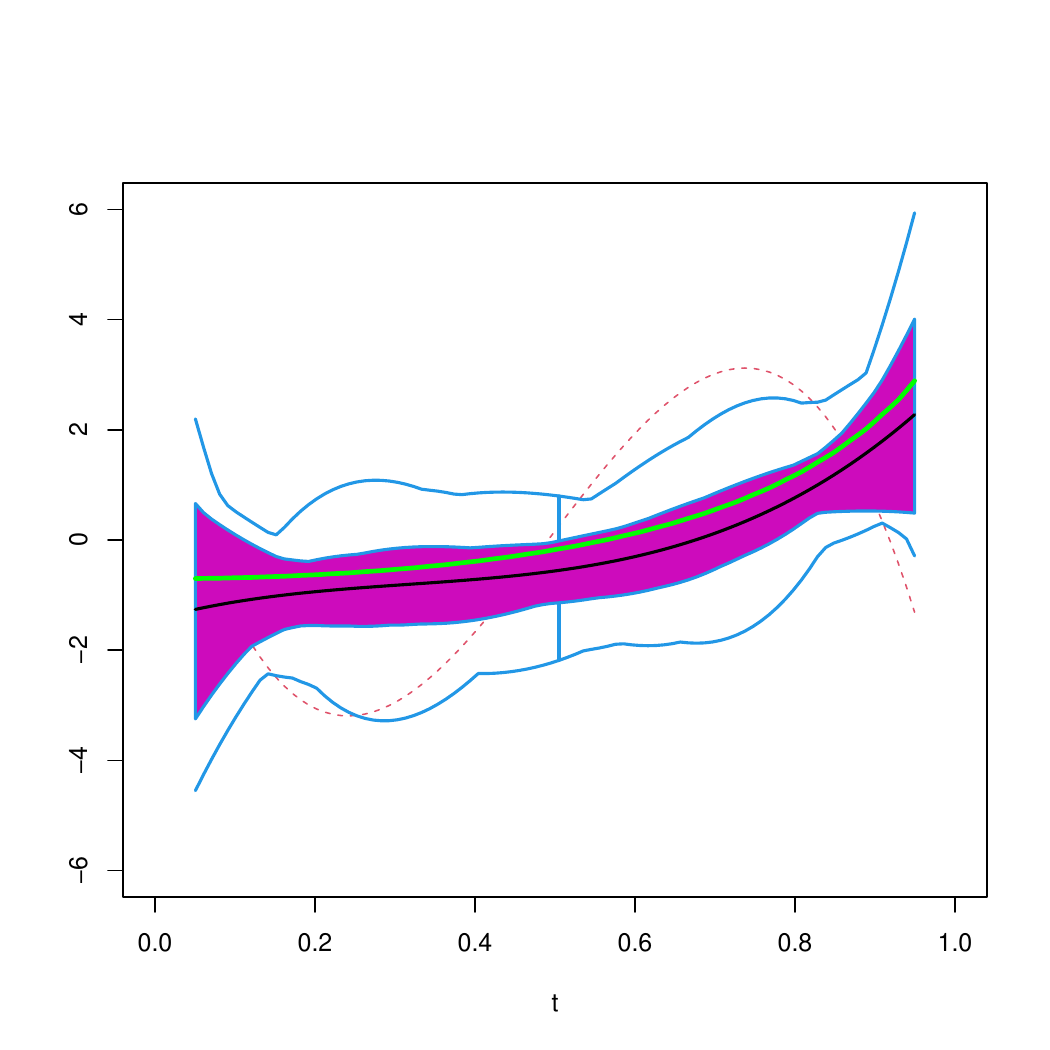}

\end{tabular}
\caption{\small \label{fig:wbeta-C310-poda5}  Functional boxplot of the estimators for $\beta_0$ under $C_{3,0.10}$  within the interval $[0.05,0.95]$. 
The true function is shown with a green dashed line, while the black solid one is the central 
curve of the $n_R = 1000$ estimates $\wbeta$.  }
\end{center} 
\end{figure}

\begin{figure}[tp]
 \begin{center}
 \footnotesize
 \renewcommand{\arraystretch}{0.2}
 \newcolumntype{M}{>{\centering\arraybackslash}m{\dimexpr.01\linewidth-1\tabcolsep}}
   \newcolumntype{G}{>{\centering\arraybackslash}m{\dimexpr.45\linewidth-1\tabcolsep}}
%\begin{tabular}{MGG}
\begin{tabular}{GG}
  $\wbeta_{\clas}$ & $\wbeta_{\eme}$   \\[-3ex]    
 
\includegraphics[scale=0.40]{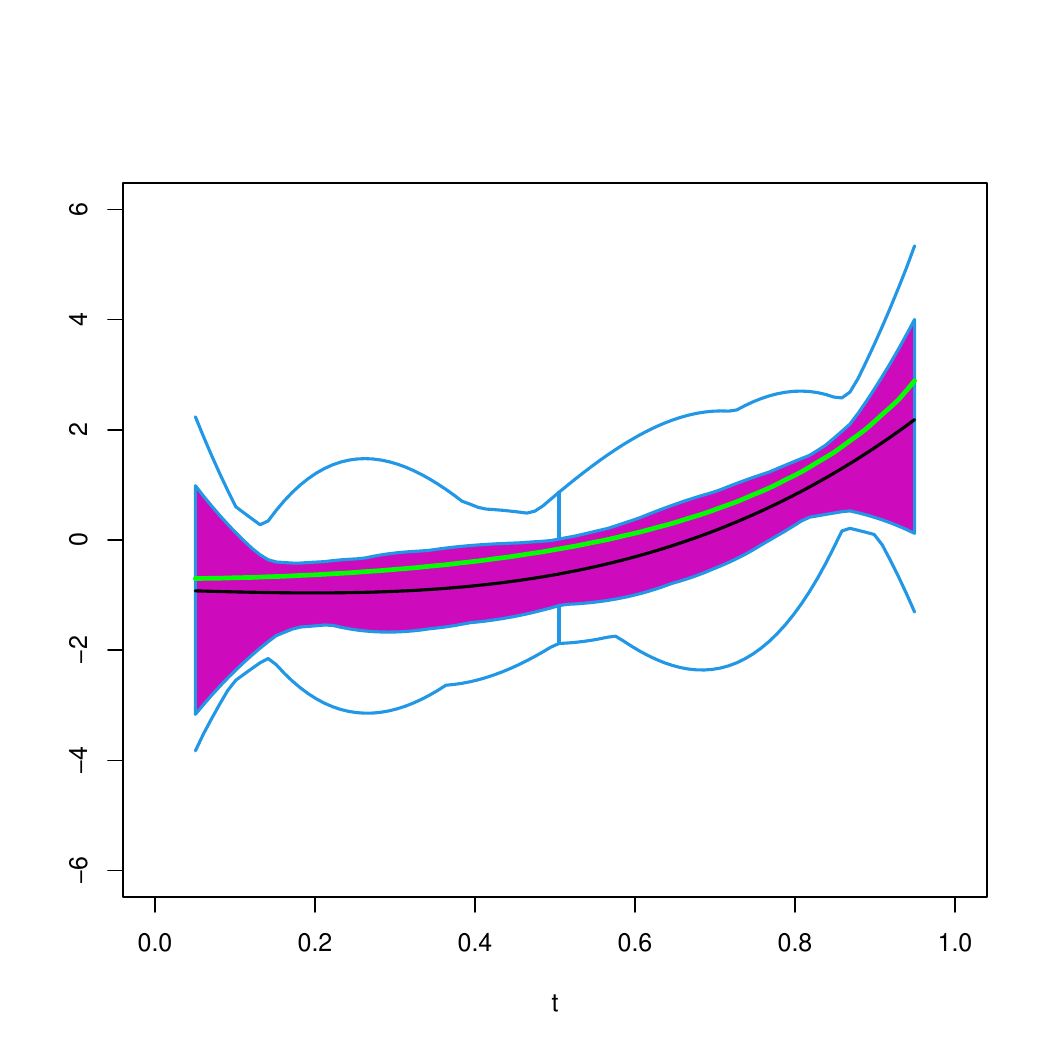}
 &  \includegraphics[scale=0.40]{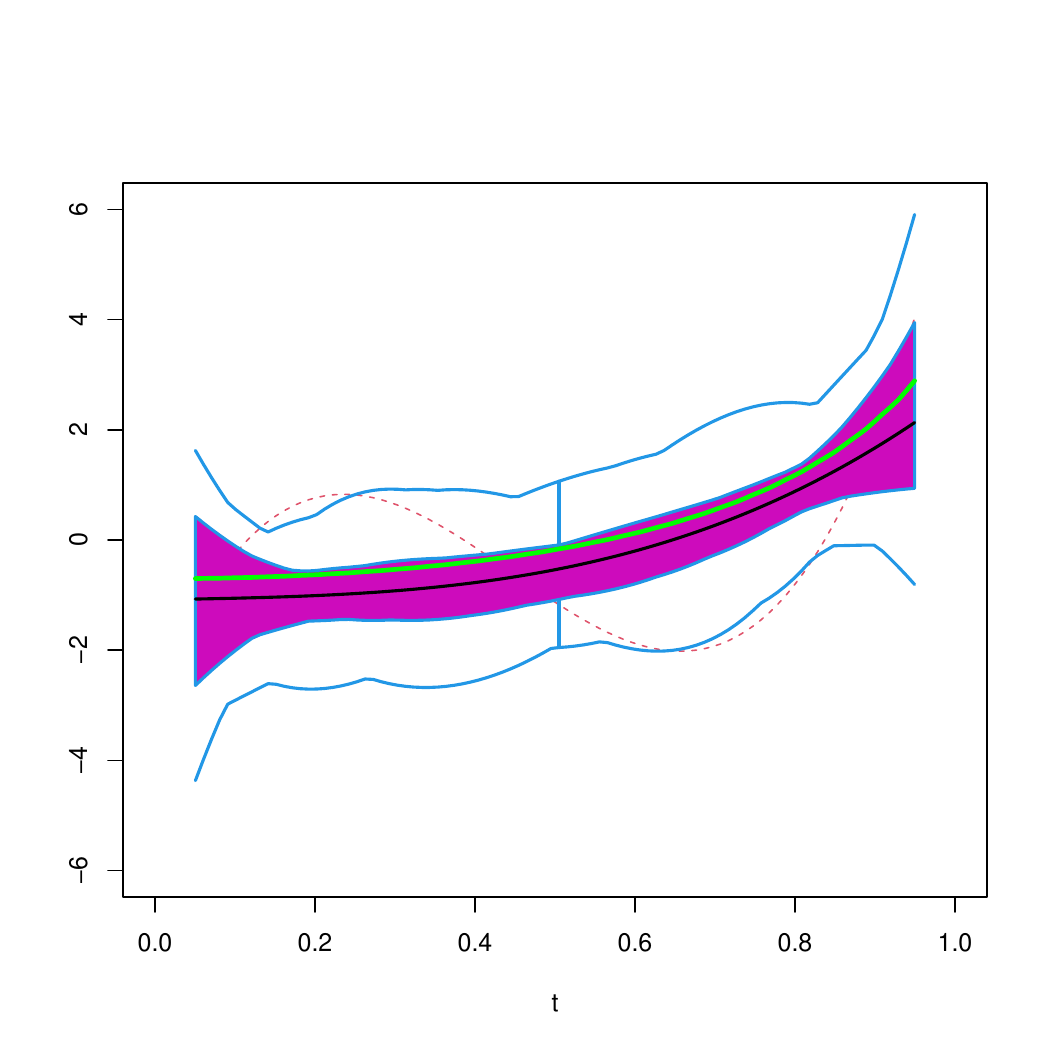}\\
   $\wbeta_{\wclHR}$ & $\wbeta_{\wemeHR}$ \\[-3ex]
    \includegraphics[scale=0.40]{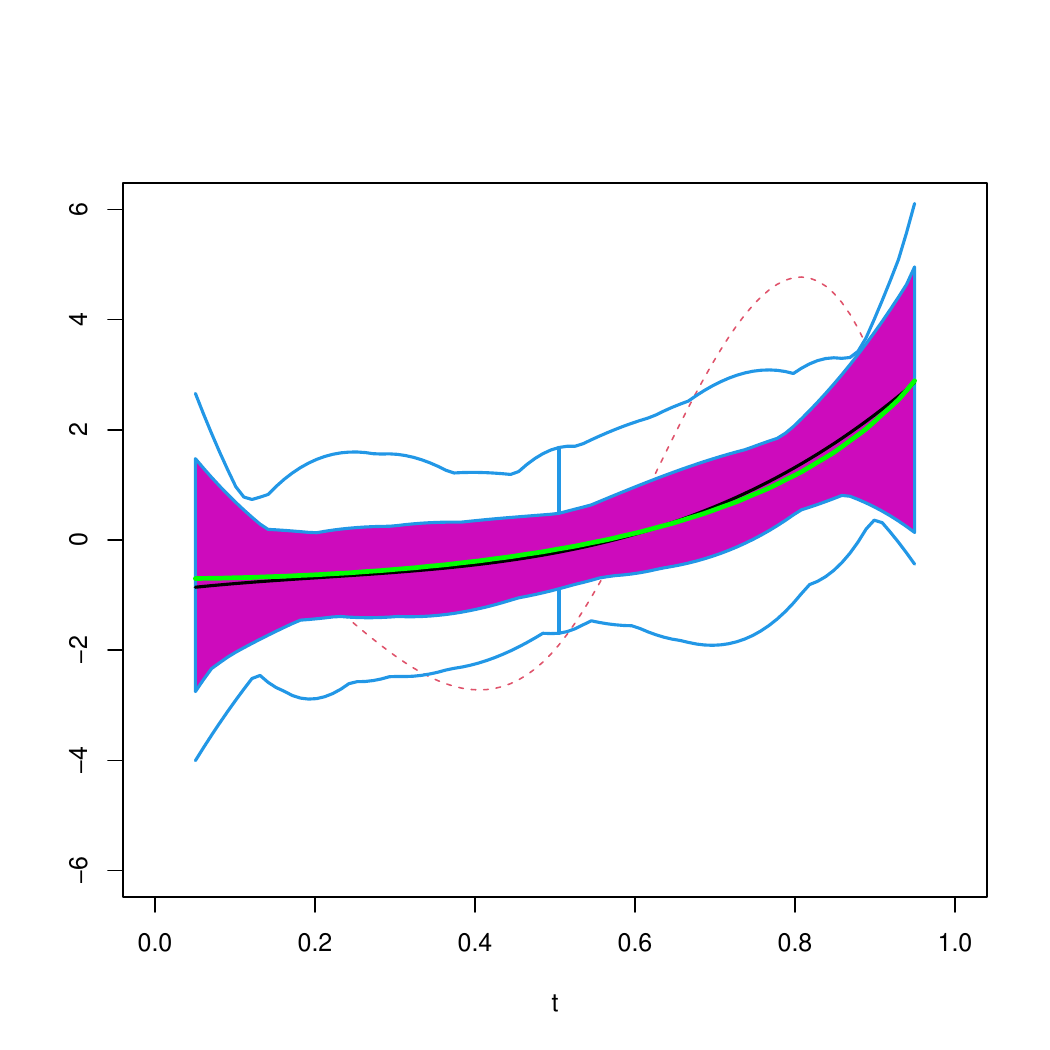}
  &  \includegraphics[scale=0.40]{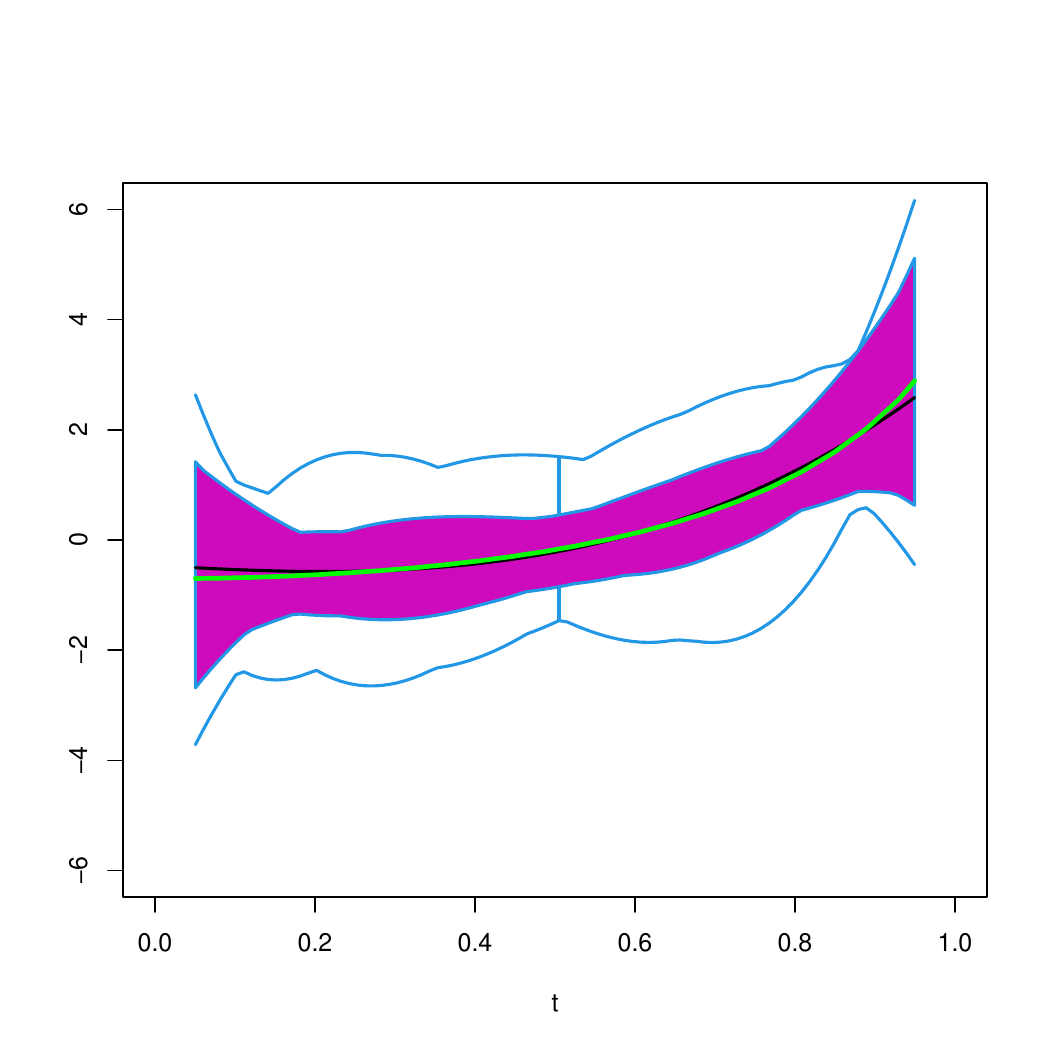}
   \\
   $\wbeta_{\wclBOX}$ & $\wbeta_{\wemeBOX}$ \\[-3ex]
  \includegraphics[scale=0.40]{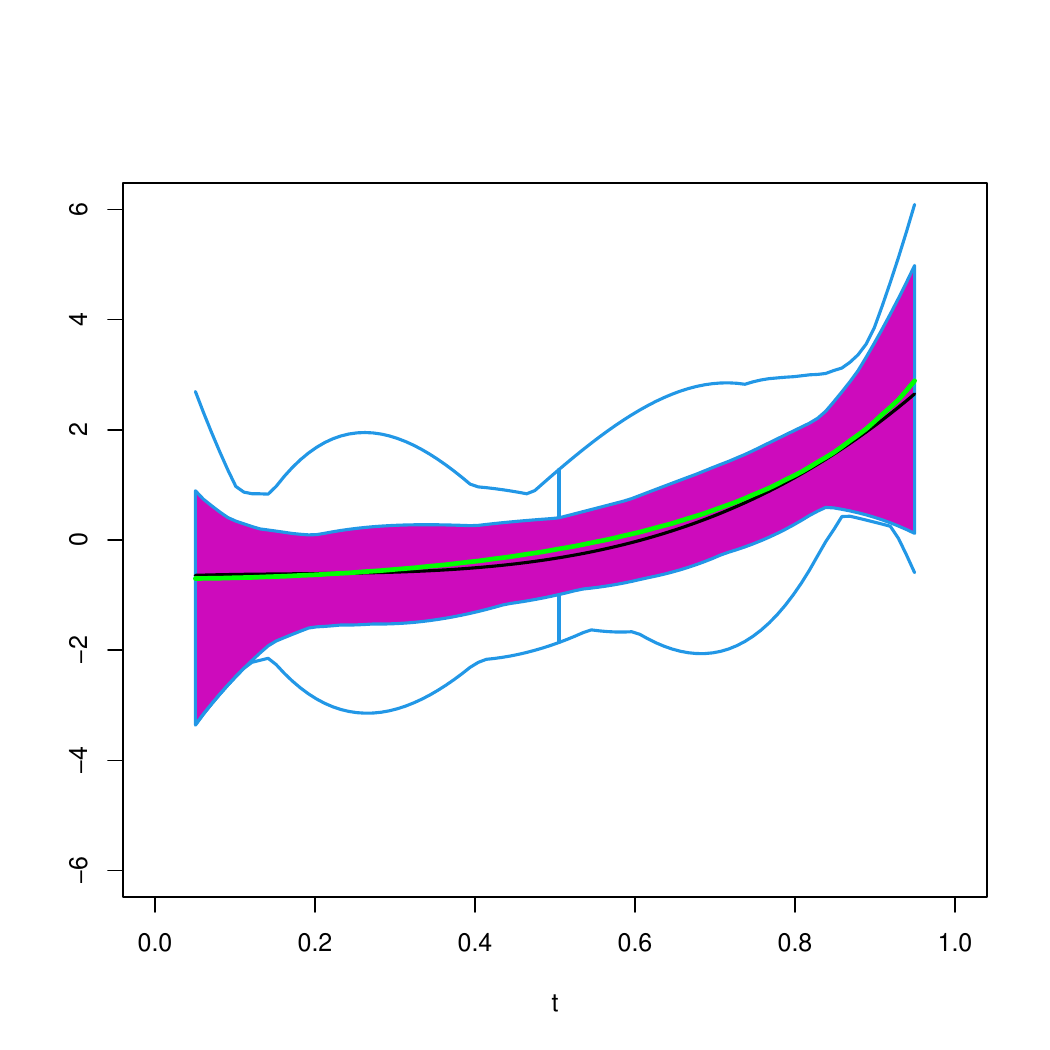}
  &  \includegraphics[scale=0.40]{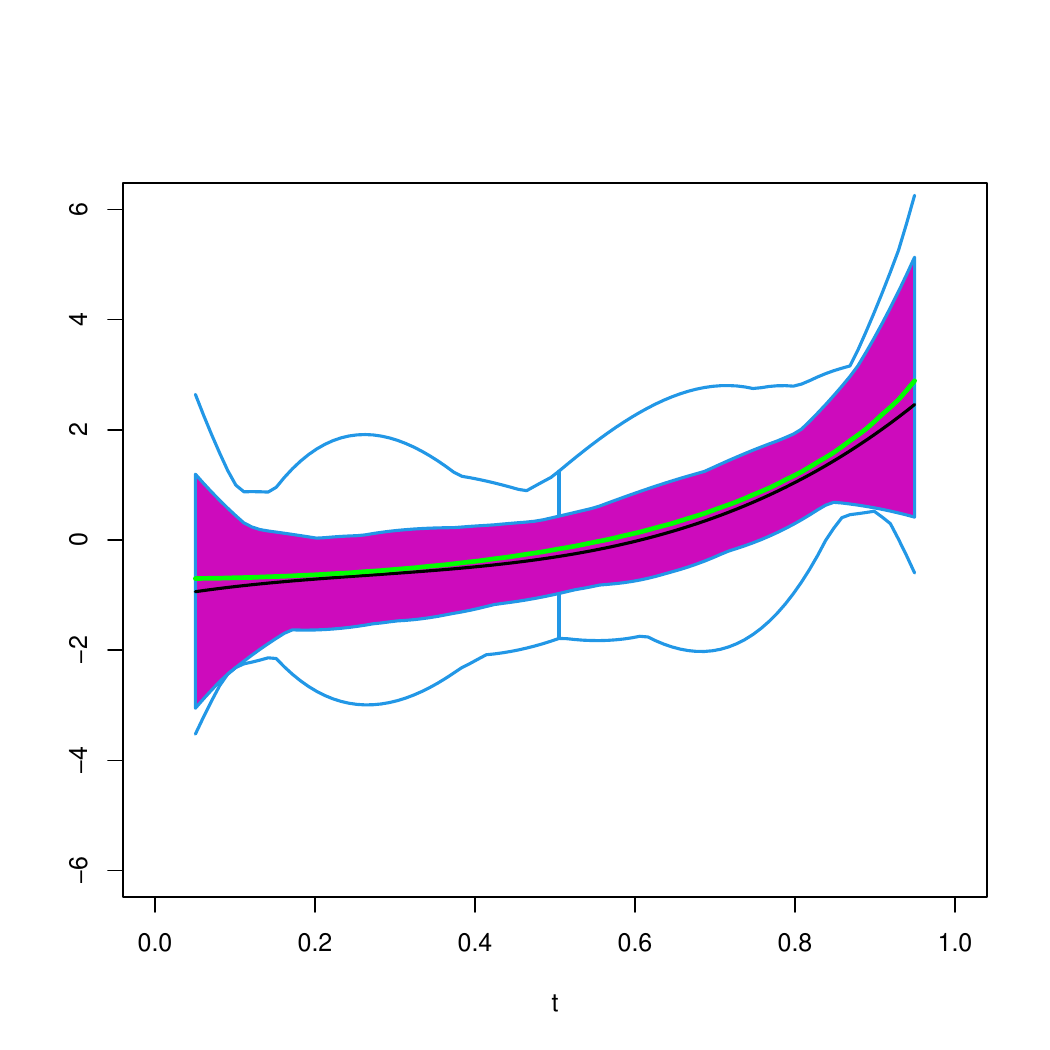}

\end{tabular}
\caption{\small \label{fig:wbeta-C45-poda5}  Functional boxplot of the estimators for $\beta_0$ under $C_{4,0.05}$  within the interval $[0.05,0.95]$. 
The true function is shown with a green dashed line, while the black solid one is the central 
curve of the $n_R = 1000$ estimates $\wbeta$.  }
\end{center} 
\end{figure}

\begin{figure}[tp]
 \begin{center}
 \footnotesize
 \renewcommand{\arraystretch}{0.2}
 \newcolumntype{M}{>{\centering\arraybackslash}m{\dimexpr.01\linewidth-1\tabcolsep}}
   \newcolumntype{G}{>{\centering\arraybackslash}m{\dimexpr.45\linewidth-1\tabcolsep}}
%\begin{tabular}{MGG}
\begin{tabular}{GG}
  $\wbeta_{\clas}$ & $\wbeta_{\eme}$   \\[-3ex]    
 
\includegraphics[scale=0.40]{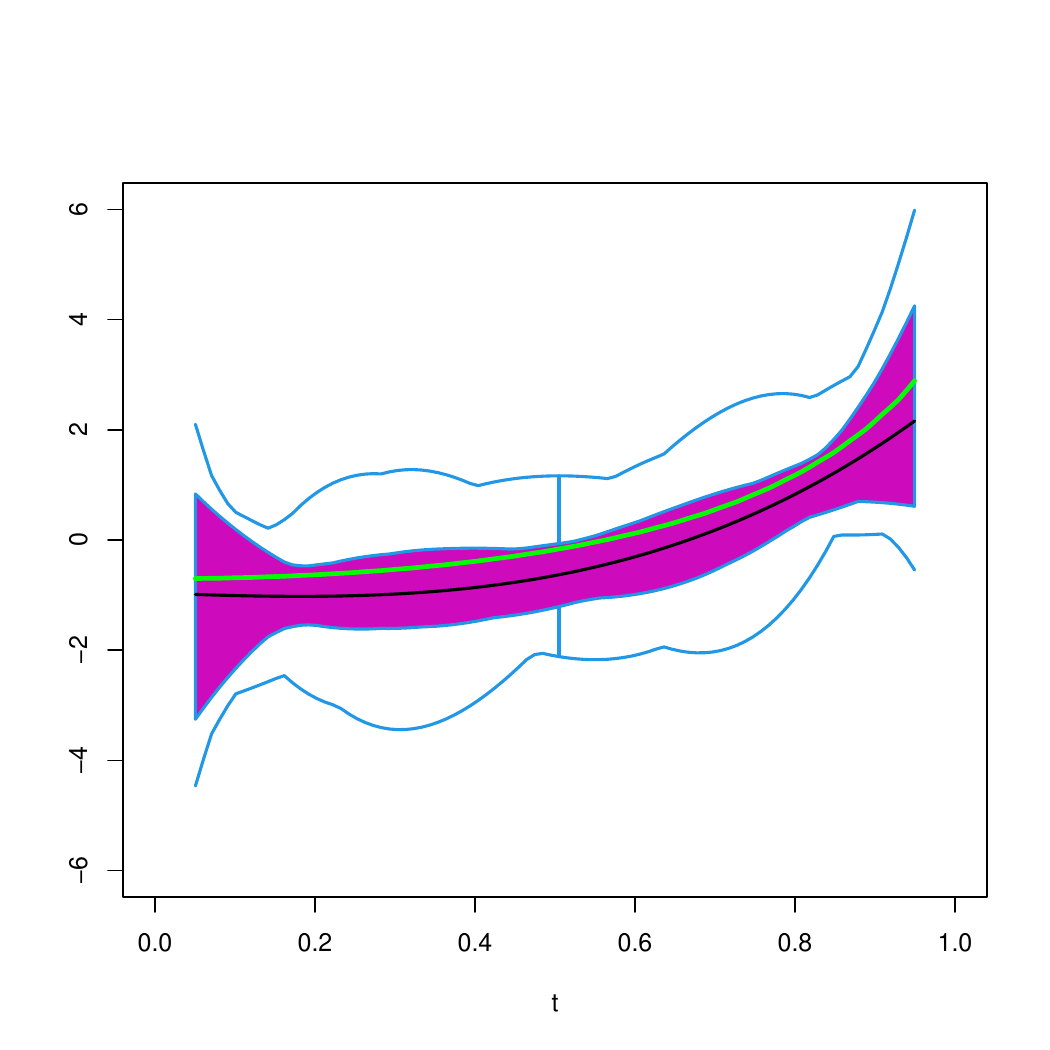}
 &  \includegraphics[scale=0.40]{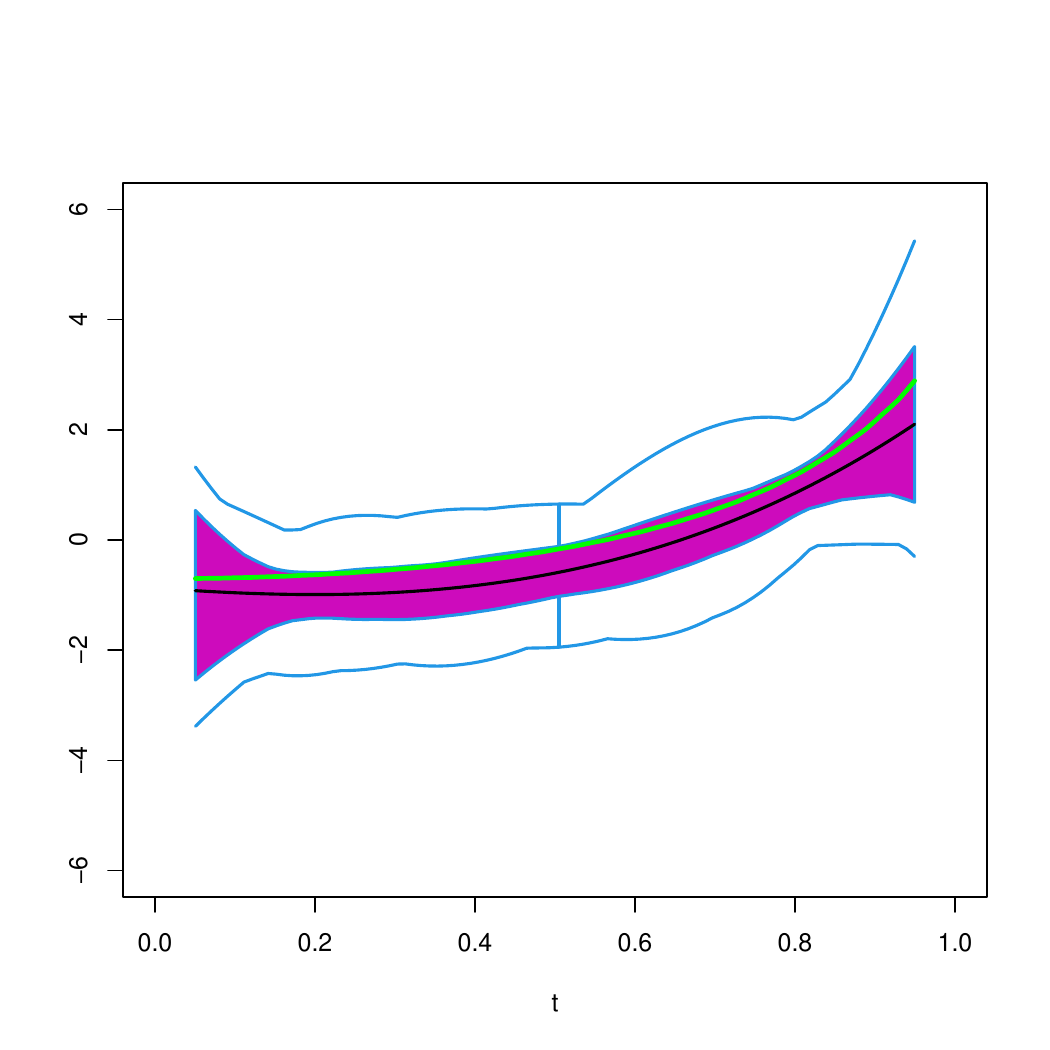}\\
   $\wbeta_{\wclHR}$ & $\wbeta_{\wemeHR}$ \\[-3ex] 
    \includegraphics[scale=0.40]{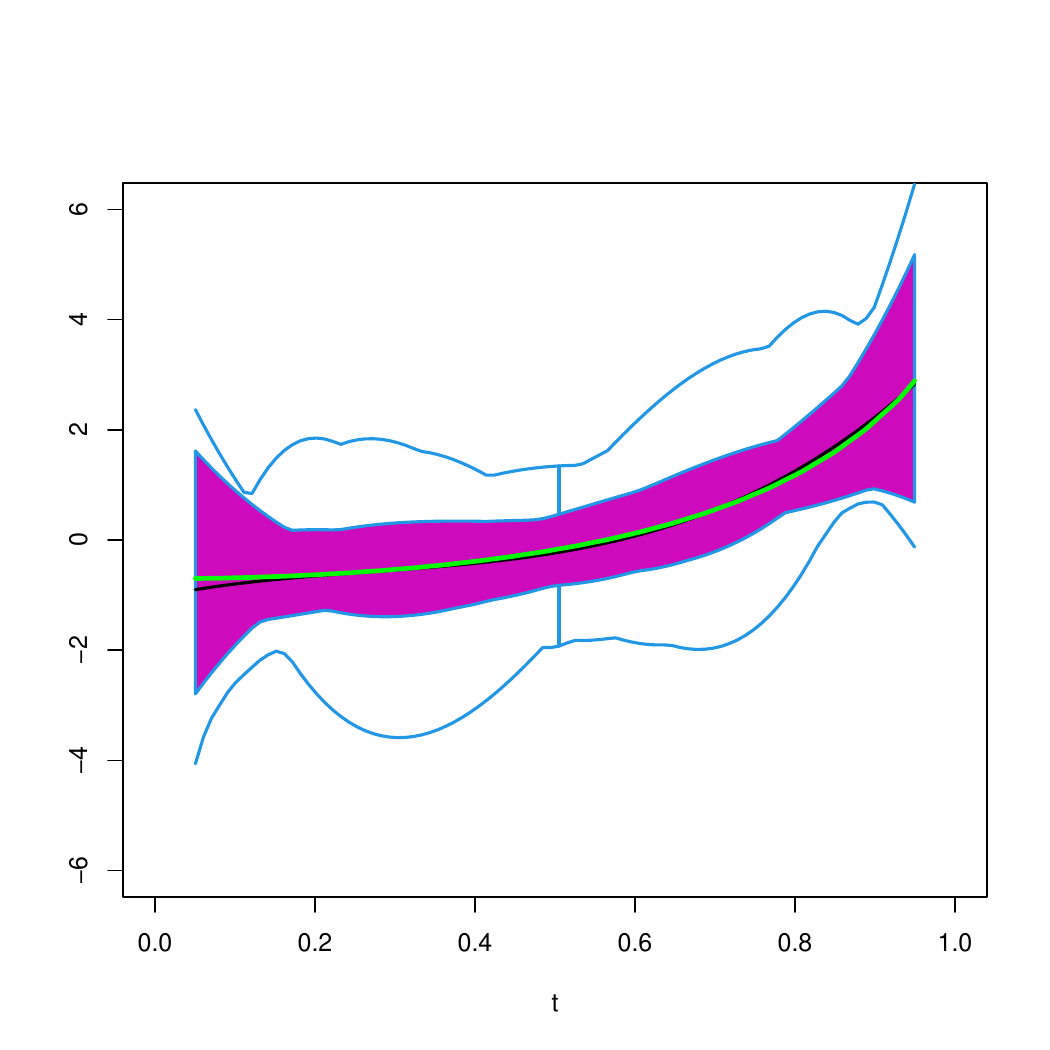}
  &  \includegraphics[scale=0.40]{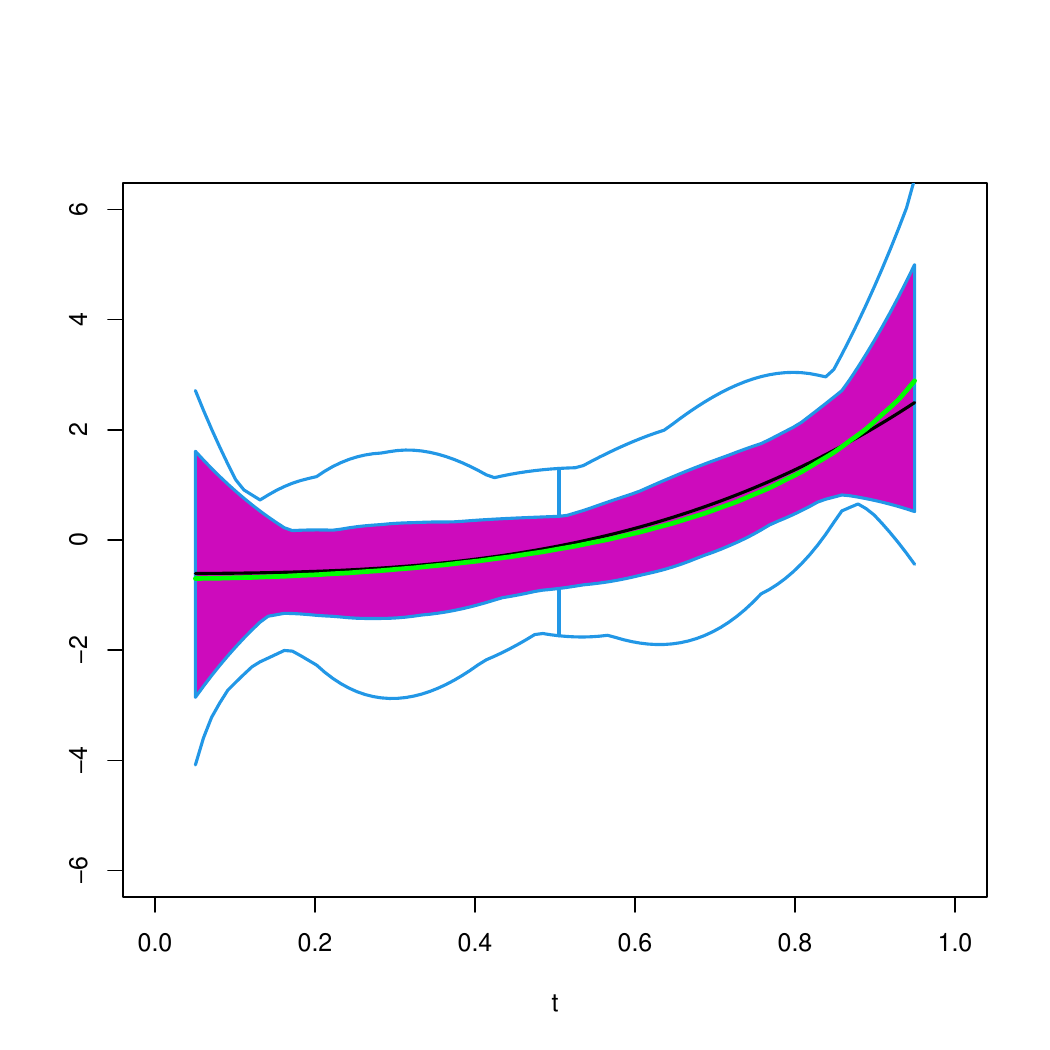}
   \\
   $\wbeta_{\wclBOX}$ & $\wbeta_{\wemeBOX}$ \\[-3ex]
  \includegraphics[scale=0.40]{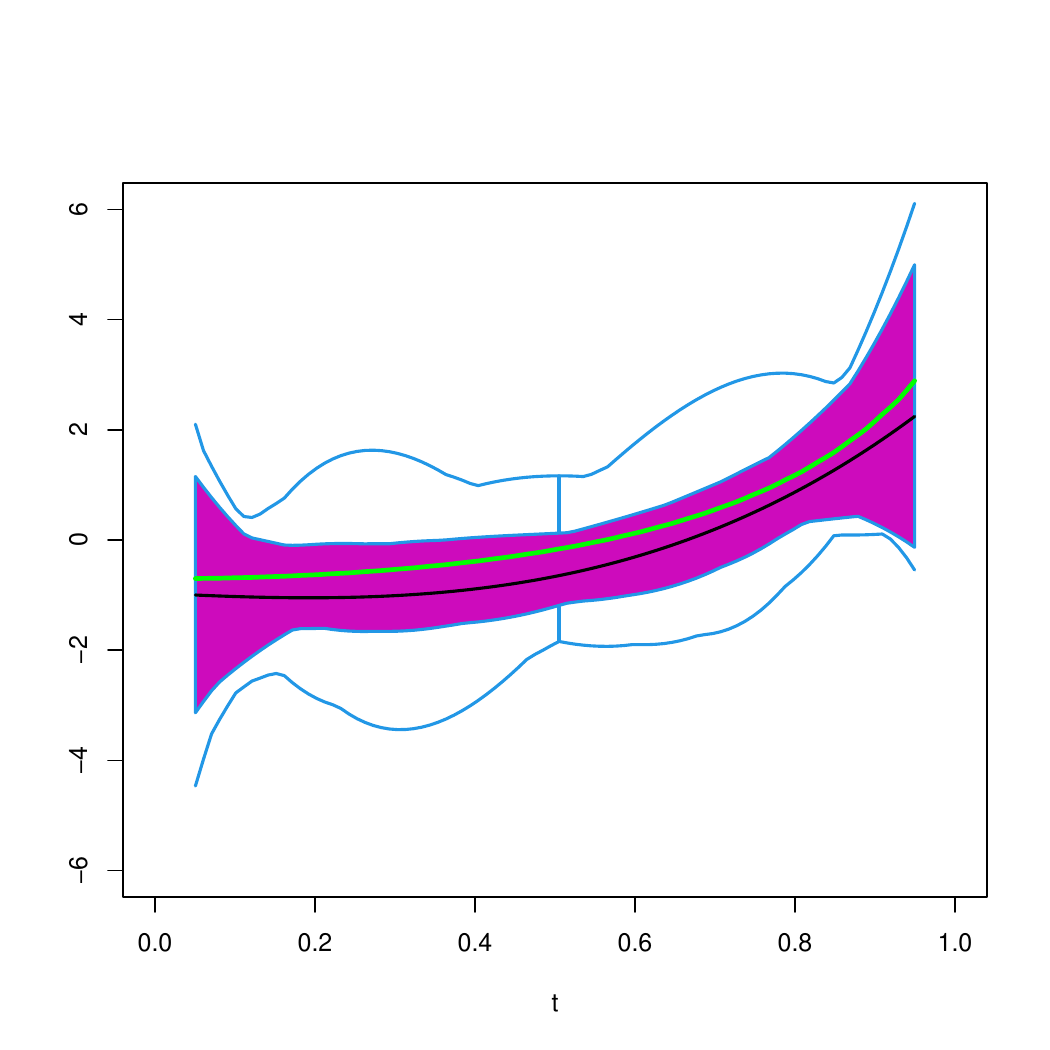}
  &  \includegraphics[scale=0.40]{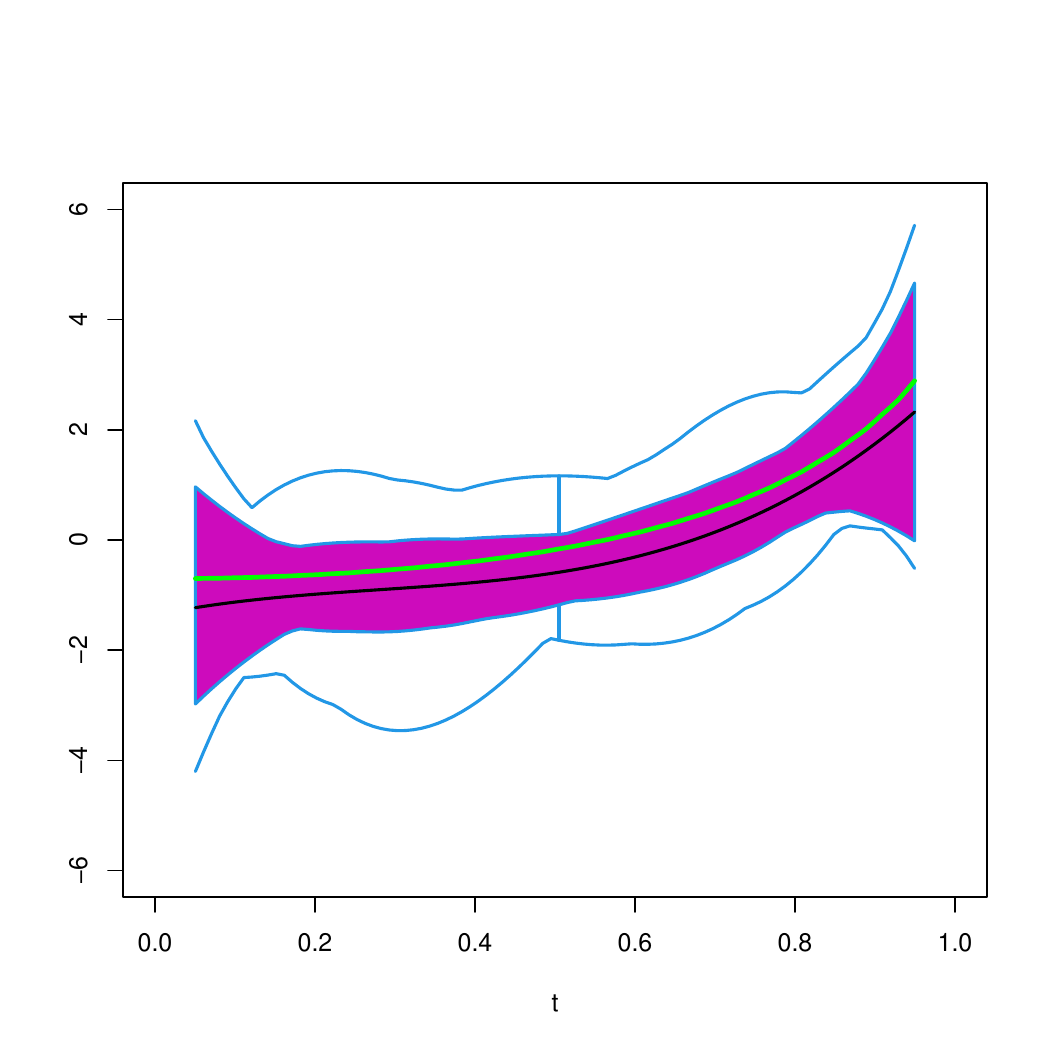}

\end{tabular}
\caption{\small \label{fig:wbeta-C410-poda5}  Functional boxplot of the estimators for $\beta_0$ under $C_{4,0.10}$  within the interval $[0.05,0.95]$. 
The true function is shown with a green dashed line, while the black solid one is the central 
curve of the $n_R = 1000$ estimates $\wbeta$.  }
\end{center} 
\end{figure}

\begin{figure}[tp]
 \begin{center}
 \footnotesize
 \renewcommand{\arraystretch}{0.2}
 \newcolumntype{M}{>{\centering\arraybackslash}m{\dimexpr.01\linewidth-1\tabcolsep}}
   \newcolumntype{G}{>{\centering\arraybackslash}m{\dimexpr.45\linewidth-1\tabcolsep}}
%\begin{tabular}{MGG}
\begin{tabular}{GG}
  $\wbeta_{\clas}$ & $\wbeta_{\eme}$   \\[-3ex]    
 
\includegraphics[scale=0.40]{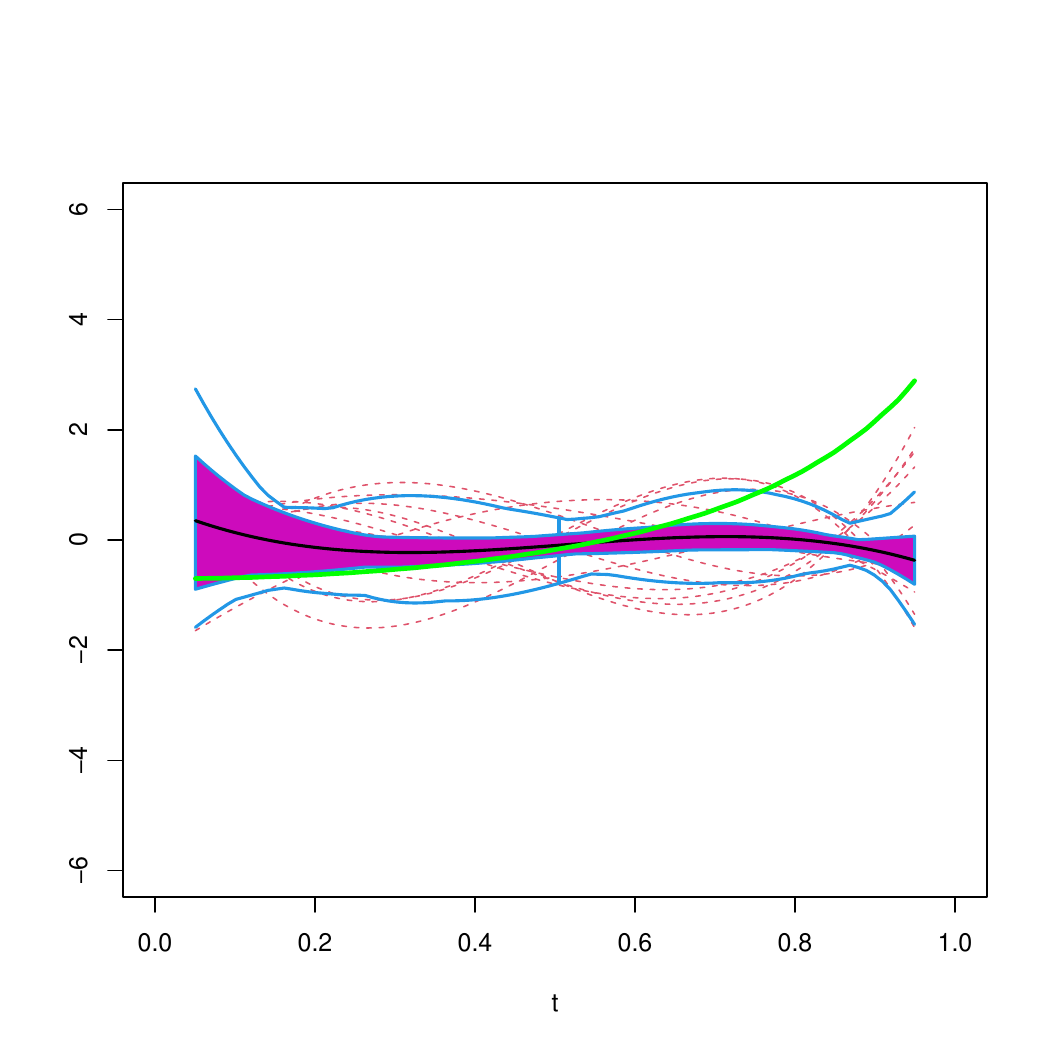}
 &  \includegraphics[scale=0.40]{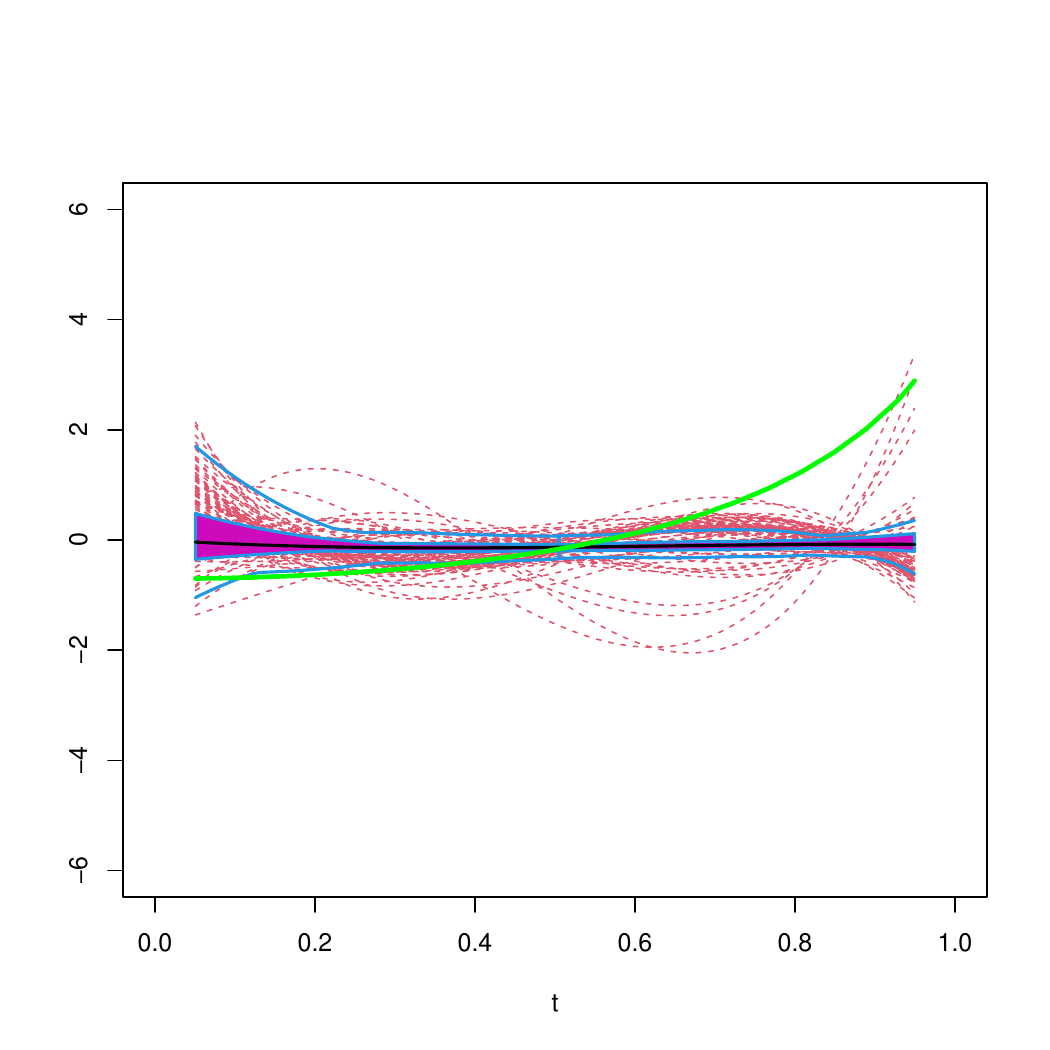}\\
   $\wbeta_{\wclHR}$ & $\wbeta_{\wemeHR}$ \\[-3ex] 
    \includegraphics[scale=0.40]{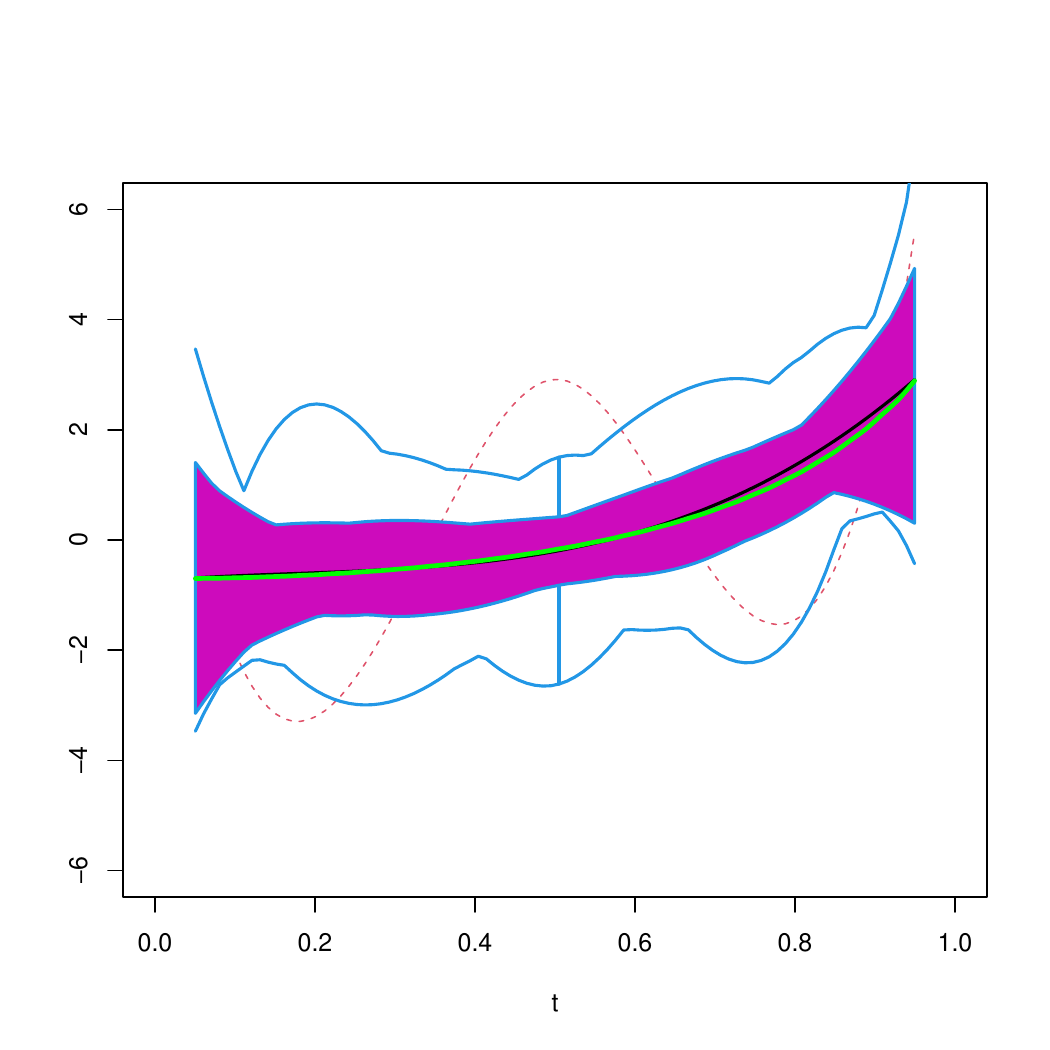}
  &  \includegraphics[scale=0.40]{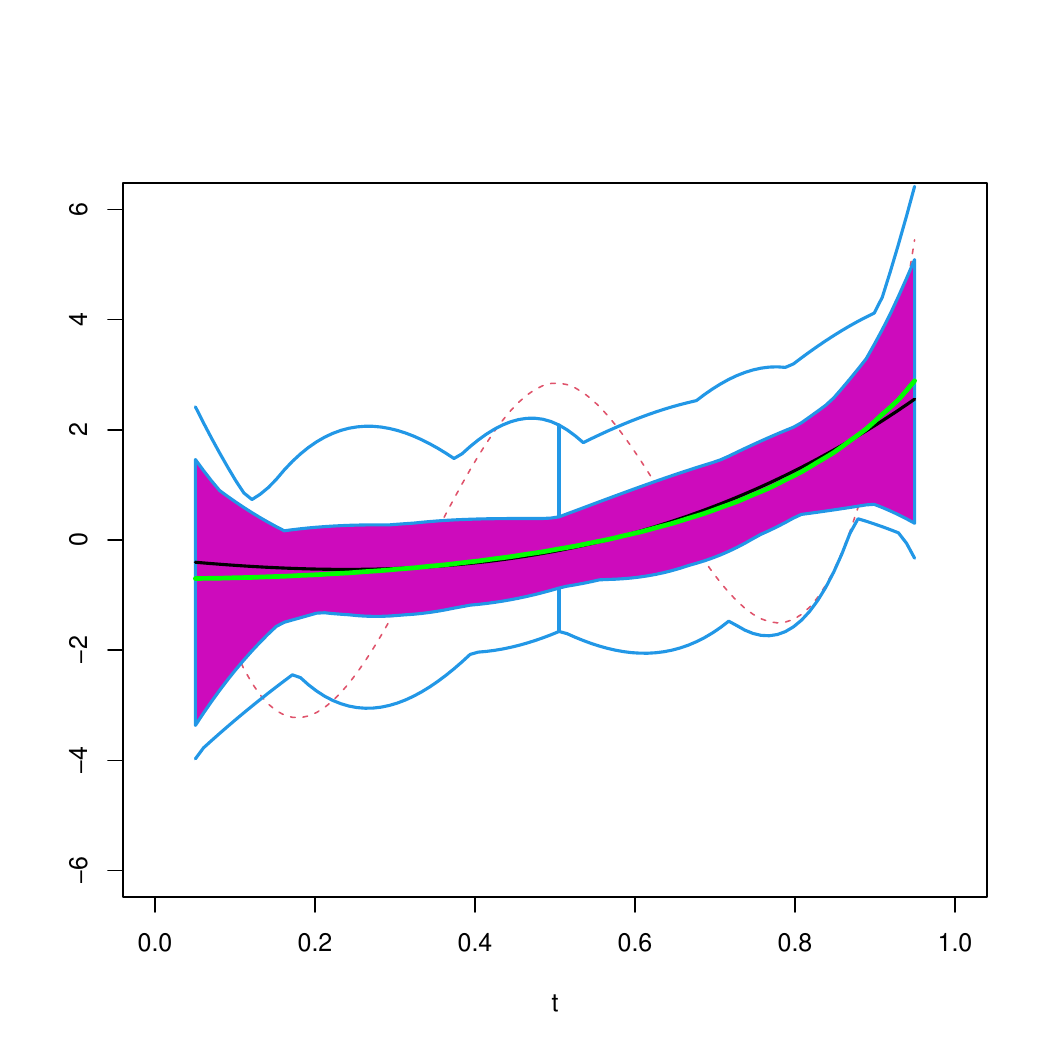}
   \\
   $\wbeta_{\wclBOX}$ & $\wbeta_{\wemeBOX}$ \\[-3ex]
  \includegraphics[scale=0.40]{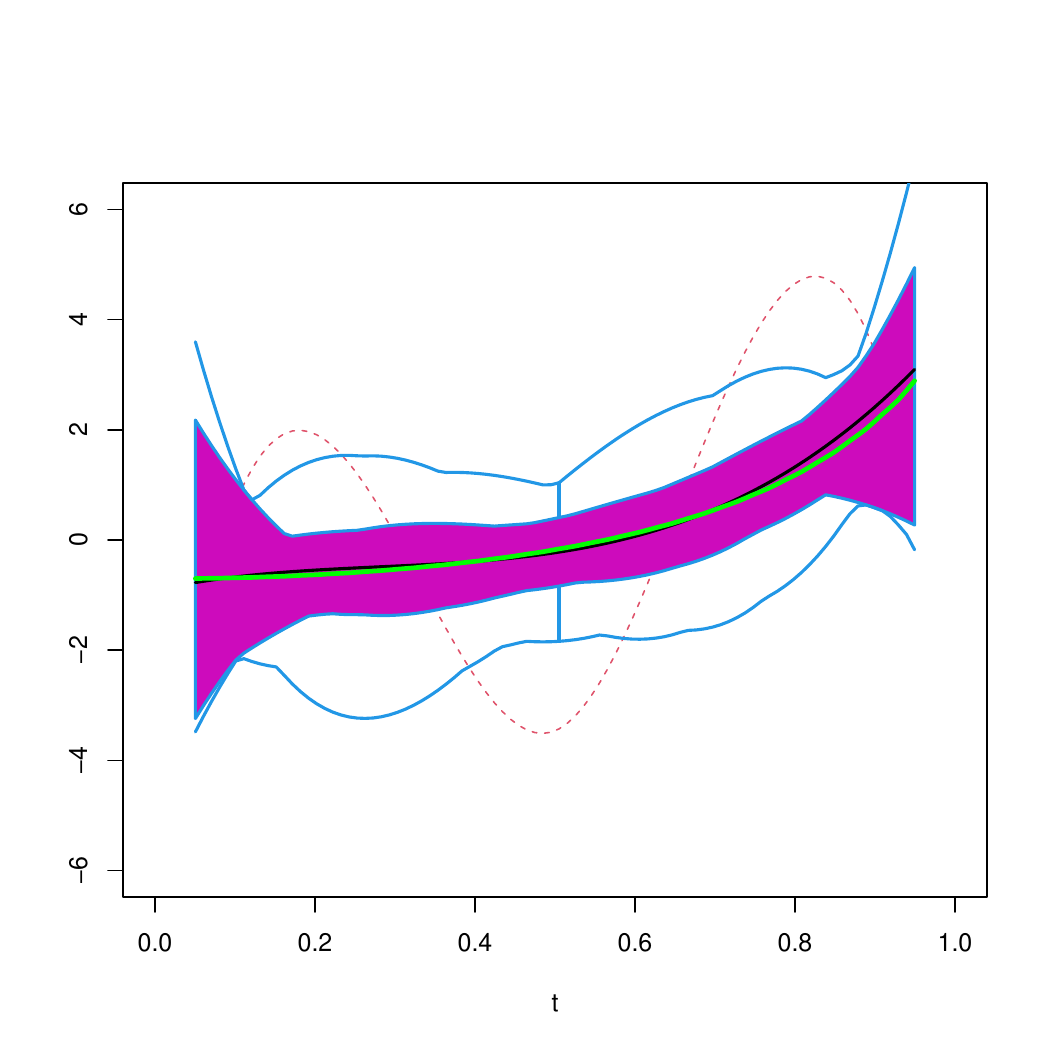}
  &  \includegraphics[scale=0.40]{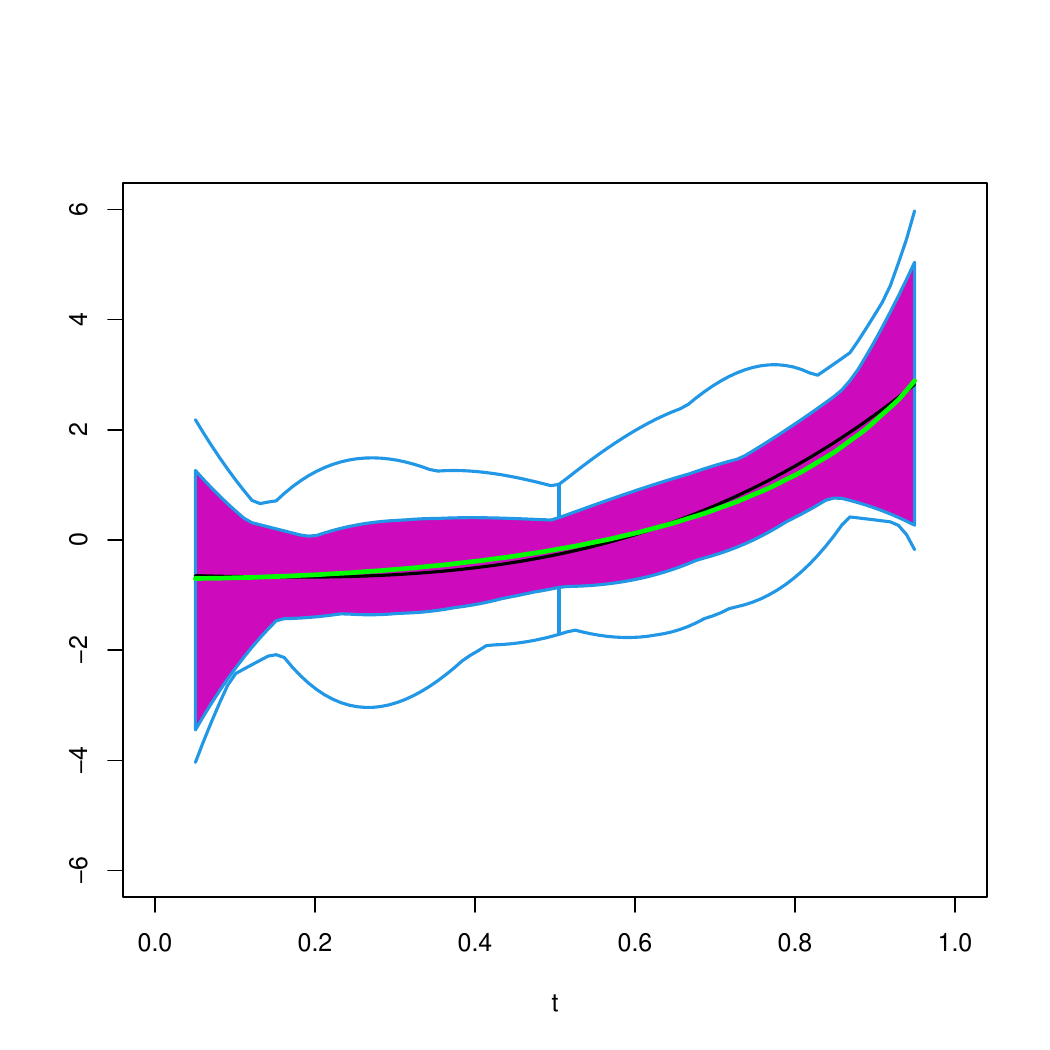}

\end{tabular}
\caption{\small \label{fig:wbeta-C55-poda5}  Functional boxplot of the estimators for $\beta_0$ under $C_{5,0.05}$  within the interval $[0.05,0.95]$. 
The true function is shown with a green dashed line, while the black solid one is the central 
curve of the $n_R = 1000$ estimates $\wbeta$.  }
\end{center} 
\end{figure}

\begin{figure}[tp]
 \begin{center}
 \footnotesize
 \renewcommand{\arraystretch}{0.2}
 \newcolumntype{M}{>{\centering\arraybackslash}m{\dimexpr.01\linewidth-1\tabcolsep}}
   \newcolumntype{G}{>{\centering\arraybackslash}m{\dimexpr.45\linewidth-1\tabcolsep}}
%\begin{tabular}{MGG}
\begin{tabular}{GG}
  $\wbeta_{\clas}$ & $\wbeta_{\eme}$   \\[-3ex]    
 
\includegraphics[scale=0.40]{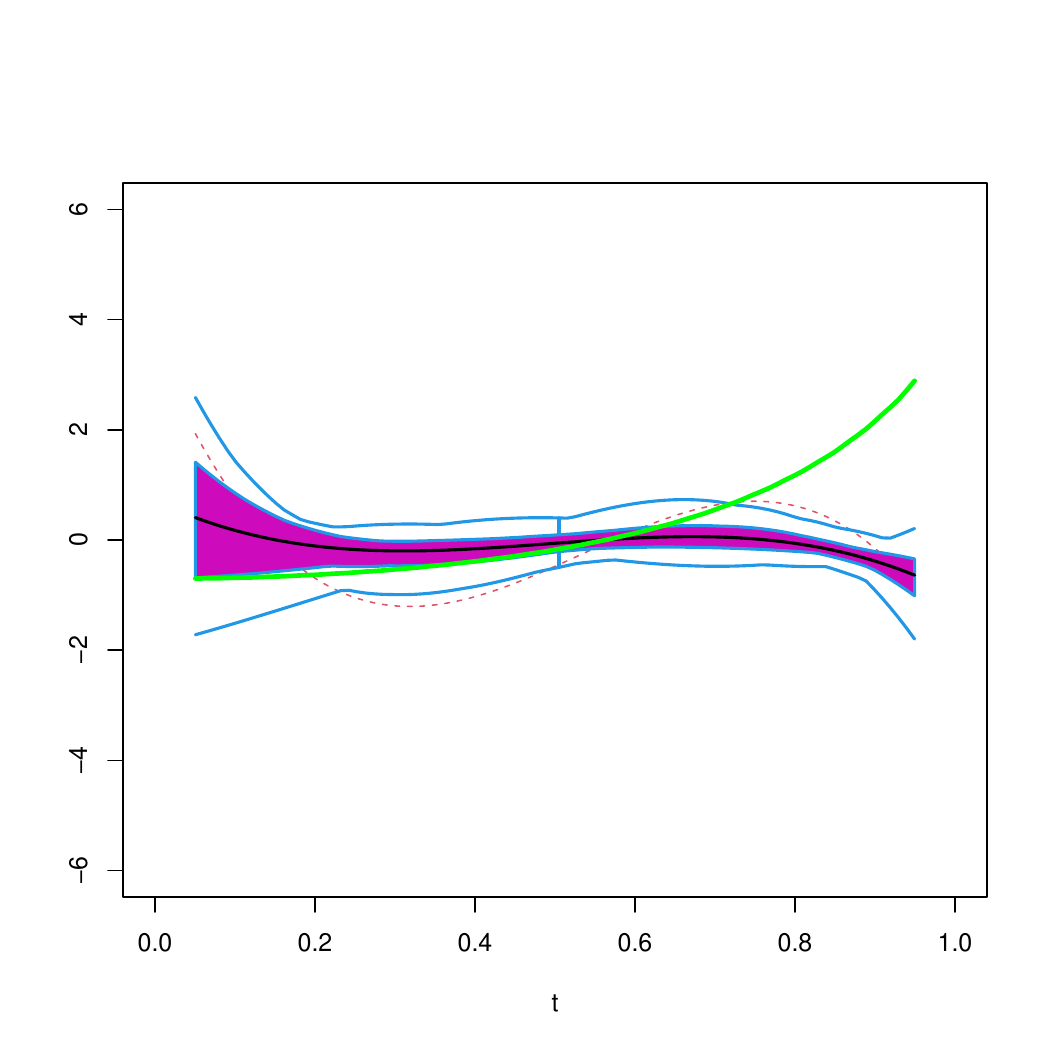}
 &  \includegraphics[scale=0.40]{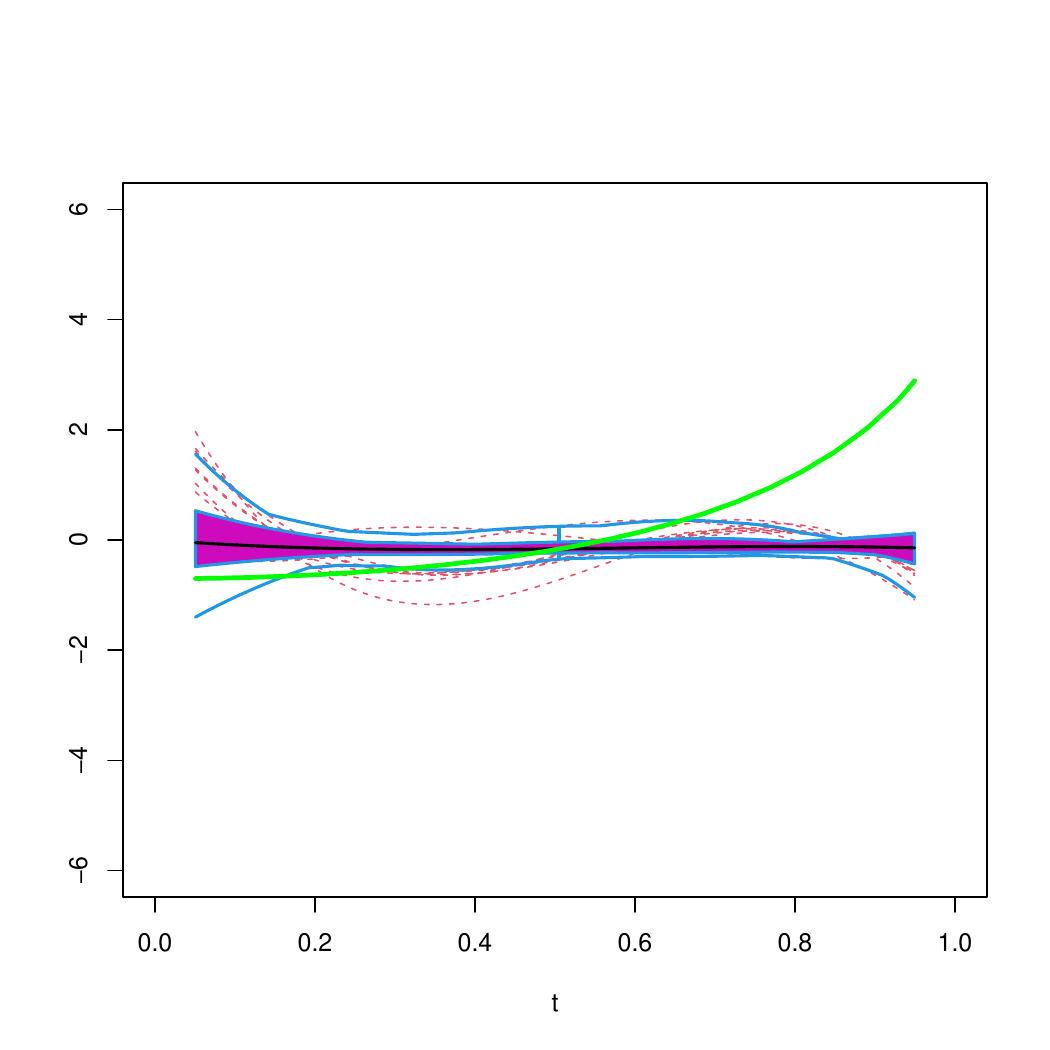}\\
   $\wbeta_{\wclHR}$ & $\wbeta_{\wemeHR}$ \\[-3ex] 
    \includegraphics[scale=0.40]{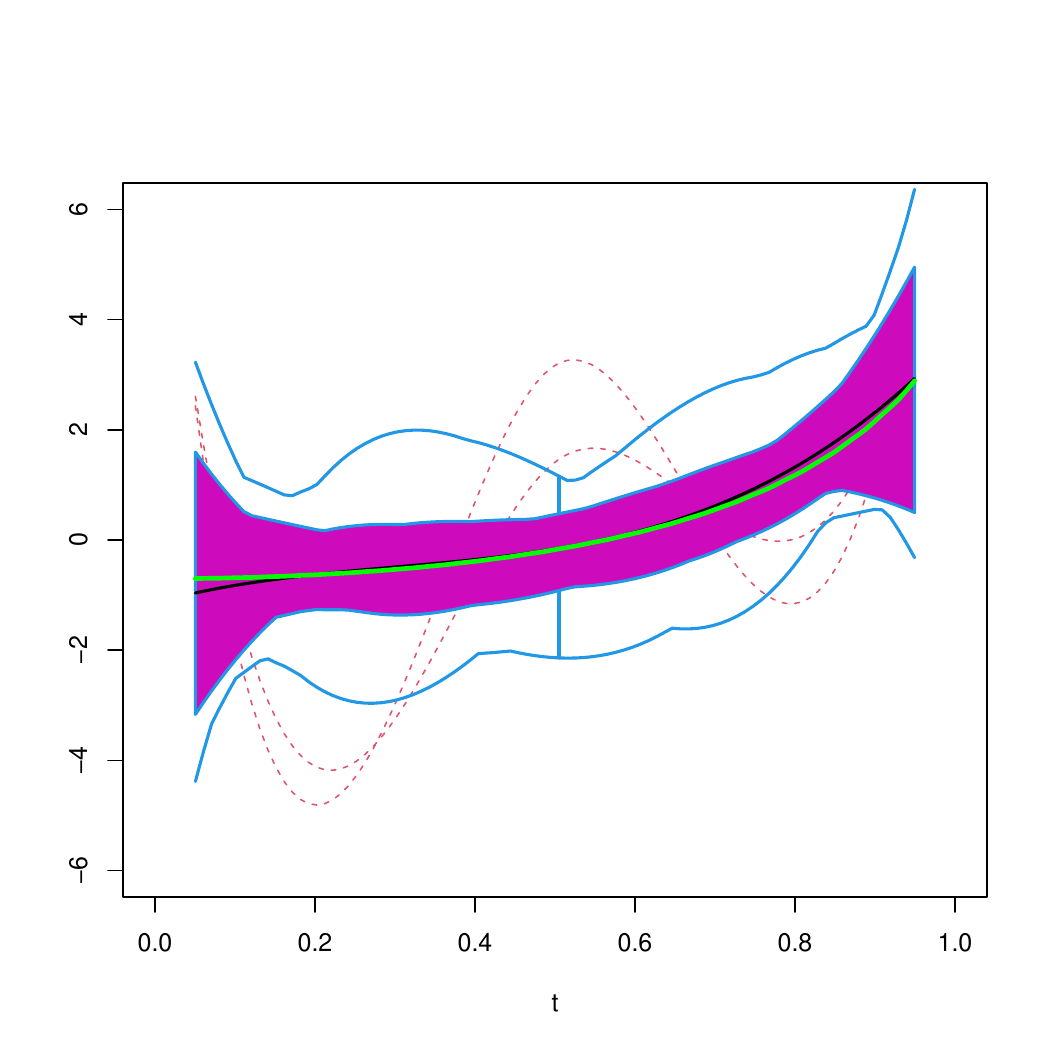}
  &  \includegraphics[scale=0.40]{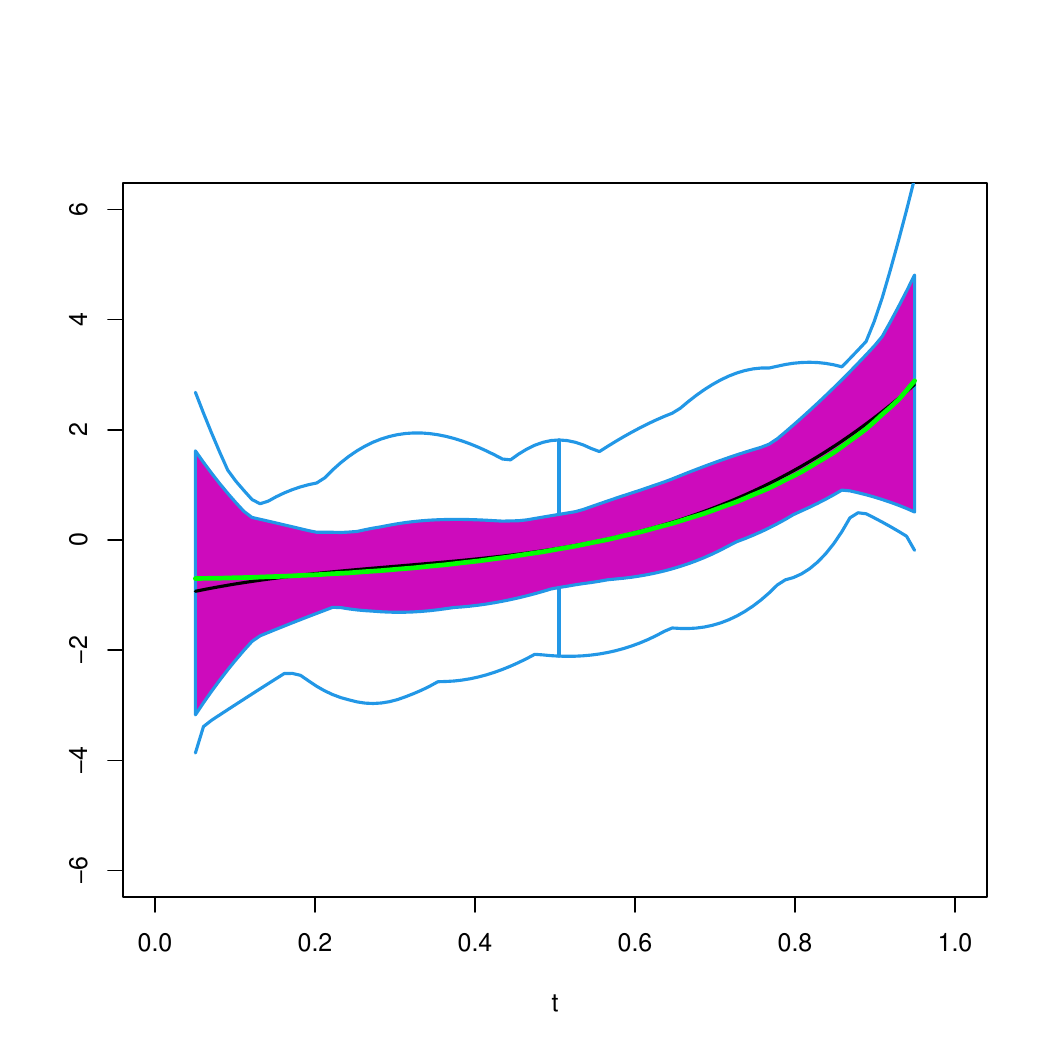}
   \\
   $\wbeta_{\wclBOX}$ & $\wbeta_{\wemeBOX}$ \\[-3ex]
  \includegraphics[scale=0.40]{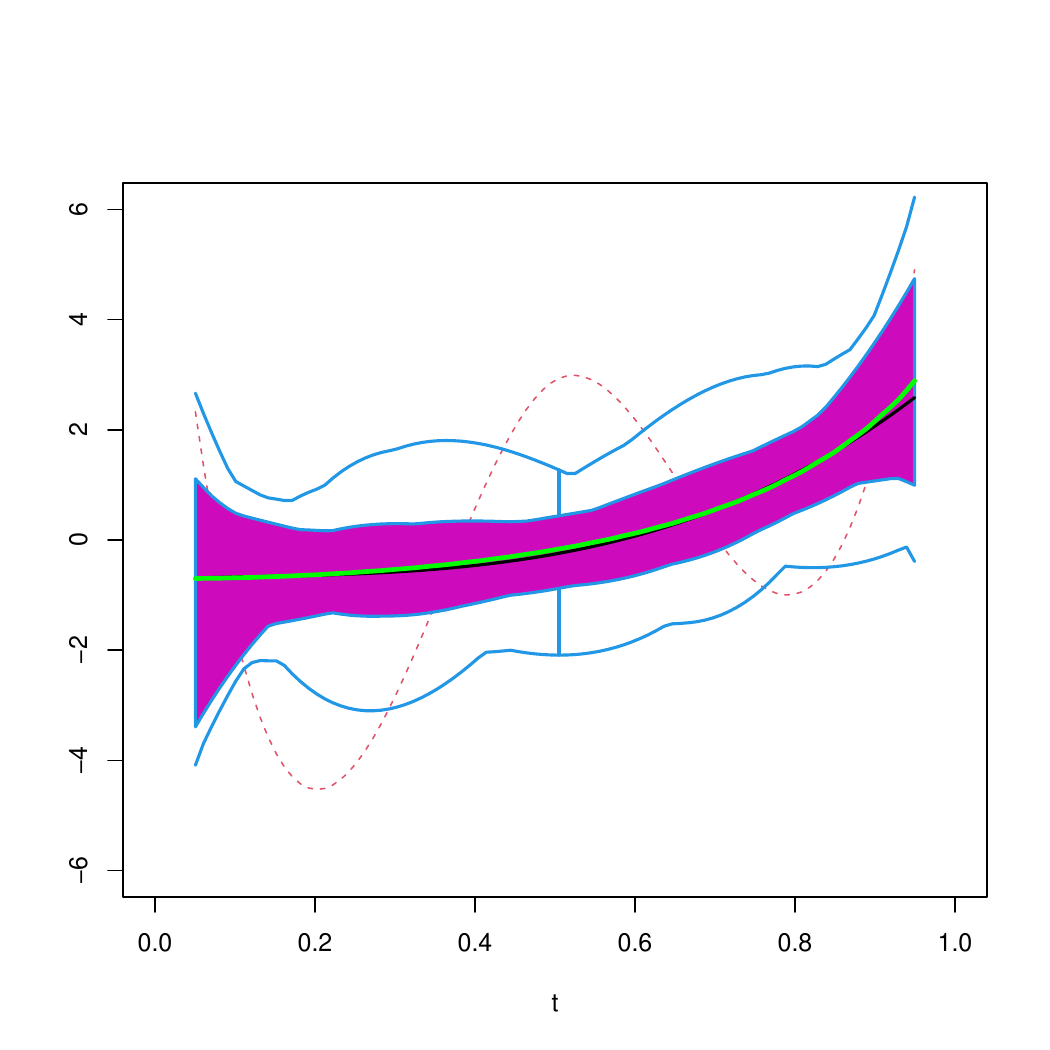}
  &  \includegraphics[scale=0.40]{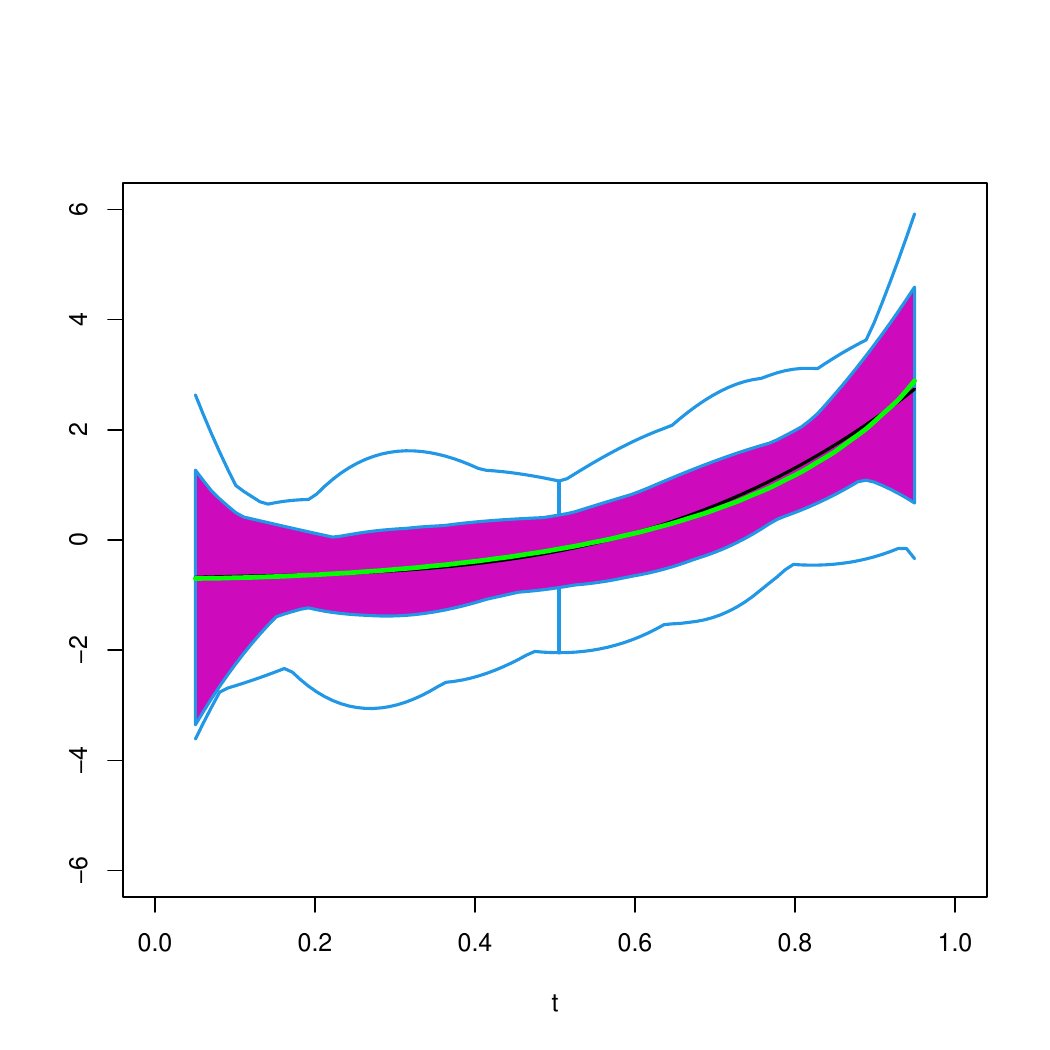}

\end{tabular}
\caption{\small \label{fig:wbeta-C510-poda5}  Functional boxplot of the estimators for $\beta_0$ under $C_{5,0.10}$  within the interval $[0.05,0.95]$. 
The true function is shown with a green dashed line, while the black solid one is the central 
curve of the $n_R = 1000$ estimates $\wbeta$.  }
\end{center} 
\end{figure}

%% file: FLogit_2023-08-15.bbl
\begin{thebibliography}{}

\bibitem[Aguilera et~al., 2008]{aguilera:etal:2008}
Aguilera, A., Escabias, M., and Valderrama, M. (2008).
\newblock Discussion of different logistic models with functional data:
  {A}pplication to systemic lupus erythematosus.
\newblock {\em Computational Statistics and Data Analysis}, 53:151--163.

\bibitem[Aguilera-Morillo et~al., 2013]{aguilera-morillo:etal:2013}
Aguilera-Morillo, M., Aguilera, A., Escabias, M., and Valderrama, M. (2013).
\newblock Penalized spline approaches for functional logit regression.
\newblock {\em Test}, 22:251--277.

\bibitem[Alin and Agostinelli, 2017]{alin:agostinelli:2017}
Alin, A. and Agostinelli, C. (2017).
\newblock Robust iteratively reweighted simpls.
\newblock {\em Journal of Chemometrics}, 31:e2881.

\bibitem[Aneiros-P{\'e}rez et~al., 2017]{aneiros:etal:2017}
Aneiros-P{\'e}rez, G., Bongiorno, E.~G., Cao, R., and Vieu, P. (2017).
\newblock {\em Functional {S}tatistics and {R}elated {F}ields}.
\newblock Springer.

\bibitem[Basu et~al., 1998]{Basu:etal:1998}
Basu, A., Harris, I.~R., Hjort, N.~L., and Jones, M.~C. (1998).
\newblock Robust and efficient estimation by minimizing a density power
  divergence.
\newblock {\em Biometrika}, 85:549--559.

\bibitem[Bianco et~al., 2022]{Bianco:Boente:Chebi:2022}
Bianco, A., Boente, G., and Chebi, G. (2022).
\newblock Penalized robust estimators in logistic regression with applications
  to sparse models.
\newblock {\em Test}, 31:563--594.

\bibitem[Bianco et~al., 2023]{Bianco:Boente:Chebi:2023}
Bianco, A., Boente, G., and Chebi, G. (2023).
\newblock Asymptotic behaviour of penalized robust estimators in logistic
  regression when dimension increases.
\newblock {\em In \textsl{Robust and Multivariate Statistical Methods:
  Festschrift in Honor of David E. Tyler, Eds: {Yi, Mengxi and Nordhausen,
  Klaus}}}, pages 323--348.
\newblock Springer International Publishing.

\bibitem[Bianco and Martinez, 2009]{Bianco:Martinez:2009}
Bianco, A. and Martinez, E. (2009).
\newblock Robust testing in the logistic regression model.
\newblock {\em Computational Statistics and Data Analysis}, 53:4095--4105.

\bibitem[Bianco and Yohai, 1996]{Bianco:yohai:1996}
Bianco, A. and Yohai, V. (1996).
\newblock Robust estimation in the logistic regression model.
\newblock {\em Lecture Notes in Statistics}, 109:17--34.

\bibitem[Boente and Martinez, 2023]{boente:martinez:2023}
Boente, G. and Martinez, A. (2023).
\newblock A robust spline approach in partially linear additive models.
\newblock {\em Computational Statistics and Data Analysis}, 178:107611.

\bibitem[Boente et~al., 2020]{boente:salibian:vena:2020}
Boente, G., Salibi{\'a}n-Barrera, M., and Vena, P. (2020).
\newblock Robust estimation for semi--functional linear regression models.
\newblock {\em Computational Statistics and Data Analysis}, 152:107041.

\bibitem[Bondell, 2005]{Bondell05}
Bondell, H.~D. (2005).
\newblock Minimum distance estimation for the logistic regression model.
\newblock {\em Biometrika}, 92:724--731.

\bibitem[Bondell, 2008]{Bondell08}
Bondell, H.~D. (2008).
\newblock A characteristic function approach to the biased sampling model, with
  application to robust logistic regression.
\newblock {\em Journal of Statistical Planning and Inference}, 138:742--755.

\bibitem[Cai and Hall, 2006]{cai:hall:2006}
Cai, T. and Hall, P. (2006).
\newblock Prediction in functional linear regression.
\newblock {\em Annals of Statistics}, 34:2159--2179.

\bibitem[Cantoni and Ronchetti, 2001]{Cantoni:ronchetti:2001}
Cantoni, E. and Ronchetti, E. (2001).
\newblock Robust inference for generalized linear models.
\newblock {\em Journal of the American Statistical Association}, 96:1022--1030.

\bibitem[Cardot et~al., 2003]{cardot:etal:2003}
Cardot, H., Ferraty, F., and Sarda, P. (2003).
\newblock Spline estimators for the functional linear model.
\newblock {\em Statistica Sinica}, 13:571--591.

\bibitem[Cardot and Sarda, 2005]{cardot:sarda:2005}
Cardot, H. and Sarda, P. (2005).
\newblock Estimation in generalized linear models for functional data via
  penalized likelihood.
\newblock {\em Journal of Multivariate Analysis}, 92:24--41.

\bibitem[Carroll and Pederson, 1993]{carroll:pederson:1993}
Carroll, R.~J. and Pederson, S. (1993).
\newblock On robust estimation in the logistic regression model.
\newblock {\em Journal of the Royal Statistical Society, Series B},
  55:693--706.

\bibitem[Croux and Haesbroeck, 2003]{Croux:H:2003}
Croux, C. and Haesbroeck, G. (2003).
\newblock Implementing the {B}ianco and {Y}ohai estimator for logistic
  regression.
\newblock {\em Computational Statistics and Data Analysis}, 44:273--295.

\bibitem[Cuevas, 2014]{cuevas:2014}
Cuevas, A. (2014).
\newblock A partial overview of the theory of statistics with functional data.
\newblock {\em Journal of Statistical Planning and Inference}, 147:1--23.

\bibitem[Denhere and Billor, 2016]{denhere:billor:2016}
Denhere, M. and Billor, N. (2016).
\newblock Robust principal component functional logistic regression.
\newblock {\em Communications in Statistics: Simulation and Computation},
  45:264--281.

\bibitem[De{V}ore and Lorentz, 1993]{devore:lorentz:1993}
De{V}ore, R. and Lorentz, G. (1993).
\newblock {\em Constructive {A}pproximation}.
\newblock Springer.

\bibitem[Escabias et~al., 2005]{escabias:etal:2005}
Escabias, M., Aguilera, A., and Valderrama, M. (2005).
\newblock Modeling environmental data by functional principal component
  logistic regression.
\newblock {\em Environmetrics}, 16:95--107.

\bibitem[Escabias et~al., 2004]{escabias:etal:2004}
Escabias, M., Aguilera, A.~M., and Valderrama, M. (2004).
\newblock Principal component estimation of functional logistic regression:
  {D}iscussion of two different approaches.
\newblock {\em Journal of Nonparametric Statistics}, 16:365--384.

\bibitem[Febrero-Bande et~al., 2017]{febrero:etal:2017}
Febrero-Bande, M., Galeano, P., and Gonz{\'a}lez-Manteiga, W. (2017).
\newblock Functional principal component regression and functional partial
  least--squares regression: {A}n overview and a comparative study.
\newblock {\em International Statistical Review}, 85:61--83.

\bibitem[Ferraty and Romain, 2010]{ferraty:romain:2010}
Ferraty, F. and Romain, Y. (2010).
\newblock {\em The {O}xford {H}andbook of {F}unctional {D}ata {A}nalysis}.
\newblock Oxford University Press.

\bibitem[Ferraty and Vieu, 2006]{ferraty:vieu:2006}
Ferraty, F. and Vieu, P. (2006).
\newblock {\em Nonparametric {F}unctional {D}ata {A}nalysis: Theory and
  Practice}.
\newblock Springer.

\bibitem[Garc{\'\i}a~Ben and Yohai, 2004]{garciaben:yohai:2004}
Garc{\'\i}a~Ben, M. and Yohai, V.~J. (2004).
\newblock Quantile--quantile plot for deviance residuals in the generalized
  linear model.
\newblock {\em Journal of Computational and Graphical Statistics}, 13:36--47.

\bibitem[Goia and Vieu, 2016]{goia:vieu:2016}
Goia, A. and Vieu, P. (2016).
\newblock An introduction to recent advances in high/infinite dimensional
  statistics.
\newblock {\em Journal of Multivariate Analysis}, 146:1--6.

\bibitem[Hall and Horowitz, 2007]{hall:horowitz:2007}
Hall, P. and Horowitz, J.~L. (2007).
\newblock Methodology and convergence rates for functional linear regression.
\newblock {\em Annals of Statistics}, 35:70--91.

\bibitem[He and Shi, 1996]{he:shi:1996}
He, X. and Shi, P. (1996).
\newblock Bivariate tensor-product {B-}spline in a partly linear model.
\newblock {\em Journal of Multivariate Analysis}, 58:162--181.

\bibitem[He and Shi, 1998]{he:shi:1998}
He, X. and Shi, P. (1998).
\newblock Monotone {B-}spline smoothing.
\newblock {\em Journal of the American statistical Association}, 93:643--650.

\bibitem[He et~al., 2002]{he:zhu:fung:2002}
He, X., Zhu, Z., and Fung, W. (2002).
\newblock Estimation in a semiparametric model for longitudinal data with
  unspecified dependence structure.
\newblock {\em Biometrika}, 89:579--590.

\bibitem[Hobza et~al., 2008]{Hobza:etal:2008}
Hobza, T., Pardo, L., and Vajda, I. (2008).
\newblock Robust median estimator in logistic regression.
\newblock {\em Journal of Statistical Planning and Inference}, 138:3822--3840.

\bibitem[Horv{\'a}th and Kokoszka, 2012]{horvath:kokoska:2012}
Horv{\'a}th, L. and Kokoszka, P. (2012).
\newblock {\em Inference for {F}unctional {D}ata with {A}pplications}.
\newblock Springer.

\bibitem[Hsing and Eubank, 2015]{hsing:eubank:2015}
Hsing, T. and Eubank, R. (2015).
\newblock {\em Theoretical foundations of {F}unctional {D}ata {A}nalysis with
  an introduction to {L}inear {O}perators}, volume 997.
\newblock John Wiley and Sons.

\bibitem[Hubert et~al., 2005]{Hubert:etal:2005}
Hubert, M., Rousseeuw, P.~J., and Vanden~Branden, K. (2005).
\newblock {ROBPCA}: {A} new approach to robust principal component analysis.
\newblock {\em Technometrics}, 47:64--79.

\bibitem[James, 2002]{james:2002}
James, G.~M. (2002).
\newblock Generalized linear models with functional predictors.
\newblock {\em Journal of the Royal Statistical Society: Series B (Statistical
  Methodology)}, 64:411--432.

\bibitem[Kalogridis, 2023]{kalogridis:2023}
Kalogridis, I. (2023).
\newblock Robust and adaptive functional logistic regression.
\newblock Available at \url{https://arxiv.org/abs/2305.01350}.

\bibitem[Kalogridis and Van~Aelst, 2019]{kalogridis:vanaelst:2019}
Kalogridis, I. and Van~Aelst, S. (2019).
\newblock Robust functional regression based on principal components.
\newblock {\em Journal of Multivariate Analysis}, 173:393--415.

\bibitem[Kalogridis and Van~Aelst, 2023]{kalogridis:vanaelst:2023}
Kalogridis, I. and Van~Aelst, S. (2023).
\newblock Robust penalized estimators for functional linear regression.
\newblock {\em Journal of Multivariate Analysis}, 194:105104.

\bibitem[Liebl, 2013]{liebl:2013}
Liebl, D. (2013).
\newblock Modelling and forecasting electricity spot prices: a functional data
  perspective.
\newblock {\em The Annals of Applied Statistics}, 7:1562--1592.

\bibitem[L\'opez-Pintado and Romo, 2009]{lopez-pintado:romo:2009}
L\'opez-Pintado, S. and Romo, J. (2009).
\newblock On the concept of depth for functional data.
\newblock {\em Journal of the American Statistical Association}, 104:718--734.

\bibitem[Lu, 2015]{lu:2015}
Lu, M. (2015).
\newblock Spline estimation of generalised monotonic regression.
\newblock {\em Journal of Nonparametric Statistics}, 27:19--39.

\bibitem[Mallat, 2009]{mallat:2009}
Mallat, S. (2009).
\newblock {\em {A Wavelet Tour of Signal Processing: the Sparse Way}}.
\newblock Academic Press.

\bibitem[Maronna et~al., 2019]{maronna:etal:2019libro}
Maronna, R., Martin, D., Yohai, V., and Salibi{\'a}n-Barrera, M. (2019).
\newblock {\em Robust {S}tatistics: {T}heory and {M}ethods (with \texttt{R})}.
\newblock John Wiley and Sons.

\bibitem[Maronna and Yohai, 2013]{maronna:yohai:2013}
Maronna, R. and Yohai, V. (2013).
\newblock Robust functional linear regression based on splines.
\newblock {\em Computational Statistics and Data Analysis}, 65:46--55.

\bibitem[Marx and Eilers, 1999]{marx:eilers:1999}
Marx, B.~D. and Eilers, P. H.~C. (1999).
\newblock Generalized linear regression on sampled signals and curves: {A}
  {P-}spline approach.
\newblock {\em Technometrics}, 4:1--13.

\bibitem[Mousavi and S{\o}rensen, 2018]{Mousavi:2018}
Mousavi, S.~N. and S{\o}rensen, H. (2018).
\newblock Functional logistic regression: {A} comparison of three methods.
\newblock {\em Journal of Statistical Computation and Simulation}, 88:250--268.

\bibitem[M{\"u}ller, 2005]{muller:2005}
M{\"u}ller, H.~G. (2005).
\newblock Functional modelling and classification of longitudinal data.
\newblock {\em Scandinavian Journal of Statistics}, 32:223--240.

\bibitem[M{\"u}ller and Stadtm{\"u}ller, 2005]{muller:stadtmuller:2005}
M{\"u}ller, H.~G. and Stadtm{\"u}ller, U. (2005).
\newblock Generalized functional linear models.
\newblock {\em Annals of Statistics}, 33:774--805.

\bibitem[Mutis et~al., 2022]{mutis:etal:2022}
Mutis, M., Beyaztas, U., Simsek, G., and Shang, H. (2022).
\newblock A robust scalar--on--function logistic regression for classification.
\newblock {\em Communications in Statistics - Theory and Methods}, pages 1--17.

\bibitem[Pollard, 1984]{Pollard:1984}
Pollard, D. (1984).
\newblock {\em {Convergence of Stochastic Processes}}.
\newblock Springer Series in Statistics.

\bibitem[Powell, 1981]{Powell:1981}
Powell, M. (1981).
\newblock {\em Approximation {T}heory and {M}ethods}.
\newblock Cambridge: Cambridge University Press.

\bibitem[Qingguo, 2015]{qingguo:2015}
Qingguo, T. (2015).
\newblock Estimation for semi-functional linear regression.
\newblock {\em Statistics}, 49:1262--1278.

\bibitem[Ramsay and Silverman, 2002]{ramsay2002applied}
Ramsay, J. and Silverman, B. (2002).
\newblock {\em {Applied {F}unctional {D}ata {A}nalysis. {M}ethods and {C}ase
  {S}tudies}}.
\newblock Springer.

\bibitem[Ramsay and Silverman, 2005]{ramsay2005functional}
Ramsay, J. and Silverman, B. (2005).
\newblock {\em {Functional {D}ata {A}nalysis, 2nd edition}}.
\newblock Springer.

\bibitem[Ramsay, 2004]{ramsay2004functional}
Ramsay, J.~O. (2004).
\newblock Functional data analysis.
\newblock {\em {Encyclopedia of Statistical Sciences}}.
%\newblock Available at
%  \url{https://onlinelibrary.wiley.com/doi/abs/10.1002/0471667196.ess0646}.

\bibitem[Ratcliffe et~al., 2002]{ratcliffe:etal:2002}
Ratcliffe, S.~J., Heller, G.~Z., and Leader, L.~R. (2002).
\newblock Functional data analysis with application to periodically stimulated
  foetal heart rate data.{ II}: {F}unctional logistic regression.
\newblock {\em Statistics in Medicine}, 21:1115--1127.

\bibitem[Reiss et~al., 2017]{reiss:etal:2017}
Reiss, P.~T., Goldsmith, J., Shang, H.~L., and Ogden, R.~T. (2017).
\newblock Methods for scalar--on--function regression.
\newblock {\em International Statistical Review}, 85:228--249.

\bibitem[Reiss et~al., 2005]{reiss:etal:2005}
Reiss, P.~T., Ogden, R.~T., Mann, J., and Parsey, R.~V. (2005).
\newblock Functional logistic regression with pet imaging data: {A} voxel-level
  clinical diagnostic tool.
\newblock {\em Journal of Cerebral Blood Flow and Metabolism}, 25:S635--S635.

\bibitem[Ronchetti, 1985]{ronchetti:1985}
Ronchetti, E. (1985).
\newblock Robust model selection in regression.
\newblock {\em Statistics and Probability Letters}, 3:21--23.

\bibitem[Schumaker, 1981]{schumaker1981spline}
Schumaker, L. (1981).
\newblock {\em {Spline {F}unctions: {B}asic {T}heory}}.
\newblock Wiley.

\bibitem[Schwarz, 1978]{schwartz:1978}
Schwarz, G. (1978).
\newblock Estimating the dimension of a model.
\newblock {\em Annals of Statistics}, 6:461--464.

\bibitem[S{\o}rensen et~al., 2013]{sorensen:etal:2013}
S{\o}rensen, H., Goldsmith, J., and Sangalli, L.~M. (2013).
\newblock An introduction with medical applications to functional data
  analysis.
\newblock {\em Statistics in Medicine}, 32:5222--5240.

\bibitem[Stone, 1982]{stone:1982}
Stone, C. (1982).
\newblock Optimal global rates of convergence for nonparametric regression.
\newblock {\em Annals of Statistics}, 10:1040--1053.

\bibitem[Stone, 1985]{stone:1985}
Stone, C. (1985).
\newblock Additive regression and other nonparametric models.
\newblock {\em Annals of Statistics}, 13:689--705.

\bibitem[Sun and Genton, 2011]{sun:genton:2011}
Sun, Y. and Genton, M.~G. (2011).
\newblock Functional boxplots.
\newblock {\em Journal of Computational and Graphical Statistics}, 20:316--334.

\bibitem[Tharmaratnam and Claeskens, 2013]{thar:claes:2013}
Tharmaratnam, K. and Claeskens, G. (2013).
\newblock A comparison of robust versions of the {AIC} based on {M-}, {S-} and
  {MM-}estimators.
\newblock {\em Statistics}, 47:216--235.

\bibitem[van~de Geer, 1988]{sara1988}
van~de Geer, S. (1988).
\newblock Regression {A}nalysis and {E}mpirical {P}rocesses.
\newblock {\em CWI Tract 45}.
\newblock Center for {M}athematics and {C}omputer {S}cience, {A}msterdam.
  Available at \url{https://ir.cwi.nl/pub/13169}.

\bibitem[van~de Geer, 2000]{sara2000}
van~de Geer, S. (2000).
\newblock {\em Empirical {P}rocesses in $M-${E}stimation}.
\newblock Cambridge Series in Statistical and Probabilistic Mathematics.

\bibitem[van~der Vaart and Wellner, 1996]{vanderVaart:wellner:1996}
van~der Vaart, A.~W. and Wellner, J.~A. (1996).
\newblock {\em Weak {C}onvergence and {E}mpirical {P}rocesses}.
\newblock Springer, New York.

\bibitem[Wang and Xiang, 2012]{Wang:Xiang:2012}
Wang, H. and Xiang, S. (2012).
\newblock On the convergence rates of {L}egendre approximation.
\newblock {\em Mathematics of Computation}, 81:861--877.

\bibitem[Wang et~al., 2016]{wang:etal:2016}
Wang, J.~L., Chiou, J., and M{\"u}ller, H. (2016).
\newblock Functional data analysis.
\newblock {\em Annual Review of Statistics and Its Application}, 3:257--295.

\bibitem[Wang et~al., 2017]{wang:etal:2017}
Wang, X., Nan, B., Zhu, J., Koeppe, R., and Frey, K. (2017).
\newblock Classification of {ADNI PET} images via regularized {3D} functional
  data analysis.
\newblock {\em Biostatistics and Epidemiology}, 1:3--19.

\bibitem[Zhao et~al., 2012]{zhao:etal:2012}
Zhao, Y., Ogden, T., and Reiss, P. (2012).
\newblock Wavelet-based {LASSO} in functional linear regression.
\newblock {\em Journal of Computational and Graphical Statistics}, 21:600--617.

\end{thebibliography}
